\documentclass[a4paper,twoside,BCOR=2cm,DIV=13,fontsize=11pt, toc=bibliography, chapterprefix=true, captions=tableheading, numbers=noenddot]{scrbook}
\usepackage[whole,autotilde]{bxcjkjatype}
\usepackage{hyperref}

\usepackage{color}

\usepackage[headsepline,automark]{scrlayer-scrpage}

\usepackage{amsmath, amssymb, amsthm}
\usepackage{mathtools}

\usepackage[looser,theoremfont,largesc]{newpxtext}
\usepackage[varg,vvarbb,bigdelims]{newpxmath}

\theoremstyle{definition}
\newtheorem{theorem}{Theorem}[section]
\newtheorem{definition}[theorem]{Definition}
\newtheorem{lemma}[theorem]{Lemma}
\newtheorem{proposition}[theorem]{Proposition}
\newtheorem{corollary}[theorem]{Corollary}
\newtheorem{remark}[theorem]{Remark}
\newtheorem{implication}[theorem]{Implication}

\newtheorem{example}[theorem]{Example}

\DeclareMathOperator{\rank}{rank}
\DeclareMathOperator{\tr}{Tr}

\DeclareMathOperator{\spn}{span}
\DeclareMathOperator{\supp}{supp}
\DeclareMathOperator{\id}{id}
\DeclareMathOperator*{\argmax}{argmax}
\allowdisplaybreaks[4]

\usepackage{braket}
\usepackage{graphicx}
\usepackage{enumerate}
\usepackage{booktabs}

\usepackage[sort&compress,numbers]{natbib}
\begin{document}

\begin{titlepage}
  \centering

  \subject{Doctorate Dissertation\\
    \vspace{0.5cm}
  博士論文}
  \title{Entanglement theory in distributed quantum information processing}
  \subtitle{（分散型量子情報処理のエンタングルメント理論）}

  \date{A Dissertation Submitted for\\ Degree of Doctor of Science\\
    \vspace{0.5cm}
    December 2018\\
    \vspace{0.5cm}
    平成30年12月博士（理学）申請\\
    \vspace{0.5cm}
  }

  \publishers{Department of Physics, Graduate School of Science,\\
    \vspace{0.5cm}
    The University of Tokyo\\
    \vspace{0.5cm}
    東京大学大学院理学系研究科物理学専攻\\
    \vspace{1.5cm}
    Hayata Yamasaki\\
    \vspace{0.5cm}
    山崎隼汰
  }

  \maketitle
\end{titlepage}

\frontmatter

\chapter{Abstract}

Distributed quantum information processing is a promising platform for scaling up quantum information processing, where small- and intermediate-scale quantum devices are connected by a network of quantum channels for communicating quantum information, so as to cooperate in achieving larger-scale information processing.
In such distributed settings, entangled states shared among the multiple devices serve as a resource for achieving nonlocal information processing tasks by local operations and classical communication (LOCC), where transformations of multipartite entangled states play central roles.
This thesis analyzes properties of quantum entanglement in these small- and intermediate-scale settings and multipartite settings.

The first part of this thesis investigates a communication task, \textit{quantum state merging}, on the small and intermediate scales.
Aiming at transferring quantum information from a sender $A$ to a receiver $B$ on these scales, this thesis analyzes entanglement cost required for one-shot quantum state merging.
Achievability bounds of entanglement cost and protocols are presented, so as to achieve one-shot state merging on the small and intermediate scales.
Improved converse bounds of the entanglement cost are also derived.
Moreover, it is proven that there is a case where $B$'s preprocessing and backward classical communication from $B$ to $A$ can be indispensable for minimizing entanglement cost in one-shot state merging from $A$ to $B$.

The second part of this thesis analyzes multipartite entanglement in distributed quantum information processing.
To quantitatively characterize nonlocal properties of multipartite state transformations for encoding and decoding quantum information in a multipartite quantum system, entanglement costs in such encoding and decoding are analyzed, where the multipartite system is distributed among spatially separated parties connected by a network.
In addition, advantage of using multipartite entanglement over bipartite entanglement is investigated, and
it is shown that when there exists a limitation on the local system size for each party, multipartite entanglement is an indispensable resource without which certain processes cannot be accomplished.

These analyses clarify fundamental limitations and potential applications of distributed quantum information processing to characterize properties of quantum entanglement in the small- and intermediate-scale settings and multipartite settings, providing a paradigm for investigating multipartite entanglement in distributed quantum information processing over networks beyond the state convertibility under LOCC\@.

\chapter{Acknowledgments}

I express my sincere thanks to my supervisor Mio Murao for all constructive suggestions and considerable supports. I thank Akihito Soeda for extensive discussions and enormous helps. I am grateful to Barbara Kraus for accepting my visit and broadening my view on multipartite quantum entanglement, and to Alexander Pirker and Wolfgang D\"{u}r for discussions and the collaboration. I thank members in the group of Mio Murao, Shojun Nakayama, Eyuri Wakakuwa, Seiseki Akibue, Jisho Miyazaki, Kohtaro Kato, Atsushi Shimbo, Yuki Mori, Ryosuke Sakai, Qingxiuxiong Dong, Marco T\'{u}lio Quintino, Paula Belzig, and Wataru Yokojima, and members in the group of Barbara Kraus, Katharina Schwaiger, David Sauerwein, Martin Hebenstreit, Yaiza Aragones Soria, Raphael Brieger, Czarnetzki Leonhard, Farid Shahandeh, and Matthias Englbrecht, for insightful discussions.

\tableofcontents
\listoffigures
\listoftables

\chapter*{List of Publications}

This thesis is based on the following papers.

\begin{description}
  \item[\cite{Y14}] H.\ Yamasaki, A.\ Pirker, M.\ Murao, W.\ D\"{u}r, and B.\ Kraus, Phys.\ Rev.\ A \textbf{98}, 052313 (2018).
  \item[\cite{Y12}] H.\ Yamasaki and M.\ Murao, \textit{Quantum state merging for arbitrarily small-dimensional systems}, (2018), arXiv:1806.07875.
  \item[\cite{Y13}] H.\ Yamasaki and M.\ Murao, \textit{Distributed Encoding and Decoding of Quantum Information over Networks}, (2018), arXiv:1807.11483.
  \item[\cite{Y19}] H.\ Yamasaki and M.\ Murao, \textit{Quantum-side-information preprocessing and backward classical communication in one-shot quantum state merging}, Unpublished.
\end{description}

This thesis also uses the results in the following previous works of mine.

\begin{description}
  \item[\cite{Y6}] H.\ Yamasaki, A.\ Soeda, and M.\ Murao, Phys.\ Rev.\ A \textbf{96}, 032330 (2017).
  \item[\cite{Y18}] H.\ Yamasaki, \textit{Distributed Construction of Multipartite Entangled States over Quantum Networks}, Master's thesis, The University of Tokyo (2016).
\end{description}

\mainmatter%

\part{Introduction and preliminaries}

\chapter{Introduction}

This chapter provides introduction and organization of the whole of this thesis.

\section{Introduction to entanglement theory in distributed quantum information processing}

Physics provides a way of understanding the world based on fundamental laws, and phenomena possibly exhibited in the world ensure consistency of such fundamental laws of theories in physics.
Such theories in physics model complex phenomena in the world as consequences of the simplified laws.
Successive efforts have led to several theories in physics, such as classical mechanics on macroscopic scales, theory of relativity on high energy scales, and quantum mechanics on microscopic scales, which have been verified on respective scales.
While identification of fundamental laws is a starting point of this type of theories, it is also fundamental to ask the following question as a next step: \textit{What kind of phenomena can be exhibited within the laws of such a theory?}
Studies of this type of question establish the base of the consistency of the laws, facilitating better understanding of the world and prediction of novel phenomena in the world described by the theory.

Quantum information theory~\cite{N1,W5,W11} studies consequences of the laws of quantum mechanics from an operational approach, answering what kind of information processing is possible and what is impossible using operations allowed in quantum mechanics.
Traditionally, such an operational approach to physics is also taken in thermodynamics, a phenomenological theory corresponding to classical mechanics. (See Reference~\cite{L} and the references therein.)
Thermodynamics answers what kind of physical processes are possible and what are impossible, when we perform operations for \textit{actively} processing states of physical systems described by classical mechanics, such as steam engines, rather than passively observing the systems.
This approach to physics in thermodynamics is operational in the sense that it abstracts Hamiltonian dynamics of the systems and introduces a class of idealized operations on the systems, such as adiabatic and isothermal processes.
Abstract properties of the physical systems, such as energy and entropy, are characterized by analyzing the fundamental limitations on appropriate tasks, such as work extraction from the Carnot cycle, performed by these idealized operations in thermodynamics.

While thermodynamics assumes that physical systems and operations are on \textit{macroscopic} scales,
suppose that operations on \textit{microscopic} scales are at hand.
On the microscopic scales, operations on the physical systems may be described by not classical but \textit{quantum} mechanics.
Making the most of such microscopic operations on quantum systems is known to have potential applications to information processing,
and this way of information processing exploiting advantage of quantum mechanics is called \textit{quantum information processing}.
In quantum information processing, instead of a classical binary bit taking $0$ or $1$ as a state, a two-level quantum system, or a \textit{qubit}, can be used as a basic unit for carrying quantum information.
Quantum information processing is performed by transforming quantum information represented by a quantum state of qubits, where in contrast with classical bits, arbitrary \textit{superposition} of two distinguishable quantum states $\Ket{0}$ and $\Ket{1}$ may be taken as a quantum state of each qubit.
In this sense, while states $0$ and $1$ of a classical bit represent classical information, a quantum state $\alpha_0\Ket{0}+\alpha_1\Ket{1}$ of a qubit can represent quantum information in addition to $\Ket{0}$ and $\Ket{1}$ of the qubit representing classical information.
Classical information can be obtained from a quantum state by a probabilistic process, a quantum measurement.
More generally than qubits, a $D$-dimensional quantum system is called a \textit{qudit} and is used for quantum information processing.
Similarly to physical processes in thermodynamics achieved by the idealized operations,
quantum information processing can also be regarded as a physical process of transforming quantum states by operations with the laws of quantum mechanics, in the sense that an initially given input of quantum or classical information to this process is represented as a quantum state and is transformed into the final output of quantum or classical information.

Quantum information theory analyzes what kind of information processing can be achieved and what cannot within the laws of quantum mechanics, providing a quantitative understanding and an operational meaning of abstract properties even characteristic of quantum mechanics.
Exploiting the quantum mechanical property of \textit{interference}, quantum information processing is potentially faster than classical one for performing some classical computational tasks, such as solving an arithmetic problem of prime factoring~\cite{P6}, simulating quantum systems~\cite{F2,L8}, and sampling from the solution of a linear system~\cite{H14}.
In other words, large-scale quantum information processing potentially provides excessive computational power that is not simulatable by any classical algorithms in an efficient way.
Moreover, spatially separated quantum systems in a superposition state may exhibit \textit{quantum entanglement}, a type of correlation characteristic of quantum mechanics, which is incompatible with any hidden-variable theory based on the paradigm of classical mechanics~\cite{B29,B27,B28}.
If a quantum state has entanglement, the state is called an \textit{entangled state}, and otherwise called a \textit{separable state}.
Entanglement appearing in a bipartite system, \textit{i.e.}, that consisting of two subsystems, is called \textit{bipartite entanglement}, and entanglement appearing in a multipartite system, \textit{i.e.}, that consisting of more than two, is called \textit{multipartite entanglement}.
In quantum information theory, information processing tasks exploiting such quantum mechanical properties are analyzed, so as to quantitatively characterize consequences of these quantum mechanical properties beyond classical mechanics.

\begin{figure}[t!]
  \centering
  \includegraphics[width=4in]{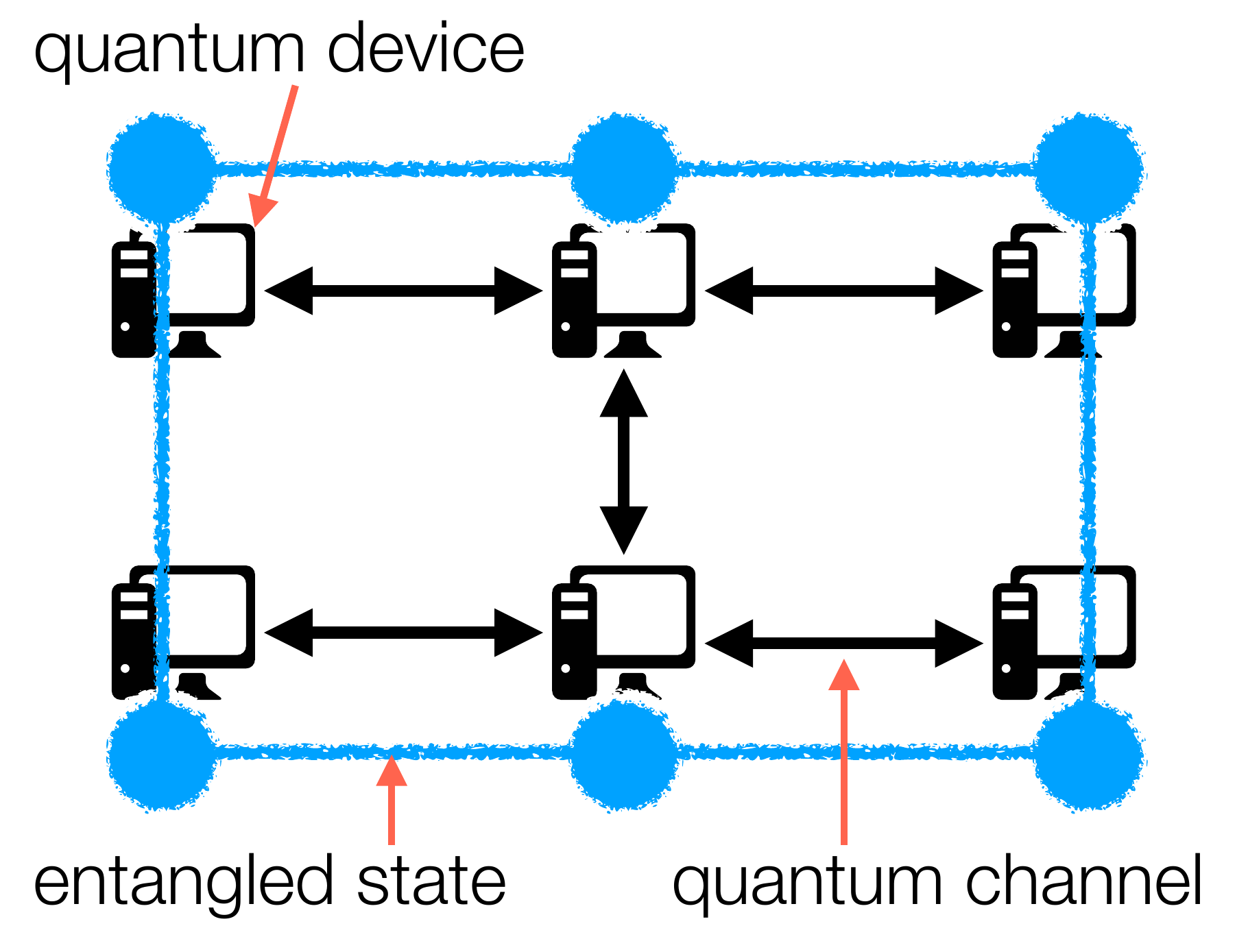}
  \caption[Distributed quantum information processing using multiple quantum devices connected by a network.]{\label{fig:distributed}Distributed quantum information processing using multiple quantum devices illustrated as each computer. For communicating quantum information between the quantum devices, these quantum devices are connected by a network of quantum channels represented by arrows. A quantum state shared among multiple quantum devices may be entangled, denoted by the blue circles connected by blue lines.}
\end{figure}

Recent advances in quantum technology facilitate quantum information processing using quantum devices capable of coherently keeping quantum states of a quantum system inside and of performing low-noise operations for transforming these quantum states.
There exists, however, technical difficulty in increasing the number of low-noise qubits built in one quantum device~\cite{P4}, and hence, the quantum system size of such a quantum device may be limited on small and intermediate scales of up to several dozens of qubits at least in the near future.
To achieve large-scale quantum information processing, larger quantum system sizes than those in such small- and intermediate-scale quantum devices are required.
For scaling up quantum information processing, \textit{distributed quantum information processing}~\cite{V3,C9,W2} is considered to be a promising platform, where larger-scale information processing is achieved using multiple quantum devices connected by a network of quantum channels for communicating quantum information, that is, quantum communication, as illustrated in Figure~\ref{fig:distributed}.
In contrast to quantum information processing performed by arbitrarily transforming quantum states of a quantum system,
the quantum devices in distributed quantum information processing share a composite quantum system whose subsystems are located in each device, and each device is allowed to perform state transformations only on the subsystem in the device.
Nonlocal state transformations over different quantum devices are performed by combining these local state transformations in each device with quantum communication.

Given multiple quantum devices in such a distributed setting,
quantum entanglement shared among the devices is considered as a correlation which cannot be generated using a class of operations consisting of arbitrary local state transformations inside each quantum device within quantum mechanics as well as arbitrary inter-device communication of classical information represented by bits.
This class of operations is called local operations and classical communication (LOCC)~\cite{D11,C19,C7}.
To generate arbitrary entangled states shared among multiple devices from a separable state,
quantum communication for transferring quantum states between the devices is necessary in addition to LOCC, and sufficiently much use of quantum communication allows the quantum devices to perform arbitrary nonlocal transformations of quantum states of the shared composite system.
Conversely, if two devices initially share a particular type of bipartite entangled state, quantum communication between these devices can be simulated by a protocol, \textit{quantum teleportation}~\cite{B5}, using LOCC assisted by the shared bipartite entanglement.
Hence, nonlocal state transformations over different quantum devices can be achieved by combining LOCC with shared entanglement.

Distributed quantum information processing can be used as a framework for operational analysis of the quantum entanglement shared among the quantum devices.
Entanglement serves as a \textit{resource} assisting LOCC in distributed quantum information processing, for achieving nonlocal state transformations over spatially separated quantum systems.
While classical communication can be reliably performed using current technology, quantum communication for sharing entanglement is more challenging and costly.
In this regard, it is natural to investigate efficient use of entanglement when cost of LOCC is negligible.
This approach of regarding entanglement as resources is a fundamental starting point of \textit{entanglement theory}~\cite{H2,P1,E5}, which has been successful in establishing operational understanding of bipartite entanglement.
Among bipartite entangled states, convertibility between these entangled states under LOCC establishes partial order of the states in terms of their usability as a resource.
This partial order yields quantifications of entanglement in terms of its value as a resource, where it is required that these quantifications are monotonically nonincreasing under LOCC\@.
Such a quantification of entanglement is called an entanglement measure.
This way of characterizing entanglement may also generalize to a more general and abstract formulation called \textit{quantum resource theory}, so that the resource-theoretic approach is applicable to investigating properties characteristic of quantum mechanics other than entanglement, such as coherence and purity~\cite{C13}.

Multipartite entanglement also serves as a resource for multiparty tasks relevant to distributed quantum information processing, such as measurement-based quantum computation~\cite{R5,R6,R7}, distributed sensing~\cite{K1,E4}, and quantum networking~\cite{P2}.
Multipartite entanglement ubiquitously appears in many-body quantum systems in condensed matter physics~\cite{A12} and quantum gravity~\cite{R4}.
However, straightforward applications of the bipartite resource-theoretic analysis are not sufficient to characterize properties of multipartite entanglement on more than two systems, because mathematical structures of multipartite entangled states are not as simple as bipartite entanglement~\cite{E2,W3,B26}.
Especially, in case of multi-qudit systems whose subsystems are of equal dimension, almost no LOCC transformation among quantum states of the system is possible~\cite{G1,S3},
and hence, the paradigm based on the partial order of bipartite entanglement under LOCC does not generalize to multipartite entanglement.

This thesis aims to characterize properties of multipartite entanglement through an operational approach, not only using the framework of LOCC, but using quantum communication networks in addition to LOCC, motivated by the settings for distributed quantum information processing.
In distributed quantum information processing, a nonlocal transformation of a quantum state shared between two quantum devices can be performed by first transferring one device's part of the state to the other device, and then performing the transformation locally on the latter device, followed by transferring the state back.
While this strategy for performing  a nonlocal state transformation is not always the most efficient in terms of a communication cost, this strategy exactly and deterministically achieves the transformation.
Given two quantum devices sharing a quantum state, the communication task of transferring one device's part of this shared state to the other is called \textit{quantum state merging}~\cite{H3,H4}.
Part~\ref{part:1} of this thesis aims to reduce the cost of achieving quantum state merging performed in this two-party LOCC setting of distributed quantum information processing.
While the original formulation of quantum state merging in References~\cite{H3,H4} and their successive works are mainly targeted at quantum communication on large scales, protocols aimed at efficient distributed quantum information processing over a network should be designed to be suitable for arbitrarily small-dimensional quantum systems, especially, on the small and intermediate scales relevant to distributed quantum information processing.
Part~\ref{part:1} of this thesis analyzes quantum state merging on the small and intermediate scales, different from the existing studies on the large scales.

Moreover, distributed quantum information processing may involve more than two quantum devices, where transformations of multipartite entangled states play central roles.
Part~\ref{part:2} of this thesis quantitatively analyzes requirements of quantum communication and quantum system sizes required for transforming multipartite entanglement in distributed quantum information processing.
The results established in Part~\ref{part:1} on quantum state merging are used for evaluating the requirements of quantum communication.
Part~\ref{part:2} also introduces and analyzes tasks of multipartite entanglement transformations in a setting where local quantum system sizes in LOCC are limited, motivated by distributed quantum information processing on the small and intermediate scales.

These analyses clarify fundamental limitations and potential applications of distributed quantum information processing to characterize properties of quantum entanglement in the small- and intermediate-scale settings and multipartite settings relevant to distributed quantum information processing, providing a paradigm for investigating multipartite entanglement in distributed quantum information processing over networks beyond the state convertibility introducing the partial order of entanglement under LOCC\@.
More detailed backgrounds and settings are given after the preliminaries in Chapter~\ref{sec:preliminaries_all}, at the beginning of Parts~\ref{part:1} and~\ref{part:2}.

\section{Technologies for distributed quantum information processing}

This section summarizes experimental technologies relevant to distributed quantum information processing, to which theoretical results in this thesis are potentially applicable.
Quantum technologies cover wide applications such as quantum computation, quantum communication, quantum simulation, and quantum sensing~\cite{A18}.
Distributed quantum information processing can be realized by combining technologies for quantum computation and quantum communication.

Ongoing experimental approaches for realizing quantum computation include superconducting circuits~\cite{W13}, ion traps~\cite{H17,C22}, photonic systems~\cite{K9}, and nitrogen-vacancy (NV) centers~\cite{D13}.
Superconducting circuits achieves control of $9$ qubits in 2015~\cite{K10}, and ion traps $5$ qubits in 2016~\cite{D14}.
A major challenge in realizing quantum computation stems from noise,
and one way to reducing effects of noise is quantum error correction~\cite{G,D,T10,B},
where quantum information is represented as a superposition of predetermined multipartite entangled states of a quantum error correcting code, so that local noise can be detected and corrected.
If noise of each quantum operation on qubits is below a given threshold, errors during quantum computation can be arbitrarily suppressed by quantum error correction.
However, it is not straightforward to increase the number of controllable qubits required for quantum error correction while keeping low noise; that is, there may exists a trade-off relation between quantity and quality of qubits.

Distributed quantum information processing is considered to be a candidate for scaling up quantum computation if a limited number of low-noise qubits are available.
This situation contrasts with that considered in theoretical research on noisy intermediate-scale quantum (NISQ) technology~\cite{P4}, which aims to find advantages and applications of intermediate-scale quantum devices from several dozens to a few hundreds of qubits that compromise on reducing noise.
In distributed quantum information processing, quantum information may be represented using a quantum error correcting code for fault tolerance, and hence, analysis of communication tasks for a state in a superposition of fixed entangled states plays essential roles.

Using such low-noise local operations, noisy entanglement at a distance generated by lossy quantum communication may be purified by means of entanglement distillation~\cite{B3}.
As for quantum communication, a quantum cryptographic task, quantum key distribution, is demonstrated using photonic systems and optical fibers over $307$ km in 2015~\cite{K11}.
However, distribution of quantum entanglement is currently more difficult due to lack of low-noise local quantum systems, as well as inefficiency in conversion between matter-based qubits and photons.
Entanglement at a distance is detected between ion traps in 2007~\cite{M7}, between NV centers in 2013~\cite{B30}, and between electron spins separated by 1.3 km in 2015~\cite{H18}.

While fault-tolerant networks required for distributed quantum information processing pose technological challenges~\cite{W2}, theoretical analysis of minimal quantum communication for achieving distributed quantum information processing is beneficial to clarifying a technological target in future experiments.

\section{Organization of this thesis}

The rest of this thesis is organized as follows.
After providing preliminaries to the rest of this thesis in Chapter~\ref{sec:preliminaries_all}, Part~\ref{part:1} analyzes a communication task, quantum state merging, between two spatially separated quantum parties having arbitrarily small-dimensional systems, so that the results are applicable to any two small- and intermediate-scale quantum devices on a network used in distributed quantum information processing.
Using the results established in Part~\ref{part:1}, Part~\ref{part:2} analyzes properties of multipartite entanglement using the framework of distributed quantum information processing over networks.
The results in Part~\ref{part:1} and Part~\ref{part:2} are summarized as follows, and the conclusion of these results is given in Part~\ref{part:conclusion}.
The structure of chapters in each part is illustrated in Figure~\ref{fig:organization}.

\begin{itemize}
  \item Part~\ref{part:1} analyzes a communication task of quantum state merging~\cite{H3,H4} on small and intermediate scales. In distributed quantum information processing, two quantum devices, namely, $A$ and $B$, may share a correlated state, and state merging is a fundamental communication task aiming at transferring $A$'s part of this shared state to $B$, where $B$'s part is called quantum side information and may be used for reducing required communication costs in state merging. Aiming at transferring quantum information on small and intermediate scales relevant to distributed quantum information processing, Part~\ref{part:1} considers a \textit{one-shot} scenario of state merging, where only a single copy of the shared state is given. While existing protocols achieving one-shot quantum state merging are costly on the small and intermediate scales, Chapter~\ref{sec:merge} in Part~\ref{part:1} establishes a protocol applicable even on the small and intermediate scales, as well as analyzing lower bounds for minimal costs in the one-shot scenario of state merging. Also, aiming at making the most of quantum side information in a one-shot scenario, Chapter~\ref{sec:two_way} in Part~\ref{part:1} proves that $B$'s preprocessing of quantum side information and backward classical communication from $B$ to $A$ can be indispensable for minimizing the cost in one-shot state merging from $A$ to $B$. These results complement existing protocols achieving nearly optimal one-shot state merging on a large scale, opening the way to another direction for future research on transferring quantum information on small and intermediate scales.
  \item Part~\ref{part:2} analyzes properties of multipartite entanglement in distributed quantum information processing, from the viewpoints of quantum communication costs over networks and the sizes of local quantum systems. Using the protocols established in Part~\ref{part:1}, Chapter~\ref{sec:distributed_encoding_decoding} in Part~\ref{part:2} evaluates costs of implementing multipartite nonlocal quantum state transformations for encoding and decoding quantum information in a multipartite quantum system, progressing beyond quantifications of bipartite and multipartite entanglement based on quantum communication costs. These encoding and decoding of quantum information are fundamental building blocks in quantum information processing, and difference between encoding and decoding is quantitatively characterized in terms of their implementation costs in distributed quantum information processing on a given tree-topology network for quantum communication. In Chapter~\ref{sec:multipartite} in Part~\ref{part:2}, advantage of the use of multipartite entanglement over bipartite entanglement is analyzed in terms of local quantum system sizes in distributed quantum information processing. Concrete examples are given to prove that multipartite entanglement outperforms bipartite entanglement when limitations on the local system sizes exist. These results facilitate operational understanding and efficient use of multipartite entanglement, from the viewpoints motivated by distributed quantum information processing over the networks beyond the state convertibility under LOCC\@.
\end{itemize}

\begin{figure}[t!]
  \centering
  \includegraphics[width=5.7in]{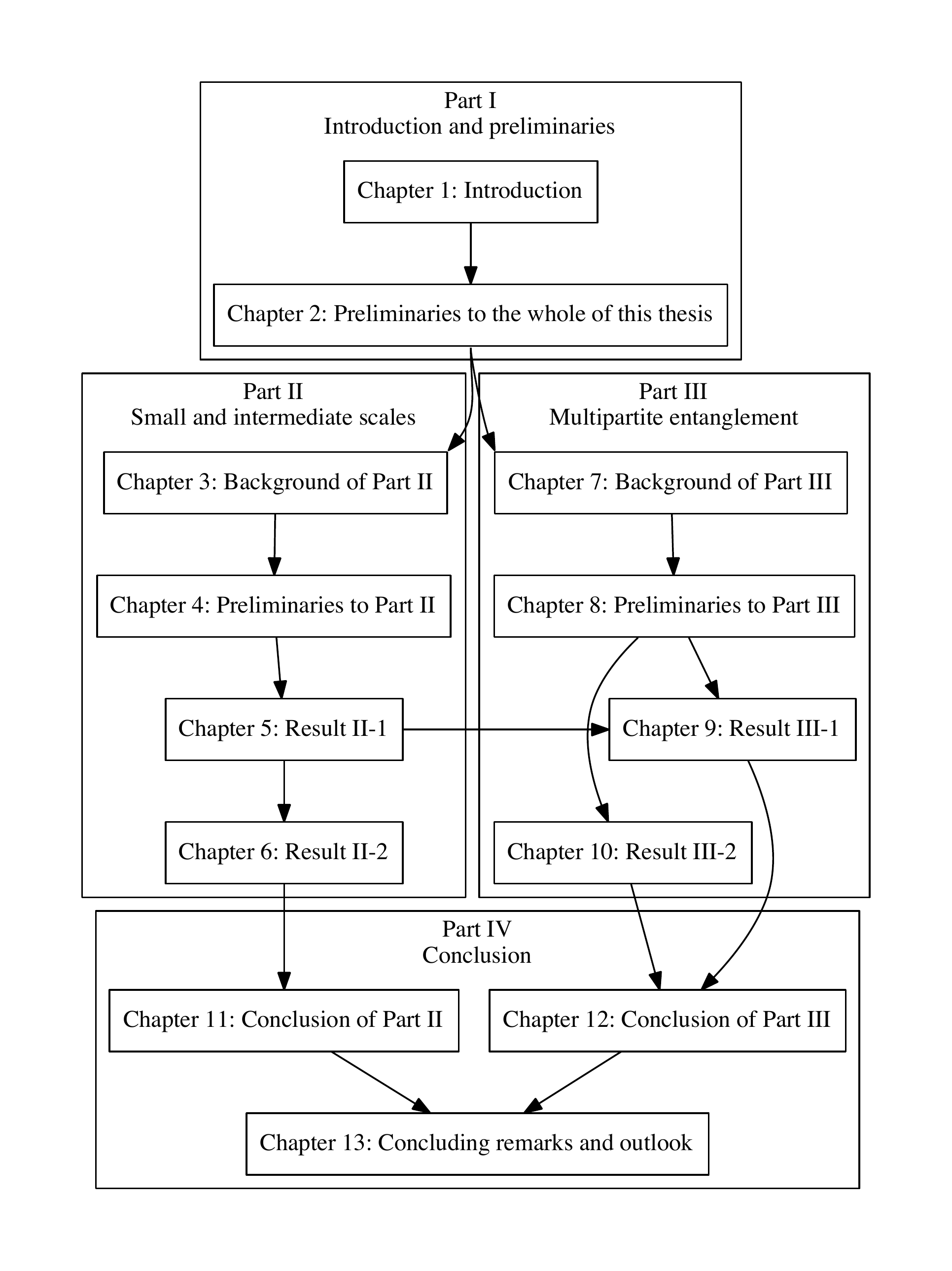}
  \caption[Organization of this thesis.]{\label{fig:organization}Organization of chapters in each part of this thesis. Part~\ref{part:1} analyzes a communication task, quantum state merging, so that the results are applicable to any two small- and intermediate-scale quantum devices on a network used in distributed quantum information processing. Using these results, Part~\ref{part:2} analyzes properties of multipartite entanglement from the viewpoint of distributed quantum information processing. The conclusion is given in Part~\ref{part:conclusion}.}
\end{figure}

\chapter{\label{sec:preliminaries_all}Preliminaries}

This chapter summarizes concepts in quantum information theory~\cite{N1,W5,W11} relevant to the rest of this thesis.
Section~\ref{sec:notations} summarizes formulation of general quantum mechanical operations used in quantum information processing.
Then, Section~\ref{sec:locc} defines local operations and classical communication (LOCC), which is a class of operations playing central roles in analyzing quantum entanglement.
After summarizing decomposition theorems used for operational analysis of entanglement in Section~\ref{sec:decomposition}, basic results of entanglement theory are summarized in Section~\ref{sec:entanglement}.

\section{\label{sec:notations}Operations in quantum mechanics}

In quantum mechanics, a physical system is represented as a complex Hilbert space.
Systems that can be represented as finite-dimensional Hilbert spaces are considered here for simplicity.
A composite system consisting of different subsystems is represented by the tensor product of Hilbert spaces representing each subsystem.
A composite system consisting of two subsystems is said to be \textit{bipartite}, and that of more than two subsystems \textit{multipartite}.
For a system labeled $A$, let $\mathcal{H}^{A}$ denote a Hilbert space representing the system, where the dimension of the Hilbert space may be written as
\begin{equation}
  D^A\coloneq\dim\mathcal{H}^A.
\end{equation}
If $D^A\geqq 2$, then the system $\mathcal{H}^{A}$ is called a \textit{qudit}, and especially if $D^A=2$, called a \textit{qubit}.
Let $\mathbb{C}$ denote the set of complex numbers, $\mathbb{Q}$ rational, and $\mathbb{R}$ real.
Then, $\mathcal{H}^A$ is isomorphic to $\mathbb{C}^{D^A}$, which is denoted by
\begin{equation}
  \mathcal{H}^A=\mathbb{C}^{D^A}.
\end{equation}
A Hilbert space $\mathcal{H}$ representing a composite system consisting of $N$ subsystems labeled $A_1\,,A_2\,,\ldots,A_N$ satisfies
\begin{equation}
  \mathcal{H}=\bigotimes_{k=1}^N\mathcal{H}^{A_k}=\mathcal{H}^{A_1}\otimes\cdots\otimes\mathcal{H}^{A_N}.
\end{equation}

For any system labeled $A$ and represented by a $D^A$-dimensional Hilbert space $\mathcal{H}^{A}$,
fix an arbitrary set of $D^A$ mutually orthogonal normalized vectors as preferred, and write this set as
\begin{equation}
  \left\{\Ket{l}^A \in\mathcal{H}^A:l\in\left\{0,\ldots,D^A-1\right\}\right\},
\end{equation}
which is called the \textit{computational basis} of $\mathcal{H}^A$.
The identity operator on $\mathcal{H}^A$ is denoted by
\begin{equation}
  \mathbb{1}^A\coloneq\sum_{l=0}^{D^A-1}\Ket{l}\Bra{l}^A,
\end{equation}
which may also be written as $\mathbb{1}_{D^A}^A$ for clarity of dimension.
The identity operators in formulas may be omitted for brevity.
A Hilbert space spanned by vectors $\Ket{\psi_0},\Ket{\psi_1},\ldots$ is denoted by
\begin{equation}
  \spn\left\{\Ket{\psi_0},\Ket{\psi_1},\ldots\right\}.
\end{equation}

A pure quantum state of a system $\mathcal{H}$ is represented by a normalized vector, denoted by a ket
\begin{align}
    &\Ket{\psi}\in\mathcal{H},\\
    &\left\|\Ket{\psi}\right\|=1,
\end{align}
where $\left\|\cdots\right\|$ represents the Euclidean norm, and $e^{\textup{i}\theta}\Ket{\psi}$ is identified with $\Ket{\psi}$ for any phase $\theta$.
In the following, a ket is normalized unless explicitly noted that it is unnormalized.
More generally, a mixed state of the system $\mathcal{H}$ is represented by a positive semidefinite operator of unit trace, which is called a density operator, and these conditions of a density operator $\psi$ are denoted by
\begin{align}
  \psi&\geqq 0, \left(\Leftrightarrow \Braket{\phi|\psi|\phi}\geqq 0,\forall\Ket{\phi}\right)\\
  \tr\psi&\coloneq\sum_{l=0}^{D-1}\Braket{l|\psi|l}=1,
\end{align}
where ${\left\{\Ket{l}\right\}}_{l=0,\ldots,D-1}$ is the computational basis of $\mathcal{H}$,
while $\tr\psi$ does not depend on the choice of the basis.
Given a system labeled $A$,
the set of density operators on $\mathcal{H}^A$ is denoted by $\mathcal{D}\left(\mathcal{H}^A\right)$,
and the set of bounded operators on $\mathcal{H}^A$ is denoted by $\mathcal{B}\left(\mathcal{H}^A\right)$.
A state of a bipartite system is called a \textit{bipartite state}, and a state of a multipartite system a \textit{multipartite state}.
Superscripts of an operator or a vector represent the labels of the corresponding Hilbert spaces, \textit{e.g.}, for a mixed state
\begin{equation}
  \psi^{A_1 A_2 A_3} \in \mathcal{D}\left(\mathcal{H}^{A_1}\otimes\mathcal{H}^{A_2}\otimes\mathcal{H}^{A_3}\right),
\end{equation}
and for a pure state
\begin{equation}
  \Ket{\psi}^{A_1 A_2 A_3} \in\mathcal{H}^{A_1}\otimes\mathcal{H}^{A_2}\otimes\mathcal{H}^{A_3}.
\end{equation}
A density operator corresponding to a pure state may be written as
\begin{equation}
  \psi^{A_1 A_2 A_3}\coloneq\Ket{\psi}\Bra{\psi}^{A_1 A_2 A_3}.
\end{equation}
Superscripts may be omitted if obvious from the context.

Operations on quantum systems can be considered to consist of unitary transformations and measurements assisted by adding and discarding auxiliary systems.
Time evolution of a quantum state may be described using an exponential function of Hamiltonians in quantum mechanics, such as that derived from the Schr\"{o}dinger equation in cases of non-relativistic closed systems.
This description using Hamiltonian is especially suited for situations where the Hamiltonian $H$ is time-independent and the time evolution in time $t_1-t_0$ is represented as $e^{-\textup{i}H\left(t_1-t_0\right)}$, or $H$ is, if time-dependent, changed according to a few parameters.
However, quantum information processing considers different situations where such Hamiltonian is actively engineered and arbitrarily controlled in a time-dependent way.
To analyze what kind of information processing is possible within quantum mechanics in such situations, the description explicitly using Hamiltonian is replaced by a unitary transformation of states,
which is represented for a system $\mathcal{H}^A$ as a unitary operator $U^A$.

Measurements are operations probabilistically obtaining classical information of measurement outcomes from a quantum state, where it is assumed that the number of the measurement outcomes is finite.
A measurement on a system $\mathcal{H}^A$ can be represented by a family of \textit{measurement operators}
\begin{equation}
  \left\{M_m^A: m=\left\{0,\ldots,n-1\right\}\right\}
\end{equation}
satisfying the completeness condition
\begin{equation}
  \sum_{m=0}^{n-1}{M_m^A}^\dag M_m^A =\mathbb{1},
\end{equation}
where $m$ is a label representing an outcome, and $n$ denotes the number of possible outcomes.
The number of the outcomes may not be explicitly written if not of interest, and the labels for the outcomes may be written using subscripts, such as ${\left\{M_m^A\right\}}_m$ in the above case.
Given any state $\psi^A$,
after performing a measurement represented by ${\left\{M_m^A\right\}}_m\,$,
the post-measurement state ${\psi_m^{\prime}}^A$ corresponding to each measurement outcome $m$ is given by
\begin{align}
  {\psi_m^{\prime}}^A&\coloneq\frac{1}{p\left(m\right)}{M_m^A \psi {M_m^A}^\dag},\\
  p\left(m\right)&\coloneq\tr{M_m^A \psi {M_m^A}^\dag},
\end{align}
where $p\left(m\right)$ is a probability distribution representing the probability of obtaining each measurement outcome.
For example, a projective measurement in the computational basis ${\left\{\Ket{l}\right\}}_{l}$ can be represented as a family of projectors
\begin{equation}
  {\left\{\Pi_l\coloneq\Ket{l}\Bra{l}\right\}}_l.
\end{equation}
By setting $n=1$, any unitary transformation $U$ can also be included in this formulation of measurement operators, where the corresponding family of measurement operators reduces to $\left\{U\right\}$.
If post-measurement states of a measurement represented by measurement operators ${\left\{M_m^A\right\}}_{m=0,\ldots,n-1}$ are not of interest,
it is also possible to consider a family of positive semidefinite operators called a \textit{positive operator-valued measure} (POVM)
\begin{equation}
  \left\{\Lambda_m^A\coloneq {M_m^A}^\dag M_m^A\geqq 0: m=\left\{0,\ldots,n-1\right\}\right\}
\end{equation}
satisfying the completeness condition
\begin{equation}
  \sum_{m=0}^{n-1}\Lambda_m =\mathbb{1}.
\end{equation}
Given any state $\psi^A$,
consider performing a measurement of $\psi^A$ represented by a POVM ${\left\{\Lambda_m^A={M_m^A}^\dag M_m^A\right\}}_m\,$,
and the probability of obtaining each measurement outcome $m$ is
\begin{equation}
  p\left(m\right)=\tr{\Lambda_m^A \psi}=\tr{M_m^A \psi {M_m^A}^\dag}.
\end{equation}
Note that for any state $\psi^A$ and any measurement,
the positivity and the completeness condition guarantee the axiom of probability:
\begin{align}
  p\left(m\right)&\geqq 0,\quad \forall m;\\
  \sum_m p\left(m\right)&=1.
\end{align}

A more general formulation of measurements may include situations where classical post-processing of the measurement outcomes is allowed, and some of the outcomes can be coarse-grained by this classical post-processing.
This situation is formulated using \textit{quantum instruments}.
The quantum instrument on $\mathcal{H}^A$ is represented using linear maps
\begin{equation}
  \mathcal{E}_m^A:\mathcal{B}\left(\mathcal{H}^A\right)\to \mathcal{B}\left(\mathcal{H}^A\right)
\end{equation}
in the form of
\begin{equation}
  \label{eq:instrument}
  \mathcal{E}_m^A\left(\psi^A\right)\coloneq\sum_{j=0}^{J_m -1} \left(M_{m,j}^A\right)\psi^A{\left(M_{m,j}^A\right)}^\dag,
\end{equation}
where $m\in\left\{0,\ldots,n-1\right\}$ is a label representing an outcome, $n$ denotes the number of possible outcomes, ${\left\{M_{m,j}^A:m\in\left\{0,\ldots,n-1\right\},j\in\left\{0,\ldots,J_m-1\right\}\right\}}$ is a family of measurement operators whose outcomes are labeled by $m$ for the quantum instrument and by $j$ to be erased by the classical post-processing, and $J_m$ may depend on $m$.
This post-processing for erasing some of the outcomes is called \textit{coarse-graining}.
More precisely, given a family of measurement operators
\begin{equation}
  \left\{M_0\,,\ldots,M_{n-1}\right\},
\end{equation}
coarse-graining divides these measurement operators into $n^\prime$ subgroups of $J_0\,,\ldots,J_{n^\prime-1}$ elements, respectively, satisfying $J_0+\cdots+J_{n^\prime-1}=n$, that is
\begin{equation}
  \left\{M_{0,0}\,,\ldots,M_{0,J_0-1}\,,M_{1,0}\,,\ldots,M_{1,J_1-1}\,,\ldots,M_{n^\prime-1,0}\,,\ldots,M_{n^\prime-1,J_{n^\prime-1}-1}\right\}.
\end{equation}
Using this coarse-graining, a quantum instrument is defined as a family of linear maps
\begin{equation}
  \left\{\mathcal{E}_m^A:m=0,\ldots,n^\prime-1\right\},
\end{equation}
where each $\mathcal{E}_m^A$ is in the form of Equation~\eqref{eq:instrument}.
Given any state $\psi^A$,
after performing a measurement represented by ${\left\{\mathcal{E}_m^A\right\}}_m\,$,
the post-measurement state for each outcome $m$ is given by
\begin{equation}
  \frac{\mathcal{E}_m^A\left(\psi^A\right)}{\tr\mathcal{E}_m^A\left(\psi^A\right)},
\end{equation}
which is obtained with probability
\begin{equation}
  {p\left(m\right)}\coloneq\tr\mathcal{E}_m^A\left(\psi^A\right).
\end{equation}
The single-outcome case of $n=1$ corresponds to the situation of performing the measurement ${\left\{M_{m,j}^A\right\}}_{m,j}$ followed by erasing all the outcomes, which can be performed deterministically.

An auxiliary system is considered to be an additionally prepared system different from an initially given system, so that operations can be performed on the composite system consisting of these systems.
When added, this auxiliary system is initialized as a fixed state, where in the following the fixed initial state is chosen as $\Ket{0}$ for simplicity.
Given any state $\psi^A$, adding an auxiliary system $\mathcal{H}^{A^\prime}$ to the given system $\mathcal{H}^A$ followed by an unitary transformation $U^{AA^\prime}$ on $\mathcal{H}^A\otimes\mathcal{H}^{A^\prime}$ is represented by an isometry transformation
\begin{equation}
  U^{A\to AA^\prime}\coloneq U^{AA^\prime}\left(\mathbb{1}^A\otimes\Ket{0}^{A^\prime}\right)
\end{equation}
satisfying
\begin{equation}
  \left(U^{A\to AA^\prime}\right)\psi^A {\left(U^{A\to AA^\prime}\right)}^\dag=U^{AA^\prime}\left(\psi^A\otimes\Ket{0}\Bra{0}^{A^\prime}\right){U^{AA^\prime}}^\dag.
\end{equation}
In other words, isometry transformations are invertible transformations from states of a smaller-dimensional quantum system to those of a larger-dimensional quantum system.
Note that such a unitary transformation $U^{AA^\prime}$ corresponding to the isometry transformation $U^{A\to AA^\prime}$ is not unique in general.
Superscripts of an operator such as $U^{A\to AA^\prime}$ represent the input system $\mathcal{H}^A$ and the output system $\mathcal{H}^A\otimes\mathcal{H}^{A^\prime}$ for the state transformation represented by the operator.

Conversely, consider a situation of discarding a part of the subsystems comprising a composite system.
Given a state of such a composite system, a state obtained by discarding a part of the subsystems is called a reduced state of the rest of subsystems.
Such a reduced state is represented as a state obtained by performing partial trace on Hilbert spaces representing the discarded subsystems.
Partial trace on a Hilbert space representing a discarded subsystem is a linear transformation such that any deterministic operation on the discarded system before the partial trace does not change the reduced state after the partial trace.
A reduced state may be represented by superscripts if obvious.
For example, given a system represented by $\mathcal{H}^{X_1}\otimes\mathcal{H}^{X_2}\otimes\mathcal{H}^{X_3}$ and a state $\psi^{X_1 X_2 X_3} \in \mathcal{D}\left(\mathcal{H}^{X_1}\otimes\mathcal{H}^{X_2}\otimes\mathcal{H}^{X_3}\right)$,
the partial trace on $\mathcal{H}^{X_2}\otimes\mathcal{H}^{X_3}$ is denoted by $\tr_{X_2 X_3}\,$,
and the reduced state of $\mathcal{H}^{X_1}$ is represented as
\begin{equation}
  \begin{split}
    \psi^{X_1}&\coloneq\tr_{X_2 X_3}\psi^{X_1 X_2 X_3}\\
              &=\sum_{l=0}^{D^{X_2}-1}\sum_{l^\prime =0}^{D^{X_3}-1}\left(\mathbb{1}^{X_1}\otimes\Bra{l}^{X_2}\otimes\Bra{l^\prime}^{X_3}\right)\psi^{X_1 X_2 X_3}\left(\mathbb{1}^{X_1}\otimes\Ket{l}^{X_2}\otimes\Ket{l^\prime}^{X_3}\right),
  \end{split}
\end{equation}
where $D^{X_2}=\dim\mathcal{H}^{X_2}$ and $D^{X_3}=\dim\mathcal{H}^{X_3}$.

Combining the above operations is sufficient for achieving the class of any operations consistent with quantum mechanics.
A transformation of quantum states is represented by a linear map, and in the following, a linear map may be simply referred to as a map.
For example, the identity map on a system $\mathcal{H}^{A}$ is denoted by $\id^A$ and due to linearity,
\begin{equation}
  \id^A\left(\psi^A\right)=\psi^A
\end{equation}
is satisfied for any $\psi^A\in\mathcal{B}\left(\psi^A\right)$.
Since adding and discarding auxiliary systems are allowed, a map representing a state transformation may have different input and output systems represented by different Hilbert spaces.
In the following, superscripts of linear maps represent the labels of input and output systems; \textit{e.g.},
for cases where input and output systems are the same, a ma may be written as
\begin{equation}
  \mathcal{N}^{A}:\mathcal{B}\left(\mathcal{H}^{A}\right)\to\mathcal{B}\left(\mathcal{H}^{A}\right),
\end{equation}
and otherwise, a map for a state transformation from a system $\mathcal{H}^{A_\textup{in}}$ to a system $\mathcal{H}^{A_\textup{out}}$ may be written as
\begin{equation}
  \mathcal{N}^{A_\textup{in}\to A_\textup{out}}:\mathcal{B}\left(\mathcal{H}^{A_\textup{in}}\right)\to\mathcal{B}\left(\mathcal{H}^{A_\textup{out}}\right).
\end{equation}
Any linear map $\mathcal{N}^{A_\textup{in}\to A_\textup{out}}$ representing a deterministic state transformation has to satisfy the following two properties for being consistent with the axiom of probability:
\begin{description}
  \item[Completely positive property] Given any auxiliary system $\mathcal{H}^R$ and any operator $\psi^{RA_\textup{in}}\in\mathcal{B}\left(\mathcal{H}^R\otimes\mathcal{H}^{A_\textup{in}}\right)$, if the operator is a positive semidefinite operator, that is,
    \begin{equation}
      \psi^{RA_\textup{in}}\geqq 0,
    \end{equation}
    then the operator obtained by performing $\mathcal{N}^{A_\textup{in}\to A_\textup{out}}$ on $\mathcal{H}^{A_\textup{in}}$ is also mapped into a positive semidefinite operator, that is,
    \begin{equation}
      \left(\id^R\otimes\mathcal{N}^{A_\textup{in}\to A_\textup{out}}\right)\left(\psi^{RA_\textup{in}}\right)\geqq 0;
    \end{equation}
  \item[Trace-preserving property] Given any operator $\psi^{A_\textup{in}}\in\mathcal{B}\left(\mathcal{H}^{A_\textup{in}}\right)$, $\mathcal{N}^{A_\textup{in}\to A_\textup{out}}$ preserves the trace of the operator, that is,
    \begin{equation}
      \tr\mathcal{N}^{A_\textup{in}\to A_\textup{out}}\left(\psi^{A_\textup{in}}\right)=\tr\psi^{A_\textup{in}}.
    \end{equation}
\end{description}
A map satisfying these properties is called a completely positive and trace-preserving (CPTP) map.
Given any CPTP map $\mathcal{N}^{A_\textup{in}\to A_\textup{out}}$, there exists a Hilbert space $\mathcal{H}^E$ representing an auxiliary system and an isometry transformation $U_\mathcal{N}^{A_\textup{in}\to A_\textup{out}E}$ from $\mathcal{H}^{A_\textup{in}}$ to $\mathcal{H}^{A_\textup{out}}\otimes\mathcal{H}^{E}$ such that for any input state $\psi^{A_\textup{in}}$
\begin{equation}
  \mathcal{N}^{A_\textup{in}\to A_\textup{out}}\left(\psi^{A_\textup{in}}\right)=\tr_E \left(U_\mathcal{N}^{A_\textup{in}\to A_\textup{out}E}\right)\psi^{A_\textup{in}}{\left(U_\mathcal{N}^{A_\textup{in}\to A_\textup{out}E}\right)}^\dag.
\end{equation}
This representation of a CPTP map is called the \textit{Stinespring dilation} of the CPTP map, implying that state transformations represented by any CPTP maps can be achieved by adding an auxiliary system, performing a unitary transformation, and discarding a part of subsystems.
The CPTP maps are also referred to as channels.
To investigate properties of a CPTP map $\mathcal{N}^{A_\textup{in}\to A_\textup{out}}$, it is useful to consider the Choi operator $J\left(\mathcal{N}^{A_\textup{in}\to A_\textup{out}}\right)\in\mathcal{B}\left(\mathcal{H}^R\otimes\mathcal{H}^{A_\textup{out}}\right)$ of $\mathcal{N}^{A_\textup{in}\to A_\textup{out}}$ defined as
\begin{equation}
  \label{eq:choi_operator}
  J\left(\mathcal{N}^{A_\textup{in}\to A_\textup{out}}\right)\coloneq\left(\id^R\otimes\mathcal{N}^{A_\textup{in}\to A_\textup{out}}\right)\left(\left(\sum_{l=0}^{D^{A_\textup{in}}-1}\Ket{l}^R\otimes\Ket{l}^{A_\textup{in}}\right)\left(\sum_{l=0}^{D^{A_\textup{in}}-1}\Bra{l}^R\otimes\Bra{l}^{A_\textup{in}}\right)\right),
\end{equation}
where this Choi operator is not normalized, $\mathcal{H}^R$ is an additional Hilbert space for introducing this Choi operator, and
\begin{equation}
  D^{A_\textup{in}}\coloneq\dim\mathcal{H}^{A_\textup{in}}.
\end{equation}
There exists one-to-one correspondence between a CPTP map and the Choi operator of the CPTP map.

In terms of these properties, the quantum instrument ${\left\{\mathcal{E}_m\right\}}_m$ can also equivalently be considered as a family of completely positive (CP) maps whose sum $\sum_m \mathcal{E}_m$ is a trace-preserving map.
In the same way as CPTP maps, different input and output systems represented by different Hilbert spaces can be considered for quantum instruments.
Given any quantum instrument ${\left\{\mathcal{E}_m^{A_\textup{in}\to A_\textup{out}}\right\}}_m$ and any state $\psi^{A_\textup{in}}$,
introduce an auxiliary system $\mathcal{H}^X$ for storing the measurement outcome of this quantum instrument,
and performing the measurement represented by ${\left\{\mathcal{E}_m^{A_\textup{in}\to A_\textup{out}}\right\}}_m$ is equivalent to performing a CPTP map $\mathcal{E}^{A_\textup{in}\to A_\textup{out}X}$ acting as
\begin{equation}
  \begin{split}
    \mathcal{E}^{A_\textup{in}\to A_\textup{out}X}\left(\psi^{A_\textup{in}}\right)&\coloneq\sum_m\mathcal{E}_m^{A_\textup{in}\to A_\textup{out}}\left(\psi^{A_\textup{in}}\right)\otimes\Ket{m}\Bra{m}^X\\
                                                                                   &=\sum_m p\left(m\right)\frac{\mathcal{E}_m^{A_\textup{in}\to A_\textup{out}}\left(\psi^{A_\textup{in}}\right)}{\tr\mathcal{E}_m^{A_\textup{in}\to A_\textup{out}}\left(\psi^{A_\textup{in}}\right)}\otimes\Ket{m}\Bra{m}^X
  \end{split}
\end{equation}
where the measurement outcomes are represented as orthogonal pure states for the computational basis of $\mathcal{H}^X$.
Note that while quantum instruments are discussed above, the same argument for different input and output systems holds for measurement operators as special cases of quantum instruments.
Any CPTP map is equivalent to a single-outcome quantum instrument.

In the rest of this thesis, the most general forms for representing operations, that is, CPTP maps and quantum instruments, and more specific forms, such as isometries and measurement operators, are suitably used for describing maps.

\section{\label{sec:locc}Local operations and classical communication}

In distributed quantum information processing, multiple quantum devices capable of coherently keeping the quantum states of a quantum system inside and of performing operations for transforming these quantum states cooperate in achieving quantum information processing.
The local quantum system in each quantum device can be regarded as a subsystem comprising a multipartite quantum system distributed among the devices.
Operations performed in each device is restricted to local operations on the subsystem held in the device.
To perform arbitrary nonlocal operations on the distributed multipartite system,
classical communication of measurement outcomes between the devices is not sufficient, while quantum communication for transferring quantum information of quantum states can be used for achieving such nonlocal operations.
However, while classical communication can be reliably performed using current technologies, quantum communication between spatially separated quantum devices is more challenging and costly.
In this regard, it is natural to ask what can be achieved only using local operations and classical communication (LOCC)~\cite{D11,C19,C7}.
This section provides definition of LOCC, and the notion of entanglement is also introduced in terms of LOCC\@.

To introduce LOCC, each of such quantum devices in performing LOCC is called a party being able to perform arbitrary local operations on the party's local quantum system.
Let $N$ be the number of the parties, and the parties are denoted by $v_1\,,\ldots,v_N$.
The set of the parties is denoted by
\begin{equation}
  V\coloneq\left\{v_1\,,\ldots,v_N\right\}.
\end{equation}
For each $v_k\in V$, let $\mathcal{H}^{v_k}$ represent the system held by the party $v_k\,$, and the whole multipartite system distributed among the parties is denoted by
\begin{equation}
  \mathcal{H}\coloneq\bigotimes_{v_k\in V}\mathcal{H}^{v_k}.
\end{equation}
LOCC consists of measurements by each party and classical communication for sending the measurement outcomes to the other parties, where each measurement can be conditioned by the former measurement outcomes obtained by other parties.
When classical communication can be freely performed, it is sufficient to consider that the measurement outcomes for each measurement are sent to all the parties.
Classical communication introduces sequential order of measurements, and the number of classical communication is referred to as the round of classical communication.

In the following, LOCC is introduced in terms of measurement operators for simplicity, while a more formal definition in terms of quantum instruments also follows from combining the following argument with classical post-processing of coarse-graining, as discussed later.
A family of measurement operators ${\left\{M_{m_1}\right\}}_{m_1}$ on $\mathcal{H}$ is called \textit{non-correcting one-way local} from a party $v_k\in V$ if each measurement operator is in the form of
\begin{equation}
  M_{m_1}\coloneq M_{m_1}^{v_k}\otimes\bigotimes_{v\in V \setminus\left\{v_k\right\}}\mathbb{1}^v.
\end{equation}
where ${\left\{M_{m_1}^{v_k}\right\}}_{m_1}$ is $v_k$'s measurement on $\mathcal{H}^{v_k}$ with outcome $m_1\,$, and for each $v\in V \setminus\left\{v_k\right\}$, $\mathbb{1}^v$ is the identity operator on $\mathcal{H}^v$.
This measurement is called one-way in the sense that classical communication is performed in the one-way direction from $v_k$ and the others, and is called non-correcting in the sense that the other parties than $v_k$ do not perform any operation for correction conditioned by $v_k$'s measurement outcome.
To describe $r$ rounds of classical communication,
write a tuple of labels for representing measurement outcomes as
\begin{equation}
  \boldsymbol{m}_r\coloneq\left(m_1\,,\ldots,m_r\right).
\end{equation}
where $r=1,2,\ldots$ and $\boldsymbol{m}_1$ is identified with $m_1$.
A family of measurement operators ${\left\{M_{\boldsymbol{m}_r}\right\}}_{\boldsymbol{m}_r}$ on $\mathcal{H}$ is said to be \textit{LOCC linked} to a family of measurement operators ${\left\{M_{\boldsymbol{m}_{r-1}}\right\}}_{\boldsymbol{m}_{r-1}}$ on $\mathcal{H}$ if there exists a party $v_k\in V$ and a non-correcting one-way local measurement ${\left\{M_{m_r}\right\}}_{m_r}$ from $v_k$ such that for each $\boldsymbol{m}_r\,$, $M_{\boldsymbol{m}_r}$ is the composition of $M_{\boldsymbol{m}_{r-1}}$ and $M_{m_r}\,$, that is,
\begin{equation}
  M_{\boldsymbol{m}_r}=M_{m_r}M_{\boldsymbol{m}_{r-1}},
\end{equation}
where each measurement operator in ${\left\{M_{m_r}\right\}}_{m_r}$ is in the form of
\begin{equation}
  M_{m_r}\coloneq M_{m_r|\boldsymbol{m}_{r-1}}^{v_k}\otimes\bigotimes_{v\in V \setminus\left\{v_k\right\}}\mathbb{1}^v,
\end{equation}
and ${\left\{M_{m_r|\boldsymbol{m}_{r-1}}^{v_k}\right\}}_{m_r}$ is $v_k$'s measurement operator on $\mathcal{H}^{v_k}$ with outcome $m_r$ possibly conditioned by all the preceding outcomes $\boldsymbol{m}_{r-1}$ having been broadcast to all the parties by classical communication.

Using these notions, LOCC is defined as follows.
Let \textit{non-correcting one-round LOCC} refer to a non-correcting one-way local measurement from some party,
and define \textit{non-correcting $r$-round LOCC} for any $r\geqq 2$ as operations represented by a family of measurement operators LOCC linked to that representing non-correcting $(r-1)$-round LOCC\@.
For any $r\geqq 1$, \textit{$r$-round LOCC} is defined as operations represented as a family of measurement operators achieved by a non-correcting $r$-round LOCC with outcome $\boldsymbol{m}_r$ followed by each party $v$'s local measurement ${\left\{M_{m^v|\boldsymbol{m}_r}^v\right\}}_{m^v}$ with measurement outcome $m^v$ conditioned by $\boldsymbol{m}_r$.
Finite-round LOCC refers to operations that are $r$-round LOCC for some finite $r$.
Considering a sequence of $r$-round LOCC for $r=1,2,\ldots$, where non-correcting $r$-round LOCC for each $r$-round LOCC is LOCC-linked to that for $(r-1)$-round LOCC,
and \textit{LOCC} is defined as operations that can be represented as a limit of this type of LOCC-linked sequence of $r$-round LOCC as $r\to\infty$.
A CPTP map achieved by LOCC is called an LOCC map.

To define LOCC in terms of quantum instruments, it suffices to modify the above definition so that whenever a measurement outcome $m_r$ is obtained for each $r=1,2,\ldots$, classical post-processing of coarse-graining of all the measurement outcomes $\boldsymbol{m}m_r$ is performed, in the same way as Reference~\cite{C7}.
Note that there exists subtle difference between the above definition of $r$-round LOCC and that in Reference~\cite{C7} in that the above definition allows each party's local measurement after non-correcting $r$-round LOCC, while Reference~\cite{C7} allows only each party's CPTP map after that.
This difference matters when separation between one-way LOCC and two-way LOCC in a task of local state discrimination is discussed in Chapter~\ref{sec:two_way}.
For any $r=1,2,\ldots$, the set of quantum instruments representing $r$-round LOCC is strictly included by that of $(r+1)$-round LOCC, and that of finite-round LOCC is strictly included in LOCC\@.
The set of CPTP maps achievable by $r$-round LOCC is known to be compact if $r$ is finite, but that representing LOCC is not, since LOCC possibly includes infinitely many rounds~\cite{C7}.

In cases of two parties denoted by $A$ and $B$ by convention, \textit{one-way LOCC} refers to one-round LOCC, and \textit{two-way LOCC} refers to LOCC other than one-way LOCC\@.
One-way LOCC from $A$ to $B$ refers to one-round LOCC defined using non-correcting one-way local measurements only from $A$ but not from $B$.

This class of operations, LOCC, naturally defines a class of states exhibiting \textit{quantum entanglement}.
Consider situations where the cost of performing LOCC is negligible compared to quantum communication.
In such situations, it is natural to assume that the parties can freely perform LOCC\@.
Then, a state $\psi\in\mathcal{D}\left(\mathcal{H}\right)$ is called a separable state if for any $\phi\in\mathcal{D}\left(\mathcal{H}\right)$, there exists an LOCC map $\mathcal{E}$ such that
\begin{equation}
  \psi=\mathcal{E}\left(\phi\right).
\end{equation}
In other words, separable states are the states that can be obtained from scratch by LOCC, in the sense that these states can be obtained from any state by LOCC\@.
Also equivalently, separable pure states are product states of local states, that is,
\begin{equation}
  \bigotimes_{v_k\in V}\Ket{\psi_k}^{v_k},
\end{equation}
where $\Ket{\psi_k}^{v_k}\in\mathcal{H}^{v_k}$ for each $v_k\,$,
and mixed separable states are convex combinations of product states.
An entangled state is defined as a state which is not separable.
\textit{Bipartite entanglement} refers to entanglement of bipartite entangled states shared between two parties, and \textit{multipartite entanglement} refers to that of multipartite entangled states shared among more than two parties.
Note that in terms of quantum resource theory~\cite{C13}, LOCC is regarded as free operations, and separable states are free states obtained from this free operations.
Entangled states are resources that cannot be obtained by LOCC from a separable state, and LOCC assisted by an initially given entangled state shared among the parties can be advantageous in distributed quantum information processing compared to performing LOCC without such assistance, as discussed in Section~\ref{sec:entanglement}.

\section{\label{sec:decomposition}Decomposition theorems for analysis of quantum state transformation}

For simplifying analysis of properties of entangled states under LOCC, mathematical decomposition theorems for operators representing quantum states can be exploited.
This section provides such decomposition theorems for later use.

\paragraph{Spectral decomposition}
Given a system $\mathcal{H}$ and a state $\psi\in\mathcal{D}\left(\mathcal{H}\right)$,
the \textit{spectral decomposition} of $\psi$ yields
\begin{equation}
  \psi=\sum_{l=0}^{R-1}\lambda_l\Ket{\psi_l}\Bra{\psi_l},
\end{equation}
where $R$ is the rank of $\psi$, each $\lambda_l>0$ is a nonzero eigenvalue, and ${\left\{\Ket{\psi_l}\right\}}_l$ is a set of normalized pure states representing eigenvectors corresponding to nonzero eigenvalues, which are orthogonal with each other.
It is assumed that the eigenvalues are sorted in descending order, that is,
\begin{equation}
  \label{eq:eigen}
  \lambda_0\geqq\lambda_1\geqq\cdots\geqq\lambda_{R-1}.
\end{equation}
More generally, if $\psi$ is a Hermitian operator, spectral decomposition of $\psi$ is in the same form as the above while each nonzero eigenvalue $\lambda_l$ can be a negative real number.
Given any function $f:\mathbb{R}\to\mathbb{R}$ and any Hermitian operator $\psi$,
define
\begin{equation}
  f\left(\psi\right)\coloneq\sum_{l=0}^{R-1}f\left(\lambda_l\right)\Ket{\psi_l}\Bra{\psi_l},
\end{equation}
where the spectral decomposition of $\psi$ is used on the right-hand side.
Note that this definition is equivalent to considering Tailor expansion of $f\left(x\right)$ and substituting $x$ with $\psi$ in this Tailor expansion to define $f\left(\psi\right)$.
For example,
\begin{equation}
  \sqrt{\psi}\coloneq\sum_{l=0}^{R-1}\sqrt{\lambda_l}\Ket{\psi_l}\Bra{\psi_l}.
\end{equation}

\paragraph{Singular value decomposition and Schmidt decomposition}
Similarly to the spectral decomposition, given any Hilbert spaces $\mathcal{H}^A$ and $\mathcal{H}^B$, and any bounded operator from $\mathcal{H}^B$ to $\mathcal{H}^A$
\begin{equation}
  \psi^{B\to A}=\sum_{l,l^\prime}\psi_{l,l^\prime}\Ket{l}^A\Bra{l^\prime}^B,
\end{equation}
where ${\left\{\Ket{l}^A\right\}}_l$ and ${\left\{\Ket{l^\prime}^A\right\}}_{l^\prime}$ are the computational bases of $\mathcal{H}^A$ and $\mathcal{H}^B$, respectively,
the singular value decomposition of $\psi$ is a decomposition in the form of
\begin{equation}
  \psi^{B\to A}=\sum_{l=0}^{R-1}\lambda_l\Ket{\psi_l}^A\Bra{\psi_l}^B,
\end{equation}
where $R$ is the rank of $\psi$, each $\lambda_l>0$ is a nonzero singular value, and ${\left\{\Ket{\psi_l}^A\right\}}_l$ and ${\left\{\Ket{\psi_l}^B\right\}}_l$ are sets of normalized pure states of $\mathcal{H}^A$ and $\mathcal{H}^B$, respectively, representing singular vectors corresponding to nonzero singular values, which are orthogonal with each other.

Analogously to singular value decomposition, given any Hilbert spaces $\mathcal{H}^A$ and $\mathcal{H}^B$, and any bipartite pure state of $\mathcal{H}^A\otimes\mathcal{H}^B$
\begin{equation}
  \Ket{\psi}^{AB}=\sum_{l,l^\prime}\psi_{l,l^\prime}\Ket{l}^A\otimes\Ket{l^\prime}^B,
\end{equation}
where ${\left\{\Ket{l}^A\right\}}_l$ and ${\left\{\Ket{l^\prime}^A\right\}}_{l^\prime}$ are the computational bases of $\mathcal{H}^A$ and $\mathcal{H}^B$, respectively,
the Schmidt decomposition of $\Ket{\psi}^{AB}$ is a decomposition in the form of
\begin{equation}
  \Ket{\psi}^{AB}=\sum_{l=0}^{R-1}\sqrt{\lambda_l}\Ket{\psi_l}^A\otimes\Ket{\psi_l}^B,
\end{equation}
where $R$ is called the Schmidt rank of $\Ket{\psi}^{AB}$, each $\lambda_l>0$ is called a nonzero Schmidt coefficient, and ${\left\{\Ket{\psi_l}^A\right\}}_l$ and ${\left\{\Ket{\psi_l}^B\right\}}_l$ are sets of normalized pure states of $\mathcal{H}^A$ and $\mathcal{H}^B$, respectively, which are orthogonal with each other.
While $\mathcal{H}^A$ and $\mathcal{H}^B$ are spanned by bases consisting of $\dim\mathcal{H}^A$ vectors and $\dim\mathcal{H}^B$ vectors, respectively,
${\left\{\Ket{\psi_l}^A\right\}}_l$ and ${\left\{\Ket{\psi_l}^B\right\}}_l\,$, consisting of $R$ vectors corresponding to nonzero Schmidt coefficients, can be used as a part of such bases, and this type of basis is called a Schmidt basis.
Schmidt-basis states may refer to states in a Schmidt basis corresponding to nonzero Schmidt coefficients.
Given the above Schmidt decomposition of $\Ket{\psi}^{AB}$, tracing out $\mathcal{H}^B$ yields the spectral decomposition of $\psi^A$
\begin{equation}
  \psi^A=\sum_{l=0}^{R-1}\lambda_l\Ket{\psi_l}\Bra{\psi_l}^A,
\end{equation}
and assume that the Schmidt coefficients are sorted in descending order in the same way Equation~\eqref{eq:eigen} for the eigenvalues.

Given any state $\psi^A$,
a bipartite pure state $\Ket{\psi}^{AB}$ is called a \textit{purification} of $\psi^A$ if
\begin{equation}
  \psi^A=\tr_B\Ket{\psi}\Bra{\psi}^{AB},
\end{equation}
where $\mathcal{H}^B$ is an auxiliary system for this purification.
A purification of $\psi^A$ may not be unique, but different purifications are related by isometries.
More precisely,
consider any state $\psi^A$ given in the spectral-decomposition form as
\begin{equation}
  \psi^A=\sum_{l=0}^{R-1}\lambda_l\Ket{\psi_l}\Bra{\psi_l}^A,
\end{equation}
and any two purifications of $\psi^A$
\begin{align}
  \Ket{\psi}^{AB}&\in\mathcal{H}^A\otimes\mathcal{H}^B,\\
  \Ket{\psi^\prime}^{AB^\prime}&\in\mathcal{H}^A\otimes\mathcal{H}^{B^\prime}.
\end{align}
Then, the corresponding Schmidt-decomposition forms are given by
\begin{align}
  \Ket{\psi}^{AB}&=\sum_{l=0}^{R-1}\sqrt{\lambda_l}\Ket{\psi_l}^A\otimes\Ket{\psi_l}^B,\\
  \Ket{\psi^\prime}^{AB^\prime}&=\sum_{l=0}^{R-1}\sqrt{\lambda_l}\Ket{\psi_l}^A\otimes\Ket{\psi_l^\prime}^{B^\prime}.
\end{align}
Hence, there exists an isometry $U^{B\to B^\prime}$
such that for each $l$
\begin{equation}
  U^{B\to B^\prime}\Ket{\psi_l}^{B}=\Ket{\psi_l^\prime}^{B^\prime},
\end{equation}
that is,
\begin{equation}
  \left(\mathbb{1}^A\otimes U^{B\to B^\prime}\right)\Ket{\psi}^{AB}=\Ket{\psi^\prime}^{AB^\prime}.
\end{equation}

\paragraph{Norms and distances}
Using the spectral decomposition,
the $1$-norm, also known as the trace norm, of any Hermitian operator $\psi$ is defined as
\begin{equation}
  {\left\|\psi\right\|}_1\coloneq\sum_{l=0}^{R-1}\left|\lambda_l\right|,
\end{equation}
where nonzero eigenvalues of $\psi$ are used on the right-hand side.
As for another norm, the $\infty$-norm, also known as the operator norm, of $\psi$ is defined as
\begin{equation}
  {\left\|\psi\right\|}_\infty\coloneq\max_l\left\{\lambda_l\right\}=\lambda_0\,,
\end{equation}
where nonzero eigenvalues of $\psi$ are used on the right-hand side.

This type of norms of Hermitian operators can be used for quantifying distance between two quantum states.
Given two states $\psi\in\mathcal{D}\left(\mathcal{H}\right)$ and $\phi\in\mathcal{D}\left(\mathcal{H}\right)$,
trace distance between these two states is defined as
\begin{equation}
  {\left\|\psi-\phi\right\|}_1\,,
\end{equation}
which characterizes the optimal success probability of discriminating these two states by a quantum measurement~\cite{H16,W11}.
There exists another commonly used quantity in quantum information theory representing closeness of two quantum states,
the (square root) fidelity between $\psi$ and $\phi$, defined as
\begin{equation}
  F\left(\psi,\phi\right)\coloneq{\left\|\sqrt{\psi}\sqrt{\phi}\right\|}_1.
\end{equation}
As for other quantities, refer to Reference~\cite{A14}.
The trace distance and the fidelity are related by Fuchs-van de Graaf inequalities
\begin{equation}
  1-\frac{1}{2}{\left\|\psi-\phi\right\|}_1\leqq F\left(\psi,\phi\right)\leqq\sqrt{1-\frac{1}{4}{\left({\left\|\psi-\phi\right\|}_1\right)}^2}.
\end{equation}
The squared form of this fidelity is written as
\begin{equation}
  F^2\left(\psi,\phi\right)\coloneq{\left({\left\|\sqrt{\psi}\sqrt{\phi}\right\|}_1\right)}^2,
\end{equation}
and is extensively used in this thesis.
This fidelity satisfies the following properties:
\begin{enumerate}
  \item $0\leqq F^2\left(\psi,\phi\right)\leqq 1$;
  \item $F^2\left(\psi,\phi\right)=1\Leftrightarrow \psi=\phi$;
  \item $F^2\left(\psi,\phi\right)=F^2\left(\phi,\psi\right)$ (symmetric);
  \item $F^2\left(\psi\otimes{\psi^\prime},\phi\otimes{\phi^\prime}\right)=F^2\left(\phi,\psi\right)F^2\left(\phi^\prime,\psi^\prime\right)$ (multiplicative);
\end{enumerate}
where $\psi$, $\phi$, ${\psi^\prime}$, and ${\phi^\prime}$ are arbitrary states.
If $\psi$ is pure, \textit{i.e.}, $\psi=\Ket{\psi}\Bra{\psi}$,
it holds that
\begin{equation}
  F^2\left(\psi,\phi\right)=\Bra{\psi}\phi\Ket{\psi}.
\end{equation}

Using the fidelity, the purified distance~\cite{T12,G11} between two normalized states $\psi$ and $\phi$ is defined as
\begin{equation}
  P\left(\psi,\phi\right)\coloneq\sqrt{1-F^2\left(\psi,\phi\right)}.
\end{equation}
The purified distance between two normalized states $\psi$ and $\phi$ can also be represented as the minimum trace distance between purifications of $\psi$ and $\phi$~\cite{T5,T11}.
Note that this thesis uses the purified distance only between normalized operators, while there also exists a generalized definition of the purified distance between sub-normalized operators~\cite{T5,T11}
\begin{align}
  \label{eq:purified_distance_subnormalized}
  &P\left(\psi,\phi\right)\coloneq\sqrt{1-F_{\ast}^2\left(\psi,\phi\right)},\\
  &F_{\ast}^2\left(\psi,\phi\right)\coloneq{\left({\left\|\sqrt{\psi}\sqrt{\phi}\right\|}_1-\sqrt{\left(1-\tr\psi\right)\left(1-\tr\phi\right)}\right)}^2,\\
  &\psi\geqq 0,\,\tr\psi\leqq 1,\,\phi\geqq 0,\,\tr\phi\leqq 1.
\end{align}
The purified distance satisfies the following properties:
\begin{enumerate}
  \item $0\leqq P\left(\psi,\phi\right)\leqq 1$;
  \item $P\left(\psi,\phi\right)=0\Leftrightarrow \psi=\phi$;
  \item $P\left(\psi,\phi\right)=P\left(\phi,\psi\right)$ (symmetric);
  \item $P\left(\psi,\phi\right)\leqq P\left(\psi,\omega\right)+P\left(\omega,\phi\right)$ (triangle inequality);
  \item $P\left(\mathcal{E}\left(\psi\right),\mathcal{E}\left(\phi\right)\right)\leqq P\left(\psi,\phi\right)$ (monotonicity);
\end{enumerate}
where $\psi$, $\phi$, and $\omega$ are arbitrary states, and $\mathcal{E}$ is any CPTP map~\cite{T5}.
Moreover, for any state $\psi$, $\phi$, and $\omega$,
\begin{equation}
  P\left(\psi\otimes\omega,\phi\otimes\omega\right)=P\left(\psi,\phi\right),
\end{equation}
due to the multiplicativity of the fidelity.
For any $\epsilon\geqq 0$,
two states $\psi$ and $\phi$ are said to be $\epsilon$-close in terms of the fidelity or the purified distance if
\begin{equation}
  \begin{split}
    &F^2\left(\psi,\phi\right)\geqq 1-\epsilon^2\\
    &\Leftrightarrow P\left(\psi,\phi\right)\leqq \epsilon.
  \end{split}
\end{equation}

\paragraph{Koashi-Imoto decomposition}
As for another decomposition, the Koashi-Imoto decomposition~\cite{K3,H6,K5,W4} is introduced in the following.
The Koashi-Imoto decomposition is first introduced in Reference~\cite{K3} to characterize a CPTP map $\mathcal{T}$ leaving any state in a given set $\left\{\psi_i^A\in\mathcal{D}\left(\mathcal{H}^A\right):i\in I\right\}$ invariant.
Note that the index set $I$ can be an infinite set.
The Koashi-Imoto decomposition of a set of states is presented in the following lemma, of which an algorithmic proof is given in Reference~\cite{K3}, and alternative proofs are given in References~\cite{H6,K5} through an operator-algebraic approach.
Note that due to the second condition in the following lemma, the Koashi-Imoto decomposition is \textit{uniquely} determined, corresponding to the decomposition said to be maximal in Reference~\cite{K3}.

\begin{lemma}
\label{lem:koashi_imoto_decomposition_set}
(Theorem~3 in Reference~\cite{K3}, Theorem~9 in Reference~\cite{H6}, and Lemma~6 in Reference~\cite{K5})
    \textit{Koashi-Imoto decomposition of a set of states.}
    Given any set
    \begin{equation}
      {\left\{\psi_i^A\in\mathcal{D}\left(\mathcal{H}^A\right): i\in I\right\}},
    \end{equation}
    there exists a \textit{unique} decomposition of $\mathcal{H}^A$
    \begin{equation}
        \mathcal{H}^A=\bigoplus_{j=0}^{J-1}\mathcal{H}^{a_j^\textup{L}}\otimes\mathcal{H}^{a_j^\textup{R}}
    \end{equation}
    such that
    \begin{enumerate}
      \item For each $i\in I$, $\psi_i^A$ is decomposed into
        \begin{equation}
          \psi_i^A=\bigoplus_{j=0}^{J-1} p\left(j\right) \omega_j^{a_j^\textup{L}}\otimes\phi_{i,j}^{a_j^\textup{R}}\, ,
        \end{equation}
        where $p\left(j\right)$ is a probability distribution and for each $j\in\{0,\ldots,J-1\}$, $\omega_j^{a_j^\textup{L}}\in\mathcal{D}\left(\mathcal{H}^{a_j^\textup{L}}\right)$ is independent of $i$, and $\phi_{i,j}^{a_j^\textup{R}}\in\mathcal{D}\left(\mathcal{H}^{a_j^\textup{R}}\right)$ depends on $i$.
      \item For any CPTP map
        \begin{equation}
          \mathcal{T}:\mathcal{B}\left(\mathcal{H}^A\right)\to\mathcal{B}\left(\mathcal{H}^A\right),
        \end{equation}
        if $\mathcal{T}$ leaves $\psi_i^A$ invariant for each $i\in I$, that is,
        \begin{equation}
          \mathcal{T}\left(\psi_i^A\right)=\psi_i^A,
        \end{equation}
        then the isometry $U_\mathcal{T}$ for the Stinespring dilation of $\mathcal{T}$ is decomposed into
        \begin{equation}
          U_\mathcal{T}=\bigoplus_{j=0}^{J-1} U_j^{a_j^\textup{L}}\otimes\mathbb{1}^{a_j^\textup{R}},
        \end{equation}
        where, for each $j\in\{0,\ldots,J-1\}$, $U_j^{a_j^\textup{L}}$ is an isometry from $\mathcal{H}^{a_j^\textup{L}}$ to $\mathcal{H}^{a_j^\textup{L}}\otimes\mathcal{H}^{A^\prime }$ satisfying
        \begin{equation}
          \tr_{A^\prime }U_\mathcal{T} \omega_j^{a_j^\textup{L}} U_\mathcal{T}^\dag = \omega_j^{a_j^\textup{L}}.
        \end{equation}
    \end{enumerate}
\end{lemma}

Using Lemma~\ref{lem:koashi_imoto_decomposition_set}, Reference~\cite{H6} considers the Koashi-Imoto decomposition of a given bipartite state $\psi^{RA}$.
The Koashi-Imoto decomposition of $\psi^{RA}$ is obtained using a set of $A$'s states that can be \textit{steered} through $\psi^{RA}$, that is, the set of states of $\mathcal{H}^A$ that can be prepared by performing a measurement of $\psi^{RA}$ on $\mathcal{H}^R$ and post-selecting an outcome.
Using an arbitrary positive semidefinite operator $\Lambda^R$,
this set of states is denoted by
\begin{equation}
  \label{eq:psi_lambda}
  \begin{split}
    S_\psi^{A|R}&\coloneq{\left\{\psi^A\left(\Lambda^R\right):\Lambda^R\geqq 0\right\}},\\
    \psi^A\left(\Lambda^R\right)&\coloneq\frac{\tr_R \left[\left(\Lambda^R\otimes\mathbb{1}^A\right)\psi^{RA}\right]}{\tr \left[\left(\Lambda^R\otimes\mathbb{1}^A\right)\psi^{RA}\right]},
  \end{split}
\end{equation}
where the post-selected outcome of a measurement of $\psi^{RA}$ on $\mathcal{H}^R$ corresponds to $\Lambda^R$.
Regard the operator $\Lambda^R$ as the index of the set $S_\psi^{A|R}$, and apply the Koashi-Imoto decomposition of a set of states shown in Lemma~\ref{lem:koashi_imoto_decomposition_set} to this set $S_\psi^{A|R}$, where $\Lambda^R$ for $\psi^A\left(\Lambda^R\right)$ corresponds to the index $i$ for $\psi_i^A$ in Lemma~\ref{lem:koashi_imoto_decomposition_set}, and the set of such positive semidefinite operators $\Lambda^R$ corresponds to the index set $I$.
Then, Reference~\cite{H6} shows that the Koashi-Imoto decomposition of the bipartite state $\psi^{RA}$ is obtained as follows.

\begin{lemma}
\label{lem:koashi_imoto_decomposition_bipartite}
    (in Proof of Theorem~6 in Reference~\cite{H6})
    \textit{Koashi-Imoto decomposition of a bipartite state.}
    Given any bipartite state $\psi^{RA}$,
    the Koashi-Imoto decomposition of the set $S_\psi^{A|R}$ defined as Equation~\eqref{eq:psi_lambda} yields a unique decomposition of $\mathcal{H}^A$ satisfying the conditions in Lemma~\ref{lem:koashi_imoto_decomposition_set}
    \begin{equation}
      \mathcal{H}^A=\bigoplus_{j=0}^{J-1}\mathcal{H}^{a_j^\textup{L}}\otimes\mathcal{H}^{a_j^\textup{R}},
    \end{equation}
    and $\psi^{RA}$ is decomposed into
    \begin{equation}
      \psi^{RA}=\bigoplus_{j=0}^{J-1} p\left(j\right) \omega_j^{a_j^\textup{L}}\otimes\phi_j^{Ra_j^\textup{R}},
    \end{equation}
    where $p\left(j\right)$ is a probability distribution.
\end{lemma}

Considering a purification $\Ket{\psi}^{RAB}$ of the bipartite state $\psi^{RA}$ in Lemma~\ref{lem:koashi_imoto_decomposition_bipartite},
Reference~\cite{W4} introduces the Koashi-Imoto decomposition of the tripartite pure state $\Ket{\psi}^{RAB}$ as follows.

\begin{lemma}
\label{lem:koashi_imoto_decomposition_tripartite}
  (Lemma~11 in Reference~\cite{W4})
  \textit{Koashi-Imoto decomposition of a tripartite pure state.}
  Given any tripartite pure state $\Ket{\psi}^{RAB}$,
  the Koashi-Imoto decomposition of the set $S_\psi^{A|R}$ defined as Equation~\eqref{eq:psi_lambda} yields a unique decomposition of $\mathcal{H}^A$ satisfying the conditions in Lemma~\ref{lem:koashi_imoto_decomposition_set}
  \begin{equation}
    \mathcal{H}^A=\bigoplus_{j=0}^{J-1}\mathcal{H}^{a_j^\textup{L}}\otimes\mathcal{H}^{a_j^\textup{R}}
  \end{equation}
  such that the support of $\psi^B$
  \begin{equation}
    \supp\left(\psi^B\right)\coloneq\spn\left\{\Ket{v}:\psi^B\Ket{v}\neq 0\right\}
  \end{equation}
  is decomposed into
  \begin{equation}
    \supp\left(\psi^B\right)=\bigoplus_{j=0}^{J-1}\mathcal{H}^{b_j^\textup{L}}\otimes\mathcal{H}^{b_j^\textup{R}},
  \end{equation}
  and $\Ket{\psi}^{RAB}$ is decomposed into
  \begin{equation}
    \Ket{\psi}^{RAB}=\bigoplus_{j=0}^{J-1}\sqrt{p\left(j\right)}\Ket{\omega_j}^{a_j^\textup{L} b_j^\textup{L}}\otimes\Ket{\phi_j}^{R a_j^\textup{R} b_j^\textup{R}},
  \end{equation}
  where $p\left(j\right)$ is a probability distribution.
\end{lemma}

Consequently,
to obtain the Koashi-Imoto decomposition of a given pure state $\Ket{\psi}^{RAB}$,
apply the algorithm presented in Reference~\cite{K3} or the operator-algebraic theorems used in References~\cite{H6,K5}
to the set of states $S_\psi^{A|R}$ defined as Equation~\eqref{eq:psi_lambda},
and then follow the above argument.
The former way of applying the algorithm in Reference~\cite{K3} is demonstrated in Appendix~\ref{sec:koashi_imoto} for concrete examples.

\section{\label{sec:entanglement}Entanglement as a resource for distributed quantum information processing}

This section summarizes examples of state transformations implementable by LOCC assisted by entangled states relevant to this thesis.
These examples show that entanglement can be used as a resource when spatially separated parties are restricted to LOCC\@.
After introducing some results on entanglement transformations under LOCC, this section also provides the notion of entanglement measures quantifying entanglement in terms of its value as a resource.

\paragraph{Entanglement swapping and quantum teleportation}
To begin, the following example is shown for demonstrating a protocol for transforming entangled states by LOCC, which is called \textit{entanglement swapping}~\cite{Y17,Z3}.
Entanglement swapping involves three parties $R$, $A$, and $B$, and systems $\mathcal{H}^R$ of $R$, $\mathcal{H}^{A}\otimes\mathcal{H}^{A^\prime}$ of $A$, and $\mathcal{H}^B$ of $B$, where dimensions of these systems are set to be
\begin{equation}
  D\coloneq\dim\mathcal{H}^R=\dim\mathcal{H}^A=\dim\mathcal{H}^{A^\prime}=\dim\mathcal{H}^B.
\end{equation}
Consider an entangled state of $\mathcal{H}^{R}\otimes\mathcal{H}^A$ with Schmidt rank $D$ shared between $R$ and $A$ defined as
\begin{equation}
  \Ket{\Phi_D^+}^{RA}\coloneq\frac{1}{\sqrt{D}}\sum_{l=0}^{D-1}\Ket{l}^{R}\otimes\Ket{l}^A,
\end{equation}
where ${\left\{\Ket{l}^R\right\}}_l$ and ${\left\{\Ket{l}^A\right\}}_l$ are the computational bases,
and the same form of entangled state $\Ket{\Phi_D^+}^{A^\prime B}\in\mathcal{H}^{A^\prime}\otimes\mathcal{H}^{B}$ shared between $A$ and $B$.
The whole state of $\mathcal{H}^R\otimes\mathcal{H}^{A}\otimes\mathcal{H}^{A^\prime}\otimes\mathcal{H}^B$ is
\begin{equation}
  \Ket{\Phi_D^+}^{RA}\otimes\Ket{\Phi_D^+}^{A^\prime B},
\end{equation}
and the reduced state of $R$ and $B$ is
\begin{equation}
  \frac{\mathbb{1}^R}{D}\otimes\frac{\mathbb{1}^B}{D},
\end{equation}
which is a separable state of $\mathcal{H}^R\otimes\mathcal{H}^B$ shared between $R$ and $B$.
This type of state proportional to $\mathbb{1}$ is called a completely mixed state.
Entanglement swapping aims to prepare an entangled state between $R$ and $B$ by LOCC assisted by these entangled states shared between $R$ and $A$, and between $A$ and $B$.

Consider $A$'s measurement on $\mathcal{H}^{A}\otimes\mathcal{H}^{A^\prime}$ in the basis
\begin{equation}
  \label{eq:max_basis}
  \left\{\left(\mathbb{1}^{A}\otimes {\left(X_D^{A^\prime}\right)}^l{\left(Z_D^{A^\prime}\right)}^{l^\prime}\right)\Ket{\Phi_D^+}^{AA^\prime}:l,l^\prime\in\left\{0,\ldots,D-1\right\}\right\},
\end{equation}
where the measurement outcome is labeled by $l$ and $l^\prime$, and $X_D^{A^\prime}$ and $Z_D^{A^\prime}$ are the generalized Pauli operators on a $D$-dimensional Hilbert space $\mathcal{H}^{A^\prime}$ defined as
\begin{align}
  X_{D}^{A^\prime}&\coloneq\sum_{l=0}^{D-1}\Ket{l+1\bmod D}\Bra{l}^{A^\prime},\\
  Z_{D}^{A^\prime}&\coloneq\sum_{l=0}^{D-1}\exp\left(\frac{\textup{i}2\pi l}{D}\right)\Ket{l}\Bra{l}^{A^\prime}.
\end{align}
In the case of qubits, \textit{i.e.}, $D=2$,
subscripts of the generalized Pauli operators may be omitted to simply write these operators as
\begin{align}
  X^{A^\prime}\coloneq X_2^{A^\prime},\\
  Z^{A^\prime}\coloneq Z_2^{A^\prime},
\end{align}
and the states in the above basis of $\mathcal{H}^{A}\otimes\mathcal{H}^{A^\prime}=\mathbb{C}^2\otimes\mathbb{C}^2$ reduce to
\begin{align}
  \Ket{\Phi_2^+}^{AA^\prime}&=\frac{1}{\sqrt{2}}\left(\Ket{0}\otimes\Ket{0}+\Ket{1}\otimes\Ket{1}\right),\\
  Z^{A^\prime}\Ket{\Phi_2^+}^{AA^\prime}&=\frac{1}{\sqrt{2}}\left(\Ket{0}\otimes\Ket{0}-\Ket{1}\otimes\Ket{1}\right),\\
  X^{A^\prime}\Ket{\Phi_2^+}^{AA^\prime}&=\frac{1}{\sqrt{2}}\left(\Ket{0}\otimes\Ket{1}+\Ket{1}\otimes\Ket{0}\right),\\
  X^{A^\prime}Z^{A^\prime}\Ket{\Phi_2^+}^{AA^\prime}&=\frac{1}{\sqrt{2}}\left(\Ket{0}\otimes\Ket{1}-\Ket{1}\otimes\Ket{0}\right).
\end{align}
For any $M^{A^\prime}\in\mathcal{B}\left(\mathcal{H}^{A^\prime}\right)$,
the state $\Ket{\Phi_D^+}^{A^\prime B}$ satisfies
\begin{equation}
  \left(M^{A^\prime}\otimes\mathbb{1}^B\right)\Ket{\Phi_D^+}^{A^\prime B}=\left(\mathbb{1}^{A^\prime}\otimes {\left(M^B\right)}^\textup{T}\right)\Ket{\Phi_D^+}^{A^\prime B},
\end{equation}
where ${\left(M^B\right)}^\textup{T}$ represents the transpose of $M^B$ with respect to the computational basis of $\mathcal{H}^B$.
Thus, the state
\begin{equation}
  \Ket{\Phi_D^+}^{RA}\otimes\Ket{\Phi_D^+}^{A^\prime B}
\end{equation}
is transformed by this measurement into a state proportional to
\begin{equation}
  \begin{split}
    &\left({\left(X_D^{A^\prime}\right)}^l{\left(Z_D^{A^\prime}\right)}^{l^\prime}\Ket{\Phi_D^+}\Bra{\Phi_D^+}^{AA^\prime}{\left(Z_D^{A^\prime}\right)}^{l^\prime}{\left(X_D^{A^\prime}\right)}^l\right)\left(\Ket{\Phi_D^+}^{RA}\otimes\Ket{\Phi_D^+}^{A^\prime B}\right)\\
    &\propto\left({\left(X_D^{A^\prime}\right)}^l{\left(Z_D^{A^\prime}\right)}^{l^\prime}\Ket{\Phi_D^+}^{AA^\prime}\right)\otimes \left({\left(X_D^B\right)}^l{\left(Z_D^B\right)}^{l^\prime}\Ket{\Phi_D^+}^{RB}\right),
  \end{split}
\end{equation}
which can be shown by using
\begin{align}
  {\left(X_{D}^{B}\right)}^\textup{T}&={X_{D}^{B}},\\
  {\left(Z_{D}^{B}\right)}^\textup{T}&={Z_{D}^{B}}.
\end{align}
Therefore, performing classical communication of $A$'s measurement outcome $l$ and $l^\prime$ from $A$ to $B$, followed by $B$'s unitary transformation ${\left(Z_D^B\right)}^{l^\prime}{\left(X_D^B\right)}^l$ conditioned by $l$ and $l^\prime$ for correction,
the parties can prepare an entangled state $\Ket{\Phi_D^+}^{RB}$ between $R$ and $B$ by LOCC assisted by the initially shared entangled states $\Ket{\Phi_D^+}^{RA}\otimes\Ket{\Phi_D^+}^{A^\prime B}$, while the reduced state initially shared between $R$ and $B$ is not entangled.
This protocol achieves entanglement swapping.
Note that no operation is performed on $\mathcal{H}^R$ throughout the protocol.

This protocol for entanglement swapping can be considered as LOCC performed by $A$ and $B$ assisted by an entangled state $\Ket{\Phi_D^+}^{A^\prime B}$ for transferring $A$'s part of $\Ket{\Phi_D^+}^{RA}$ from $A$ to $B$, keeping coherence between $R$ and $B$ to obtain $\Ket{\Phi_D^+}^{RB}$.
If $A$ and $B$ performing LOCC assisted by $\Ket{\Phi_D^+}^{A^\prime B}$ achieve a CPTP map $\mathcal{E}^{A\to B}$ transferring $A$'s part of $\Ket{\Phi_D^+}^{RA}$ from $A$ to $B$, that is,
\begin{equation}
  \left(\id^R\otimes\mathcal{E}^{A\to B}\right)\left({\Phi_D^+}^{RA}\right)={\Phi_D^+}^{RB},
\end{equation}
then, due to the linearity of the CPTP map, the same CPTP map $\mathcal{E}^{A\to B}$ can transfer an arbitrary state given from $\mathcal{D}\left(\mathcal{H}^{A}\right)$, that is,
\begin{equation}
  \mathcal{E}^{A\to B}\left(\psi^{A}\right)=\psi^B,\quad\forall\psi^{A}\in\mathcal{D}\left(\mathcal{H}^{A}\right),
\end{equation}
and \textit{vice versa}.
The protocol for transferring an arbitrary state on $\mathcal{H}^{A}$ by entanglement-assisted LOCC is known as \textit{quantum teleportation}~\cite{B5}.
This equivalence between transferring $A$'s part of $\Ket{\Phi_D^+}^{RA}$ from $A$ to $B$ and transferring arbitrary states of $\mathcal{H}^{A}$ is known as the relative state method~\cite{P3}, where in the former case, $R$ is regarded as reference on which neither $A$ nor $B$ can perform any operation, and $A$ and $B$ keeps coherence between $R$ and $AB$.

Quantum teleportation simulates noiseless quantum communication transferring an arbitrary state $\psi^A$ of a $D$-dimensional system $\mathcal{H}^A$ from $A$ to $B$ by LOCC assisted by shared entanglement in the form of $\Ket{\Phi_D^+}$.
Conversely, $\Ket{\Phi_D^+}$ shared between $A$ and $B$ can be prepared by quantum communication, where such a protocol can be $A$'s preparing $\Ket{\Phi_D^+}^{AA^\prime}$ by local operations, followed by transferring a part of this bipartite state corresponding to $\mathcal{H}^{A^\prime}$ from $A$ to $B$ by quantum communication.
Thus, when LOCC can be freely performed, shared entanglement and quantum communication can be used as an equivalent resource for assisting LOCC\@.
The entangled state in such a form with the minimal Schmidt rank, that is, $\Ket{\Phi_2^+}$, can be used as a basic unit of entanglement and called an \textit{ebit}.

\paragraph{Quantum state transformation by LOCC}
Given that entanglement may serve as a resource assisting LOCC for performing distribute quantum information processing, it is natural to analyze which entangled state has more capability as a resource than others under LOCC\@.
For two states $\phi$ and $\psi$ shared among $N$ parties, if there exists an LOCC map $\mathcal{E}_\textup{LOCC}$ achieving
\begin{equation}
  \mathcal{E}_\textup{LOCC}\left(\phi\right)=\psi,
\end{equation}
then $\phi$ is said to be convertible, or transformable, into $\psi$ by LOCC\@.
This state transformation by LOCC is denoted by
\begin{equation}
  \phi\xrightarrow{\textup{LOCC}}\psi.
\end{equation}
In this case, $\phi$ can be considered to have more capability as a resource for assisting LOCC than $\psi$.
If such a resource state having more capability, such as $\phi$ in the above case, is shared among parties, the parties may transform the shared resource state by LOCC into another suitable form, such as $\psi$, for assisting LOCC\@.

This paradigm yields a \textit{common resource state}~\cite{S18,G2} transformable into any state in a given set, that is, a resource state having more capability than any state in the set.
This set of states is called the \textit{target set} in the context of common resource states.
Common resource states are assumed to be \textit{fully entangled}, that is, entangled with respect to any bipartition of the parties.
More formally, given any target set $S$, a fully entangled state $\Ket{\phi}$ is called a common resource state for $S$ if for any $\psi\in S$, it holds that
\begin{equation}
  \phi\xrightarrow{\textup{LOCC}}\psi.
\end{equation}
Similarly, Reference~\cite{M6} also introduces common resource states in terms of state convertibility by stochastic LOCC, that is, performing an LOCC measurement followed by post-selecting an outcome.

For bipartite pure states of $\mathcal{H}^A\otimes\mathcal{H}^B$ shared between $A$ and $B$, convertibility of the states under LOCC is characterized by \textit{majorization}, as summarized in the following.
Consider two states having Schmidt decomposition
\begin{align}
  \Ket{\phi}^{AB}&\coloneq\sum_{l=0}^{R_\phi-1}\sqrt{\lambda_l^\phi}\Ket{\phi_l}^A\otimes\Ket{\phi_l}^B,\\
  \Ket{\psi}^{AB}&\coloneq\sum_{l=0}^{R_\psi-1}\sqrt{\lambda_l^\psi}\Ket{\psi_l}^A\otimes\Ket{\psi_l}^B.
\end{align}
Reduced states on $\mathcal{H}^A$ of these states are
\begin{align}
  \phi^A&=\sum_{l=0}^{R_\phi-1}\lambda_l^\phi \Ket{\phi_l}\Bra{\phi_l}^A,\\
  \psi^A&=\sum_{l=0}^{R_\psi-1}\lambda_l^\psi \Ket{\psi_l}\Bra{\psi_l}^A.
\end{align}
A Hermitian operator $\phi^A$ is said to be \textit{majorized} by a Hermitian operator $\psi^A$, which is denoted by
\begin{equation}
  \phi^A\prec\psi^A,
\end{equation}
if there exists a CPTP map $\mathcal{U}^A$ in the following form called a mixed unitary channel
\begin{equation}
  \begin{split}
    \phi^A&=\mathcal{U}^A\left(\psi^A\right)\\
          &\coloneq\sum_{j}p\left(j\right)U_j^A\psi^A{U_j^A}^\dag
  \end{split}
\end{equation}
where $p\left(j\right)$ is a probability distribution, and $U_j^A$ for each $j$ is a unitary operator on $\mathcal{H}^A$.
Such a probability distribution introduces randomness,
and hence,
if two quantum states $\phi$ and $\psi$ satisfy $\phi^A\prec\psi^A$,
$\phi$ can be considered to be a more randomized state than $\psi$ in the sense that $\phi$ can be obtained from $\psi$ by a mixed unitary channel.
To investigate properties of mixed unitary channels, it is useful to consider a unital channel $\tilde{\mathcal{U}}^A$ defined as a channel transforming the identity to the identity, that is,
\begin{equation}
  \tilde{\mathcal{U}}^A\left(\mathbb{1}^A\right)=\mathbb{1}^A.
\end{equation}
For qubits, a channel is a mixed unitary channel if and only if it is a unital channel.
But in general, any mixed unitary channel is a unital channel, but not \textit{vice versa}~\cite{L4,W11}.

The majorization condition of Hermitian operators can also be represented in terms of real vectors defined using eigenvalues of $\phi^A$ and $\psi^A$.
For a Hermitian operator $\phi^A$, define a real vector of $\dim\mathcal{H}^A$ elements
\begin{equation}
  \boldsymbol{\lambda}\left(\phi^A\right)\coloneq\left(\lambda_0^\phi,\ldots,\lambda_{R_\phi-1}^\phi,0,\ldots,0\right),
\end{equation}
where the first $R_\phi$ elements are eigenvalues of $\phi^A$ in descending order, and the rest is filled with zero.
Let $\boldsymbol{\lambda}\left(\psi^A\right)$ denote a real vector of $\dim\mathcal{H}^A$ elements defined for $\psi^A$ in the same way.
The zero elements of $\boldsymbol{\lambda}\left(\phi^A\right)$ and $\boldsymbol{\lambda}\left(\psi^A\right)$ may be denoted by
\begin{align}
  \lambda_{R_\phi}^{\phi}&=\lambda_{R_\phi+1}^{\phi}=\cdots=\lambda_{D-1}^{\phi}=0,\\
  \lambda_{R_\psi}^{\psi}&=\lambda_{R_\psi+1}^{\psi}=\cdots=\lambda_{D-1}^{\psi}=0.
\end{align}
A real vector $\boldsymbol{\lambda}\left(\phi^A\right)$ of $D$ elements is said to be \textit{majorized} by $\boldsymbol{\lambda}\left(\psi^A\right)$, which is denoted by
\begin{equation}
  \boldsymbol{\lambda}\left(\phi^A\right)\prec\boldsymbol{\lambda}\left(\psi^A\right),
\end{equation}
if it holds that
\begin{align}
  \sum_{l=0}^{m}\lambda_l^\phi&\leqq\sum_{l=0}^{m}\lambda_l^\psi,\quad\forall m\in\left\{0,\ldots,D-2\right\},\\
  \sum_{l=0}^{D-1}\lambda_l^\phi&=\sum_{l=0}^{D-1}\lambda_l^\psi,
\end{align}
where the elements of $\boldsymbol{\lambda}\left(\phi\right)$ and $\boldsymbol{\lambda}\left(\psi\right)$ are assumed to be in descending order.
A Hermitian operator $\phi^A$ is majorized by a Hermitian operator $\psi^A$ if and only if the corresponding real vector $\boldsymbol{\lambda}\left(\phi\right)$ is majorized by $\boldsymbol{\lambda}\left(\psi\right)$.

Using majorization, convertibility between two bipartite pure states by LOCC is characterized as follows.

\begin{lemma}
\label{lem:pure_convertibility}
  (Reference~\cite{N2})
  \textit{Convertibility between bipartite pure states by LOCC.}
  For any two bipartite pure states $\Ket{\phi}^{AB}$ and $\Ket{\psi}^{AB}$,
  \begin{equation}
    \phi^{AB}\xrightarrow{\textup{LOCC}}\psi^{AB}
  \end{equation}
  if and only if
  \begin{equation}
    \phi^A\prec\psi^A.
  \end{equation}
\end{lemma}

This characterization of transformations between bipartite pure states under LOCC generalizes to transformations of a bipartite pure state into a bipartite mixed state as follows.

\begin{lemma}
\label{lem:mixed}
  (Reference~\cite{J1})
  \textit{Convertibility of a bipartite pure state into a bipartite mixed by LOCC.}
  For a bipartite pure state $\Ket{\phi}^{AB}$ and a bipartite mixed state $\psi^{AB}$,
  \begin{equation}
    \phi^{AB}\xrightarrow{\textup{LOCC}}\psi^{AB}
  \end{equation}
  if and only if
  \begin{equation}
    \boldsymbol{\lambda}\left(\phi^A\right)\prec\min\sum_j p(j)\boldsymbol{\lambda}\left(\psi_j^A\right),
  \end{equation}
  where the minimization is taken over any ensemble ${\left\{p(j),\Ket{\psi_j}^{AB}\right\}}_j$ of pure states which are not necessarily orthogonal to each other and satisfy
  \begin{equation}
    \psi^{AB}=\sum_j p(j)\Ket{\psi_j}\Bra{\psi_j}^{AB}.
  \end{equation}
\end{lemma}

Convertibility of bipartite states under LOCC establishes partial order of entangled states in terms of capability as a resource,
in the sense that the set of quantum states can be regarded as a partially ordered set if a relation between two quantum states defined according to whether one state can be convertible into the other by LOCC is considered as the partial order of this set.
Given a bipartite system $\mathcal{H}^A\otimes\mathcal{H}^B$,
entanglement in this partial order can be quantified by a function $E:\mathcal{D}\left(\mathcal{H}^A\otimes\mathcal{H}^B\right)\to\mathbb{R}$ satisfying
\begin{equation}
  \begin{split}
    &\phi^{AB}\xrightarrow{\textup{LOCC}}\psi^{AB}\\
    &\Rightarrow E\left(\phi^{AB}\right)\geqq E\left(\psi^{AB}\right).
  \end{split}
\end{equation}
and a function having this property is called an \textit{entanglement measure}.
Note that one can additionally impose other properties to identify theoretically tractable entanglement measures, such as $E\left(\psi^{AB}\right)=0$ for any separable state $\psi$, as reviewed in References~\cite{H2,P1,E5}.
Various entanglement measures are known in bipartite cases, such as distillable entanglement~\cite{B3}, entanglement cost~\cite{B3,H1}, relative entropy of entanglement~\cite{V7}, and squashed entanglement~\cite{C21},
and these entanglement measures coincide for any pure state $\Ket{\psi}^{AB}$ with the entanglement entropy defined as
\begin{equation}
  -\sum_{l=0}^{R_\psi-1}\lambda_l^\psi\log_2\lambda_l^\psi,
\end{equation}
where $\lambda_l^\psi$ for each $l$ corresponds to a Schmidt coefficient appearing in the Schmidt decomposition
\begin{equation}
  \Ket{\psi}^{AB}=\sum_{l=0}^{R_\psi-1}\sqrt{\lambda_l^\psi}\Ket{\psi_l}^A\otimes\Ket{\psi_l}^B.
\end{equation}
A basic unit of these entanglement measures is \textit{ebit}, that is, the entanglement entropy of $\Ket{\Phi_2^+}$.
For bipartite pure states, the Schmidt rank is also monotonically nonincreasing under LOCC~\cite{L2}, and this property is referred to as the LOCC monotonicity of the Schmidt rank.
Hence, the Schmidt rank, or its generalization to mixed states~\cite{T1}, can be regarded as an entanglement measure, while they are discrete.

As a special case of local operations,
a unitary transformation on $\mathcal{H}^A\otimes\mathcal{H}^B$ in the form of $U^A\otimes U^B$ is called a \textit{local unitary} transformation.
Given two states $\phi^{AB}$ and $\psi^{AB}$, consider a case where there exists a local unitary transformation $U^A\otimes U^B$ such that
\begin{equation}
  \phi^{AB}=\left(U^A\otimes U^B\right)\psi^{AB}\left({U^A}^\dag\otimes {U^B}^\dag\right).
\end{equation}
In this case, $\phi^{AB}$ and $\psi^{AB}$ are said to be \textit{locally unitarily equivalent}.
Since unitary transformations are invertible, for any locally unitarily equivalent states $\phi^{AB}$ and $\psi^{AB}$ and any entanglement measure $E$, it holds that
\begin{equation}
  E\left(\phi^{AB}\right)=E\left(\psi^{AB}\right).
\end{equation}

Lemmas~\ref{lem:pure_convertibility} and~\ref{lem:mixed} imply that for any state $\psi^{AB}$, there exists an LOCC map achieving
\begin{equation}
  \Ket{\Phi_D^+}^{AB}\xrightarrow{\textup{LOCC}}\psi^{AB},
\end{equation}
where
\begin{equation}
  D=\min\left\{\dim\mathcal{H}^A,\dim\mathcal{H}^B\right\}.
\end{equation}
Hence, $\Ket{\Phi_D^+}^{AB}$ and its locally unitarily equivalent states maximize any entanglement measure $E$,
and in this sense, $\Ket{\Phi_D^+}^{AB}$ and its locally unitarily equivalent states are called \textit{maximally entangled states}.
Maximally entangled states of two qubits, that is, $\Ket{\Phi_2^+}^{AB}$ and its locally unitarily equivalent states, are called Bell states.
The maximally entangled state of a bipartite system is unique up to these local unitary transformations.

In contrast to these well-established results on bipartite entanglement,
properties of multipartite entanglement are more involved~\cite{E2,W3,B26}.
For a multipartite system in general, there may not exist a single maximally entangled state in the multipartite system itself transformable by LOCC into all the states in the system~\cite{V1,S1,S2,M1,G1,S3}.
In particular, given a multipartite system where each local dimension is $d$, almost no LOCC transformation among pure states of the system is possible~\cite{G1,S3}.
Due to these facts, applicability of resource-theoretic analysis based on state convertibility under LOCC is limited if multipartite entanglement is concerned.
In contrast, the analysis of multipartite entanglement in Part~\ref{part:1} and~\ref{part:2} adopt a different perspective, based on settings relevant to distributed quantum information processing where the parties can be restricted to having small- and intermediate-scale quantum systems of up to several dozens of qubits and connected by a network for quantum communication.
Part~\ref{part:1} analyzes a fundamental communication task, quantum state merging~\cite{H3,H4}, under such small- and intermediate-scale settings, and Part~\ref{part:2} analyzes manipulation of multipartite entanglement on networks.

\part{\label{part:1}One-shot quantum state merging on small and intermediate scales under one-way and two-way communication}

\chapter{Background and overview of Part~\ref{part:1}}

Quantum state merging~\cite{H3,H4} is a communication task playing crucial roles in distributed quantum information processing~\cite{W8,W9,W10,W12} and multipartite entanglement transformations~\cite{A3,D8,Y8,D9,S4}.
Quantum state merging, or quantum state redistribution~\cite{D2,D3} as a generalized task including state merging,
was originally introduced in the context of quantum Shannon theory, and they have also applied to the analyses of a family of other quantum communication tasks in quantum Shannon theory, such as derivation of a capacity of noisy quantum channels~\cite{D4,A7,H5,A2,H11,A8,P3,W5}.
In the task of state merging originally formulated using the framework of local operations and classical communication (LOCC)~\cite{H3,H4},
two spatially separated parties $A$ and $B$ initially share an entangled resource state and are given $n$ shared states whose purification with reference $R$ is represented as ${\left(\Ket{\psi}^{RAB}\right)}^{\otimes n}$, where $A$ and $B$ know classical description of $\Ket{\psi}^{RAB}$.
The goal of the task is to asymptotically transfer $A$'s part of $\Ket{\psi}^{RAB}$ from $A$ to $B$ and obtain $\Ket{\psi}^{RB'B}$,
keeping coherence between $B$ and $R$, by $A$ and $B$'s LOCC assisted by shared entanglement within an error in fidelity approaching to zero as $n \rightarrow \infty$.
This type of scenario of achieving a task for infinitely many times within a vanishing error is called the asymptotic scenario.
When $A$ and $B$ are initially given a shared maximally entangled resource state in addition to $\Ket{\psi}^{RAB}$,
quantum communication can be simulated by LOCC assisted by this maximally entangled resource state by means of quantum teleportation~\cite{B5}.
Given a protocol for state merging, the amount of this shared entanglement required for the protocol, or equivalently, that of quantum communication when LOCC is free, is called \textit{entanglement cost} of the protocol, regarded as the cost to be minimized.

It is an essential feature of state merging that the parties may exploit classical description of the initially given states $\Ket{\psi}^{RAB}$ for reducing entanglement cost required for the protocols.
Without classical description,
there exists a trivial protocol achieving state merging by quantum teleportation~\cite{B5} for transferring $A$'s part of $\Ket{\psi}^{RAB}$ from $A$ to $B$.
This trivial protocol does not require the classical description, and as the result, it requires the same entanglement cost for any given state.
In contrast, entanglement cost in state merging can be reduced compared to quantum teleportation and can even be negative when the protocol provides a net gain of shared entanglement.

Quantum state merging can also be regarded as an analogue of source coding with decoder's side information in classical information theory established by Slepian and Wolf~\cite{S7}.
Reference~\cite{S7} introduces and analyzes a situation involving three parties $A_1\,$, $A_2\,$, and $B$, where each of $A_1$ and $A_2$ is given classical information that is correlated with the other's, and the classical information of $A_1$ and $A_2$ is to be transferred to $B$.
Then, Reference~\cite{S7} characterizes the minimal amount of classical communication from $A_1$ and $A_2$ to $B$ required for achieving this task.
If all of the $A_1$'s classical information is first transferred to $B$, this classical information possibly correlated with $A_2$'s is called \textit{side information} at $B$, which can be used for reducing the amount of classical communication for transferring the rest of classical information from $A_2$ to $B$ compared to the case without this side information.
In quantum state merging, if $B$'s part of $\Ket{\psi}^{RAB}$ is correlated with $A$'s, $B$'s part may contribute to reducing entanglement cost required for transferring $A$'s, compared to the cases without $B$'s part.
Such a task with $B$'s ability to use a part of the shared quantum state is called a task with \textit{quantum side information} at $B$.
Similar notions of quantum side information are also widely used in the contexts other than state merging, such as entropic uncertainty relations~\cite{C12}, state exchange~\cite{Y16}, and classical-quantum Slepian-Wolf problems~\cite{D10,T5,R3,T9,L1,M5,C11}.
Properties of $B$'s quantum side information in state merging can be quantitatively captured in terms of entanglement cost in state merging.
In the asymptotic scenario of state merging, the minimal entanglement cost in state merging is given by the conditional quantum entropy ${H\left(A|B\right)}_{\psi}$ per copy~\cite{H3,H4,B20}, which clarifies an operational meaning of the conditional quantum entropy.

While the asymptotic scenario is well-established in quantum Shannon theory, there have also been studied zero-error scenarios~\cite{G5}, which are originally established in a classical setting by Shannon~\cite{S6} and first introduced into a quantum setting in Reference~\cite{M2}.
Regarding the zero-error scenarios of classical source coding with decoder's side information, optimal zero-error code design is proven to be $NP$-hard~\cite{K7}.
However, in classical coding theory, explicit construction of zero-error coding protocols such as Shannon coding~\cite{S9} and Huffman coding~\cite{H9}, if not necessarily optimal, establishes a foundation of theoretical analyses as well as practical applications.
In this direction, explicit zero-error coding protocols for classical source coding with decoder's side information are given in References~\cite{K7,W7,J2,Y11,Z2,M3,M4}.

Aside from this regime where infinitely many copies of $\Ket{\psi}^{RAB}$ are given, another regime is the one-shot regime where only a single copy of $\Ket{\psi}^{RAB}$ is given.
The scenarios in the one-shot regime can also be classified into two scenarios: one is an exact scenario with zero error, and the other is an approximate scenario in which a nonzero error is tolerated for reducing entanglement cost.
Analysis in the one-shot regime clarifies the structure of protocols achieving the communication tasks at a single-copy level and is more relevant to practical situations such as distributed quantum information processing.
In addition to the asymptotic scenario,
state merging and redistribution have been defined and analyzed in various one-shot scenarios~\cite{Y9,B9,B12,D7,D6,H10,B10,D5,M,N3,A4,A5,B15,B13,A16,A17}.
There also exist other derivatives of state merging and redistribution in modified settings~\cite{B1,B2,S11,S12,S14,A10,A11}.

In this part, after providing preliminaries in Chapter~\ref{sec:preliminaries_1},
the following two results on one-shot state merging are presented in Chapters~\ref{sec:merge} and~\ref{sec:two_way},
aiming at investigating entanglement cost of transferring quantum information of unknown states on the small and intermediate scales.
Chapter~\ref{sec:merge} constructs protocols achieving one-shot state merging even on small and intermediate scales, as well as analyzing lower bounds of the minimal achievable entanglement cost.
While the protocols constructed in Chapter~\ref{sec:merge} use only one-way LOCC from $A$ to $B$,
Chapter~\ref{sec:two_way} discusses advantage of two-way LOCC over one-way LOCC in one-shot state merging.

\section*{Quantum state merging for arbitrarily small-dimensional systems}

\begin{figure}[t!]
  \centering
  \includegraphics[width=4in]{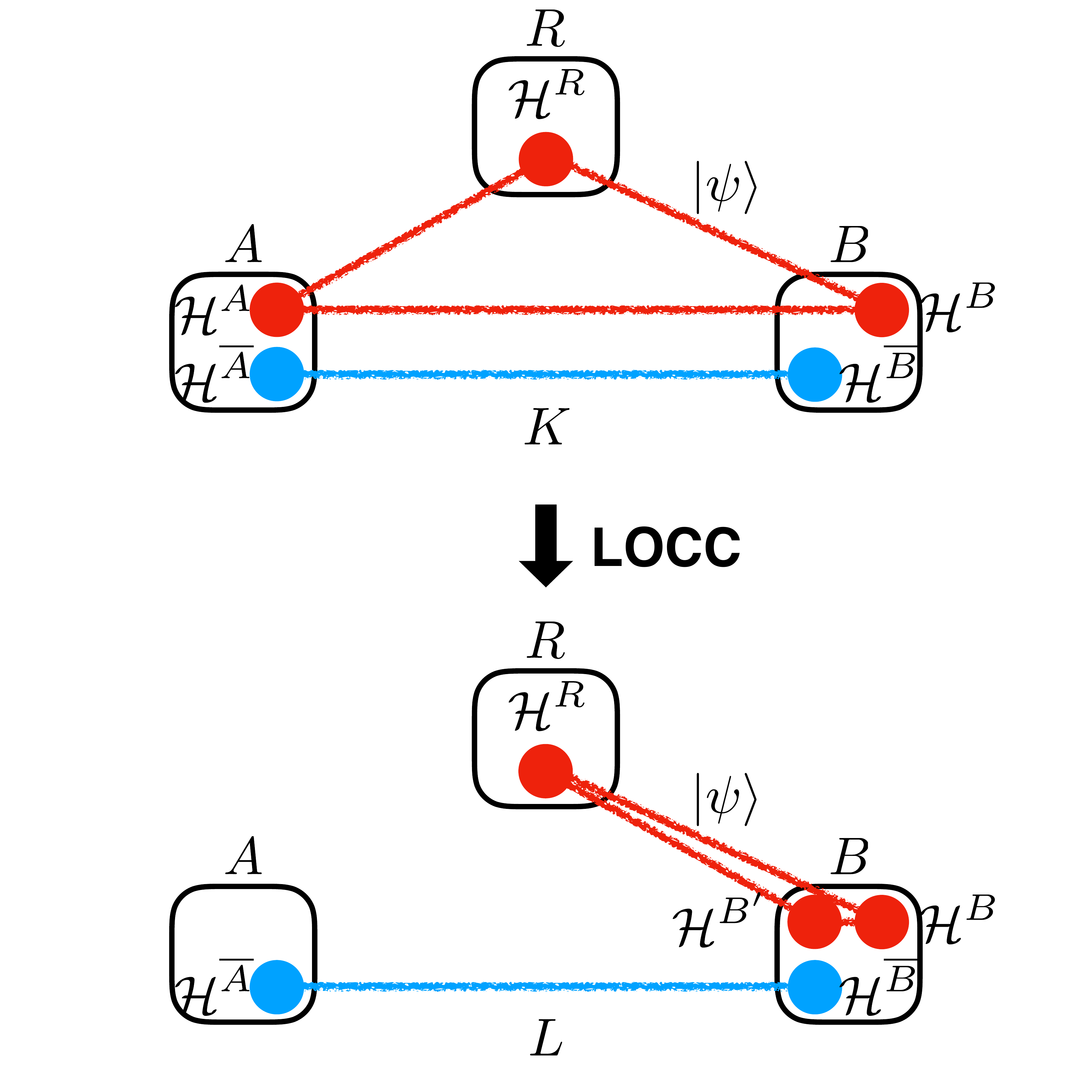}
  \caption[Exact state merging.]{\label{fig:merge}Exact state merging of a given state $\Ket{\psi}^{RAB}$ denoted by the red circles. Parties $A$ and $B$ perform LOCC assisted by a maximally entanglement resource state $\Ket{\Phi_K^+}^{\overline{AB}}$ with the Schmidt rank $K$ denoted by the top blue circles to transfer the reduced state of $\Ket{\psi}^{RAB}$ on $A$ to $B$ and obtain $\Ket{\psi}^{RB'B}$ while $\Ket{\Phi_L^+}^{\overline{AB}}$ with the Schmidt rank $L$ denoted by the bottom blue circles is also obtained.}
\end{figure}

Chapter~\ref{sec:merge} investigates general bounds of entanglement cost required for quantum state merging on the small and intermediate scales relevant to distributed quantum information processing.
The existing protocols for one-shot quantum state merging or redistribution~\cite{B9,Y9,B12,D7,D6,H10,B10,D5,M,N3,A4,A5,B15,B13,A16,A17} achieve near optimality only on a large scale relevant to \textit{one-shot quantum information theory}
where functions of states called the smooth conditional min- and max-entropies~\cite{R2,T5} are used for evaluating entanglement cost.
Definitions of these functions are summarized in Appendix~\ref{sec:one_shot_entropies}.
These protocols also cause a nonzero approximation error in fidelity, since the vital techniques for these protocols, namely, one-shot decoupling~\cite{D6} and the convex-split lemma~\cite{A4}, cannot avoid errors.
As higher fidelity is pursued in state merging of a fixed single copy of $\Ket{\psi}^{RAB}$, entanglement cost required for the protocols diverges to infinity. Hence, there always exists a region of error close to zero where the protocols do not contribute to reducing the entanglement cost.
Moreover,
in cases where the system size for the reduced state of $\Ket{\psi}^{RAB}$ on $A$ is as small as up to a few dozens of qubits,
the protocols require more entanglement cost than quantum teleportation even if the error tolerance is reasonably large. (See Remark~\ref{remark:usefulness} for more discussion.)
In this sense, strategies in state merging to exploit the classical description of $\Ket{\psi}^{RAB}$ for reducing entanglement cost have \textit{not} yet been established for arbitrarily small-dimensional systems or arbitrarily high fidelity.

In contrast, Chapter~\ref{sec:merge} explicitly constructs protocols for one-shot state merging
with the following features:
\begin{enumerate}
  \item Applicable to any state of an arbitrarily small-dimensional system, including small- and intermediate-scale states;
  \item Fulfilling arbitrarily high fidelity requirement, including zero error;
  \item Retaining the essential feature of state merging, that is, exploiting classical description of $\Ket{\psi}^{RAB}$ for reducing entanglement cost.
\end{enumerate}
The tasks of one-shot state merging investigated here are achieved exactly, that is, without approximation, and is called \textit{exact state merging}, as illustrated in Figure~\ref{fig:merge}.
Entanglement cost of the obtained protocols for exact state merging is not larger than, and can be strictly smaller than, that for its inverse task, \textit{exact state splitting} summarized in Section~\ref{sec:split}, depending on the Koashi-Imoto decomposition of $\Ket{\psi}^{RAB}$~\cite{W4,K3,H6,K5}.
Multiple examples of states including those relevant to distributed quantum information processing are also shown, where the obtained protocols for exact state merging can reduce entanglement cost, since these states have \textit{nontrivial} Koashi-Imoto decompositions.
In the same way as the asymptotic scenarios, entanglement cost of the obtained protocol can even be negative.
In addition to providing achievability bounds by constructing the protocols, Chapter~\ref{sec:merge} also obtains converse bounds, that is, lower bounds of the minimal entanglement cost required for any protocol for exact state merging.
The obtained converse bounds improve the existing converse bound~\cite{B9} given in terms of the conditional max-entropy~\cite{R2,T5,T11}, and it is shown that the converse bound is achievable when the state to be merged is represented by qubits.
By means of \textit{smoothing}~\cite{R2,T5,T11},
these results on exact state merging are straightforwardly extended to \textit{approximate state merging}, where arbitrarily small approximation error in fidelity is allowed so that the entanglement cost can further be reduced compared to exact state merging.
The obtained converse bound of entanglement cost in approximate state merging improves the existing converse bound~\cite{B10}.

\section*{One-shot quantum state merging under one-way and two-way communication}

The minimal entanglement cost ${H\left(A|B\right)}_{\psi}$ in the asymptotic scenario of state merging can be achieved by only \textit{one-way} LOCC, using one-way classical communication only from $A$ to $B$, even if $A$ and $B$ are allowed to perform \textit{two-way} LOCC, using two-way classical communication both from $A$ to $B$ and from $B$ to $A$.
Indeed, ${H\left(A|B\right)}_{\psi}$ is monotonically nondecreasing under a class of operations consisting of $B$'s preprocessing and backward classical communication from $B$ to $A$, as shown in Section~\ref{sec:cost}.
A conventional interpretation of ${H\left(A|B\right)}_{\psi}\coloneq {H\left(AB\right)}_{\psi}-{H\left(B\right)}_{\psi}$ is that the first term ${H\left(AB\right)}_{\psi}$ quantifies quantum information encoded in $A$ and $B$'s shared state, and the second term ${H\left(B\right)}_{\psi}$ quantifies quantum information initially located at $B$~\cite{H4}, while this interpretation is based on one-way communication from $A$ to $B$ in analogy to classical source coding with $B$'s classical side information~\cite{S7}.
As for one-shot scenarios of state merging, the existing protocols~\cite{B9,Y9,B12,D7,D6,H10,B10,D5,M,N3,A4,A5,B15,B13,A16,A17}, based on either technique of one-shot decoupling~\cite{D6} or the convex-split lemma~\cite{A4},
use only one-way communication similarly to the asymptotic scenario,
and whether protocols using two-way communication may outperform those using only one-way communication has been unknown~\cite{B10}.

In contrast, Chapter~\ref{sec:two_way} demonstrates a provable advantage of two-way LOCC over one-way LOCC in exploiting $B$'s quantum side information in a one-shot scenario of state merging,
showing that the minimal entanglement cost in state merging under two-way LOCC can be strictly smaller than that under one-way LOCC, while they coincide in the asymptotic scenario.
The results in Chapter~\ref{sec:two_way} suggest that under a one-shot regime, $B$'s preprocessing and backward classical communication from $B$ to $A$ can be indispensable for making the most of quantum side information, that is, for minimizing entanglement cost in state merging.

\chapter{\label{sec:preliminaries_1}Preliminaries to Part~\ref{part:1}}

This chapter provides preliminaries to Part~\ref{part:1}.
Section~\ref{sec:def_merge} defines tasks of one-shot state merging to be analyzed in the rest of Part~\ref{part:1}.
For comparison, results of the preceding works~\cite{H3,H4} on the asymptotic scenario of state merging are summarized in Section~\ref{sec:asymptotic},
and definitions and properties of state splitting, the inverse task of state merging, are summarized in Section~\ref{sec:split}.
Sections~\ref{sec:def_merge} and~\ref{sec:split} are based on Reference~\cite{Y12}.

\section{\label{sec:def_merge}Definition of one-shot state merging}

Definitions of the tasks of one-shot state merging to be analyzed in this part are presented in this section.

Quantum state merging involves three spatially separated parties $A$, $B$, and $R$,
where by convention, $A$ is a sender, $B$ is a receiver, and $R$ is a reference to consider a purification.
Let $\mathcal{H}^A$ and $\mathcal{H}^{\overline{A}}$ be $A$'s systems, $\mathcal{H}^B$, $\mathcal{H}^{B'}$, and $\mathcal{H}^{\overline{B}}$ be $B$'s, and $\mathcal{H}^R$ be $R$'s,
where
\begin{equation}
  \dim\mathcal{H}^{A}=\dim\mathcal{H}^{B'}.
\end{equation}
Given any tripartite pure state $\Ket{\psi}^{RAB}$ shared among $R$, $A$, and $B$, state merging of $\Ket{\psi}^{RAB}$ aims to transfer $A$'s part of $\Ket{\psi}^{RAB}$ to $B$, keeping coherence between $R$ and $B$, to obtain $\Ket{\psi}^{RB^\prime B}$.
Assume that the parties $A$ and $B$ are allowed to freely perform local operations and classical communication (LOCC), and quantum communication is performed by LOCC assisted by a maximally entangled resource state shared between $A$ and $B$ in $\mathcal{H}^{\overline{A}}\otimes\mathcal{H}^{\overline{B}}$.
Note that $A$ and $B$ cannot perform any operation on $R$.

While a trivial protocol for state merging of $\Ket{\psi}^{RAB}$ is quantum teleportation of $A$'s part of $\Ket{\psi}^{RAB}$,
the initially given state $\Ket{\psi}^{RAB}$ in state merging may have entanglement between $A$ and $B$, and hence, there are cases where $A$ and $B$ can distill this entanglement of $\Ket{\psi}^{RAB}$ in achieving state merging of $\Ket{\psi}^{RAB}$ to reduce the required amount of entanglement.
A maximally entangled state initially shared between $A$ and $B$ is written as
\begin{equation}
  \Ket{\Phi_K^+}^{\overline{A}\overline{B}}\coloneq\sum_{l=0}^{K-1}\Ket{l}^{\overline{A}}\otimes\Ket{l}^{\overline{B}},
\end{equation}
where $K$ denotes the Schmidt rank of the initially shared maximally entangled resource state.
After achieving state merging, $A$ and $B$ are still allowed to share a maximally entangled state
\begin{equation}
  \Ket{\Phi_L^+}^{\overline{A}\overline{B}}\coloneq\sum_{l=0}^{L-1}\Ket{l}^{\overline{A}}\otimes\Ket{l}^{\overline{B}},
\end{equation}
where $L$ denotes the Schmidt rank of the finally shared maximally entangled resource state.
The amount of entanglement for $\Ket{\Phi_K^+}$ and $\Ket{\Phi_L^+}$ is measured in terms of the entanglement entropy, that is, $\log_2 K$ and $\log_2 L$, respectively.
If $\log_2 K-\log_2 L > 0$, $\log_2 K-\log_2 L$ is regarded as the amount of entanglement consumed in the protocol, and otherwise, $\log_2 L-\log_2 K$ as a net gain of entanglement.
When $\log_2 K > 0$, $\log_2 L > 0$, and $\log_2 K-\log_2 L > 0$ in state merging, a part of entanglement of the initially shared maximally entangled resource state, that is, $\log_2 L$ ebits out of $\log_2 K$ ebits, can be considered to be used catalytically.
This setting is called a \textit{catalytic setting}.

In one-shot scenarios, it is also beneficial to minimize the initially required amount of entanglement for achieving state merging.
Hence, another setting can also be considered by fixing $\log_2 L=0$ in the above catalytic setting of state merging.
This setting is called a \textit{non-catalytic setting}.
Note that protocols for state merging in the non-catalytic setting also works in the catalytic setting, but not necessarily \textit{vice versa}.
In the following of this chapter, the catalytic setting is considered unless stated otherwise explicitly.

A simple one-shot scenario of state merging is that requiring zero error in the protocol.
This task is called \textit{exact state merging} and defined as follows.
The task of exact state merging is also illustrated in Figure~\ref{fig:merge}.

\begin{definition}
\label{def:merging}
    \textit{Exact state merging.}
    Exact state merging of a purified given state $\Ket{\psi}^{RAB}$ is a task for parties $A$ and $B$ to achieve a transformation
    \begin{equation}
            \id^R\otimes\mathcal{M}\left({\psi}^{RAB}\otimes{\Phi^+_K}^{\overline{A}\overline{B}}\right) ={\psi}^{RB'B}\otimes{\Phi^+_L}^{\overline{A}\overline{B}}
    \end{equation}
    by an LOCC map
    \begin{equation}
      \mathcal{M}:\mathcal{B}\left(\mathcal{H}^A\otimes\mathcal{H}^B\otimes\mathcal{H}^{\overline{A}}\otimes\mathcal{H}^{\overline{B}}\right)\to\mathcal{B}\left(\mathcal{H}^{B'}\otimes\mathcal{H}^B\otimes\mathcal{H}^{\overline{A}}\otimes\mathcal{H}^{\overline{B}}\right),
    \end{equation}
    which can be constructed depending on the classical description of $\Ket{\psi}^{RAB}$.
    The definition of exact state merging in the non-catalytic setting is also obtained by setting $\log_2 L=0$ in the above definition.
    Entanglement cost of a protocol for exact state merging in the catalytic setting is defined as
    \begin{equation}
      \log_2 K-\log_2 L,
    \end{equation}
    and that in the non-catalytic setting is defined as
    \begin{equation}
      \log_2 K.
    \end{equation}
\end{definition}

The minimal entanglement cost among all the protocols for exact state merging of $\Ket{\psi}^{RAB}$ may be simply referred to as entanglement cost in exact state merging of $\Ket{\psi}^{RAB}$.
If
\begin{equation}
  \log_2 K\geqq\log_2 \dim\mathcal{H}^A,
\end{equation}
there exists a trivial protocol for exact state merging by quantum teleportation to transfer $\psi^{A}$ from $A$ to $B$.
The results given in Chapters~\ref{sec:merge} and~\ref{sec:two_way} provide protocols at less entanglement cost using the classical description of $\Ket{\psi}^{RAB}$.

There exist following tasks achievable at the same entanglement cost using the same protocol as those in exact state merging of a given state $\Ket{\psi}^{RAB}$.
Consider the Schmidt decomposition of $\Ket{\psi}^{RAB}$ with respect to bipartition between $\mathcal{H}^R$ and $\mathcal{H}^{A}\otimes\mathcal{H}^{B}$
\begin{equation}
  \label{eq:schmidt}
  \Ket{\psi}^{RAB}=\sum_{l=0}^{D-1} \sqrt{\lambda_l}\Ket{l}^R\otimes\Ket{\psi_l}^{AB},
\end{equation}
where $D$ is the Schmidt rank, and $\lambda_l>0$ for each $l\in\{0,\ldots,D-1\}$.
Then, entanglement cost in exact state merging of $\Ket{\psi}^{RAB}$ equals to that of a maximally entangled state $\Ket{\Phi_D^+\left(\psi\right)}^{RAB}$ with Schmidt rank $D$ corresponding to $\Ket{\psi}^{RAB}$
\begin{equation}
  \label{eq:max}
  \Ket{\Phi_D^+\left(\psi\right)}^{RAB}\coloneq\sum_{l=0}^{D-1} \frac{1}{\sqrt{D}}\Ket{l}^R\otimes\Ket{\psi_l}^{AB},
\end{equation}
where the Schmidt basis on the right-hand side is the same as that in Equation~\eqref{eq:schmidt}, and this maximally entangled state is independent of the Schmidt coefficients $\sqrt{\lambda_0},\ldots,\sqrt{\lambda_{D-1}}$ in Equation~\eqref{eq:schmidt}.

This equivalence also implies that entanglement cost in exact state merging of $\Ket{\psi}^{RAB}$ is the same as that required for merging arbitrary bipartite states shared between $A$ and $B$ on a subspace of $\mathcal{H}^A\otimes\mathcal{H}^B$ spanned by the Schmidt-basis states ${\left\{\Ket{\psi_l}^{AB}\right\}}_l$ corresponding to nonzero Schmidt coefficients in Equation~\eqref{eq:schmidt}.
The equivalence between considering the maximally entangled state with $R$ in Equation~\eqref{eq:max} and considering arbitrary bipartite states on the corresponding subspace is also known as the relative state method~\cite{P3}.
Note that in general, entanglement cost in exact state merging of $\Ket{\psi}^{RAB}$ is different from that required for merging arbitrary bipartite states given from an ensemble ${\left\{p\left(l\right),\Ket{\psi_l}^{AB}\right\}}_l$ for a probability distribution $p\left(l\right)$, since coherence of arbitrary superposition of ${\left\{\Ket{\psi_l}^{AB}\right\}}_l$ has to be kept in state merging of $\Ket{\psi}^{RAB}$.

These equivalences are shown as the following proposition, and see Appendix~\ref{sec:equivalence} for the proof.

\begin{proposition}
\label{prp:max}
  \textit{Equivalence of exact state merging of an arbitrary tripartite pure state, a corresponding maximally entangled state, and a corresponding set of bipartite states.}
  Given any fixed integer $K$, $L$, and any pure state $\Ket{\psi}^{RAB}$ whose Schmidt rank with respect to bipartition between $\mathcal{H}^R$ and $\mathcal{H}^{A}\otimes\mathcal{H}^{B}$ is $D$ and whose Schmidt decomposition is given by Equation~\eqref{eq:schmidt},
  the following statements are equivalent:
  \begin{enumerate}
    \item An LOCC map $\mathcal{M}$ achieves the following exact state merging of $\Ket{\psi}^{RAB}$
      \begin{equation}
          \id^R\otimes\mathcal{M}\left({\psi}^{RAB}\otimes{\Phi^+_K}^{\overline{A}\overline{B}}\right) ={\psi}^{RB^\prime B}\otimes{\Phi^+_L}^{\overline{A}\overline{B}};
      \end{equation}
    \item The same LOCC map $\mathcal{M}$ as the above achieves the following exact state merging of $\Ket{\Phi_D^+\left(\psi\right)}^{RAB}$
      \begin{equation}
        \id^R\otimes\mathcal{M}\left({\Phi_D^+\left(\psi\right)}^{RAB}\otimes{\Phi^+_K}^{\overline{A}\overline{B}}\right)
        ={\Phi_D^+\left(\psi\right)}^{RB'B}\otimes{\Phi^+_L}^{\overline{A}\overline{B}},
      \end{equation}
      where $\Ket{\Phi_D^+\left(\psi\right)}^{RAB}$ is the maximally entangled state corresponding to $\Ket{\psi}^{RAB}$, defined as Equation~\eqref{eq:max}.
    \item Define a set $S_\psi^{AB}\subset\mathcal{H}^A\otimes\mathcal{H}^B$ of arbitrary bipartite states on a subspace of $\mathcal{H}^A\otimes\mathcal{H}^B$ spanned by the Schmidt-basis states ${\left\{\Ket{\psi_l}^{AB}\right\}}_l$ corresponding to nonzero Schmidt coefficients of $\Ket{\psi}^{RAB}$ in Equation~\eqref{eq:schmidt}, that is,
      \begin{equation}
          S_\psi^{AB}
          \coloneq{\left\{\psi_{\boldsymbol{\alpha}}^{AB}\coloneq\sum_{l=0}^{D-1}\sum_{l^\prime=0}^{D-1}\alpha_{l,l^\prime} \Ket{\psi_l}\Bra{\psi_{l^\prime}}^{AB}\in\mathcal{D}\left(\mathcal{H}^A\otimes\mathcal{H}^B\right)\right\}}_{\boldsymbol{\alpha}},
      \end{equation}
      where $\boldsymbol{\alpha}$ denotes the tuple of the parameters $\alpha_{l,l^\prime}$ for all $l$ and $l^\prime$.
      Then, the same LOCC map $\mathcal{M}$ as the above achieves the following state transformation for any bipartite state $\psi_{\boldsymbol{\alpha}}^{AB}\in S_\psi^{AB}$
      \begin{equation}
          \mathcal{M}\left(\psi_{\boldsymbol{\alpha}}^{AB}\otimes{\Phi^+_K}^{\overline{A}\overline{B}}\right)
          =\psi_{\boldsymbol{\alpha}}^{B^\prime B}\otimes{\Phi^+_L}^{\overline{A}\overline{B}},
      \end{equation}
      where $\mathcal{M}$ is independent of $\boldsymbol{\alpha}$.
  \end{enumerate}
  The same equivalence also holds in the non-catalytic setting by fixing $\log_2 L=0$.
\end{proposition}

While this zero-error scenario is fundamental,
a sufficiently small error in fidelity of quantum states does not largely affect outcomes of any measurement in quantum mechanics.
Hence, another scenario can be considered, where a nonzero error for approximation may be tolerated for reducing entanglement cost in comparison with exact state merging.
This task is called \textit{approximate state merging} and defined as follows.

\begin{definition}
\label{def:approxiamte_state_merging}
  \textit{Approximate state merging.}
    Approximate state merging of a given state $\Ket{\psi}^{RAB}$ within a given error $\epsilon\geqq 0$ is a task of parties $A$ and $B$ performing an LOCC map
    \begin{equation}
      \tilde{\mathcal{M}}: \mathcal{B}\left(\mathcal{H}^A\otimes\mathcal{H}^{B}\otimes\mathcal{H}^{\overline{A}}\otimes\mathcal{H}^{\overline{B}}\right) \to \mathcal{B}\left(\mathcal{H}^{B^\prime}\otimes\mathcal{H}^B\otimes\mathcal{H}^{\overline{A}}\otimes\mathcal{H}^{\overline{B}}\right)
    \end{equation}
    achieving a transformation
    \begin{equation}
        \id^R \otimes\tilde{\mathcal{M}}\left({\psi}^{RAB}\otimes{\Phi^+_K}^{\overline{A}\overline{B}}\right)
        ={\tilde\psi}^{RB^\prime B\overline{A}\overline{B}},
    \end{equation}
    where the fidelity of the final state satisfies
    \begin{equation}
      F^2\left({\tilde{\psi}}^{RB^\prime B\overline{A}\overline{B}}, \psi^{RB^\prime B}\otimes{\Phi_L^+}^{\overline{A}\overline{B}}\right)\coloneq\left(\Bra{\psi}\otimes\Bra{\Phi_L^+}\right){\tilde{\psi}}\left(\Ket{\psi}\otimes\Ket{\Phi_L^+}\right)\geqq 1-\epsilon^2.
    \end{equation}
    Given a protocol for approximate state merging, entanglement cost of the protocol is defined as
    \begin{equation}
      \log_2 K - \log_2 L.
    \end{equation}
    Approximate state merging of $\Ket{\psi}^{RAB}$ within $\epsilon$ in the non-catalytic setting is defined by fixing
    \begin{equation}
      \log_2 L = 0
    \end{equation}
    in the above definition.
\end{definition}

\section{\label{sec:asymptotic}Asymptotic scenario of quantum state merging and quantum entropy}

The results of the preceding works~\cite{H3,H4} on entanglement cost required for the asymptotic scenario of quantum state merging is summarized in this section.
The minimal entanglement cost in the asymptotic scenario is characterized by entropic functions, which are also summarized in this section.

The asymptotic scenario of quantum state merging is originally analyzed in References~\cite{H3,H4},
which aims at achieving approximate state merging of many copies of the same state so as to use the law of large numbers; that is, instead of $\Ket{\psi}^{RAB}$ in Definition~\ref{def:approxiamte_state_merging}, the given state is in the form of
\begin{equation}
  {\left(\Ket{\psi}^{RAB}\right)}^{\otimes n}.
\end{equation}
For any $\epsilon>0$,
the rate of entanglement cost of a protocol achieving approximate state merging of ${\left(\Ket{\psi}^{RAB}\right)}^{\otimes n}$ within $\epsilon$ is evaluated by
\begin{equation}
  \frac{1}{n}\left(\log_2 K-\log_2 L\right).
\end{equation}
References~\cite{H3,H4} analyze the minimal achievable rate of the entanglement cost under asymptotically vanishing error, that is,
\begin{equation}
  \lim_{\epsilon\to 0}\lim_{n\to\infty}\inf\left\{\frac{1}{n}\left(\log_2 K-\log_2 L\right)\right\},
\end{equation}
where the infimum is taken over all the protocols achieving approximate state merging of ${\left(\Ket{\psi}^{RAB}\right)}^{\otimes n}$ within $\epsilon$.

This minimal achievable rate of the entanglement cost in state merging in the asymptotic scenario is evaluated using quantum entropy summarized in the following.
For any discrete random variable $X$ with a probability distribution $p\left(x\right)$,
the entropy of $X$ is defined as
\begin{equation}
  H\left(X\right)\coloneq -\sum_x p\left(x\right)\log_2 p\left(x\right).
\end{equation}
Similarly, for any state $\psi^A$,
the quantum entropy of $\psi^A$ is defined as
\begin{equation}
  {H\left(A\right)}_\psi\coloneq-\tr\psi^A\log_2\psi^A,
\end{equation}
which is also written as
\begin{equation}
  H\left(\psi^A\right)\coloneq {H\left(A\right)}_\psi.
\end{equation}
Quantum entropy is invariant under isometry; that is, for any isometry $U^{A\to A^\prime}$,
it holds that
\begin{equation}
  H\left(\psi^A\right)={H\left(\left(U^{A\to A^\prime}\right)\psi^A{\left(U^{A\to A^\prime}\right)}^\dag\right)}.
\end{equation}

For any bipartite state $\psi^{AB}$,
the joint quantum entropy of $\psi^{AB}$ for the bipartite system $\mathcal{H}^A\otimes\mathcal{H}^B$ is defined as
\begin{equation}
  {H\left(AB\right)}_\psi\coloneq H\left(\psi^{AB}\right).
\end{equation}
The joint quantum entropy of multipartite states is generally defined in the same way.
If a given bipartite state $\psi^{XA}$ is in the form of
\begin{equation}
  \psi^{XA}=\sum_x p\left(x\right)\Ket{x}\Bra{x}^X\otimes\psi_x^A,
\end{equation}
where $p\left(x\right)$ is a probability distribution, ${\left\{\Ket{x}\right\}}_x$ is the computational basis of $\mathcal{H}^X$, and $\psi_x^A\in\mathcal{D}\left(\mathcal{H}^A\right)$ for each $x$,
then $\psi^{XA}$ is called a classical-quantum state.
This classical-quantum state $\psi^{XA}$ can be regarded as a mixed state representing an ensemble ${\left\{p\left(x\right),\Ket{x}\Bra{x}^X\otimes\psi_x^A\right\}}_x$.
The joint quantum entropy of this classical-quantum state $\psi^{XA}$ is evaluated as
\begin{equation}
  {H\left(XA\right)}_\psi= H\left(X\right)+\sum_x p\left(x\right) H\left(\psi_x^A\right).
\end{equation}

For any bipartite state $\psi^{AB}$,
the conditional quantum entropy of $\psi^{AB}$ is defined as
\begin{equation}
  {H\left(A|B\right)}_\psi\coloneq {H\left(AB\right)}_\psi - {H\left(B\right)}_\psi.
\end{equation}
Conditional quantum entropy satisfies for any CPTP map $\mathcal{N}^{B\to B^\prime}$
\begin{align}
  {H\left(A|B\right)}_\psi&\geqq{H\left(A|B^\prime\right)}_{\psi^\prime}\,,\\
  {\psi^\prime}^{B^\prime}&\coloneq\mathcal{N}^{B\to B^\prime}\left(\psi^B\right),
\end{align}
and this type of inequality is called the data processing inequality.

Using these notations, the minimal achievable rate of entanglement cost in the asymptotic scenario of state merging is evaluated as follows.

\begin{lemma}
  (References~\cite{H3,H4}.)
  \textit{Entanglement cost in state merging in the asymptotic scenario}
  For any pure state $\Ket{\psi}^{RAB}$, the minimal achievable rate of the entanglement cost in the asymptotic scenario of state merging is given by the conditional quantum entropy of $\psi^{AB}$, that is,
  \begin{equation}
    \lim_{\epsilon\to 0}\lim_{n\to\infty}\inf\left\{\frac{1}{n}\left(\log_2 K-\log_2 L\right)\right\}={H\left(A|B\right)}_\psi
  \end{equation}
  where the infimum is taken over all the protocols achieving approximate state merging of ${\left(\Ket{\psi}^{RAB}\right)}^{\otimes n}$ within $\epsilon$.
\end{lemma}

\section{\label{sec:split}Quantum state splitting for arbitrarily small-dimensional systems}

Before proceeding to analysis of one-shot state merging in Chapters~\ref{sec:merge} and~\ref{sec:two_way},
an inverse task of one-shot state merging, one-shot state splitting, is introduced in this section for comparison.
State splitting involves three spatially separated parties $A$, $B$, and $R$,
where by convention, $A$ is a sender, $B$ is a receiver, and $R$ is a reference to consider purification.
Let $A$ have systems $\mathcal{H}^A$, $\mathcal{H}^{A^\prime }$, and $\mathcal{H}^{\overline{A}}$, $B$ have $\mathcal{H}^B$ and $\mathcal{H}^{\overline{B}}$, and $R$ have $\mathcal{H}^R$,
where
\begin{equation}
  \dim\mathcal{H}^{A^\prime }=\dim\mathcal{H}^B.
\end{equation}
In state splitting, $A$ and $B$ initially share a maximally entangled resource state, and $A$ also holds a given bipartite state of $\mathcal{H}^A\otimes\mathcal{H}^{A^\prime}$, whose purification with reference $R$ is represented as $\Ket{\psi}^{RAA^\prime}$.
Assume that $A$ and $B$ can freely perform local operations and classical communication (LOCC) assisted by a maximally entangled resource state initially shared between $\mathcal{H}^{\overline{A}}$ of $A$ and $\mathcal{H}^{\overline{B}}$ of $B$, that is,
\begin{equation}
  \Ket{\Phi^+_K}^{\overline{A} \overline{B}}\coloneq\frac{1}{\sqrt{K}}\sum_{l=0}^{K-1}\Ket{l}^{\overline{A}} \otimes\Ket{l}^{\overline{B}},
\end{equation}
where $K$ denotes the Schmidt rank of this resource state.
Note that $A$ and $B$ cannot perform any operation on $\mathcal{H}^R$.
State splitting of $\Ket{\psi}^{RAA^\prime}$ aims to transfer a part of $\Ket{\psi}^{RAA^\prime}$ corresponding to $\mathcal{H}^{A^\prime}$ from $A$ to $B$, keeping coherence between $R$ and $AB$, to obtain $\Ket{\psi}^{RAB}$, where $\mathcal{H}^B$ is $B$'s system corresponding to $\mathcal{H}^{A^\prime}$.

In the same way as exact state merging,
a simple one-shot scenario of state splitting is that requiring zero error.
This task is called \textit{exact state splitting}, illustrated in Figure~\ref{fig:split}, and defined as follows.
Note that it suffices to consider no catalytic use of entanglement in exact state splitting.

\begin{figure}[t!]
  \centering
  \includegraphics[width=4in]{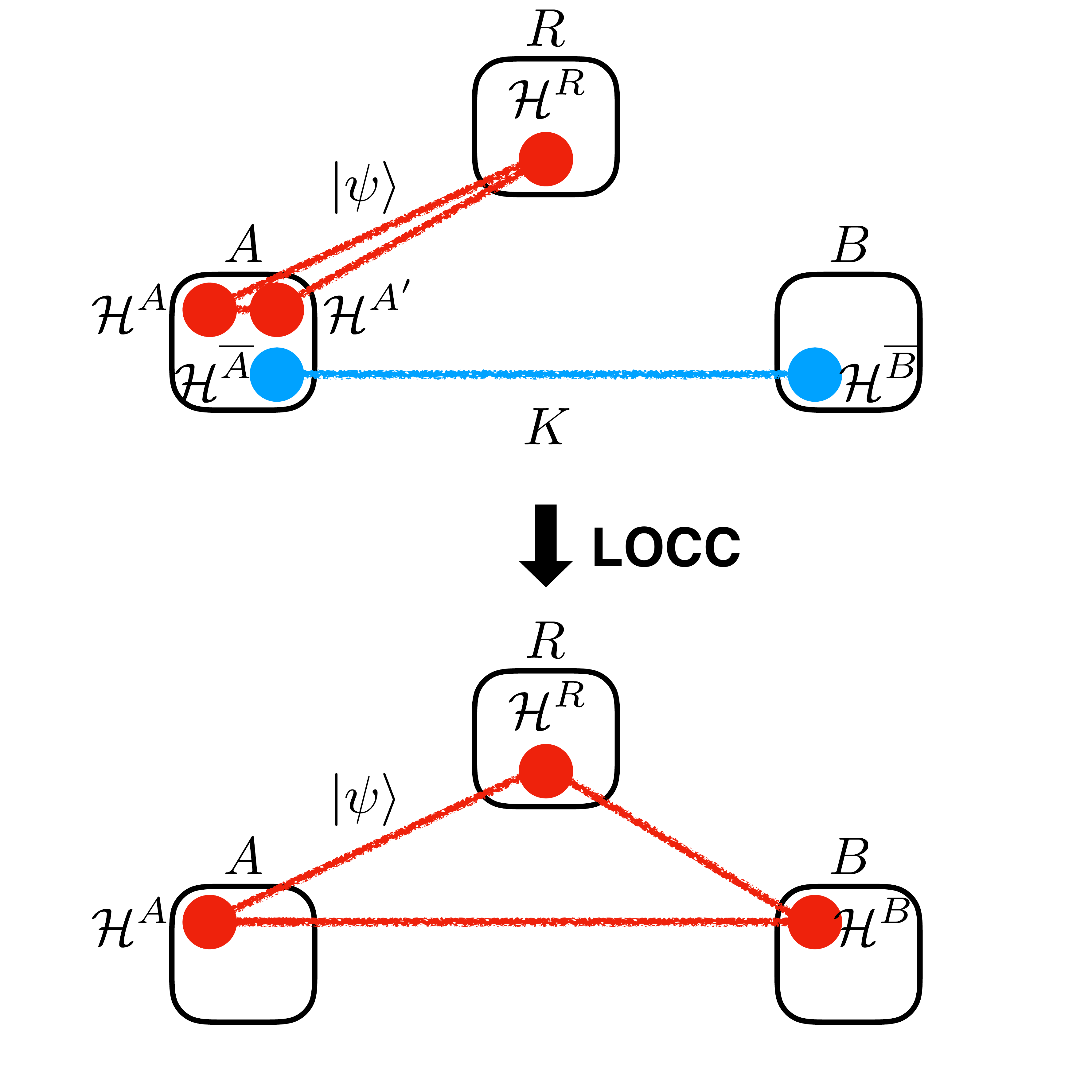}
  \caption[Exact state splitting]{\label{fig:split}Exact state splitting of a given state $\Ket{\psi}^{RAA^\prime }$ denoted by the red circles. Parties $A$ and $B$ perform LOCC assisted by a maximally entanglement resource state $\Ket{\Phi_K^+}^{\overline{AB}}$ with the Schmidt rank $K$ denoted by the blue circles to transfer the reduced state $\psi^{A^\prime }$ from $A$ to $B$ and obtain $\Ket{\psi}^{RAB}$. }
\end{figure}

\begin{definition}
\label{def:splitting}
    \textit{Exact state splitting.}
    Exact state splitting of a purified given state $\Ket{\psi}^{R A A^\prime }$ is a task for parties $A$ and $B$ to achieve a transformation
    \begin{equation}
        \id^R \otimes\mathcal{S}\left({\psi}^{RAA^\prime }\otimes{\Phi^+_K}^{\overline{A}\overline{B}}\right)
        ={\psi}^{RAB}
    \end{equation}
    by an LOCC map
    \begin{equation}
      \mathcal{S}: \mathcal{B}\left(\mathcal{H}^A\otimes\mathcal{H}^{A^\prime }\otimes\mathcal{H}^{\overline{A}}\otimes\mathcal{H}^{\overline{B}}\right) \to \mathcal{B}\left(\mathcal{H}^A\otimes\mathcal{H}^B\otimes\mathcal{H}^{\overline{A}}\otimes\mathcal{H}^{\overline{B}}\right),
    \end{equation}
    which can be constructed depending on the classical description of $\Ket{\psi}^{RAA^\prime}$.
    Given a protocol for exact state splitting, entanglement cost of the protocol is defined as $\log_2 K$.
\end{definition}

If
\begin{equation}
  \log_2 K\geqq\log_2\dim\mathcal{H}^{A^\prime},
\end{equation}
there exists a trivial protocol for state splitting by quantum teleportation to transfer $\psi^{A^\prime }$ from $A$ to $B$.
In contrast, there are cases where classical description of $\Ket{\psi}^{RAA^\prime}$ can be used for reducing the entanglement cost.
The following theorem shows the minimal entanglement cost in exact state splitting and a protocol achieving the minimal entanglement cost.

\begin{theorem}
\label{thm:split}
    \textit{Optimal entanglement cost in exact state splitting.}
    Given any $\Ket{\psi}^{RAA^\prime }$,
    exact state splitting of $\Ket{\psi}^{RAA^\prime }$ is achievable if and only if
    \begin{equation}
        \log_2 K \geqq \log_2 \rank \psi^{A^\prime }.
    \end{equation}
\end{theorem}

\begin{proof}
    \textit{If part}:
    The proof is by construction, and an LOCC protocol achieving
    \begin{equation}
        \label{eq:split_upper}
        \log_2 K = \log_2 \rank \psi^{A^\prime }
    \end{equation}
    is constructed.
    Note that the trivial protocol, that is, quantum teleportation of $\psi^{A^\prime}$, requires entanglement cost $\log_2 K = \log_2\dim\mathcal{H}^{A^\prime }$,
    and the protocol achieving Equation~\eqref{eq:split_upper} outperforms this trivial protocol when $\psi^{A^\prime}$ is not a full-rank state, that is, $\rank\psi^{A^\prime }<\dim\mathcal{H}^{A^\prime}$; \textit{e.g.},
    when $\psi^{A^\prime}$ is locally represented as a code state of a quantum error correcting code~\cite{G,D,T10,B} using a larger-dimensional system $\mathcal{H}^{A^\prime}$ than the rank of $\psi^{A^\prime}$.

    To achieve Equation~\eqref{eq:split_upper}, a method for compressing $\psi^{A^\prime }$ is provided.
    Consider the Schmidt decomposition of the given state $\Ket{\psi}^{RAA^\prime }$ with respect to the bipartition between $\mathcal{H}^R\otimes\mathcal{H}^A$ and $\mathcal{H}^{A^\prime }$, that is,
    \begin{equation}
        \Ket{\psi}^{RAA^\prime }=\sum_{l\in R_{\psi}} \sqrt{\lambda_l^\psi}\Ket{l}^{RA}\otimes\Ket{l}^{A^\prime },
    \end{equation}
    where $R_{\psi}\coloneq\left\{0,\ldots,\rank \psi^{A^\prime }-1\right\}$, each $\sqrt{\lambda_l^\psi}>0$ is a nonzero Schmidt coefficient, and ${\left\{\Ket{l}^{RA}: l\in R_{\psi}\right\}}$ and ${\left\{\Ket{l}^{A^\prime }: l\in R_{\psi}\right\}}$ are subsets of the Schmidt bases of $\mathcal{H}^R\otimes\mathcal{H}^A$ and $\mathcal{H}^{A^\prime }$, respectively, corresponding to the nonzero Schmidt coefficients.
    Let $\mathcal{H}^{A^{\prime\prime}}$ be $A$'s auxiliary system satisfying
    \begin{equation}
      \dim\mathcal{H}^{A^{\prime\prime}}=\rank \psi^{A^\prime },
    \end{equation}
    and ${\left\{\Ket{l}^{A^{\prime\prime}}: l\in R_{\psi}\right\}}$ be the computational basis of $\mathcal{H}^{A^{\prime\prime}}$.
    Consider an isometry $U_\textup{split}$ from $\mathcal{H}^{A^\prime }$ to $\mathcal{H}^{A^{\prime\prime}}$ satisfying for each $l\in R_{\psi}$
    \begin{equation}
      \Ket{l}^{A^{\prime\prime}}=U_\textup{split}\Ket{l}^{A^\prime }.
    \end{equation}
    By performing $U_\textup{split}$, $\psi^{A^\prime }$ is compressed into a state on $\mathcal{H}^{A^{\prime\prime}}$, that is,
    \begin{equation}
        \begin{split}
          \Ket{\psi'}^{RAA^{\prime\prime}}&\coloneq\mathbb{1}^{RA}\otimes U_\textup{split}\Ket{\psi}^{RAA^\prime }\\
                             &=\sum_{l\in R_{\psi}} \sqrt{\lambda_l^\psi}\Ket{l}^{RA}\otimes\Ket{l}^{A^{\prime\prime}}.
        \end{split}
    \end{equation}
    By performing $U_\textup{split}^\dag$, the given state $\Ket{\psi}$ can be recovered from the compressed state $\Ket{\psi'}$.

    The LOCC protocol achieving Equation~\eqref{eq:split_upper} is as follows.
    First, $A$ performs $U_\textup{split}$ to transform the given state $\Ket{\psi}^{RAA^\prime }$ into the compressed state $\Ket{\psi'}^{RAA^{\prime\prime}}$.
    Next, the reduced state ${\psi'}^{A^{\prime\prime}}$ is sent from $A$ to $B$ by quantum teleportation using the resource state satisfying Equation~\eqref{eq:split_upper}.
    After performing quantum teleportation, $B$ performs $U_\textup{split}^\dag$ on the system for the received state to recover $\Ket{\psi}^{RAB}$.

    \textit{Only if part}:
    The proof uses the LOCC monotonicity of the Schmidt rank.
    The Schmidt rank of $\Ket{\psi}^{RAA^\prime }\otimes\Ket{\Phi^+_K}^{\overline{A}\overline{B}}$ between the party $B$ and the other parties $R$ and $A$ is $K$.
    After performing an LOCC map $\id^R\otimes\mathcal{S}$, the Schmidt rank of $\Ket{\psi}^{RAB}$ between the party $B$ and the other parties $R$ and $A$ is $\rank\psi^{A^\prime }$.
    Since the Schmidt rank of pure states is monotonically nonincreasing under LOCC,
    it holds that
    \begin{equation}
      K\geqq\rank\psi^{A^\prime }.
    \end{equation}
    Therefore, it holds that
    \begin{equation}
      \log_2 K \geqq \log_2 \rank \psi^{A^\prime }.
    \end{equation}
\end{proof}

\chapter{\label{sec:merge}Quantum state merging for arbitrarily small-dimensional systems}

This chapter investigates general bounds of entanglement cost required for one-shot state merging defined in Chapter~\ref{sec:preliminaries_1} and illustrated in Figure~\ref{fig:merge}, aiming at achieving this task on the small and intermediate scales.
Sections~\ref{sec:merge_achievability_exact} and~\ref{sec:achievability_approximate} construct protocols for exact state merging and approximate state merging, respectively, so that these protocols can be applied to arbitrarily small-dimensional systems.
Sections~\ref{sec:converse} and~\ref{sec:converse_approximate} provides improved converse bounds of entanglement cost in exact state merging and approximate state merging, respectively.
Implications are discussed in Section~\ref{sec:examples}.

\section{\label{sec:merge_achievability_exact}Achievability bound for exact state merging}

In this section, protocols for exact state merging are constructed, which are applicable to any state of an arbitrarily small-dimensional system.
To construct the protocols, the Koashi-Imoto decomposition introduced in Section~\ref{sec:decomposition} is used.
Given any state $\Ket{\psi}^{RAB}$,
Lemma~\ref{lem:koashi_imoto_decomposition_tripartite} implies that there exists a unique decomposition of $\mathcal{H}^A$ and $\supp\left(\psi^B\right)$
\begin{align}
  \label{eq:notation_space}
  \mathcal{H}^A=\bigoplus_{j=0}^{J-1}\mathcal{H}^{a_j^\textup{L}}\otimes\mathcal{H}^{a_j^\textup{R}},\quad
  \supp\left(\psi^B\right)=\bigoplus_{j=0}^{J-1}\mathcal{H}^{b_j^\textup{L}}\otimes\mathcal{H}^{b_j^\textup{R}},
\end{align}
such that $\Ket{\psi}^{RAB}$ is uniquely decomposed into
\begin{equation}
  \label{eq:notation_state}
  \Ket{\psi}^{RAB}=\bigoplus_{j=0}^{J-1}\sqrt{p\left(j\right)}\Ket{\omega_j}^{a_j^\textup{L} b_j^\textup{L}}\otimes\Ket{\phi_j}^{R a_j^\textup{R} b_j^\textup{R}},
\end{equation}
where $p\left(j\right)$ is a probability distribution.
Using this Koashi-Imoto decomposition,
a protocol for exact state merging can be constructed, as shown in the following theorem.

\begin{theorem}
\label{thm:merge}
    \textit{An achievability bound of entanglement cost in exact state merging applicable to arbitrarily small-dimensional systems.}
    Given any $\Ket{\psi}^{RAB}$ and any $\delta > 0$,
    there exists a protocol for exact state merging of $\Ket{\psi}^{RAB}$ achieving
    \begin{equation}
        \label{eq:merge_cost}
        \log_2 K-\log_2 L \leqq \max_{j}\left\{\log_2\left(\lambda_0^{a_j^\textup{L}}\dim\mathcal{H}^{a_j^\textup{R}}\right)\right\} + \delta,
    \end{equation}
    where $\lambda_0^{a_j^\textup{L}}$ is the largest eigenvalue of
    \begin{equation}
      \omega_j^{a_j^\textup{L}}=\tr_{b_j^\textup{L}}\Ket{\omega_j}\Bra{\omega_j}^{a_j^\textup{L} b_j^\textup{L}},
    \end{equation}
    and the other notations are the same as those in Equations~\eqref{eq:notation_space} and~\eqref{eq:notation_state}.
\end{theorem}

As for the non-catalytic setting of exact state merging, entanglement cost $\log_2 K$ for the initially shared maximally entangled resource state can be reduced compared to $\log_2 K$ in the catalytic setting in Theorem~\ref{thm:merge}.
Note that $\log_2 K$ in the non-catalytic setting may be more than the net entanglement cost $\log_2 K - \log_2 L$ in the catalytic setting in Theorem~\ref{thm:merge}.

\begin{theorem}
\label{thm:merge_without_catalyst}
    \textit{An achievability bound of entanglement cost in the non-catalytic setting of exact state merging applicable to arbitrarily small-dimensional systems.}
    Given any $\Ket{\psi}^{RAB}$,
    there exists a protocol for exact state merging of $\Ket{\psi}^{RAB}$ in the non-catalytic setting achieving
    \begin{equation}
        \label{eq:merge_without_catalyst_cost}
        \log_2 K \leqq \max_{j}\left\{\log_2\left\lceil\lambda_0^{a_j^\textup{L}}\dim\mathcal{H}^{a_j^\textup{R}}\right\rceil\right\},
    \end{equation}
    where $\lceil{}\cdots{}\rceil$ is the ceiling function, and the other notations are the same as those in Theorem~\ref{thm:merge}.
\end{theorem}

In the following, the proofs of Theorems~\ref{thm:merge} and~\ref{thm:merge_without_catalyst} are provided.

\begin{proof}[Proof of Theorem~\ref{thm:merge}]
    The proof is by construction, and a protocol for exact state merging of $\Ket{\psi}^{RAB}$ achieving Inequality~\eqref{eq:merge_cost} is shown in the following.
    Define
    \begin{align}
      j_0&\coloneq\argmax_{j}\left\{\log_2\left(\lambda_0^{a_j^\textup{L}}\dim\mathcal{H}^{a_j^\textup{R}}\right)\right\},\\
      D^{a_j^\textup{R}}&\coloneq\dim\mathcal{H}^{a_j^\textup{R}}\quad
        \text{for each $j\in\left\{0,\ldots,J-1\right\}$.}
    \end{align}

    The protocol uses the Koashi-Imoto decomposition in the following tensor-product form of the Koashi-Imoto decomposition of $\Ket{\psi}^{RAB}$, which is equivalent to that shown in Lemma~\ref{lem:koashi_imoto_decomposition_tripartite} as well as Equations~\eqref{eq:notation_space} and~\eqref{eq:notation_state}.
    Given the Koashi-Imoto decomposition of $\Ket{\psi}^{RAB}$ in the form of Equation~\eqref{eq:notation_state}, this decomposition can also be written using auxiliary systems $\mathcal{H}^{a_0}$ and $\mathcal{H}^{b_0}$ as
    \begin{equation}
      \label{eq:ki_tripartite_isometry}
      \begin{split}
        \left(\mathbb{1}^R\otimes U^A \otimes U^B\right) \Ket{\psi}^{RAB}
        =\sum_{j=0}^{J-1}\sqrt{p\left(j\right)}\Ket{j}^{a_0}\otimes\Ket{j}^{b_0}\otimes\Ket{\omega_j}^{a^L b^L}\otimes\Ket{\phi_j}^{R a^R b^R},
      \end{split}
    \end{equation}
    where $\mathcal{H}^{a_0}$, $\mathcal{H}^{b_0}$, $\mathcal{H}^{a^L}$, $\mathcal{H}^{b^L}$, $\mathcal{H}^{a^R}$, and $\mathcal{H}^{b^R}$ satisfy
    \begin{align}
      \dim\mathcal{H}^{a_0}&=J,\\
      \dim\mathcal{H}^{b_0}&=J,\\
      \dim\mathcal{H}^{a^L}&=\max_j\left\{\dim\mathcal{H}^{a_j^L}\right\},\\
      \dim\mathcal{H}^{b^L}&=\max_j\left\{\dim\mathcal{H}^{b_j^L}\right\},\\
      \dim\mathcal{H}^{a^R}&=\max_j\left\{\dim\mathcal{H}^{a_j^R}\right\},\\
      \dim\mathcal{H}^{b^R}&=\max_j\left\{\dim\mathcal{H}^{b_j^R}\right\},
    \end{align}
    $U^A$ is an isometry from $\mathcal{H}^A$ to $\mathcal{H}^{a_0}\otimes\mathcal{H}^{a^L}\otimes\mathcal{H}^{a^R}$,
    $U^B$ is an isometry from $\mathcal{H}^B$ to $\mathcal{H}^{b_0}\otimes\mathcal{H}^{b^L}\otimes\mathcal{H}^{b^R}$,
    and ${\{\Ket{j}^{a_0}:j=0,\ldots,J-1\}}$ and ${\{\Ket{j}^{b_0}:j=0,\ldots,J-1\}}$ are the computational basis of $\mathcal{H}^{a_0}$ and $\mathcal{H}^{b_0}$, respectively.
    As stressed in Reference~\cite{K3},
    information on $\psi^A$ is encoded in three parts of the Koashi-Imoto decomposition in Equation~\eqref{eq:ki_tripartite_isometry}, namely, $\mathcal{H}^{a_0}$, $\mathcal{H}^{a^\textup{R}}$, and $\mathcal{H}^{a^\textup{L}}$, which can be regarded as the classical part, the quantum part, and the redundant part, respectively.
    In the rest of the proof, the following three subprocesses are presented:
    \begin{enumerate}
      \item Entanglement distillation from the \textit{redundant} part;
      \item Quantum teleportation of the \textit{quantum} part;
      \item Coherent merging of the \textit{classical} part by a measurement.
    \end{enumerate}
    Then, these three subprocesses are combined using controlled measurements and controlled isometries, which are controlled by computational-basis states of $\mathcal{H}^{a_0}$ and $\mathcal{H}^{b_0}$.

    \textit{Subprocess~1: Entanglement distillation from the redundant part.}
    Due to the continuity of $\log_2\,$, there exists a rational number $\tilde{\lambda}_0^{a_{j_0}^\textup{L}}\in\mathbb{Q}$ such that
    \begin{equation}
        \log_2\left(\lambda_{0}^{a_{j_0}^\textup{L}}D^{a^\textup{R}_{j_0}}\right)
        \leqq\log_2\left(\tilde{\lambda}_{0}^{a_{j_0}^\textup{L}}D^{a^\textup{R}_{j_0}}\right)
        \leqq\log_2\left(\lambda_{0}^{a_{j_0}^\textup{L}}D^{a^\textup{R}_{j_0}}\right)+\delta.
    \end{equation}
    Thus, for any $j\in\left\{0,\ldots,J-1\right\}$, it holds that
    \begin{equation}
        \lambda_{0}^{a_{j}^\textup{L}}D^{a^\textup{R}_j}
        \leqq\lambda_{0}^{a_{j_0}^\textup{L}}D^{a^\textup{R}_{j_0}}
        \leqq\tilde{\lambda}_{0}^{a_{j_0}^\textup{L}}D^{a^\textup{R}_{j_0}}.
    \end{equation}
    Hence, it is obtained that
    \begin{equation}
        \lambda_0^{a_{j}^\textup{L}}\leqq\frac{D^{a^\textup{R}_{j_0}}}{D^{a^\textup{R}_j}}{\tilde{\lambda}_0^{a_{j_0}^\textup{L}}},
    \end{equation}
    and since $\tilde{\lambda}_0^{a_{j_0}^\textup{L}}\in\mathbb{Q}$, there exist integers $K_j$ and $L_j$ such that the right-hand side of the above inequality is written as
    \begin{equation}
         \frac{D^{a^\textup{R}_{j_0}}}{D^{a^\textup{R}_j}}{\tilde{\lambda}_0^{a_{j_0}^\textup{L}}}=\frac{K_j}{L_j}.
    \end{equation}
    Therefore, it holds that
    \begin{equation}
        \frac{\lambda_0^{a_{j}^\textup{L}}}{K_j}\leqq \frac{1}{L_j}.
    \end{equation}
    For each $j\in\left\{0,\ldots,J-1\right\}$, the majorization condition for LOCC convertibility between bipartite pure states in Lemma~\ref{lem:pure_convertibility} guarantees that there exists an LOCC map represented by a family of operators 
    \begin{equation}
      {\left\{M_{j,m_1}\otimes U_{j,m_1}\right\}}_{m_1}
    \end{equation}
    achieving for each $m_1\,$,
    \begin{equation}
        \left(M_{j,m_1}\otimes U_{j,m_1}\right)\left(\Ket{\omega_j}^{a^\textup{L} b^\textup{L}}\otimes\Ket{\Phi^+_{K_j}}^{\overline{A}\overline{B}}\right)
        =\Ket{\Phi^+_{L_j}}^{\overline{A}\overline{B}},
    \end{equation}
    where ${\left\{M_{j,m_1}\right\}}_{m_1}$ represents $A$'s measurement from $\mathcal{H}^{a^\textup{L}}\otimes\mathcal{H}^{\overline{A}}$ to $\mathcal{H}^{\overline{A}}$ with outcome $m_1$ satisfying the completeness
    \begin{equation}
      \sum_{m_1} M_{j,m_1}^\dag M_{j,m_1}=\mathbb{1},
    \end{equation}
    and $U_{j,m_1}$ represents $B$'s isometry from $\mathcal{H}^{b^\textup{L}}\otimes\mathcal{H}^{\overline{B}}$ to $\mathcal{H}^{\overline{B}}$ conditioned by $m_1$.
    Regarding an explicit form of ${\left\{M_{j,m_1}\otimes U_{j,m_1}\right\}}_{m_1}\,$, refer to References~\cite{N2,T3}.

    \textit{Subprocess~2: Quantum teleportation of the quantum part.}
    While quantum teleportation for sending the full reduced state $\phi_j^{a^R}\coloneqq\tr_{R b^R}\Ket{\phi_j}\Bra{\phi_j}^{R a^R b^R}$ requires a maximally entangled resource state with Schmidt rank
    \begin{equation}
      \dim\mathcal{H}^{a^R}=\max_j\dim\mathcal{H}^{a_j^R},
    \end{equation}
    a compression method is adopted here instead of just performing quantum teleportation of $\phi_j^{a^R}$, so that each $\phi_j^{a^R}$ is transferred from $A$ to $B$ using a maximally entangled resource state with Schmidt rank $\dim\mathcal{H}^{a_j^R}$, which is smaller than or equal to $\dim\mathcal{H}^{a^R}$.
    Consider $A$'s auxiliary system $\bigotimes_{j=0}^{J-1}\mathcal{H}^{{\left(a^\prime\right)}_j^\textup{R}}$, where $\dim\mathcal{H}^{{\left(a^\prime\right)}_j^\textup{R}}=D^{a_j^\textup{R}}$ for each $j$.
    The state $\Ket{\phi_j}^{R a^\textup{R} b^\textup{R}}$ can be compressed into
    \begin{equation}
            \Ket{\phi_j}^{R{\left(a^\prime\right)}_j^\textup{R} b^\textup{R}}=U'_j\Ket{\phi_j}^{R a^\textup{R} b^\textup{R}},
    \end{equation}
    where $U'_j$ is an isometry from $\mathcal{H}^{a^\textup{R}}$ to $\mathcal{H}^{{\left(a^\prime\right)}_j^\textup{R}}$, and $\Ket{\phi_j}^{R{\left(a^\prime\right)}_j^\textup{R} b^\textup{R}}$ represents the same state as $\Ket{\phi_j}^{R a^\textup{R} b^\textup{R}}$.
    Quantum teleportation~\cite{B5} to send states of $\mathcal{H}^{{\left(a^\prime\right)}_j^\textup{R}}$ consists of $A$'s projective measurement in the maximally entangled basis ${\left\{\Ket{\Phi_{j,m_2}}\right\}}_{m_2}$ on $\mathcal{H}^{{\left(a^\prime\right)}_j^\textup{R}}\otimes\mathcal{H}^{\overline{A}}$ with measurement outcome $m_2$ and $B$'s generalized Pauli correction $\sigma_{j,m_2}$ from $\mathcal{H}^{\overline{B}}$ to $\mathcal{H}^{{(b')}^\textup{R}}$ conditioned by $m_2\,$, where $\mathcal{H}^{{(b')}^\textup{R}}$ is $B$'s auxiliary system corresponding to $\mathcal{H}^{a^\textup{R}}$.
    The map for quantum teleportation is represented by
    \begin{equation}
      {\left\{\Bra{\Phi_{j,m_2}}\otimes\sigma_{j,m_2}\right\}}_{m_2}\,,
    \end{equation}
    which traces out the post-measurement state of $A$ and achieves for each $m_2\,$,
    \begin{equation}
        \begin{split}
            &\left(\Bra{\Phi_{j,m_2}}\otimes\sigma_{j,m_2}\right)\left(\Ket{\phi_j}^{R {\left(a^\prime\right)}_j^\textup{R} b^\textup{R}}\otimes\Ket{\Phi_{D^{a_j^\textup{R}}}^+}^{\overline{A}\overline{B}}\right)\\
            &=\left[\left(\Bra{\Phi_{j,m_2}}U'_j\right)\otimes\sigma_{j,m_2}\right]\left(\Ket{\phi_j}^{R a^\textup{R} b^\textup{R}}\otimes\Ket{\Phi_{D^{a_j^\textup{R}}}^+}^{\overline{A}\overline{B}}\right)\\
            &=\Ket{\phi_j}^{R {(b')}^\textup{R} b^\textup{R}}.
        \end{split}
    \end{equation}

    \textit{Subprocess~3: Coherent merging of the classical part by a measurement.}
    As for the classical part $\mathcal{H}^{a_0}$,
    $A$ performs a measurement to merge the classical part.
    This measurement should be performed without breaking coherence between $R$ and $B$.
    This contrasts with the protocol proposed in Reference~\cite{K4} for transferring a state drawn from a given ensemble, in which a projective measurement onto each of the subspaces of the Koashi-Imoto decomposition indexed by $j$ in Lemma~\ref{lem:koashi_imoto_decomposition_set} destroys superposition of states among different subspaces.
    For coherent merging of the classical part, $A$ performs a projective measurement on $\mathcal{H}^{a_0}$ with outcome $m_3$ in the Fourier basis ${\left\{\Ket{m_3}\right\}}_{m_3}$ defined in terms of the computational basis ${\left\{\Ket{j}^{a_0}\right\}}_j\,$, that is, for each $m_3\,$,
    \begin{equation}
        \Ket{m_3}^{a_0}\coloneq \sum_{j=0}^{J-1}\exp\left(\frac{\textup{i}{\pi}jm_3}{J}\right)\Ket{j}^{a_0}.
    \end{equation}
    After sending the measurement outcome $m_3$ by classical communication from $A$ to $B$,
    the originally given state of $\mathcal{H}^{a_0}\otimes\mathcal{H}^{a_L}\otimes\mathcal{H}^{b_L}$
    can be recovered from $B$'s classical part $\mathcal{H}^{b_0}$ of the post-measurement state by $B$'s local isometry conditioned by $m_3$
    \begin{equation}
      \label{eq:subprocess_3}
      \sum_{j=0}^{J-1}\exp\left(\frac{\textup{i}{\pi}jm_3}{J}\right)\Ket{j}^{{(b')}_0}\otimes\Ket{j}\Bra{j}^{b_0}\otimes\Ket{\omega_j}^{{(b')}^L b^L},
    \end{equation}
    where $\mathcal{H}^{{(b')}_0}\otimes\mathcal{H}^{{(b')}^L}$ is $B$'s auxiliary system corresponding to $\mathcal{H}^{a_0}\otimes\mathcal{H}^{a^L}$.

    Subprocesses~1--3 are combined using controlled measurements and controlled isometries.
    Regarding $A$'s measurement,
    the measurements used in Subprocesses~1 and~2 are performed by extending each measurement to a measurement controlled coherently by the computational-basis state $\Ket{j}^{a_0}$.
    Regarding Subprocess~1 for the redundant part, the controlled version of the measurement is given by
    \begin{equation}
      \sum_{j=0}^{J-1}\Ket{j}\Bra{j}^{a_0}\otimes M_{j,m_1}\,,
    \end{equation}
    and regarding Subprocess~2 for the quantum part, given by
    \begin{equation}
      \sum_{j=0}^{J-1}\Ket{j}\Bra{j}^{a_0}\otimes\left(\Bra{\Phi_{j,m_2}}U'_j\right).
    \end{equation}
    The measurement in Subprocess~3 for the classical part is also represented in terms of the computational basis as
    \begin{equation}
      \sum_{j=0}^{J-1}\Bra{m_3}^{a_0}\left(\Ket{j}\Bra{j}^{a_0}\right)=\sum_{j=0}^{J-1}\exp\left(\frac{-\textup{i}{\pi}jm_3}{J}\right)\Bra{j}^{a_0}.
    \end{equation}
    Combining these three together, $A$'s measurement ${\left\{M_{m_1\,,m_2\,,m_3}\right\}}_{m_1\,,m_2\,,m_3}$ is given by
    \begin{equation}
      M_{m_1\,,m_2\,,m_3}
      =\sum_{j=0}^{J-1}\left[\exp\left(\frac{-\textup{i}{\pi}jm_3}{J}\right)\Bra{j}^{a_0}\right]\otimes\left[\Bra{\Phi_{j,m_2}}U'_j M_{j,m_1}\right].
    \end{equation}
    The completeness of this measurement follows from
    \begin{equation}
        \begin{split}
            &\sum_{m_1\,,m_2\,,m_3}M_{m_1\,,m_2\,,m_3}^\dag M_{m_1\,,m_2\,,m_3}\\
            &=\sum_{m_1\,,m_2\,,m_3}\sum_{j,j'}\left[\exp\left(\frac{\textup{i}{\pi}m_3(j'-j)}{J}\right)\Ket{j'}\Bra{j}\right]
            \otimes\left[M_{j,m_1}^\dag {U'_{j'}}^\dag\Ket{\Phi_{j',m_2}}\Bra{\Phi_{j,m_2}}U'_j M_{j,m_1}\right]\\
            &=\sum_{j}\Ket{j}\Bra{j}
            \otimes\left[\sum_{m_1\,,m_2}M_{j,m_1}^\dag{U'_j}^\dag\Ket{\Phi_{j,m_2}}\Bra{\Phi_{j,m_2}}U'_j M_{j,m_1}\right]\\
            &=\mathbb{1},
        \end{split}
    \end{equation}
    where $\mathbb{1}$ is the identity operator on $\mathcal{H}^{a_0}\otimes\mathcal{H}^{a^\textup{L}}\otimes\mathcal{H}^{a^\textup{R}}\otimes\mathcal{H}^{\overline{A}}$.

    As for $B$'s isometry,
    the isometries in Subprocesses~1 and~2 are also controlled coherently by the computational-basis state $\Ket{j}^{b_0}$.
    Regarding Subprocess~1 for the redundant part, the controlled version of the isometry is given by
    \begin{equation}
      \sum_{j=0}^{J-1}\Ket{j}\Bra{j}^{b_0}\otimes U_{j,m_1}\,,
    \end{equation}
    and regarding Subprocess~2 for the quantum part, given by
    \begin{equation}
      \sum_{j=0}^{J-1}\Ket{j}\Bra{j}^{b_0}\otimes\sigma_{j,m_2}.
    \end{equation}
    The isometry in Subprocess~3 is given by Equation~\eqref{eq:subprocess_3}.
    Combining these three together, $B$'s isometry $U_{m_1\,,m_2\,,m_3}$ is given by
    \begin{equation}
        \begin{split}
          U_{m_1\,,m_2\,,m_3}
          =\sum_{j=0}^{J-1}\exp\left(\frac{\textup{i}{\pi}jm_3}{J}\right)\Ket{j}^{{(b')}_0}\otimes\Ket{j}\Bra{j}^{b_0}\otimes\Ket{\omega_j}^{{(b')}^L b^\textup{L}}
          \otimes\sigma_{j,m_2} U_{j,m_1}.
        \end{split}
    \end{equation}

    Consequently, for any combination $(m_1\,,m_2\,,m_3)$,
    the LOCC map represented by a family of operators ${\left\{M_{m_1\,,m_2\,,m_3}\otimes U_{m_1\,,m_2\,,m_3}\right\}}_{m_1\,,m_2\,,m_3}$ achieves state merging of $\Ket{\psi}^{RAB}$
    \begin{equation}
        \begin{split}
            &\left(M_{m_1\,,m_2\,,m_3}\otimes U_{m_1\,,m_2\,,m_3}\right)\\
            &\quad\left(\Ket{j}^{a_0}\otimes\Ket{j}^{b_0}\otimes\Ket{\omega_j}^{a^\textup{L} b^\textup{L}}\otimes\Ket{\phi_j}^{R a^\textup{R} b^\textup{R}}
            \otimes\Ket{\Phi_{D^{a_j^\textup{R}}}^+}^{\overline{A}\overline{B}}\otimes\Ket{\Phi_{K_j}^+}^{\overline{A}\overline{B}}\right)\\
            &=\Ket{j}^{{(b')}_0}\otimes\Ket{j}^{b_0}\otimes\Ket{\omega_j}^{{(b')}^L b^\textup{L}}\otimes\Ket{\phi_j}^{R {(b')}^\textup{R} b^\textup{R}}
            \otimes\Ket{\Phi_{L_j}^+}^{\overline{A}\overline{B}}.
        \end{split}
    \end{equation}
    Choose $K$ as the least common multiple of the integers $\left\{D^{a_0^R}K_0\,,\ldots,D^{a_{J-1}^R}K_{J-1}\right\}$,
    and this state transformation yields
    \begin{equation}
      \begin{split}
        &\left(M_{m_1\,,m_2\,,m_3}\otimes U_{m_1\,,m_2\,,m_3}\right)\\
        &\quad\left(\Ket{j}^{a_0}\otimes\Ket{j}^{b_0}\otimes\Ket{\omega_j}^{a^L b^L}\otimes\Ket{\phi_j}^{R a^R b^R}\otimes\Ket{\Phi_{K}^+}^{\overline{A}\overline{B}}\right)\\
        &=\Ket{j}^{{(b')}_0}\otimes\Ket{j}^{b_0}\otimes\Ket{\omega_j}^{{(b')}^L b^L}\otimes\Ket{\phi_j}^{R {(b')}^R b^R}\otimes\Ket{\Phi_{L}^+}^{\overline{A}\overline{B}},
      \end{split}
    \end{equation}
    where $L$ is an integer defined as
    \begin{equation}
      L\coloneq\frac{K}{\tilde{\lambda}_0^{a_{j_0}^L}D^{a_{j_0}^R}}=\frac{K}{D^{a_j^R} K_j}L_j\,,\quad \forall j\in\left\{0,\ldots,J-1\right\}.
    \end{equation}
    Then, an LOCC map represented as
    \begin{equation}
      \label{eq:merge}
        {\left\{\left[ M_{m_1\,,m_2\,,m_3}U^A\right] \otimes\left[ \left({\left(U^{B'}\right)}^\dag\otimes {\left(U^B\right)}^\dag\right) U_{m_1\,,m_2\,,m_3}U^B\right]\right\}}_{m_1\,,m_2\,,m_3}
    \end{equation}
    achieves for each $(m_1\,,m_2\,,m_3)$,
    \begin{equation}
      \begin{split}
        &\left(\left[ M_{m_1\,,m_2\,,m_3}U^A\right]\otimes\left[ \left({\left(U^{B'}\right)}^\dag\otimes {\left(U^B\right)}^\dag\right) U_{m_1\,,m_2\,,m_3}U^B\right]\right)\\
        &\quad\left(\Ket{\psi}^{RAB}\otimes\Ket{\Phi_{K}^+}^{\overline{A}\overline{B}}\right)\\
        &=\Ket{\psi}^{RB'B}\otimes\Ket{\Phi_{L}^+}^{\overline{A}\overline{B}},
      \end{split}
    \end{equation}
    where $U^A$ and $U^B$ are those in Equation~\eqref{eq:ki_tripartite_isometry}, and ${\left(U^{B'}\right)}^\dag$ from $\mathcal{H}^{{(b')}_0}\otimes\mathcal{H}^{{(b')}^L}\otimes\mathcal{H}^{{(b')}^\textup{R}}$ to $\mathcal{H}^{B'}=\bigoplus_{j=0}^{J-1}\mathcal{H}^{{(b')}_j^L}\otimes\mathcal{H}^{{(b')}_j^R}$ acts in the same way as ${\left(U^A\right)}^\dag$.

    The protocol represented by the LOCC map shown in Equation~\eqref{eq:merge} achieves the condition on entanglement cost given in Inequality~\ref{eq:merge_cost}, as shown in the following.
    For each $j$, entanglement cost amounts to
    \begin{equation}
      \begin{split}
        &\log_2 D^{a_j^\textup{R}} + \log_2 K_j - \log_2 L_j\\
        &= \log_2 \left(\frac{K_j}{L_j} D^{a_j^\textup{R}}\right)\\
        &= \log_2 \left(\tilde{\lambda}_0^{a_{j_0}^\textup{L}}D^{a_{j_0}^\textup{R}}\right),
      \end{split}
    \end{equation}
    which is independent of $j$.
    Thus, entanglement cost of the whole protocol is given by
    \begin{equation}
      \begin{split}
      &\log_2 K-\log_2 L\\
      &=\log_2 \left(\tilde{\lambda}_0^{a_{j_0}^\textup{L}}D^{a_{j_0}^\textup{R}}\right)\\
      &\leqq \log_2 \left(\lambda_0^{a_{j_0}^\textup{L}}D^{a^\textup{R}_{j_0}}\right)+\delta\\
      &=\max_j\left\{\log_2 \left(\lambda_0^{a_{j_0}^\textup{L}}\dim\mathcal{H}^{a^\textup{R}_{j_0}}\right)\right\}+\delta,
      \end{split}
    \end{equation}
    which yields the conclusion.
\end{proof}

\begin{proof}[Proof of Theorem~\ref{thm:merge_without_catalyst}]
    The proof is by construction, and a protocol for exact state merging of $\Ket{\psi}^{RAB}$ in the non-catalytic setting achieving the equality in~\eqref{eq:merge_without_catalyst_cost} is shown in the following.
    Define for each $j\in\{0,\ldots,J-1\}$,
    \begin{equation}
        D^{a_j^\textup{R}}\coloneq\dim\mathcal{H}^{a_j^\textup{R}}.
    \end{equation}
    The core idea of the protocol is similar to that in Theorem~\ref{thm:merge} using the Koashi-Imoto decomposition in the form of Equation~\eqref{eq:ki_tripartite_isometry}.
    The rest of the proof is given in the same way as the proof of Theorem~\ref{thm:merge}, where Subprocess~2 and Subprocess~3 are the same as those in Theorem~\ref{thm:merge}, and Subprocess~1 is modified as follows, since the resource state is not used catalytically in the entanglement distillation from the redundant part in Subprocess~1.

    \textit{Subprocess~1:}
    For each $j\in\{0,\ldots,J-1\}$, it holds that
    \begin{equation}
        \lambda_{0}^{a_{j}^\textup{L}}D^{a^\textup{R}_j}
        \leqq\left\lceil\lambda_{0}^{a_j^\textup{L}}D^{a_j^\textup{R}}\right\rceil
        \leqq\max_j\left\{ \left\lceil \lambda_0^{a_j^\textup{L}}D^{a^\textup{R}_j}\right\rceil\right\}.
    \end{equation}
    Then, given the resource state $\Ket{\Phi_K^+}$, where
    \begin{equation}
        K=\max_j\left\{\left\lceil \lambda_0^{a_j^\textup{L}}D^{a_j^\textup{R}}\right\rceil\right\},
    \end{equation}
    it holds that
    \begin{equation}
        \frac{\lambda_{0}^{a_{j}^\textup{L}}}{K}\leqq\frac{1}{D^{a_j^\textup{R}}}.
    \end{equation}
    For each $j\in\left\{0,\ldots,J-1\right\}$, the majorization condition for LOCC convertibility between bipartite pure states in Lemma~\ref{lem:pure_convertibility} guarantees that there exists an LOCC map represented by a family of operators ${\left\{M_{j,m_1}\otimes U_{j,m_1}\right\}}_{m_1}$ achieving, for each $m_1\,$,
    \begin{equation}
        \left(M_{j,m_1}\otimes U_{j,m_1}\right)\left(\Ket{\omega_j}^{a^\textup{L} b^\textup{L}}\otimes\Ket{\Phi^+_{K}}^{\overline{A}\overline{B}}\right)
        =\Ket{\Phi^+_{D^{a_j^\textup{R}}}}^{\overline{A}\overline{B}},
    \end{equation}
    where ${\left\{M_{j,m_1}\right\}}_{m_1}$ represents $A$'s measurement from $\mathcal{H}^{a^\textup{L}}\otimes\mathcal{H}^{\overline{A}}$ to $\mathcal{H}^{\overline{A}}$ with outcome $m_1$ satisfying the completeness
    \begin{equation}
      \sum_{m_1} M_{j,m_1}^\dag M_{j,m_1}=\mathbb{1},
    \end{equation}
    and $U_{j,m_1}$ represents $B$'s isometry from $\mathcal{H}^{b^\textup{L}}\otimes\mathcal{H}^{\overline{B}}$ to $\mathcal{H}^{\overline{B}}$ conditioned by $m_1$.

    In the same way as Theorem~\ref{thm:merge},
    $A$'s combined measurement
    \begin{equation}
      {\left\{\Bra{m_1\,,m_2\,,m_3}\right\}}_{m_1\,,m_2\,,m_3}\,,
    \end{equation}
    where the post-measurement state is traced out, is given by
    \begin{equation}
      \Bra{m_1\,,m_2\,,m_3}
      =\sum_{j=0}^{J-1}\left[\exp\left(\frac{-\textup{i}{\pi}jm_3}{J}\right)\Bra{j}^{a_0}\right]\otimes\left[\Bra{\Phi_{j,m_2}}U'_j M_{j,m_1}\right].
    \end{equation}

    Also, $B$'s combined isometry $U_{m_1\,,m_2\,,m_3}$ is given by
    \begin{equation}
        U_{m_1\,,m_2\,,m_3}
        =\sum_{j=0}^{J-1}\exp\left(\frac{\textup{i}{\pi}jm_3}{J}\right)\Ket{j}^{{(b')}_0}\otimes\Ket{j}\Bra{j}^{b_0}\otimes\Ket{\omega_j}^{{(b')}^L b^\textup{L}}
        \otimes\sigma_{j,m_2} U_{j,m_1}.
    \end{equation}

    Consequently, the LOCC map represented by
    \begin{equation}
      \label{eq:merge_without_catalyst}
        {\left\{\left[\Bra{m_1\,,m_2\,,m_3}U^A\right]\otimes
        \left[\left({\left(U^{B'}\right)}^\dag\otimes{\left(U^B\right)}^\dag\right) U_{m_1\,,m_2\,,m_3}U^B\right]\right\}}_{m_1\,,m_2\,,m_3}\,,
    \end{equation}
    achieves for any combination $(m_1\,,m_2\,,m_3)$,
    \begin{equation}
      \begin{split}
        &\left(\left[\Bra{m_1\,,m_2\,,m_3}U^A\right]\otimes \left[\left({\left(U^{B'}\right)}^\dag\otimes{\left(U^B\right)}^\dag\right) U_{m_1\,,m_2\,,m_3}U^B\right]\right)\\
        &\quad\left(\Ket{\psi}^{RAB}\otimes\Ket{\Phi_{K}^+}^{\overline{A}\overline{B}}\right)\\
        &=\Ket{\psi}^{RB'B},
      \end{split}
    \end{equation}
    where $U^A$, $U^B$, and $U^{B'}$ are the same as those in Equation~\eqref{eq:merge}.

    Entanglement cost of the protocol represented by the LOCC map shown in Equation~\eqref{eq:merge_without_catalyst} is given by
    \begin{equation}
      \begin{split}
        &\log_2 K\\
        &= \max_j\left\{\log_2\left\lceil \lambda_0^{a_j^\textup{L}}D^{a_j^\textup{R}}\right\rceil\right\}\\
        &=\log_2\max_j\left\{\log_2\left\lceil \lambda_0^{a_j^\textup{L}}\dim\mathcal{H}^{a_j^\textup{R}}\right\rceil\right\},
      \end{split}
    \end{equation}
    which yields the conclusion.
\end{proof}

\begin{remark}
\label{remark:merge}
    \textit{Comparison between exact state merging and splitting.}
    Entanglement cost in exact state merging is not larger than that in its inverse task, that is, exact state splitting summarized in Section~\ref{sec:split}.
    For any $\Ket{\psi}^{RAB}$, it holds that
    \begin{align}
      &\max_{j}\left\{\log_2\lambda_0^{a_j^\textup{L}}\dim\mathcal{H}^{a_j^\textup{R}}\right\}\leqq\log_2\rank\psi^{A},\\
      &\max_{j}\left\{\log_2\left\lceil\lambda_0^{a_j^\textup{L}}\dim\mathcal{H}^{a_j^\textup{R}}\right\rceil\right\}\leqq\log_2\rank\psi^{A},
    \end{align}
    where the left-hand sides are the optimal entanglement cost in exact state merging shown in Theorems~\ref{thm:merge} and~\ref{thm:merge_without_catalyst}, the right-hand sides are the optimal entanglement cost in exact state splitting shown in Theorem~\ref{thm:split}, $\mathcal{H}^{a_j^\textup{R}}$ is a Hilbert space defined according to the Koashi-Imoto decomposition of $\Ket{\psi}^{RAB}$ in Equation~\eqref{eq:notation_space}, and $\lambda_0^{a_j^\textup{L}}$ is the largest eigenvalue of a reduced state
    \begin{equation}
      \omega_j^{a_j^\textup{L}}=\tr_{b_j^\textup{L}}\Ket{\omega_j}\Bra{\omega_j}^{a_j^\textup{L} b_j^\textup{L}}
    \end{equation}
    of $\Ket{\omega_j}^{a_j^\textup{L} b_j^\textup{L}}$ also defined according to the Koashi-Imoto decomposition of $\Ket{\psi}^{RAB}$ in Equation~\eqref{eq:notation_state}.
    These inequalities can be derived from
    \begin{equation}
      \dim\mathcal{H}^{a_j^\textup{R}}\leqq\rank\psi^{A}
    \end{equation}
    and
    \begin{equation}
      \lambda_0^{a_j^\textup{L}}\leqq 1,
    \end{equation}
    where the former inequality holds by construction of the Koashi-Imoto decomposition, and the latter follows from the normalization of $\omega_j^{a_j^\textup{L}}$.
    Moreover, as will be shown in Implication~\ref{ex:1} in Section~\ref{sec:examples}, entanglement cost in exact state merging can be strictly smaller than that in spitting.
\end{remark}

\begin{remark}
\label{remark:usefulness}
    \textit{Usefulness of the protocols for exact state merging on small and intermediate scales.}
    The obtained protocols for exact state merging outperforms the existing protocols for approximate state merging~\cite{B9,Y9,B12,D7,D6,H10,B10,D5,M,N3,A4,A5,B15,B13,A16,A17} in terms of entanglement cost, as discussed in the following.
    While some of the existing protocols are fully quantum protocols achieved by local operations and quantum communication assisted by shared entanglement,
    replacing the quantum communication in a fully quantum protocol with quantum teleportation yields an entanglement-assisted LOCC protocol corresponding to the fully quantum protocol.
    Using this replacement, the entanglement cost in state merging of $\Ket{\psi}^{RAB}$, which is denoted here by $E_\textup{merge}(\psi)$, is compared in the LOCC framework.

    The protocols in Theorem~\ref{thm:merge} and~\ref{thm:merge_without_catalyst} for exact state merging of $\Ket{\psi}^{RAB}$ require at most as much entanglement cost as that required for quantum teleportation of $\psi^A$, and when the system size for $\psi^A$ is small, these protocols cost less than the existing protocols for approximate state merging in one-shot scenarios.
    Regarding the existing protocols,
    the achievability bounds of $E_\textup{merge}(\psi)$ of the corresponding entanglement-assisted LOCC protocols can be calculated from the analyses in References~\cite{B9,B12,D7,D6,H10,B10,N3}.
    Given $\epsilon > 0$, these achievability bounds are in the form
    \begin{equation}
      E_\textup{merge}(\psi)=\cdots+O\left(\log\frac{1}{\epsilon}\right) \quad\textup{as } \epsilon\to 0,
    \end{equation}
    which diverges to infinity as higher fidelity is pursued.
    For example, from Theorem~4 in Reference~\cite{B10}, the achievability bound of $E_\textup{merge}(\psi)$ of one-shot state merging of $\Ket{\psi}^{RAB}$ within an error $\epsilon>0$ is given by
    \begin{equation}
      {H_{\max}^{\epsilon_1}\left(A|B\right)}_{\psi}+2\log_2 \frac{1}{\epsilon_4}+3,
    \end{equation}
    where $\epsilon=8\epsilon_1+\sqrt{3\epsilon_4}$, and the first term is represented by the smooth conditional max-entropy~\cite{R2,T5,T11} summarized in Appendix~\ref{sec:one_shot_entropies}.
    To achieve $\epsilon=0.02$, the second and third terms amount to
    \begin{equation}
      2\log_2 \frac{1}{\epsilon_4}+3>28.7.
    \end{equation}
    Note that $\epsilon=0.02$ guarantees, in the task of state discrimination of $\Ket{\psi}$ and the final state, the optimal success probability
    \begin{equation}
      P_\textup{succ}=\frac{1}{2}+\frac{1}{4}{\left\|\psi-\psi_\textup{final}\right\|}_1\leqq 51\%,
    \end{equation}
    which is obtained from the Fuchs-van de Graaf inequalities
    \begin{equation}
      \frac{1}{4}{\left\|\psi-\psi_\textup{final}\right\|}_1\leqq \frac{1}{2}\sqrt{1-F^2}.
    \end{equation}
    Thus, given $\Ket{\psi}^{RAB}$ where
    \begin{equation}
      \dim\mathcal{H}^A\leqq 2^{28},
    \end{equation}
    even if
    \begin{equation}
      {H_{\max}^{\epsilon_1}\left(A|B\right)}_{\psi}=0,
    \end{equation}
    the approximate protocols requires more entanglement cost than the protocols in Theorem~\ref{thm:merge} and~\ref{thm:merge_without_catalyst} and even than quantum teleportation.

    Also, as will be discussed in Implication~\ref{ex:1} in Section~\ref{sec:examples}, useful states for distributed quantum information processing, including the Greenberger-Horne-Zeilinger (GHZ) states and multipartite code states for quantum error correcting codes~\cite{G,D,T10,B}, have \textit{nontrivial} Koashi-Imoto decompositions, that is, $J\neq 1$, when these states are regarded as tripartite states.
    In this regard, the protocols for exact state merging are already sufficient for reducing entanglement cost compared to quantum teleportation in these cases relevant to distributed quantum information processing.
\end{remark}

\section{\label{sec:achievability_approximate}Achievability bound for approximate state merging}

The protocols on exact state merging presented in the previous section is extended to its approximate versions, using smoothing~\cite{R2,T5,T11}.
In the following, the catalytic setting in Theorem~\ref{thm:merge} is considered, while extension of the protocol in the non-catalytic setting in Theorem~\ref{thm:merge_without_catalyst} is also possible in the same way.
Note that while allowing small error in smoothing may provide better bounds, the bounds obtained by smoothing usually include optimization over a ball of close states, and exact state merging already suffices for useful examples including those relevant to distributed quantum information processing, as discussed in Remark~\ref{remark:usefulness}.

Given any pure state $\Ket{\psi}^{RAB}$ and an error $\epsilon\geqq 0$,
an achievability bound of entanglement cost in approximate state merging of $\Ket{\psi}^{RAB}$ within $\epsilon$ is obtained as follows.
Consider the Koashi-Imoto decomposition of any pure state $\Ket{\tilde\psi}^{RAB}$ satisfying
\begin{equation}
  {F^2\left(\psi^{RAB},\tilde\psi^{RAB}\right)}\coloneq{\left|\Braket{\psi|\tilde\psi}\right|}^2\geqq 1-{\left(\frac{\epsilon}{2}\right)}^2.
\end{equation}
Due to the Koashi-Imoto decomposition of $\Ket{\tilde\psi}^{RAB}$ shown in Lemma~\ref{lem:koashi_imoto_decomposition_tripartite}, $\mathcal{H}^A$ and $\mathcal{H}^B$ are uniquely decomposed into
\begin{align}
\label{eq:notation_space_approx}
  \mathcal{H}^A=\bigoplus_{j=0}^{J-1}{\mathcal{H}}^{\tilde{a}_j^\textup{L}}\otimes{\mathcal{H}}^{\tilde{a}_j^\textup{R}},\quad
  \mathcal{H}^B=\bigoplus_{j=0}^{J-1}{\mathcal{H}}^{\tilde{b}_j^\textup{L}}\otimes{\mathcal{H}}^{\tilde{b}_j^\textup{R}},
\end{align}
and $\Ket{\tilde\psi}^{RAB}$ is uniquely decomposed into
\begin{equation}
\label{eq:notation_state_approx}
\Ket{\tilde\psi}^{RAB}=\bigoplus_{j=0}^{\tilde{J}-1}\sqrt{\tilde{p}\left(j\right)}\Ket{\tilde{\omega}_j}^{\tilde{a}_j^\textup{L} \tilde{b}_j^\textup{L}}\otimes\Ket{\tilde{\phi}_j}^{R \tilde{a}_j^\textup{R} \tilde{b}_j^\textup{R}},
\end{equation}
where $\tilde{p}\left(j\right)$ is a probability distribution.
Using these notations, Theorem~\ref{thm:merge} on exact state merging is extended to approximate state merging as follows.

\begin{theorem}
\label{thm:approximate}
  \textit{An achievability bound of entanglement cost in approximate state merging applicable to arbitrarily small-dimensional systems.}
  Given any $\Ket{\psi}^{RAB}$,
  any $\epsilon \geqq 0$,
  and any $\delta > 0$,
  there exists a protocol for approximate state merging of $\Ket{\psi}^{RAB}$ within $\epsilon$ achieving
  \begin{equation}
    \label{eq:approximate}
    \log_2 K-\log_2 L\leqq \min_{\Ket{\tilde\psi}}\max_{j}\left\{\log_2\left(\lambda_0^{\tilde{a}_j^\textup{L}}\dim\mathcal{H}^{\tilde{a}_j^\textup{R}}\right)\right\} + \delta,
  \end{equation}
  where the notations are the same as those in Equations~\eqref{eq:notation_space_approx} and~\eqref{eq:notation_state_approx}, $\lambda^{\tilde{a}_j^\textup{L}}_0$ is the largest eigenvalue of $\tilde{\omega}_j^{\tilde{a}_j^\textup{L}}$, and the minimization is over any normalized pure state $\Ket{\tilde\psi}^{RAB}$ satisfying
  \begin{equation}
    {F^2\left(\psi^{RAB},\tilde\psi^{RAB}\right)}\coloneq{\left|\Braket{\psi|\tilde\psi}\right|}^2\geqq 1-{\left(\frac{\epsilon}{2}\right)}^2.
  \end{equation}
\end{theorem}

\begin{proof}
  The proof is by construction, and it is shown that the LOCC map $\tilde{\mathcal{M}}$ for exact state merging of an approximate state $\Ket{\tilde\psi}$ for the minimum in Inequality~\eqref{eq:approximate} achieves approximate state merging of $\Ket{\psi}^{RAB}$ within $\epsilon$.
  To calculate the error in approximate state merging,
  the purified distance $P\left(\rho,\sigma\right)$ of any two states $\rho$ and $\sigma$ is used.
  Using the properties of the purified distance summarized in Section~\ref{sec:decomposition}, it is straightforward to obtain
  \begin{equation}
    \begin{split}
      &P\left(\id^R\otimes\tilde{\mathcal{M}}\left({\psi}^{RAB}\otimes{\Phi^+_K}^{\overline{A}\overline{B}}\right),{\psi}^{RB'B}\otimes{\Phi^+_L}^{\overline{A}\overline{B}}\right)\\
      &\leqq P\Big(\id^R\otimes\tilde{\mathcal{M}}\left({\psi}^{RAB}\otimes{\Phi^+_K}^{\overline{A}\overline{B}}\right), \id^R\otimes\tilde{\mathcal{M}}\left({\tilde\psi}^{RAB}\otimes{\Phi^+_K}^{\overline{A}\overline{B}}\right)\Big)\\
      &\quad +P\Big(\id^R\otimes\tilde{\mathcal{M}}\left({\tilde\psi}^{RAB}\otimes{\Phi^+_K}^{\overline{A}\overline{B}}\right), {\psi}^{RB'B}\otimes{\Phi^+_L}^{\overline{A}\overline{B}}\Big)\\
      &\leqq P\left({\psi}^{RAB}\otimes{\Phi^+_K}^{\overline{A}\overline{B}},{\tilde\psi}^{RAB}\otimes{\Phi^+_K}^{\overline{A}\overline{B}}\right)
      +P\left({\tilde\psi}^{RB'B}\otimes{\Phi^+_L}^{\overline{A}\overline{B}},{\psi}^{RB'B}\otimes{\Phi^+_L}^{\overline{A}\overline{B}}\right)\\
      &= P\left({\psi}^{RAB},{\tilde\psi}^{RAB}\right) +P\left({\tilde\psi}^{RB'B},{\psi}^{RB'B}\right)\\
      &\leqq\frac{\epsilon}{2}+\frac{\epsilon}{2}\\
      &=\epsilon.
    \end{split}
  \end{equation}
  Therefore, it holds that
  \begin{equation}
    \begin{split}
      &F^2\left(\id^R\otimes\tilde{\mathcal{M}}\left({\psi}^{RAB}\otimes{\Phi^+_K}^{\overline{A}\overline{B}}\right),{\psi}^{RB'B}\otimes{\Phi^+_L}^{\overline{A}\overline{B}}\right)\\
      &\geqq 1-\epsilon^2.
    \end{split}
  \end{equation}
\end{proof}

\section{\label{sec:converse}Converse bound for exact state merging}

A converse bound of entanglement cost in exact state merging illustrated in Figure~\ref{fig:merge} is derived in this section.
In other words, this section aims at obtaining a lower bound of the entanglement cost of any possible protocol for exact state merging.
This converse bound improves the existing converse bound in terms of the conditional max-entropy originally shown in Reference~\cite{B9}.
After showing this converse bound, comparison with the existing bound is analyzed, followed by discussing the tightness of the obtained converse bound.

A converse bound for exact state merging is obtained as follows.

\begin{theorem}
\label{thm:new}
\textit{A converse bound of entanglement cost in exact state merging.}
For any state $\Ket{\psi}^{RAB}$ and any protocol for exact state merging of $\Ket{\psi}^{RAB}$,
it holds that
\begin{equation}
    \label{eq:lower_catalytic}
      \log_2 K - \log_2 L
      \geqq \inf\left\{\log_2 K - \log_2 L: \frac{\mathbb{1}_K}{K}\otimes\psi^{B}\prec\frac{\mathbb{1}_L}{L}\otimes\psi^{AB}\right\},
\end{equation}
where $\prec$ denotes majorization for Hermitian operators summarized in Section~\ref{sec:entanglement}.
Also, for any protocol for exact state merging of $\Ket{\psi}^{RAB}$ in the non-catalytic setting,
it holds that
\begin{equation}
    \label{eq:lower_non_catalytic}
      \log_2 K
      \geqq \min\left\{\log_2 K: \frac{\mathbb{1}_K}{K}\otimes\psi^{B}\prec\psi^{AB}\right\},
\end{equation}
where the notations are the same as those in Inequality~\eqref{eq:lower_catalytic}.
\end{theorem}

\begin{proof}
  The proof of Inequality~\eqref{eq:lower_catalytic} is given, while Inequality~\eqref{eq:lower_non_catalytic} can be shown in a similar way by substituting $L$ in the following proof with $1$.

  Any protocol for exact state merging transforms $\Ket{\psi}^{RAB}\otimes\Ket{\Phi_K^+}^{\overline{A}\overline{B}}$ into $\Ket{\psi}^{RB'B}\otimes\Ket{\Phi_L^+}^{\overline{A}\overline{B}}$ by LOCC\@.
  Hence, with respect to the bipartition between $\mathcal{H}^R\otimes\mathcal{H}^A\otimes\mathcal{H}^{\overline{A}}$ and $\mathcal{H}^B\otimes\mathcal{H}^{B^\prime}\otimes\mathcal{H}^{\overline{B}}$, LOCC convertibility between bipartite pure states yields
  the majorization condition in Lemma~\ref{lem:pure_convertibility}
  \begin{equation}
    \frac{\mathbb{1}_K}{K}\otimes\psi^{B}\prec\frac{\mathbb{1}_L}{L}\otimes\psi^{AB}
  \end{equation}
  in terms of Hermitian operators representing their reduced states.
  Since this majorization holds for any $K$ and $L$ achieving exact state merging of $\Ket{\psi}^{RAB}$, Inequality~\eqref{eq:lower_catalytic} is obtained.
\end{proof}

As a corollary of Theorem~\ref{thm:new}, the following converse bound for maximally entangled states in the form of Equation~\eqref{eq:max} is obtained, which is easier to calculate compared to that in Theorem~\ref{thm:new}.
The following analysis in this section may assume that $\psi^R=\frac{\mathbb{1}^R}{D}$ holds for a given state $\Ket{\psi}^{RAB}$ for simplicity,
based on the fact that entanglement cost in exact state merging of $\Ket{\psi}^{RAB}$ and that of $\Ket{\Phi_D^+\left(\psi\right)}^{RAB}$ are the same, as shown in Proposition~\ref{prp:max}.
Note that to calculate the converse bound in the following corollary for any given state $\Ket{\psi}^{RAB}$, first calculate the Schmidt decomposition of $\Ket{\psi}^{RAB}$ to obtain the corresponding maximally entangled state $\Ket{\Phi_D^+\left(\psi\right)}^{RAB}$ from Equation~\eqref{eq:max}, and then apply the corollary to $\Ket{\Phi_D^+\left(\psi\right)}^{RAB}$.

\begin{corollary}
\label{col:tractable_converse}
\textit{A converse bound of entanglement cost in exact state merging derived from Theorem~\ref{thm:new}}.
For any state $\Ket{\psi}^{RAB}$ satisfying $\psi^R=\frac{\mathbb{1}^R}{D}$, and any protocol for exact state merging of $\Ket{\psi}^{RAB}$, it holds that
\begin{equation}
  \label{eq:tractable_lower_catalytic}
  \log_2 K - \log_2 L \geqq \log_2 \left({\lambda_0^B}D\right),
\end{equation}
where $\lambda_0^B$ is the largest eigenvalue of $\psi^B$.
Also, for any protocol for exact state merging in the non-catalytic setting of $\Ket{\psi}^{RAB}$ satisfying $\psi^R=\frac{\mathbb{1}^R}{D}$,
it holds that
\begin{equation}
  \label{eq:tractable_lower_non_catalytic}
  \log_2 K \geqq \log_2 \left\lceil\lambda_0^B D\right\rceil,
\end{equation}
where $\lceil{}\cdots{}\rceil$ is the ceiling function, and $\lambda_0^B$ is the same as that in Equation~\eqref{eq:tractable_lower_catalytic}.
\end{corollary}

\begin{proof}[\textit{Proof of Inequality~\eqref{eq:tractable_lower_catalytic}}]
  Due to Theorem~\ref{thm:new}, exact state merging implies
  \begin{equation}
    \frac{\mathbb{1}_K}{K}\otimes\psi^{B}\prec\frac{\mathbb{1}_L}{L}\otimes\psi^{AB}.
  \end{equation}
  Thus, the largest eigenvalues of the both sides of this majorization satisfy
  \begin{equation}
    \frac{\lambda_0^B}{K}\leqq\frac{1}{DL}.
  \end{equation}
  Hence, it holds that
  \begin{equation}
    \log_2 K - \log_2 L \geqq \log_2\left({\lambda_0^B}D\right).
  \end{equation}
\end{proof}

\begin{proof}[\textit{Proof of Inequality~\eqref{eq:tractable_lower_non_catalytic}}]
    From the same argument as the above, it is obtained that
    \begin{equation}
        \frac{\lambda_0^B}{K}\leqq\frac{1}{D}.
    \end{equation}
    Hence, it holds that
    \begin{equation}
        K\geqq\lambda_0^B D,
    \end{equation}
    and since $K$ is an integer, it holds that
    \begin{equation}
        K\geqq \left\lceil \lambda_0^B D\right\rceil.
    \end{equation}
    Therefore, it is obtained that
    \begin{equation}
        \log_2 K \geqq \log_2 \left\lceil\lambda_0^B D\right\rceil.
    \end{equation}
\end{proof}

Reference~\cite{B9} also provides a converse bound of entanglement cost in exact state merging of any given state $\Ket{\psi}^{RAB}$ in terms of the conditional max-entropy as follows. Note that this converse bound in Reference~\cite{B9} is only shown for one-way LOCC,
while the converse bounds in Theorem~\ref{thm:new} and Corollary~\ref{col:tractable_converse} are applicable to any LOCC map including two-way LOCC\@.

\begin{lemma}
\label{lem:old}
    (Corollary 4.12.\ in Reference~\cite{B9})
    \textit{A converse bound of entanglement cost in exact state merging in Reference~\cite{B9}.}
    For any state $\Ket{\psi}^{RAB}$ and any one-way LOCC protocol for exact state merging of $\Ket{\psi}^{RAB}$,
    where classical communication is performed only from $A$ to $B$,
    it holds that
    \begin{equation}
        \log_2 K - \log_2 L \geqq {H_{\max}(A|B)}_\psi,
    \end{equation}
    where the right-hand side is the conditional max-entropy summarized in Appendix~\ref{sec:one_shot_entropies}.
\end{lemma}

For states in the form of Equation~\eqref{eq:max}, the converse bounds in Theorem~\ref{thm:new} and Corollary~\ref{col:tractable_converse} are at least as tight as the existing bound in Lemma~\ref{lem:old} as shown in the following proposition.
Moreover, Implication~\ref{ex:2} in Section~\ref{sec:examples} will show a case where these converse bounds are strictly tighter than the existing bound.
It is sufficient to show that the converse bound in Corollary~\ref{col:tractable_converse} is at least as tight as that in Lemma~\ref{lem:old}, since Theorem~\ref{thm:new} provides at least as tight bound as that in Corollary~\ref{col:tractable_converse}.

\begin{proposition}
  \textit{Comparison of converse bounds of entanglement cost in exact state merging.}
    For any state $\Ket{\psi}^{RAB}$ satisfying $\psi^R=\frac{\mathbb{1}^R}{D}$,
    it holds that
    \begin{equation}
        \log_2\left({\lambda_0^B}D\right)\geqq{H_{\max}(A|B)}_\psi,
    \end{equation}
    where the notations are the same as those in Corollary~\ref{col:tractable_converse} and Lemma~\ref{lem:old}.
\end{proposition}

\begin{proof}
    Consider the Schmidt decomposition of $\Ket{\psi}^{RAB}$
    \begin{equation}
        \Ket{\psi}^{RAB}=\sum_{l=0}^{D-1} \frac{1}{\sqrt{D}}\Ket{l}^R\otimes\Ket{\psi_l}^{AB}.
    \end{equation}
    Reference~\cite{V2} shows that $2^{{H_{\max}(A|B)}_\psi}$ equals to the following optimization problem, which is a type of optimization problem called semidefinite programming:
    minimize ${\left\|Z^B\right\|}_\infty$ subject to $\mathbb{1}^R \otimes Z^{AB}\geqq\Ket{\psi}\Bra{\psi}^{RAB}$ and $Z^{AB}\geqq 0$.
    The case
    \begin{equation}
      Z^{AB}=D\psi^{AB}
    \end{equation}
    satisfies these constraints:
    \begin{align}
      &\mathbb{1}^R\otimes D\psi^{AB}=\sum_l\Ket{l}\Bra{l}^R\otimes\sum_l\Ket{\psi_l}\Bra{\psi_l}^{AB}\geqq\Ket{\psi}\Bra{\psi}^{RAB};\\
      &D\psi^{AB}\geqq 0.
    \end{align}
    Therefore, it holds that
    \begin{equation}
      \begin{split}
        &\log_2 \left({\lambda_0^B}D\right) = \log_2 {\left\|D\psi^B\right\|}_\infty\\
        &\geqq\min_{Z^{AB}} \log_2{\left\|Z^B\right\|}_\infty ={{H_{\max}(A|B)}_\psi}.
      \end{split}
    \end{equation}
\end{proof}

It is natural to ask how tight the converse bounds in Theorem~\ref{thm:new} and Corollary~\ref{col:tractable_converse} are.
In the following analysis of the tightness, the non-catalytic setting of exact state merging using one-way LOCC from $A$ to $B$ is considered for simplicity,
and the following proposition simplifies the analysis.

\begin{proposition}
\label{prp:equivalence}
    \textit{A necessary and sufficient condition for exact state merging in the non-catalytic setting by one-way LOCC\@.}
    Given any pure state $\Ket{\psi}^{RAB}$ satisfying $\psi^R=\frac{\mathbb{1}^R}{D}$,
    there exists one-way LOCC map $\mathcal{M}^{A\to B}$ from $A$ to $B$ achieving
    \begin{equation}
        \id^R\otimes\mathcal{M}^{A\to B}\left(\psi^{RAB}\otimes{\Phi_K^+}^{\overline{A}\overline{B}}\right)=\psi^{RB'B}
    \end{equation}
    if and only if
    there exists a mixed-unitary channel
    \begin{equation}
      \mathcal{U}(\rho)=\sum_m p\left(m\right) U_m\rho U_m^\dag,
    \end{equation}
    where $p\left(m\right)$ is a probability distribution and $U_m$ for each $m$ is a unitary, achieving
    \begin{equation}
      \label{eq:mixed_unitary}
      \id^R\otimes\mathcal{U}^{\hat{B}}\left({\Phi_D^+}^{R\hat{B}}\right)=\psi^{RB}\otimes\frac{\mathbb{1}_K^{\overline{B}}}{K},
    \end{equation}
    where $\mathcal{H}^{\hat{B}}=\mathcal{H}^{B}\otimes\mathcal{H}^{\overline{B}}$ and $\Ket{\Phi_D^+}^{R\hat{B}}\coloneq\frac{1}{\sqrt{D}}\sum_{l=0}^{D-1}\Ket{l}^R\otimes\Ket{l}^{\hat{B}}$.
\end{proposition}

\begin{proof}
  \textit{If part:} Assume that
  \begin{equation}
    \psi^{RB}\otimes\frac{\mathbb{1}_K^{\overline{B}}}{K}=\sum_m p\left(m\right)\left(\mathbb{1}^R\otimes U_m^{\hat{B}}\right){\Phi_D^+}^{R{\hat{B}}}{\left(\mathbb{1}^R\otimes U_m^{\hat{B}}\right)}^\dag.
  \end{equation}
  A purification yields
  \begin{equation}
    \left(\mathbb{1}^{RB \overline{B}}\otimes U \right)\left(\Ket{\psi}^{RAB}\otimes\Ket{\Phi_K^+}^{\overline{A}\overline{B}}\right)
    =\sum_m \sqrt{p\left(m\right)}\Ket{m}^{A_0}\otimes\left(\mathbb{1}^R\otimes U_m^{\hat{B}}\right)\Ket{\Phi_D^+}^{R{\hat{B}}},
  \end{equation}
  where $\mathcal{H}^{A_0}$ is $A$'s auxiliary system, and $U$ is an isometry performed by $A$.
  Hence, one-way LOCC from $A$ to $B$ represented by ${\left\{\left(\Bra{m}^{A_0} U\right)\otimes {\left(U_m^{\hat{B}}\right)}^\dag\right\}}_m\,$, where the post-measurement state of $A$ is traced out, achieves, for each $m$,
  \begin{equation}
    \mathbb{1}^{R}\otimes\left[\left(\Bra{m}^{A_0} U\right)\otimes {\left(U_m^{\hat{B}}\right)}^\dag\right]\left(\Ket{\psi}^{RAB}\otimes\Ket{\Phi_K^+}^{\overline{A}\overline{B}}\right)
    \propto\Ket{\Phi_D^+}^{R{\hat{B}}},
  \end{equation}
  and $\Ket{\Phi_D^+}^{R{\hat{B}}}$ on the right-hand side can be transformed into $\Ket{\psi}^{RB^\prime B}$ by $B$'s local isometry.

  \textit{Only if part:} Assume that there exists $A$'s POVM ${\left\{\Lambda_m\right\}}_m$ on $\mathcal{H}^A\otimes\mathcal{H}^{\overline{A}}$ satisfying for each $m$
  \begin{equation}
    \tr_A\left[\left(\mathbb{1}^{RB\overline{B}}\otimes\Lambda_m\right)\left(\psi^{RAB}\otimes{\Phi_K^+}^{\overline{A}\overline{B}}\right)\right]
    =p\left(m\right){\left(\mathbb{1}^{R}\otimes U_m^{\hat{B}}\right)}{\Phi_D^+}^{R{\hat{B}}}{\left(\mathbb{1}^{R}\otimes U_m^{\hat{B}}\right)}^\dag,
  \end{equation}
  where $p\left(m\right)$ is a probability distribution, and $U_m^{\hat{B}}$ is $B$'s unitary correction conditioned by $m$.
  Note that ${\Phi_D^+}^{R{\hat{B}}}$ on the right-hand side can be transformed into $\Ket{\psi}^{RB^\prime B}$ by $B$'s local isometry.
  Then, it holds that
  \begin{equation}
    \begin{split}
      &\psi^{RB}\otimes\frac{\mathbb{1}_K^{\overline{B}}}{K}\\
      &=\sum_m\tr_A\left[\left(\mathbb{1}^{RB\overline{B}}\otimes\Lambda_m\right)\left(\psi^{RAB}\otimes{\Phi_K^+}^{\overline{A}\overline{B}}\right)\right]\\
      &=\sum_m p\left(m\right)\left(\mathbb{1}^R\otimes U_m^{\hat{B}}\right){\Phi_D^+}^{R{\hat{B}}}{\left(\mathbb{1}^R\otimes U_m^{\hat{B}}\right)}^\dag\\
      &=\id^R\otimes\mathcal{U}^{\hat{B}}\left({\Phi_D^+}^{R\hat{B}}\right).
    \end{split}
  \end{equation}
\end{proof}

Note that it is straightforward to generalize the above proof of Proposition~\ref{prp:equivalence} in the non-catalytic setting to the catalytic setting, that is,
\begin{equation}
  \begin{split}
    &\id^R\otimes\mathcal{M}^{A\to B}\left(\psi^{RAB}\otimes{\Phi_K^+}^{\overline{A}\overline{B}}\right)={\psi}^{RB^\prime B}\otimes{\Phi_L^+}^{\overline{A}\overline{B}}\\
    &\Leftrightarrow\id^R\otimes\mathcal{U}^{\hat{B}}\left({\Phi_D^+}^{R{\hat{B}}}\otimes\frac{\mathbb{1}_L^{\hat{B}}}{L}\right)=\psi^{RB}\otimes\frac{\mathbb{1}_K^{\overline{B}}}{K},
  \end{split}
\end{equation}
which can also be shown for quantum state redistribution in approximate scenarios~\cite{B15,B13}.

For qubits, the converse bound in Corollary~\ref{col:tractable_converse} is tight enough to provide the optimal entanglement cost, as shown in the following.
Note that an equivalent condition in terms of Schmidt coefficients of $\Ket{\psi_l}^{AB}$ in Equation~\eqref{eq:max} is also given in Theorem~II.1.\ in Reference~\cite{O}.
\begin{theorem}
\label{thm:qubit}
    \textit{Optimal entanglement cost in exact state merging in the non-catalytic setting for qubits.}
    Consider any three-qubit pure state $\Ket{\psi}^{RAB}\in{\left(\mathbb{C}^2\right)}^{\otimes 3}$ satisfying $\psi^R=\frac{\mathbb{1}^R}{2}$,
    exact state merging of $\Ket{\psi}^{RAB}$ in the non-catalytic setting is achievable if and only if
    \begin{equation}
        \log_2 K \geqq \log_2 \left\lceil\lambda_0^B D\right\rceil,
    \end{equation}
    where the notations are the same as those in Corollary~\ref{col:tractable_converse}.
    Equivalently,
    the exact state merging of $\Ket{\psi}^{RAB}$ is achievable at entanglement cost $\log_2 K = 0$ if and only if $\psi^B=\frac{\mathbb{1}^B}{2}$,
    and otherwise, entanglement cost $\log_2 K = 1$ is required.
\end{theorem}

\begin{proof}
    \textit{If part:}
    Assume that
    \begin{equation}
      \psi^B=\frac{\mathbb{1}^B}{2},
    \end{equation}
    and the existence of an LOCC protocol for exact state merging of $\Ket{\psi}^{RAB}$ achieving $\log_2 K=0$ is shown.
    Note that otherwise, quantum teleportation of $\psi^A$ achieves $\log_2 K = 1$.
    To show the existence of the LOCC protocol, Proposition~\ref{prp:equivalence} implies that it is sufficient to prove the existence of a mixed-unitary channel $\mathcal{U}$ achieving
    \begin{equation}
      \label{eq:qubit}
      \id^R\otimes\mathcal{U}^B\left({\Phi_2^+}^{RB}\right)=\psi^{RB}.
    \end{equation}
    Note that $\mathcal{H}^{\hat{B}}$ in Equation~\eqref{eq:mixed_unitary} in Proposition~\ref{prp:equivalence} is simply written as $\mathcal{H}^B$ in Equation~\eqref{eq:qubit} since $\mathcal{H}^{\hat{B}}=\mathcal{H}^B$ in this proof.

    Given $\psi^{RB}$ satisfying $\psi^R=\frac{\mathbb{1}^R}{2}$,
    $\psi^{RB}$ can be regarded as a normalized operator of the Choi operator of a CPTP map $\mathcal{U}^B$ defined as Equation~\eqref{eq:choi_operator}.
    Tracing out $\mathcal{H}^R$ for $\psi^{RB}$ yields
    \begin{equation}
        \mathcal{U}^B\left(\frac{\mathbb{1}^B}{2}\right)=\psi^B=\frac{\mathbb{1}^B}{2},
    \end{equation}
    and hence, $\mathcal{U}^B$ is a unital channel.
    Since any unital channel on a qubit is a mixed-unitary channel, $\mathcal{U}^B$ is a mixed-unitary channel, which yield the conclusion.
\end{proof}

As for qudits of more than two dimension, the converse bound in Theorem~\ref{thm:new} is not necessarily achievable, since the following proposition shows an example of exact state merging that does not satisfy the equality of~\eqref{eq:lower_non_catalytic}.
The following proposition shows a three-qutrit state of which any one-way LOCC protocol for exact state merging in the non-catalytic setting fails to achieve
\begin{equation}
    \log_2 K =\min\left\{\log_2 K: \frac{\mathbb{1}_K}{K}\otimes\psi^{B}\prec\psi^{AB}\right\}.
\end{equation}
\begin{proposition}
\label{prp:qutrit}
  \textit{Impossibility of achieving the converse bound of entanglement cost in exact state merging in the non-catalytic setting for qutrits.}
  There exists a three-qutrit pure state $\Ket{\psi}^{RAB}\in{\left(\mathbb{C}^3\right)}^{\otimes 3}$ satisfying $\psi^R=\frac{\mathbb{1}^R}{D}$ where $D=3$, such that exact state merging of $\Ket{\psi}^{RAB}$ in the non-catalytic setting cannot be achieved by any one-way LOCC protocol at entanglement cost
  \begin{equation}
    \log_2 K =\min\left\{\log_2 K: \frac{\mathbb{1}_K}{K}\otimes\psi^{B}\prec\psi^{AB}\right\},
  \end{equation}
  where the notations are the same as those in Theorem~\ref{thm:new}.
\end{proposition}

\begin{proof}
    Consider a CPTP map
    \begin{equation}
      \mathcal{N}(\rho)=\frac{1}{2}(\tr\rho)\mathbb{1}-\frac{1}{2}\rho^\textup{T},
    \end{equation}
    where $\rho^\textup{T}$ is the transpose of $\rho$ with respect to the computational basis.
    The Choi operator of $\mathcal{N}$ defined as Equation~\eqref{eq:choi_operator} is written as
    \begin{equation}
        \begin{split}
          &J(\mathcal{N})\coloneq\\
          &\frac{1}{2}\left(\Ket{2}\otimes\Ket{1}-\Ket{1}\otimes\Ket{2}\right){\left(\Bra{2}\otimes\Bra{1}-\Bra{1}\otimes\Bra{2}\right)}+\\
          &\frac{1}{2}\left(\Ket{0}\otimes\Ket{2}-\Ket{2}\otimes\Ket{0}\right){\left(\Bra{0}\otimes\Bra{2}-\Bra{2}\otimes\Bra{0}\right)}+\\
          &\frac{1}{2}\left(\Ket{1}\otimes\Ket{0}-\Ket{0}\otimes\Ket{1}\right){\left(\Bra{1}\otimes\Bra{0}-\Bra{0}\otimes\Bra{1}\right)}.
        \end{split}
    \end{equation}
    This map $\mathcal{N}$ is a unital channel but not a mixed-unitary channel~\cite{L4,W11}.

    Consider a normalized state
    \begin{equation}
        \psi^{RB}\coloneq\frac{J(\mathcal{N})}{3}.
    \end{equation}
    A purification of $\psi^{RB}$ is
    \begin{equation}
        \begin{split}
          &\Ket{\psi}^{RAB}=\\
          &\frac{1}{\sqrt{3}}\Ket{0}^A\otimes{\left(\frac{1}{\sqrt{2}}\Ket{2}^R\otimes\Ket{1}^B-\frac{1}{\sqrt{2}}\Ket{1}^R\otimes\Ket{2}^B\right)}+\\
          &\frac{1}{\sqrt{3}}\Ket{1}^A\otimes{\left(\frac{1}{\sqrt{2}}\Ket{0}^R\otimes\Ket{2}^B-\frac{1}{\sqrt{2}}\Ket{2}^R\otimes\Ket{0}^B\right)}+\\
          &\frac{1}{\sqrt{3}}\Ket{2}^A\otimes{\left(\frac{1}{\sqrt{2}}\Ket{1}^R\otimes\Ket{0}^B-\frac{1}{\sqrt{2}}\Ket{0}^R\otimes\Ket{1}^B\right)}
        \end{split}
    \end{equation}
    This state satisfies
    \begin{align}
      &\psi^R=\frac{\mathbb{1}^R}{3},\\
      &\psi^B=\frac{\mathbb{1}^B}{3}.
    \end{align}
    Hence, it holds that
    \begin{equation}
        \min\left\{\log_2 K: \frac{\mathbb{1}_K}{K}\otimes\psi^{B}\prec\psi^{AB}\right\}=0.
    \end{equation}

    Assume that there exists a one-way LOCC protocol for exact state merging of $\Ket{\psi}^{RAB}$ in the non-catalytic setting at entanglement cost $\log_2 K=0$, to derive a contradiction.
    Due to Proposition~\ref{prp:equivalence}, this assumption is equivalent to the existence of a mixed-unitary channel $\mathcal{U}^B$ such that
    \begin{equation}
        \id^R\otimes\mathcal{U}^B\left({\Phi_3^+}^{RB}\right)=\psi^{RB}=\frac{J(\mathcal{N})}{3},
    \end{equation}
    where, in the same way as Equation~\eqref{eq:qubit}, $\mathcal{H}^{\hat{B}}$ in Equation~\eqref{eq:mixed_unitary} in Proposition~\ref{prp:equivalence} is written as $\mathcal{H}^B$.
    Due to the one-to-one correspondence between a CPTP map and the Choi operator of the CPTP map, $\mathcal{N}=\mathcal{U}$ is necessary, which contradicts to the fact that $\mathcal{N}$ is not a mixed-unitary channel, and the conclusion is obtained.
\end{proof}

\section{\label{sec:converse_approximate}Converse bound for approximate state merging}

Given any pure state $\Ket{\psi}^{RAB}$ and an error $\epsilon\geqq 0$,
Theorem~\ref{thm:new} is extended in this section, to obtain a converse bound of entanglement cost in approximate state merging of $\Ket{\psi}^{RAB}$ within $\epsilon$.
While the catalytic setting is analyzed in the following, the same argument holds for the non-catalytic setting.
It is also shown that this converse bound for approximate state merging improves the converse bound derived from the previous study on one-shot approximate state redistribution~\cite{B10} when $\epsilon$ is sufficiently small.

In the same way as the proof of Theorem~\ref{thm:new} on exact state merging,
a converse bound of entanglement cost in approximate state merging is obtained by applying a majorization condition for LOCC convertibility to the bipartition between $B$ and $RA$.
While the proof of Theorem~\ref{thm:new} on exact state merging uses the majorization condition for LOCC convertibility between bipartite pure states in Lemma~\ref{lem:pure_convertibility},
approximate state merging requires another majorization condition for LOCC convertibility from a bipartite pure state to a bipartite mixed state shown in Lemma~\ref{lem:mixed}, since the final state in approximate state merging can be a mixed state.
Given any pure state $\Ket{\psi}^{RAB}$ and an error $\epsilon\geqq 0$,
a converse bound of entanglement cost in approximate state merging of $\Ket{\psi}^{RAB}$ within $\epsilon$ is obtained using Lemma~\ref{lem:mixed} as follows.

\begin{theorem}
\label{thm:approximate_converse}
\textit{A converse bound of entanglement cost in approximate state merging.}
For any state $\Ket{\psi}^{RAB}$, any error $\epsilon\geqq 0$, and any protocol for approximate state merging of $\Ket{\psi}^{RAB}$ within $\epsilon$,
\begin{equation}
  \begin{split}
    &\log_2 K - \log_2 L \geqq\\
    &\inf\Big\{\log_2 K - \log_2 L:\\
    &\quad\boldsymbol{\lambda}\left(\psi^{B}\otimes\frac{\mathbb{1}_K^{\overline{B}}}{K}\right)\prec\sum_j p(j)\boldsymbol{\lambda}\left(\psi_j^{B^\prime B\overline{B}}\right),\\
    &\quad\left.F^2\left(\sum_j p(j)\Ket{\psi_j}\Bra{\psi_j}^{RB^\prime B\overline{A}\overline{B}},\psi^{RB^\prime B}\otimes{\Phi_L^+}^{\overline{A}\overline{B}}\right)\geqq 1-\epsilon^2\right\}.
  \end{split}
\end{equation}
\end{theorem}

\begin{proof}
  Any protocol for approximate state merging transforms $\Ket{\psi}^{RAB}\otimes\Ket{\Phi_K^+}^{\overline{A}\overline{B}}$ into ${\tilde\psi}^{RB^\prime B \overline{A}\overline{B}}$ by LOCC, where ${\tilde\psi}$ satisfies
  \begin{equation}
    F^2\left({\tilde\psi}^{RB^\prime B \overline{A}\overline{B}},\psi^{RB^\prime B}\otimes{\Phi_L^+}^{\overline{A}\overline{B}}\right)\geqq 1-\epsilon_1^2.
  \end{equation}
  Substituting $A$, $B$, $\Ket{\phi}^{AB}$, and $\psi^{AB}$ in Lemma~\ref{lem:mixed} with $R\overline{A}$, $B^\prime B \overline{B}$, $\Ket{\psi}^{RAB}\otimes\Ket{\Phi_K^+}^{\overline{A}\overline{B}}$, and ${\tilde\psi}^{RB^\prime B \overline{A}\overline{B}}$, respectively, yields an ensemble $\left\{p(j),\Ket{\psi_j}^{RB^\prime B \overline{A}\overline{B}}\right\}$ satisfying
  \begin{align}
    &{\tilde\psi}^{RB^\prime B \overline{A}\overline{B}}=\sum_j p(j)\Ket{\psi_j}\Bra{\psi_j}^{RB^\prime B \overline{A}\overline{B}},\\
    &\boldsymbol{\lambda}\left(\psi^{B}\otimes\frac{\mathbb{1}_K^{\overline{B}}}{K}\right)\prec\sum_j p(j)\boldsymbol{\lambda}\left(\psi_j^{B^\prime B\overline{B}}\right).
  \end{align}
  Therefore, the conclusion is obtained.
\end{proof}

Reference~\cite{B10} also analyzes a converse bound for fully quantum protocols for one-shot approximate state redistribution, which is a generalized task including approximate state merging as a special case.
As discussed in Remark~\ref{remark:usefulness}, it is straightforward to convert this converse bound for fully quantum protocols to the converse bound of entanglement cost in the LOCC framework, which yields the following.

\begin{lemma}
\label{lem:existing_approxiamte_converse}
  (Proposition~12 in Reference~\cite{B10})
  \textit{A converse bound of entanglement cost in approximate state merging given in Reference~\cite{B10}.}
  For any state $\Ket{\psi}^{RAB}$, any errors $\epsilon_1\in(0,1)$, $\epsilon_2\in(0,1-\epsilon_1)$, and any protocol for approximate state merging of $\Ket{\psi}^{RAB}$ within $\epsilon_1\,$,
  it holds that
  \begin{equation}
    \log_2 K - \log_2 L \geqq {H_{\min}^{\epsilon_2}(AB)}_\psi-{H_{\min}^{\epsilon_1+\epsilon_2}(B)}_\psi,
  \end{equation}
  where $H_{\min}^{\epsilon}$ is the smooth min-entropy summarized in Appendix~\ref{sec:one_shot_entropies}.
\end{lemma}

When the error tolerance in approximate state merging is sufficiently small,
the converse bound shown in Theorem~\ref{thm:approximate_converse} improves the converse bound shown in Lemma~\ref{lem:existing_approxiamte_converse} in the following sense.
\begin{proposition}
  \textit{Comparison of converse bounds of entanglement cost in approximate state merging.}
  For any state $\Ket{\psi}^{RAB}$, any errors $\epsilon_1\in(0,1)$, $\epsilon_2\in(0,1-\epsilon_1)$, and any protocol for approximate state merging of $\Ket{\psi}^{RAB}$ within $\epsilon_1\,$,
  \begin{equation}
    \begin{split}
      &\lim_{\epsilon_1\to 0}\inf\Big\{\log_2 K - \log_2 L:\\
      &\quad\boldsymbol{\lambda}\left(\psi^{B}\otimes\frac{\mathbb{1}_K^{\overline{B}}}{K}\right)\prec\sum_j p(j)\boldsymbol{\lambda}\left(\psi_j^{B^\prime B\overline{B}}\right),\\
      &\quad F^2\Big(\sum_j p(j)\Ket{\psi_j}\Bra{\psi_j}^{RB^\prime B\overline{A}\overline{B}},\psi^{RB^\prime B}\otimes{\Phi_L^+}^{\overline{A}\overline{B}}\Big)\geqq 1-\epsilon^2\Big\}\\
      &\geqq\lim_{\epsilon_1,\epsilon_2\to 0}\left({H_{\min}^{\epsilon_2}(AB)}_\psi-{H_{\min}^{\epsilon_1+\epsilon_2}(B)}_\psi\right),
    \end{split}
  \end{equation}
  where the notations are the same as those in Theorem~\ref{thm:approximate_converse} and Lemma~\ref{lem:existing_approxiamte_converse}.
\end{proposition}

\begin{proof}
  Regarding the converse bound shown in Theorem~\ref{thm:approximate_converse}, it holds that
  \begin{equation}
    \begin{split}
      &\lim_{\epsilon_1\to 0}\inf\Big\{\log_2 K - \log_2 L:\\
      &\quad\boldsymbol{\lambda}\left(\psi^{B}\otimes\frac{\mathbb{1}_K^{\overline{B}}}{K}\right)\prec\sum_j p(j)\boldsymbol{\lambda}\left(\psi_j^{B^\prime B\overline{B}}\right),\\
      &\quad F^2\Big(\sum_j p(j)\Ket{\psi_j}\Bra{\psi_j}^{RB^\prime B\overline{A}\overline{B}},\psi^{RB^\prime B}\otimes{\Phi_L^+}^{\overline{A}\overline{B}}\Big)\geqq 1-\epsilon_1^2\Big\}\\
      &=\inf\left\{\log_2 K - \log_2 L: \frac{\mathbb{1}_K}{K}\otimes\psi^{B}\prec\frac{\mathbb{1}_L}{L}\otimes\psi^{AB}\right\}.
    \end{split}
  \end{equation}
  As for the converse bound shown in Lemma~\ref{lem:existing_approxiamte_converse}, the limit can be calculated as~\cite{R2,T5,T11}
  \begin{equation}
    \lim_{\epsilon_1,\epsilon_2\to 0}\left({H_{\min}^{\epsilon_2}(AB)}_\psi-{H_{\min}^{\epsilon_1+\epsilon_2}(B)}_\psi\right)=\log_2 \frac{1}{\lambda_0^{AB}} - \log_2 \frac{1}{\lambda_0^{B}},
  \end{equation}
  where $\lambda_0^{AB}$ and $\lambda_0^{B}$ are the largest eigenvalues of $\psi^{AB}$ and $\psi^B$, respectively.

  The majorization
  \begin{equation}
    \frac{\mathbb{1}_K}{K}\otimes\psi^{B}\prec\frac{\mathbb{1}_L}{L}\otimes\psi^{AB}
  \end{equation}
  implies that the largest eigenvalues of this majorization satisfy
  \begin{equation}
    \frac{\lambda_0^{B}}{K}\leqq\frac{\lambda_0^{AB}}{L},
  \end{equation}
  and hence,
  \begin{equation}
    \log_2 K - \log_2 L \geqq\log_2 \frac{1}{\lambda_0^{AB}} - \log_2 \frac{1}{\lambda_0^{B}}.
  \end{equation}
  Due to this implication, it holds that
  \begin{equation}
    \begin{split}
      &\inf\left\{\log_2 K - \log_2 L: \frac{\mathbb{1}_K}{K}\otimes\psi^{B}\prec\frac{\mathbb{1}_L}{L}\otimes\psi^{AB}\right\}
      \geqq\log_2 \frac{1}{\lambda_0^{AB}} - \log_2 \frac{1}{\lambda_0^{B}},
    \end{split}
  \end{equation}
  which yields the conclusion.
\end{proof}

\section{\label{sec:examples}Implications}

Implications of the results in this chapter are discussed. In the following, $\otimes$ in representing the tensor product of states may be omitted for brevity.
Define
\begin{align}
  \Ket{+}&\coloneq\frac{1}{\sqrt{2}}\left(\Ket{0}+\Ket{1}\right),\\
  \Ket{\Psi^\pm}&\coloneq\frac{1}{\sqrt{2}}\left(\Ket{0}\Ket{1}\pm\Ket{1}\Ket{0}\right),\\
  \Ket{\Phi^\pm}&\coloneq\frac{1}{\sqrt{2}}\left(\Ket{0}\Ket{0}\pm\Ket{1}\Ket{1}\right).
\end{align}

\begin{implication}
\label{ex:1}
\textit{Reduced entanglement cost in exact state merging compared with quantum teleportation and exact state splitting by performing a measurement on the classical part followed by classical communication.}
Consider a tripartite Greenberger-Horne-Zeilinger (GHZ) state of $d$-dimensional systems for any $d\geqq 2$
\begin{equation}
  \Ket{\textup{GHZ}_d}^{RAB}\coloneq\frac{1}{\sqrt{d}}\sum_{l=0}^{d-1}\Ket{l}^R\Ket{l}^A\Ket{l}^B.
\end{equation}
Quantum teleportation of the reduced state of $\Ket{\textup{GHZ}_d}^{RAB}$ on $A$ requires $\log_2 d$ ebits, that is, $\Ket{\Phi_d^+}$ for an initial resource state.
Note that exact state splitting summarized in Section~\ref{sec:split} also requires $\log_2 d$ ebits due to Theorem~\ref{thm:split}.
By contrast, the protocols for exact state merging of $\Ket{\textup{GHZ}_d}^{RAB}$ in Theorems~\ref{thm:merge} and~\ref{thm:merge_without_catalyst} achieve respectively
\begin{equation}
  \log_2 K - \log_2 L = 0<\log_2 d
\end{equation}
and
\begin{equation}
  \log_2 K  = 0<\log_2 d.
\end{equation}
In a similar way, as will be shown in Chapter~\ref{sec:distributed_encoding_decoding}, the protocol for exact state merging can be used for achieving \textit{zero} entanglement cost in exact state merging of multipartite code states of quantum error correcting codes~\cite{G,D,T10,B}.
\end{implication}

\begin{implication}
\label{ex:3}
\textit{Negative entanglement cost in exact state merging by entanglement distillation from the redundant part.}
Consider a pure state
\begin{equation}
  \begin{split}
    \Ket{\psi}^{RAB}=\frac{1}{\sqrt{3}}\Big(&\Ket{0}^R\Ket{\Psi^+}^{A_1B_1}\Ket{\Phi^-}^{A_2B_2}\Ket{\Phi^+}^{A_3B_3}+\\
                                            &\Ket{1}^R\Ket{0}^{A_1}\Ket{0}^{B_1}\Ket{\Phi^-}^{A_2B_2}\Ket{\Phi^+}^{A_3B_3}+\\
                                            &\Ket{2}^R\Ket{2}^{A_1}\Ket{2}^{B_1}\Ket{0}^{A_2}\Ket{0}^{B_2}\Ket{\Psi^-}^{A_3B_3}\Big),
  \end{split}
\end{equation}
where each of $\mathcal{H}^A=\mathcal{H}^{A_1}\otimes\mathcal{H}^{A_2}\otimes\mathcal{H}^{A_3}$ and $\mathcal{H}^B=\mathcal{H}^{B_1}\otimes\mathcal{H}^{B_2}\otimes\mathcal{H}^{B_3}$ is of $3\times 2\times 2=12$ dimension.
Quantum teleportation of $\psi^A$ requires $\log_2 12$ ebits, that is, $\Ket{\Phi_{12}^+}$ for an initial resource state.
By contrast, the protocols for exact state merging of $\Ket{\psi}^{RAB}$ in Theorems~\ref{thm:merge} and~\ref{thm:merge_without_catalyst} achieve respectively
\begin{equation}
  \log_2 K - \log_2 L = -1 < 0
\end{equation}
and
\begin{equation}
  \log_2 K = 0.
\end{equation}
The former negative entanglement cost leads to a net gain of shared entanglement.
\end{implication}

\begin{implication}
\label{ex:2}
\textit{Improvement in converse bounds of entanglement cost in exact state merging.}
Consider a three-qubit pure state
\begin{equation}
    \Ket{\psi}^{RAB}=\frac{1}{\sqrt{2}}\Big(\Ket{0}^R\Ket{\Psi^+}^{AB} +\Ket{1}^R\Ket{0}^A\Ket{0}^{B}\Big).
\end{equation}
The protocols for exact state merging of $\Ket{\psi}^{RAB}$ in Theorems~\ref{thm:merge} and~\ref{thm:merge_without_catalyst} require respectively
\begin{equation}
  \log_2 K - \log_2 L = 1
\end{equation}
and
\begin{equation}
  \log_2 K = 1.
\end{equation}
Since $\psi^B\neq\frac{\mathbb{1}^B}{2}$,
the latter equality for exact state merging in the non-catalytic setting is optimal due to Theorem~\ref{thm:qubit}.
As for the former, this example shows the difference between the converse bounds of entanglement cost in exact state merging in Theorem~\ref{thm:new} and Lemma~\ref{lem:old}.
In this case,
\begin{align}
  &\log_2 \left({\lambda_0^B}D\right)=\log_2 \frac{3}{2} > 0.5849,\\
  &{H_{\max}(A|B)}_\psi < 0.5432,
\end{align}
where the notations are the same as those in Theorem~\ref{thm:new} and Lemma~\ref{lem:old},
and the value of ${H_{\max}(A|B)}_\psi$ is calculated by a semidefinite programming~\cite{V2} using Split Conic Solver (SCS)~\cite{S8} and YALMIP~\cite{L5}.
These calculations imply that the converse bounds in Theorem~\ref{thm:new} and Corollary~\ref{col:tractable_converse} can be strictly tighter than the existing converse bound obtained from Lemma~\ref{lem:old}.
\end{implication}

\begin{implication}
\label{ex:4}
\textit{Asymmetry between $A$ and $B$ in exact state merging.}
Consider a three-qubit pure state
\begin{equation}
    \Ket{\psi}^{RAB}=\frac{1}{\sqrt{2}}\Big(\Ket{0}^R\Ket{0}^{A}\Ket{0}^{B} +\Ket{1}^R\Ket{1}^A\Ket{+}^{B}\Big).
\end{equation}
The protocols for exact state merging of $\Ket{\psi}^{RAB}$ in Theorems~\ref{thm:merge} and~\ref{thm:merge_without_catalyst} require respectively
\begin{equation}
  \log_2 K - \log_2 L = 1
\end{equation}
and
\begin{equation}
  \log_2 K = 1.
\end{equation}
Since $\psi^B\neq\frac{\mathbb{1}^B}{2}$, the latter equality for exact state merging in the non-catalytic setting is optimal due to Theorem~\ref{thm:qubit}.

In contrast, interchange $A$ and $B$ for $\Ket{\psi}^{RAB}$ to consider
\begin{equation}
    \Ket{\psi^\prime}^{RAB}=\frac{1}{\sqrt{2}}\Big(\Ket{0}^R\Ket{0}^{A}\Ket{0}^{B} +\Ket{1}^R\Ket{+}^A\Ket{1}^{B}\Big).
\end{equation}
In the same way as the above case of $\Ket{\psi}^{RAB}$,
the protocols for exact state merging of $\Ket{\psi^\prime}^{RAB}$ in Theorems~\ref{thm:merge} and~\ref{thm:merge_without_catalyst} require respectively
\begin{equation}
  \log_2 K - \log_2 L = 1
\end{equation}
and
\begin{equation}
  \log_2 K = 1.
\end{equation}
However, since $\psi^B=\frac{\mathbb{1}^B}{2}$, Theorem~\ref{thm:qubit} implies that there exists a protocol
for exact state merging in the non-catalytic setting of $\Ket{\psi^\prime}^{RAB}$ achieving
\begin{equation}
  \log_2 K = 0 < 1.
\end{equation}
Indeed, $\Ket{\psi^\prime}^{RAB}$ can also be written as
\begin{equation}
    \begin{split}
      \Ket{\psi^\prime}^{RAB}=&\sqrt{\frac{1}{2}+\frac{\sqrt{2}}{4}}{\left[\frac{\left(1+\sqrt{2}\right)\Ket{0}+\Ket{1}}{\sqrt{4+2\sqrt{2}}}\right]}^A\Ket{\Phi^-}^{RB}+\\
                       &\sqrt{\frac{1}{2}-\frac{\sqrt{2}}{4}}{\left[\frac{\left(1-\sqrt{2}\right)\Ket{0}+\Ket{1}}{{\sqrt{4-2\sqrt{2}}}}\right]}^A\Ket{\Phi^+}^{RB},
    \end{split}
\end{equation}
and hence, $A$'s measurement in basis
\begin{equation}
    \left\{\frac{\left(1+\sqrt{2}\right)\Ket{0}+\Ket{1}}{\sqrt{4+2\sqrt{2}}}, \frac{\left(1-\sqrt{2}\right)\Ket{0}+\Ket{1}}{{\sqrt{4-2\sqrt{2}}}}\right\}
\end{equation}
yields a maximally entangled state between $R$ and $B$.

These cases imply that the difference in entanglement costs between the optimal protocol and the protocols presented in Theorems~\ref{thm:merge} and~\ref{thm:merge_without_catalyst} may arise depending on whether the quantum part of the Koashi-Imoto decomposition can be merged at less entanglement cost than performing quantum teleportation.
Note that the optimal protocol obtained in Theorem~\ref{thm:qubit} works only for qubits, and Proposition~\ref{prp:qutrit} implies that extension to qudits is not straightforward.
\end{implication}

\begin{implication}
\label{ex:5}
\textit{Special cases where the achievability and converse bounds for exact state merging coincide.}
Special cases are discussed where one of the subsystems of the system $\mathcal{H}^R\otimes\mathcal{H}^A\otimes\mathcal{H}^B$ for a given state $\Ket{\psi}^{RAB}$ is initially decoupled from the others.
In these cases, the achievability bound for exact state merging in Theorem~\ref{thm:merge} coincides with the converse bound in Theorem~\ref{thm:new}.
Note that in general, there may exist a gap between these bounds as discussed in Implications~\ref{ex:2} and~\ref{ex:4}, while full characterization of the cases where this gap closes is unknown.

Consider the case where the system $\mathcal{H}^R$ is initially decoupled with the others, and a given pure state is in the form
\begin{equation}
  \Ket{\psi_{R\textup{-}AB}}^{RAB}=\Ket{\mu}^{R}\otimes\Ket{\nu}^{AB}.
\end{equation}
Due to the Koashi-Imoto decomposition of $\Ket{\psi}^{RAB}$ in Lemma~\ref{lem:koashi_imoto_decomposition_tripartite}, the decomposition of $\mathcal{H}^A$ is
\begin{equation}
  \mathcal{H}^A=\mathcal{H}^{a_0^L},
\end{equation}
where in terms of the notations of Lemma~\ref{lem:koashi_imoto_decomposition_tripartite}, $J=1$, and $\mathcal{H}^{a_0^R}$ does not explicitly appear since in this case
\begin{equation}
  \dim\mathcal{H}^{a_0^L}=\dim\mathcal{H}^A,\quad\dim\mathcal{H}^{a_0^R}=1.
\end{equation}
As for $\Ket{\psi_{R\textup{-}AB}}^{RAB}$, the decomposition yields
\begin{equation}
  \Ket{\psi_{R\textup{-}AB}}^{RAB}=\Ket{\mu}^{R}\otimes\Ket{\nu}^{a_0^L b_0^L},
\end{equation}
and define
\begin{equation}
  \lambda_0\coloneq\lambda_0^{a_0^L}=\lambda_0^B,
\end{equation}
where the notations are the same as those in Theorems~\ref{thm:merge} and~\ref{thm:new}.
The protocol in Theorem~\ref{thm:merge} for exact state merging of $\Ket{\psi_{R\textup{-}AB}}^{RAB}$ achieves for any $\delta > 0$
\begin{equation}
  \log_2 K - \log_2 L \leqq \log_2\lambda_0 + \delta,
\end{equation}
where shared entanglement is distilled by Subprocess~1 in the proof of Theorems~\ref{thm:merge}.
The converse bound in Theorem~\ref{thm:new} shows for any protocol for exact state merging of $\Ket{\psi_{R\textup{-}AB}}^{RAB}$
\begin{equation}
  \log_2 K - \log_2 L \geqq \log_2\lambda_0.
\end{equation}

Next, consider the case where the system $\mathcal{H}^B$ is initially decoupled with the others, and a given pure state is in the form
\begin{equation}
  \Ket{\psi_{B\textup{-}RA}}^{RAB}=\Ket{\mu}^{B}\otimes\Ket{\nu}^{RA}.
\end{equation}
Due to the Koashi-Imoto decomposition of $\Ket{\psi_{B\textup{-}RA}}^{RAB}$ in Lemma~\ref{lem:koashi_imoto_decomposition_tripartite}, the decomposition of $\mathcal{H}^A$ is
\begin{equation}
  \mathcal{H}^A=\mathcal{H}^{a_0^R}\oplus\mathcal{H}^{a_1^L},
\end{equation}
where in terms of the notations of Lemma~\ref{lem:koashi_imoto_decomposition_tripartite}, $J=1$, and $\mathcal{H}^{a_0^L}$ and $\mathcal{H}^{a_1^R}$ do not explicitly appear since in this case
\begin{align}
  \dim\mathcal{H}^{a_0^L}&=1,\\
  \dim\mathcal{H}^{a_0^R}&=\rank\psi_{B\textup{-}RA}^A,\\
  \dim\mathcal{H}^{a_1^L}&=\dim\mathcal{H}^A-\rank\psi_{B\textup{-}RA}^A,\\
  \dim\mathcal{H}^{a_1^R}&=1.
\end{align}
As for $\Ket{\psi_{B\textup{-}RA}}^{RAB}$, the decomposition yields
\begin{equation}
  \Ket{\psi_{B\textup{-}RA}}^{RAB}=\Ket{\mu}^{b_0^L}\otimes\Ket{\nu}^{R a_0^R}.
\end{equation}
The protocol in Theorem~\ref{thm:merge} for exact state merging of $\Ket{\psi_{B\textup{-}RA}}^{RAB}$ achieves
\begin{equation}
  \log_2 K = \rank \nu^{a_0^R} = \rank\psi_{B\textup{-}RA}^A,\quad \log_2 L = 0.
\end{equation}
where $\nu^{a_0^R}$ is transferred using quantum teleportation in Subprocess~2 in the proof of Theorems~\ref{thm:merge}.
The converse bound in Theorem~\ref{thm:new} shows for any protocol for exact state merging of $\Ket{\psi_{B\textup{-}RA}}^{RAB}$
\begin{equation}
  \log_2 K - \log_2 L \geqq \rank\psi_{B\textup{-}RA}^A.
\end{equation}

Finally, consider the case where the system $\mathcal{H}^A$ is initially decoupled with the others, and a given pure state is in the form
\begin{equation}
  \Ket{\psi_{A\textup{-}RB}}^{RAB}=\Ket{\mu}^{A}\otimes\Ket{\nu}^{RB}.
\end{equation}
Due to the Koashi-Imoto decomposition of $\Ket{\psi_{A\textup{-}RB}}^{RAB}$ in Lemma~\ref{lem:koashi_imoto_decomposition_tripartite}, the decomposition of $\mathcal{H}^A$ is
\begin{equation}
  \mathcal{H}^A=\mathcal{H}^{a_0^L},
\end{equation}
where in terms of the notations of Lemma~\ref{lem:koashi_imoto_decomposition_tripartite}, $J=1$, and $\mathcal{H}^{a_0^R}$ does not explicitly appear since in this case
\begin{align}
  \dim\mathcal{H}^{a_0^L}=\dim\mathcal{H}^A,\quad \dim\mathcal{H}^{a_0^R}=1.
\end{align}
As for $\Ket{\psi_{A\textup{-}RB}}^{RAB}$, the decomposition yields
\begin{equation}
  \Ket{\psi_{A\textup{-}RB}}^{RAB}=\Ket{\mu}^{a_0^L}\otimes\Ket{\nu}^{R b_0^R}.
\end{equation}
The protocol in Theorem~\ref{thm:merge} for exact state merging of $\Ket{\psi_{A\textup{-}RB}}^{RAB}$ achieves
\begin{equation}
  \log_2 K = \log_2 L = 0,
\end{equation}
where $B$ locally prepares a state corresponding to $\Ket{\mu}^{a_0^L}$ due to Subprocess~3 in the proof of Theorems~\ref{thm:merge}.
The converse bound in Theorem~\ref{thm:new} shows for any protocol for exact state merging of $\Ket{\psi_{A\textup{-}RB}}^{RAB}$
\begin{equation}
  \log_2 K - \log_2 L \geqq 0.
\end{equation}
\end{implication}

\chapter{\label{sec:two_way}One-shot quantum state merging under one-way and two-way communication}

This chapter proves that in one-shot state merging from $A$ to $B$, $B$'s preprocessing of quantum side information and backward classical communication from $B$ to $A$ can be indispensable for minimizing the entanglement cost.
The setting and the statement are presented in Section~\ref{sec:statement_two_way},
and the proof is given in Section~\ref{sec:proof} using interconnection between state merging and another relevant task, local state discrimination.
Based on this interconnection between state merging and local state discrimination, interpretation of entanglement cost in state merging is discussed in Section~\ref{sec:cost}

\section{\label{sec:statement_two_way}Separation between one-way and two-way LOCC in a one-shot state merging.}
The main result of this chapter shows a provable advantage of two-way LOCC over one-way LOCC in a one-shot scenario of state merging,
which contrasts with the existing protocols for one-shot state merging using only one-way communication~\cite{B9,Y9,B12,D7,D6,H10,B10,D5,M,N3,A4,A5,B15,B13,A16,A17}.
This advantage is shown for approximate state merging of a particular given state in the non-catalytic setting introduced in Definition~\ref{def:approxiamte_state_merging}.
Note that this result straightforwardly shows that the advantage also exists for exact state merging.
In the following of this chapter, state merging may refer to this approximate state merging in the non-catalytic setting, if obvious.
The main result is illustrated in Figure~\ref{fig:result} and shown as follows.

\begin{theorem}
\label{thm:result}
  \textit{Separation between one-way LOCC and two-way LOCC in a one-shot state merging}
  There exists a state $\Ket{\psi}^{RAB}$ (defined later in Eq.~\eqref{eq:phi}) and a nonzero error threshold $\epsilon_0>0$ such that for any $\epsilon\in\left[0,\epsilon_0\right]$, the following hold.
  \begin{enumerate}
    \item The \textit{optimal one-way} LOCC protocol for non-catalytic approximate state merging of $\Ket{\psi}^{RAB}$ within $\epsilon$ requires \textit{one} ebit of entanglement cost, that is,
      \begin{equation}
        \log_2 K = 1;
      \end{equation}
    \item There exists a \textit{two-way} LOCC protocol for non-catalytic approximate state merging of $\Ket{\psi}^{RAB}$ within $\epsilon$ achieving \textit{zero} entanglement cost, that is,
      \begin{equation}
        \log_2 K = 0 < 1.
      \end{equation}
  \end{enumerate}
\end{theorem}

\begin{figure}[t!]
  \centering
  \includegraphics[width=4in]{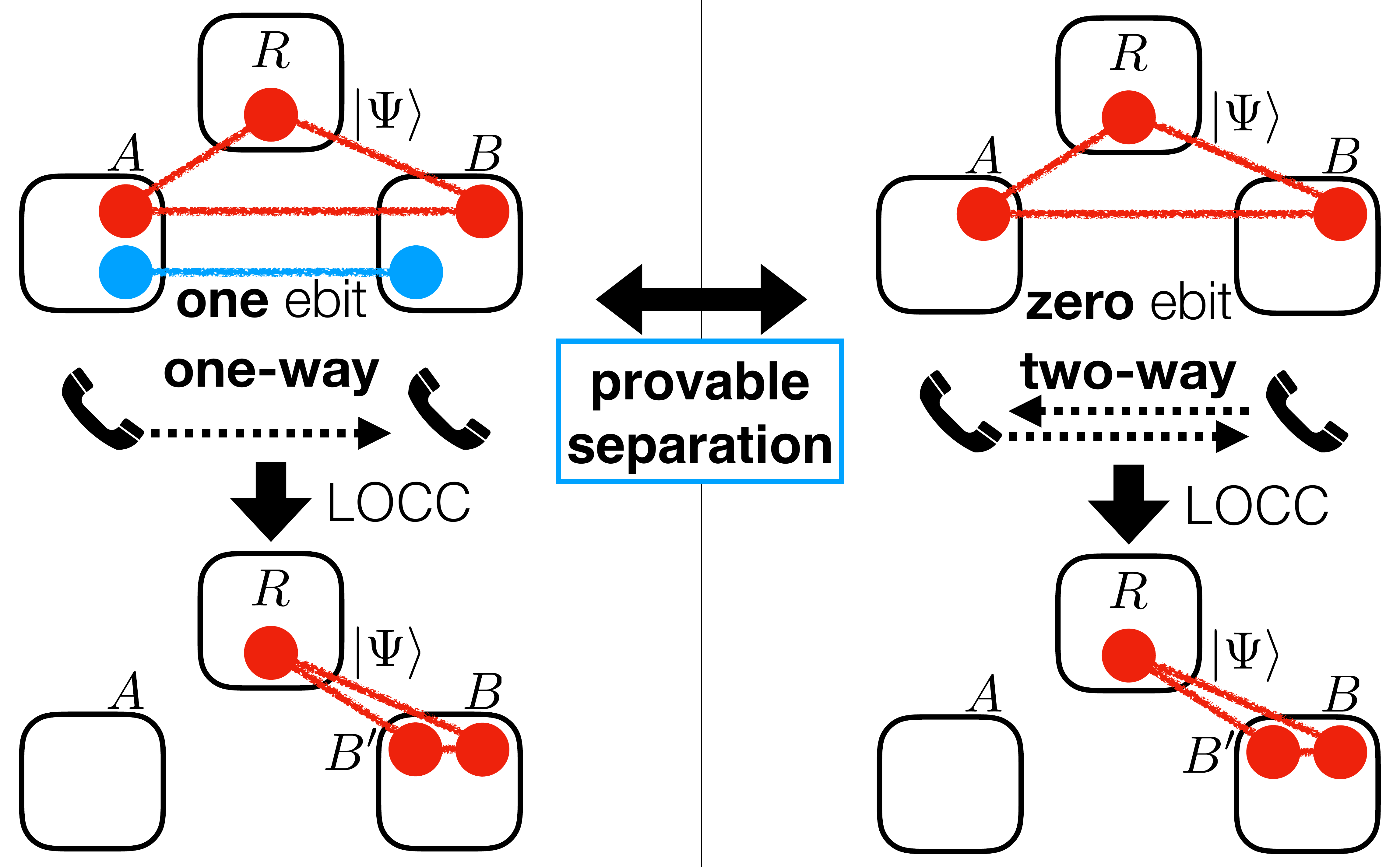}
  \caption[Provable separation between one-way LOCC and two-way LOCC in a one-shot state merging.]{\label{fig:result}The result shown in Theorem~\ref{thm:result} demonstrating provable separation between one-way and two-way local operations and classical communication (LOCC) in a one-shot state merging of $\Ket{\psi}^{RAB}$ defined as Equation~\eqref{eq:phi} represented by the red circles, where classical communication is represented by the dotted arrows. While an optimal one-way LOCC protocol for this task requires one ebit of entanglement cost represented by the connected blue circles, there exists a two-way LOCC protocol achieving zero entanglement cost.}
\end{figure}

\begin{table}[t!]
  \centering
  \caption[Is there a case where separation between one-way LOCC and two-way LOCC can be shown?]{\label{table:compare}Is there a case where separation between one-way LOCC and two-way LOCC can be shown? The separations are in terms of achievability of deterministic transformations between two fixed bipartite pure states, entanglement cost in state spitting, entanglement cost in state merging, distillable entanglement from bipartite mixed states, and success probability of local state discrimination among bipartite states. State merging provides the contrast between the asymptotic and one-shot scenarios.}
  \begin{tabular}{@{}lll@{}}
    \toprule
    task  & asymptotic scenario & one-shot scenario \\
    \midrule
    \begin{tabular}{@{}l@{}}state transformation (bipartite pure)\end{tabular} & No~\cite{B3}.  & No~\cite{N2}.\\
    state splitting & No~\cite{A2}. & No~(Theorem~\ref{thm:split}).\\
    state merging & No~\cite{H3,H4}. & Yes~(Theorem~\ref{thm:result}).\\
    \begin{tabular}{@{}l@{}} entanglement distillation \end{tabular} & Yes~\cite{B14}. & Yes~\cite{C2}.\\
    \begin{tabular}{@{}l@{}}local state discrimination \end{tabular} & Yes~\cite{O3}. & Yes~\cite{G4,C4,O2,C5,N4,T7,T8,C3}.\\
    \bottomrule
  \end{tabular}
\end{table}

Note that for proving Theorem~\ref{thm:result}, it is not sufficient to apply the proof techniques having used for obtaining converse bounds of entanglement cost in state merging that are based on the monotonicity of entropic functions~\cite{H4,H10} or the majorization condition used in Chapter~\ref{sec:merge}, since these techniques are based on no-go theorems applicable to any LOCC map including two-way LOCC\@.
The proof of Theorem~\ref{thm:result} requires a no-go theorem that is \textit{only applicable to one-way} LOCC and is \textit{provably false for two-way} LOCC, and hence, another proof technique than these existing ones has to be established.

Regarding provable separation between one-way LOCC and two-way LOCC in achievability of a given task, only several examples are known to date, such as entanglement distillation and local state discrimination, as shown in Table~\ref{table:compare}.
Note that while the set of one-way LOCC maps is strictly included in that of two-way LOCC maps~\cite{C7}, this difference does not necessarily affect achievability of a given task; \textit{e.g.}, one-way LOCC suffices for deterministic transformations between two fixed bipartite pure states and state splitting.
Among the known separations, the separation in local state discrimination based on hypothesis testing is first proven in a one-shot scenario~\cite{O2}, but whether the separation still survives in the corresponding asymptotic scenario was open until it is shown in Reference~\cite{O3} that the separation \textit{does survive}.
In contrast to such known separations shown in both asymptotic and one-shot scenarios, Theorem~\ref{thm:result} on state merging provides a case where provable separation in a one-shot scenario \textit{does not asymptotically survive}, in the sense that one-way LOCC suffices in the corresponding asymptotic scenario.

As for another remark, this chapter evaluates the amount of \textit{initially} shared entanglement and does not allow catalytic use of this shared entanglement, for simplicity.
There are known only a few tasks of which catalytic use of entanglement is proven to affect achievability, such as entanglement transformation~\cite{J1,E1}, distributed implementation of a nonlocal bipartite unitary~\cite{V4}, and local state discrimination~\cite{Y15}.
While state merging can be regarded as a transformation of tripartite pure states,
problems on state transformations in such a catalytic setting are hard to solve analytically in general, even in bipartite cases as pointed out in Reference~\cite{J1}.
As for catalyst in state merging, even if catalyst is allowed in the definition itself, asymptotic optimality can be achieved with an inconsiderable amount of catalyst~\cite{H4}, while there exists no quantitative study in one-shot scenarios.

\section{\label{sec:proof}Interconnection between state merging and local state discrimination}
To prove separation between one-way LOCC and two-way LOCC in a one-shot scenario of state merging in Theorem~\ref{thm:result}, local state discrimination is used.
In local state discrimination, two parties $A$ and $B$ initially share an unknown state $\Ket{\psi_l}^{AB}$ given from a known set
\begin{equation}
  {\left\{\Ket{\psi_l}^{AB}\right\}}_{l=0,\ldots,D-1}
\end{equation}
of $D$ orthogonal pure states, and the task aims to determine the index $l$ of $\Ket{\psi_l}^{AB}$ with unit probability by an LOCC measurement.
Note that for the analysis in this chapter, it suffices to consider local state discrimination without using initially shared entangled resource states, while generalization to that using resources of shared entanglement is straightforward, as discussed in References~\cite{C1,B11,A6,Z1,B7,G3,B8}.
There exists a set of orthogonal pure states for which local state discrimination is not achievable by one-way LOCC but is achievable by two-way LOCC, which is called a $2$-LOCC set.
References~\cite{N4,T7,T8} provide $2$-LOCC sets for any possible dimensional systems.

State merging can be viewed as a generalized task of local state discrimination, in the sense that achievability of the former implies that of the latter.
Proposition~\ref{prp:equivalence} shows that
if there exists a protocol achieving state merging of a tripartite state having the Schmidt decomposition
\begin{equation}
  \Ket{\psi}^{RAB}\coloneq\frac{1}{\sqrt{D}}\sum_{l=0}^{D-1}\Ket{l}^R\otimes\Ket{\psi_l}^{AB}
\end{equation}
at zero entanglement cost,
then this protocol transforms any superposition of the $D$ orthogonal states
\begin{equation}
  {\left\{\Ket{\psi_l}^{AB}\right\}}_{l=0,\ldots,D-1}
\end{equation}
into that of
\begin{equation}
  {\left\{\Ket{\psi_l}^{B^\prime B}\right\}}_{l=0,\ldots,D-1}\,,
\end{equation}
\textit{i.e.},
\begin{equation}
  \label{eq:relative}
  \sum_{l=0}^{D-1}\alpha_l\Ket{\psi_l}^{AB}\xrightarrow{\textup{LOCC}}\sum_{l=0}^{D-1}\alpha_l\Ket{\psi_l}^{B^\prime B}.
\end{equation}
Thus, local state discrimination for ${\left\{\Ket{\psi_l}^{AB}\right\}}_{l}$ can be achieved by first performing the protocol for state merging of $\Ket{\psi}^{RAB}$ to transform $\Ket{\psi_l}^{AB}$ into $\Ket{\psi_l}^{B^\prime B}$ for any $l$, and then performing $B$'s measurement for discriminating $B$'s orthogonal states ${\left\{\Ket{\psi_l}^{B^\prime B}\right\}}_{l}$.
Note that a similar interconnection is also pointed out in the asymptotic scenario~\cite{A9}.

In contrast, achievability of local state discrimination does not necessarily imply that of state merging if a protocol achieving local state discrimination uses a technique called \textit{elimination}, \textit{i.e.}, the measurement for excluding some of the possibilities of ${\left\{\Ket{\psi_l}^{AB}\right\}}_{l}$.
For example, consider a set of states
\begin{equation}
  \big\{\Ket{\psi_0}^{AB}\coloneq\Ket{0}^A\otimes\Ket{0}^B,\Ket{\psi_1}^{AB}\coloneq\Ket{0}^A\otimes\Ket{1}^B,\Ket{\psi_2}^{AB}\coloneq\Ket{1}^A\otimes\Ket{+}^B\big\},
\end{equation}
where
\begin{equation}
  \Ket{+}\coloneq\frac{1}{\sqrt{2}}\left(\Ket{0}+\Ket{1}\right).
\end{equation}
If $A$ eliminates some of the possibilities by a measurement in basis $\left\{\Ket{0},\Ket{1}\right\}$, $B$'s local measurement conditioned by $A$'s outcome can discriminate the remaining orthogonal states on $B$.
In contrast, state merging of the corresponding tripartite state
\begin{equation}
  \frac{1}{\sqrt{3}}\sum_{l=0}^{2}\Ket{l}^R\otimes\Ket{\psi_l}^{AB}
\end{equation}
is not achievable at zero entanglement cost due to the converse bound shown in Corollary~\ref{col:tractable_converse}.
In this way, a protocol for local state discrimination using elimination does not generalize to that for state merging, because elimination destroys coherence between $R$ and the others.
As for the known $2$-LOCC sets,
two-way LOCC protocols shown in References~\cite{N4,T7,T8} for achieving local state discrimination require elimination, and hence, \textit{do not generalize} to state merging in a straightforward way.

In contrast, the following analysis identifies a $2$-LOCC set for which a two-way LOCC protocol for local state discrimination can be constructed \textit{without elimination},
and the corresponding two-way LOCC protocol for state merging can also be constructed.
Consider a set ${\left\{\Ket{\psi_l}^{AB}\right\}}_{l=0,1,2}$ of three orthogonal states of $\mathbb{C}^{11}\otimes\mathbb{C}^{11}$,
and define each state as
\begin{equation}
    \label{eq:s}
    \begin{split}
      \Ket{\psi_0}^{AB}\coloneq&\sqrt{\frac{2}{11}}\Ket{\Phi_2^+}^{AB}\oplus\sqrt{\frac{9}{11}}\Ket{\Phi_9^+}^{AB},\\
      \Ket{\psi_1}^{AB}\coloneq&\sqrt{\frac{2}{11}}\gamma_1 X_2^A\Ket{\Phi_2^+}^{AB}\oplus\sqrt{\frac{9}{11}}{\left(X_9^A\right)}^3\Ket{\Phi_9^+}^{AB},\\
      \Ket{\psi_2}^{AB}\coloneq&\sqrt{\frac{2}{11}}\gamma_2 Z_2^A\Ket{\Phi_2^+}^{AB}\oplus\sqrt{\frac{9}{11}}{\left(X_9^A\right)}^6\Ket{\Phi_9^+}^{AB},
    \end{split}
\end{equation}
where each subsystem is decomposed into subspaces
\begin{equation}
  \mathbb{C}^{11}=\mathbb{C}^{2}\oplus\mathbb{C}^{9},
\end{equation}
$X_k^A$ and $Z_k^A$ are the generalized Pauli operator on a subspace $\mathbb{C}^k$ of $A$'s system for $A$'s part of
\begin{equation}
  \Ket{\Phi_k^+}^{AB}\coloneq\frac{1}{\sqrt{k}}\sum_{l=0}^{k-1}\Ket{l}^A\otimes\Ket{l}^B,
\end{equation}
and $\gamma_1$ and $\gamma_2$ are nonreal complex numbers satisfying
\begin{align}
  &{\left|\gamma_1\right|}^2=1,\\
  &{\left|\gamma_2\right|}^2=1,\\
  &\gamma_2\neq\pm\textup{i}\gamma_1^2.
\end{align}
The corresponding tripartite state is
\begin{equation}
\label{eq:phi}
  \Ket{\psi}\coloneq\frac{1}{\sqrt{3}}\sum_{l=0}^{2}\Ket{l}^R\otimes\Ket{\psi_l}^{AB},
\end{equation}
where ${\left\{\Ket{\psi_l}^{AB}\right\}}_{l=0,1,2}$ is defined as Equation~\eqref{eq:s}.
This state $\Ket{\psi}^{RAB}$ yields Theorem~\ref{thm:result} as follows.

\begin{proof}[Proof of the first statement in Theorem~\ref{thm:result}.]
The set
\begin{equation}
  {\left\{\Ket{\psi_l}^{AB}\right\}}_{l=0,1,2}
\end{equation}
defined as Equation~\eqref{eq:s} is shown to be a $2$-LOCC set in Reference~\cite{N4},
and hence, impossibility of local state discrimination by one-way LOCC yields impossibility of \textit{exact} state merging in the non-catalytic setting of $\Ket{\psi}^{RAB}$ defined as Equation~\eqref{eq:phi} at zero entanglement cost by one-way LOCC\@.
Since the set of one-way LOCC maps is compact,
this impossibility of exact state merging in the non-catalytic setting by one-way LOCC implies that there exists a sufficiently small but nonzero error $\epsilon>0$ such that \textit{approximate} state merging in the non-catalytic setting of $\Ket{\psi}^{RAB}$ within $\epsilon$ is still impossible at zero entanglement cost by one-way LOCC\@.
Note that the no-go theorem on local state discrimination by one-way LOCC in Reference~\cite{N4} does not generalize in a straightforward way to scenarios where catalytic use of entanglement is allowed, due to the fact that there may exist local state discrimination that is achievable at zero entanglement by using shared entanglement catalytically, but is not achievable without catalytic use of entanglement~\cite{Y20}.

The rest of the proof constructs a one-way LOCC protocol for state merging of $\Ket{\psi}^{RAB}$ achieving one ebit of entanglement cost and zero error, \textit{i.e.},
\begin{equation}
  \label{eq:one_ebit}
  \begin{split}
    \log_2 K&=1,\\
    F^2\left(\tilde\psi,\Ket{\psi}\Bra{\psi}\right)&=1,
  \end{split}
\end{equation}
based on the general protocol established in Theorem~\ref{thm:merge_without_catalyst} using the Koashi-Imoto decomposition.
Note that this one-way LOCC protocol is less costly than the trivial protocol of performing quantum teleportation of $A$'s part of $\Ket{\psi}^{RAB}$ of an eleven-dimensional system.

While the general protocol shown in Theorem~\ref{thm:merge_without_catalyst} requires $\log_2 3$ ebits of entanglement cost for $\Ket{\psi}^{RAB}$, this protocol can be modified using a specific structure of $\Ket{\psi}^{RAB}$, to achieve one ebit of entanglement cost.
The following construction of this protocol mainly discusses this specific part in the particular case of $\Ket{\psi}^{RAB}$.
For brevity, define
\begin{align}
  &\Ket{\Psi_0}\coloneq\Ket{\Phi_2^+},\\
  &\Ket{\Psi_1}\coloneq\left(\gamma_1 X_2^A\otimes\mathbb{1}^B\right)\Ket{\Phi_2^+},\\
  &\Ket{\Psi_2}\coloneq\left(\gamma_2 Z_2^A\otimes\mathbb{1}^B\right)\Ket{\Phi_2^+}.
\end{align}

Using Lemma~\ref{lem:koashi_imoto_decomposition_tripartite},
the following Koashi-Imoto decomposition of $\Ket{\psi}^{RAB}$ is obtained.
The Hilbert spaces $\mathcal{H}^A=\mathbb{C}^{11}$ of $A$ and $\supp\left(\psi^B\right)=\mathcal{H}^B=\mathbb{C}^{11}$ of $B$ are decomposed into
\begin{equation}
  \label{eq:h_decomposition}
  \begin{split}
    &\mathcal{H}^A=\bigoplus_{j=0}^{3}\mathcal{H}^{a_j^\textup{R}},\\
    &\mathcal{H}^B=\bigoplus_{j=0}^{3}\mathcal{H}^{b_j^\textup{R}},
  \end{split}
\end{equation}
where
\begin{align}
    &\dim\mathcal{H}^{a_0^\textup{R}}=\dim\mathcal{H}^{b_0^\textup{R}}=2,\\
    &\dim\mathcal{H}^{a_1^\textup{R}}=\dim\mathcal{H}^{b_1^\textup{R}}=3,\\
    &\dim\mathcal{H}^{a_2^\textup{R}}=\dim\mathcal{H}^{b_2^\textup{R}}=3,\\
    &\dim\mathcal{H}^{a_3^\textup{R}}=\dim\mathcal{H}^{b_3^\textup{R}}=3.
\end{align}
Note that $\mathcal{H}^{a_j^\textup{L}}$ and $\mathcal{H}^{b_j^\textup{L}}$ in Lemma~\ref{lem:koashi_imoto_decomposition_tripartite} do not explicitly appear in the decomposition in Equation~\eqref{eq:h_decomposition}, since $\mathcal{H}^{a_j^\textup{L}}=\mathbb{C}$ and $\mathcal{H}^{b_j^\textup{L}}=\mathbb{C}$ for each $j\in\left\{0,\ldots,3\right\}$ in this case.
The state $\Ket{\psi}^{RAB}$ is decomposed into
\begin{equation}
    \Ket{\psi}^{RAB}=\sqrt{\frac{2}{11}}\Ket{\phi_0}^{Ra_0^\textup{R} b_0^\textup{R}}\oplus\bigoplus_{j=1}^{3}\sqrt{\frac{3}{11}}\Ket{\phi_j}^{Ra_j^\textup{R} b_j^\textup{R}},
\end{equation}
where
\begin{equation}
    \Ket{\phi_0}^{Ra_0^\textup{R} b_0^\textup{R}}\coloneq\sqrt{\frac{1}{3}}\sum_{l=0}^{2}\Ket{l}^R\otimes\Ket{\Psi_l}^{a_0^\textup{R} b_0^\textup{R}},
\end{equation}
and for each $j\in\left\{1,2,3\right\}$,
\begin{equation}
  \Ket{\phi_j}^{Ra_j^\textup{R} b_j^\textup{R}}\coloneq\sqrt{\frac{1}{9}}\sum_{l,m=0}^{2}\Ket{l}^R\otimes\Ket{l+m\bmod 3}^{a_j^\textup{R}}\otimes\Ket{m}^{b_j^\textup{R}}.
\end{equation}
While the definition of ${\left\{\Ket{\psi_l}^{AB}\right\}}_{l=0,1,2}$ in Equation~\eqref{eq:s} uses the decomposition of each system $\mathbb{C}^{11}=\mathbb{C}^2\oplus\mathbb{C}^9$,
$\mathcal{H}^{a_0^\textup{R}}$ and $\mathcal{H}^{b_0^\textup{R}}$ in Equation~\eqref{eq:h_decomposition} correspond to $\mathbb{C}^2$, $\mathcal{H}^{a_1^\textup{R}}$ and $\mathcal{H}^{b_1^\textup{R}}$ in Equation~\eqref{eq:h_decomposition} correspond to a three-dimensional subspace of $\mathbb{C}^9$ spanned by $\left\{\Ket{0},\Ket{3},\Ket{6}\right\}$,
$\mathcal{H}^{a_2^\textup{R}}$ and $\mathcal{H}^{b_2^\textup{R}}$ correspond to that by $\left\{\Ket{1},\Ket{4},\Ket{7}\right\}$,
and $\mathcal{H}^{a_3^\textup{R}}$ and $\mathcal{H}^{b_3^\textup{R}}$ correspond to that by $\left\{\Ket{2},\Ket{5},\Ket{8}\right\}$.
Introducing auxiliary systems $\mathcal{H}^{a_0}$ of $A$ and $\mathcal{H}^{b_0}$ of $B$,
this decomposition can also be written as
\begin{equation}
  \label{eq:psi_decomposition}
  \begin{split}
    \left(U^A\otimes U^B\right)\Ket{\psi}^{RAB}
    =\sqrt{\frac{2}{11}}\Ket{0}^{a_0}\otimes\Ket{0}^{b_0}\otimes\Ket{\phi_0}^{Ra^\textup{R} b^\textup{R}}
    +\sum_{j=1}^{3}\sqrt{\frac{3}{11}}\Ket{j}^{a_0}\otimes\Ket{j}^{b_0}\otimes\Ket{\phi_j}^{Ra^\textup{R} b^\textup{R}}
  \end{split}
\end{equation}
where
\begin{align}
  &\dim\mathcal{H}^{a_0}=\dim\mathcal{H}^{b_0}=4,\\
  &\dim\mathcal{H}^{a^\textup{R}}=\max_j\left\{\dim\mathcal{H}^{a_j^\textup{R}}\right\}=3,\\
  &\dim\mathcal{H}^{b^\textup{R}}=\max_j\left\{\dim\mathcal{H}^{b_j^\textup{R}}\right\}=3,
\end{align}
$U^A$ is $A$'s local isometry from $\mathcal{H}^A$ to $\mathcal{H}^{a_0}\otimes\mathcal{H}^{a^\textup{R}}$, and
$U^B$ is $B$'s local isometry from $\mathcal{H}^B$ to $\mathcal{H}^{b_0}\otimes\mathcal{H}^{b^\textup{R}}$.

Using the Koashi-Imoto decomposition in the form of Equation~\eqref{eq:psi_decomposition},
the protocol for exact state merging shown in Theorem~\ref{thm:merge_without_catalyst} performs three subprocesses $1$, $2$, and $3$, which are combined using controlled measurements and controlled isometries.
In the following, these three subprocesses in the case of $\Ket{\psi}^{RAB}$ are discussed.
In particular, Subprocess~2 is modified using a specific structure of $\Ket{\psi}^{RAB}$ to achieve one ebit of entanglement cost.

\textit{Subprocess 1:} The first subprocess is concerned with reduced states on $\mathcal{H}^{a_j^\textup{L}}\otimes\mathcal{H}^{b_j^\textup{L}}$, and since $\mathcal{H}^{a_j^\textup{L}}$ and $\mathcal{H}^{b_j^\textup{L}}$ do not explicitly appear in the decomposition in Equation~\eqref{eq:h_decomposition}, this subprocess is not performed in this case.

\textit{Subprocess 2:} The second subprocess is for transferring $A$'s part of $\Ket{\phi_j}^{Ra^\textup{R} b^\textup{R}}$ to $B$, so that $\Ket{\phi_j}^{R{\left(b^\prime\right)}^\textup{R} b^\textup{R}}$ is obtained, where $\mathcal{H}^{{\left(b^\prime\right)}^\textup{R}}$ is $B$'s auxiliary system corresponding to $\mathcal{H}^{a^\textup{R}}$.
While quantum teleportation is used for this subprocess in the proofs of Theorems~\ref{thm:merge} and~\ref{thm:merge_without_catalyst} to provide a general protocol,
there may exist cases where this subprocess can be achieved at less entanglement cost than performing quantum teleportation, as pointed out in Implication~\ref{ex:4}.
As for the case of $\Ket{\psi}^{RAB}$, $\Ket{\phi_0}^{Ra^\textup{R} b^\textup{R}}$ is merged using quantum teleportation, which requires one ebit of an initially shared maximally entangled state $\Ket{\Phi_2^+}^{\overline{A}\overline{B}}$, where $\mathcal{H}^{\overline{A}}$ and $\mathcal{H}^{\overline{B}}$ are systems for the shared maximally entangled states of $A$ and $B$, respectively.
If $\Ket{\phi_1}^{Ra^\textup{R} b^\textup{R}}$, $\Ket{\phi_2}^{Ra^\textup{R} b^\textup{R}}$, or $\Ket{\phi_3}^{Ra^\textup{R} b^\textup{R}}$ are also merged in the same way,
$\log_2 3$ ebits are required.
Instead, by performing $A$'s measurement on $\mathcal{H}^{a^\textup{R}}$ in the computational basis
\begin{equation}
  {\left\{\Ket{m}^{a^\textup{R}}\right\}}_{m=0,1,2}
\end{equation}
followed by $B$'s isometry correction conditioned by $A$'s measurement outcome, no entanglement is required for merging $\Ket{\phi_1}^{Ra^\textup{R} b^\textup{R}}$, $\Ket{\phi_2}^{Ra^\textup{R} b^\textup{R}}$, and $\Ket{\phi_3}^{Ra^\textup{R} b^\textup{R}}$.
However, to coherently combine Subprocess~2 for $\Ket{\phi_0}^{Ra^\textup{R} b^\textup{R}}$, $\Ket{\phi_1}^{Ra^\textup{R} b^\textup{R}}$, $\Ket{\phi_2}^{Ra^\textup{R} b^\textup{R}}$, and $\Ket{\phi_3}^{Ra^\textup{R} b^\textup{R}}$,
one ebit of entanglement $\Ket{\Phi_2^+}^{\overline{A}\overline{B}}$ has to be consumed by $A$'s measurement on $\mathcal{H}^{\overline{A}}$ in the computational basis ${\left\{\Ket{m}^{\overline{A}}\right\}}_{m=0,1}$ followed by $B$'s isometry correction.
Consequently,
the LOCC map for Subprocess~2 can be written as a family of operators
\begin{equation}
  {\left\{\Bra{j,m_2}\otimes\sigma_{j,m_2}\right\}}_{m_2}
\end{equation}
tracing out the post-measurement state of $A$, where $\Ket{0,m_2}$ and $\sigma_{0,m_2}$ corresponds to ${\left(U_j^\prime\right)}^\dag\Ket{\Phi_{j,m_2}}$ and $\sigma_{j,m_2}$ in Subprocess~2 used for Theorems~\ref{thm:merge} and~\ref{thm:merge_without_catalyst} based on quantum teleportation, and for each $j\in\left\{1,2,3\right\}$, ${\left\{\Ket{j,m_2}\right\}}_{m_2}$ and $\sigma_{j,m_2}$ are the computational basis for $A$'s measurement and the isometry for $B$'s correction conditioned by $A$'s measurement outcome $m_2\,$, respectively.

\textit{Subprocess 3:} The third subprocess is for merging states on $\mathcal{H}^{a_0}\otimes\mathcal{H}^{b_0}$, and this subprocess can be performed in the same way as Theorems~\ref{thm:merge} and~\ref{thm:merge_without_catalyst}.

Combining these three subprocesses in the same way as Theorems~\ref{thm:merge} and~\ref{thm:merge_without_catalyst}, obtain the one-way LOCC protocol achieving Equation~\eqref{eq:one_ebit} is obtained, which yields the conclusion.

\end{proof}

\begin{proof}[Proof of the second statement in Theorem~\ref{thm:result}]
  The proof is by construction, and a two-way LOCC protocol for exact state merging of $\Ket{\psi}$ in the non-catalytic setting achieving zero entanglement cost
  \begin{equation}
    \log_2 K=0
  \end{equation}
  is constructed.
  This two-way LOCC protocol works as follows, and the explicit forms of measurements are shown later.
  While three maximally entangled two-qubit states
  \begin{equation}
    \left\{\Ket{\Phi_2^+}^{AB},\gamma_1 X_2^A\Ket{\Phi_2^+}^{AB},\gamma_2 Z_2^A\Ket{\Phi_2^+}^{AB}\right\}
  \end{equation}
  used in Equation~\eqref{eq:s} cannot be discriminated by any LOCC measurement by themselves~\cite{G10},
  $B$ can perform an appropriate three-outcome measurement
  \begin{equation}
    {\left\{M_j^B\right\}}_{j=0,1,2}\,,
  \end{equation}
  so that the additional terms on $A$'s subspace $\mathbb{C}^9$ in Equation~\eqref{eq:s} become orthogonal.
  Using this orthogonality, $A$ can also perform an appropriate thirty-three-outcome measurement
  \begin{equation}
    {\left\{M_{k|j}^A\right\}}_{k=0,\ldots,32}
  \end{equation}
  conditioned by $B$'s measurement outcome $j$, so that for each measurement outcome $j$ and $k$ of the LOCC measurement
  \begin{equation}
    {\left\{M_{k|j}^A\otimes M_{j}^B\right\}}_{j,k}\,,
  \end{equation}
  orthogonal states ${\left\{\Ket{\psi_l}^{AB}\right\}}_{l=0,1,2}$ defined as Equation~\eqref{eq:s} are transformed into orthogonal states of $B$.
  Thus, $B$'s local isometry correction conditioned by $j$ and $k$ yields $\Ket{\psi}^{RB^\prime B}$.

  In the following, $B$'s measurement
  \begin{equation}
    {\left\{M_j^B\right\}}_{j=0,1,2}
  \end{equation}
  and $A$'s measurement
  \begin{equation}
    {\left\{M_{k|j}^A\right\}}_{k=0,\ldots,32}
  \end{equation}
  conditioned by $B$'s measurement outcome $j$ are shown explicitly.
  To present these measurements, consider that $\mathcal{H}^A$ and $\mathcal{H}^B$ are decomposed in the same way as Equation~\eqref{eq:s} for defining ${\left\{\Ket{\psi_l}^{AB}\right\}}_l\,$, that is,
  \begin{align}
    \mathcal{H}^A&=\mathbb{C}^2\oplus\mathbb{C}^9,\\
    \mathcal{H}^B&=\mathbb{C}^2\oplus\mathbb{C}^9.
  \end{align}

The measurement ${\left\{M_j^B\right\}}_{j=0,1,2}$ performed by $B$ is
\begin{align}
    &M_0^B\coloneq\sqrt{\frac{1}{3}}\left(\Ket{0}\Bra{0}+\Ket{1}\Bra{1}\right)\oplus\left(\Ket{0}\Bra{0}+\Ket{1}\Bra{1}+\Ket{2}\Bra{2}\right),\\
    &M_1^B\coloneq\sqrt{\frac{1}{3}}\left(\Ket{0}\Bra{0}+\Ket{1}\Bra{1}\right)\oplus\left(\Ket{3}\Bra{3}+\Ket{4}\Bra{4}+\Ket{5}\Bra{5}\right),\\
    &M_2^B\coloneq\sqrt{\frac{1}{3}}\left(\Ket{0}\Bra{0}+\Ket{1}\Bra{1}\right)\oplus\left(\Ket{6}\Bra{6}+\Ket{7}\Bra{7}+\Ket{8}\Bra{8}\right),
\end{align}
where each operator on the right-hand side is on $\mathbb{C}^2\oplus\mathbb{C}^9$.
This measurement satisfies the completeness condition
\begin{equation}
  \sum_{j=0}^{2}M_j^\dag M_j=\mathbb{1}.
\end{equation}

As for $A$'s measurement ${\left\{M_{k|j}^A\right\}}_{k=0,\ldots,32}$ conditioned by $j\in\left\{0,1,2\right\}$,
the case of $j=0$, that is, ${\left\{M_{k|0}^A\right\}}_{k=0,\ldots,32}\,$, is shown first, while a similar  construction applies to the cases of $j=1,2$, as discussed later.
For brevity, define a bipartite pure state $\Ket{\Psi}\in\mathbb{C}^9\otimes\mathbb{C}^9$ with Schmidt rank three as
\begin{equation}
  \Ket{\Psi}\coloneq\sqrt{\frac{1}{3}}\left(\Ket{0}\otimes\Ket{0}+\Ket{1}\otimes\Ket{1}+\Ket{2}\otimes\Ket{2}\right),
\end{equation}
and also define the Fourier-basis states of three-dimensional subspaces of $\mathbb{C}^9$
\begin{align}
  \Ket{\omega_{n}^{\left(0,4,8\right)}}&\coloneqq\frac{1}{\sqrt{3}}\Ket{0}+\frac{\exp\left(\frac{\textup{i}\pi n}{3}\right)}{\sqrt{3}}\Ket{4}+\frac{\exp\left(\frac{\textup{i}\pi 2n}{3}\right)}{\sqrt{3}}\Ket{8},\\
  \Ket{\omega_{n}^{\left(1,5,6\right)}}&\coloneqq\frac{1}{\sqrt{3}}\Ket{1}+\frac{\exp\left(\frac{\textup{i}\pi n}{3}\right)}{\sqrt{3}}\Ket{5}+\frac{\exp\left(\frac{\textup{i}\pi 2n}{3}\right)}{\sqrt{3}}\Ket{6},\\
  \Ket{\omega_{n}^{\left(2,3,7\right)}}&\coloneqq\frac{1}{\sqrt{3}}\Ket{2}+\frac{\exp\left(\frac{\textup{i}\pi n}{3}\right)}{\sqrt{3}}\Ket{3}+\frac{\exp\left(\frac{\textup{i}\pi 2n}{3}\right)}{\sqrt{3}}\Ket{7},
\end{align}
where $n\in\left\{0,1,2\right\}$.
If $B$'s measurement outcome is $j=0$, the post-measurement state is
\begin{equation}
  \Ket{\psi^{\left(0\right)}}^{RAB}=\frac{1}{\sqrt{3}}\sum_{l=0}^{2}\Ket{l}^R\otimes\Ket{\psi_l^{\left(0\right)}}^{AB},
\end{equation}
where
\begin{equation}
  \begin{split}
    \Ket{\psi_0^{\left(0\right)}}\coloneq&\sqrt{\frac{2}{11}}\Ket{\Phi_2^+}\oplus\sqrt{\frac{9}{11}}\Ket{\Psi}\\
    =&\sqrt{\frac{1}{11}}\left(\Ket{0}\otimes\Ket{0}+\Ket{1}\otimes\Ket{1}\right)\oplus\sqrt{\frac{3}{11}}\left(\Ket{0}\otimes\Ket{0}+\Ket{1}\otimes\Ket{1}+\Ket{2}\otimes\Ket{2}\right),\\
    \Ket{\psi_1^{\left(0\right)}}\coloneq&\sqrt{\frac{2}{11}}\left(\gamma_1 X_2\otimes\mathbb{1}\right)\Ket{\Phi_2^+}\oplus\sqrt{\frac{9}{11}}\left({\left(X_9\right)}^3\otimes\mathbb{1}\right)\Ket{\Psi}\\
    =&\sqrt{\frac{1}{11}}\gamma_1\left(\Ket{1}\otimes\Ket{0}+\Ket{0}\otimes\Ket{1}\right)\oplus\sqrt{\frac{3}{11}}\left(\Ket{3}\otimes\Ket{0}+\Ket{4}\otimes\Ket{1}+\Ket{5}\otimes\Ket{2}\right),\\
    \Ket{\psi_2^{\left(0\right)}}\coloneq&\sqrt{\frac{2}{11}}\left(\gamma_2 Z_2\otimes\mathbb{1}\right)\Ket{\Phi_2^+}\oplus\sqrt{\frac{9}{11}}\left({\left(X_9\right)}^6\otimes\mathbb{1}\right)\Ket{\Psi}\\
    =&\sqrt{\frac{1}{11}}\gamma_2\left(\Ket{0}\otimes\Ket{0}-\Ket{1}\otimes\Ket{1}\right)\oplus\sqrt{\frac{3}{11}}\left(\Ket{6}\otimes\Ket{0}+\Ket{7}\otimes\Ket{1}+\Ket{8}\otimes\Ket{2}\right).
  \end{split}
\end{equation}
In this case, $A$'s measurement ${\left\{M_{k|0}^A\right\}}_{k=0,\ldots,32}$ is in the form of
\begin{equation}
  M_{k|0}\coloneq\Bra{\phi_{k|0}},
\end{equation}
where $k\in\left\{0,\ldots,32\right\}$, the post-measurement state of $A$ is traced out, and $\Ket{\phi_{k|0}}\in\mathbb{C}^2\oplus\mathbb{C}^9$ is an unnormalized vector.
Each $\Ket{\phi_{k|0}}$ is defined as
\begin{align}
  \Ket{\phi_{0 |0}}   &\coloneqq& \sqrt{\frac{3}{36}}\Ket{0}&\oplus\sqrt{\frac{1}{36}}\left( \Ket{0}+\Ket{4}-\overline{\gamma_2}\Ket{6}\right),\\
  \Ket{\phi_{1 |0}}   &\coloneqq&-\sqrt{\frac{3}{36}}\Ket{0}&\oplus\sqrt{\frac{1}{36}}\left( \Ket{0}+\Ket{4}-\overline{\gamma_2}\Ket{6}\right),\\
  \Ket{\phi_{2 |0}}   &\coloneqq& \sqrt{\frac{3}{36}}\Ket{1}&\oplus\sqrt{\frac{1}{36}}\left(-\Ket{0}+\Ket{4}-\overline{\gamma_2}\Ket{6}\right),\\
  \Ket{\phi_{3 |0}}   &\coloneqq&-\sqrt{\frac{3}{36}}\Ket{1}&\oplus\sqrt{\frac{1}{36}}\left(-\Ket{0}+\Ket{4}-\overline{\gamma_2}\Ket{6}\right),\\
  \Ket{\phi_{4 |0}}   &\coloneqq& \sqrt{\frac{3}{36}}\Ket{0}&\oplus\sqrt{\frac{1}{36}}\left( \Ket{0}-\Ket{4}-\overline{\gamma_2}\Ket{6}\right),\\
  \Ket{\phi_{5 |0}}   &\coloneqq&-\sqrt{\frac{3}{36}}\Ket{0}&\oplus\sqrt{\frac{1}{36}}\left( \Ket{0}-\Ket{4}-\overline{\gamma_2}\Ket{6}\right),\\
  \Ket{\phi_{6 |0}}   &\coloneqq& \sqrt{\frac{3}{36}}\Ket{1}&\oplus\sqrt{\frac{1}{36}}\left(-\Ket{0}-\Ket{4}-\overline{\gamma_2}\Ket{6}\right),\\
  \Ket{\phi_{7 |0}}   &\coloneqq&-\sqrt{\frac{3}{36}}\Ket{1}&\oplus\sqrt{\frac{1}{36}}\left(-\Ket{0}-\Ket{4}-\overline{\gamma_2}\Ket{6}\right),\\
  \Ket{\phi_{8 |0}}   &\coloneqq& \sqrt{\frac{3}{36}}\Ket{0}&\oplus\sqrt{\frac{1}{36}}\left( \Ket{1}+\Ket{5}-\overline{\gamma_2}\Ket{7}\right),\\
  \Ket{\phi_{9 |0}}   &\coloneqq&-\sqrt{\frac{3}{36}}\Ket{0}&\oplus\sqrt{\frac{1}{36}}\left( \Ket{1}+\Ket{5}-\overline{\gamma_2}\Ket{7}\right),\\
  \Ket{\phi_{10|0}}   &\coloneqq& \sqrt{\frac{3}{36}}\Ket{1}&\oplus\sqrt{\frac{1}{36}}\left(-\Ket{1}+\Ket{5}-\overline{\gamma_2}\Ket{7}\right),\\
  \Ket{\phi_{11|0}}   &\coloneqq&-\sqrt{\frac{3}{36}}\Ket{1}&\oplus\sqrt{\frac{1}{36}}\left(-\Ket{1}+\Ket{5}-\overline{\gamma_2}\Ket{7}\right),\\
  \Ket{\phi_{12|0}}   &\coloneqq& \sqrt{\frac{3}{36}}\Ket{0}&\oplus\sqrt{\frac{1}{36}}\left( \Ket{1}-\Ket{5}-\overline{\gamma_2}\Ket{7}\right),\\
  \Ket{\phi_{13|0}}   &\coloneqq&-\sqrt{\frac{3}{36}}\Ket{0}&\oplus\sqrt{\frac{1}{36}}\left( \Ket{1}-\Ket{5}-\overline{\gamma_2}\Ket{7}\right),\\
  \Ket{\phi_{14|0}}   &\coloneqq& \sqrt{\frac{3}{36}}\Ket{1}&\oplus\sqrt{\frac{1}{36}}\left(-\Ket{1}-\Ket{5}-\overline{\gamma_2}\Ket{7}\right),\\
  \Ket{\phi_{15|0}}   &\coloneqq&-\sqrt{\frac{3}{36}}\Ket{1}&\oplus\sqrt{\frac{1}{36}}\left(-\Ket{1}-\Ket{5}-\overline{\gamma_2}\Ket{7}\right),\\
  \Ket{\phi_{16|0}}   &\coloneqq& \sqrt{\frac{3}{36}}\Ket{0}&\oplus\sqrt{\frac{1}{36}}\left( \Ket{2}+\Ket{3}-\overline{\gamma_2}\Ket{8}\right),\\
  \Ket{\phi_{17|0}}   &\coloneqq&-\sqrt{\frac{3}{36}}\Ket{0}&\oplus\sqrt{\frac{1}{36}}\left( \Ket{2}+\Ket{3}-\overline{\gamma_2}\Ket{8}\right),\\
  \Ket{\phi_{18|0}}   &\coloneqq& \sqrt{\frac{3}{36}}\Ket{1}&\oplus\sqrt{\frac{1}{36}}\left(-\Ket{2}+\Ket{3}-\overline{\gamma_2}\Ket{8}\right),\\
  \Ket{\phi_{19|0}}   &\coloneqq&-\sqrt{\frac{3}{36}}\Ket{1}&\oplus\sqrt{\frac{1}{36}}\left(-\Ket{2}+\Ket{3}-\overline{\gamma_2}\Ket{8}\right),\\
  \Ket{\phi_{20|0}}   &\coloneqq& \sqrt{\frac{3}{36}}\Ket{0}&\oplus\sqrt{\frac{1}{36}}\left( \Ket{2}-\Ket{3}-\overline{\gamma_2}\Ket{8}\right),\\
  \Ket{\phi_{21|0}}   &\coloneqq&-\sqrt{\frac{3}{36}}\Ket{0}&\oplus\sqrt{\frac{1}{36}}\left( \Ket{2}-\Ket{3}-\overline{\gamma_2}\Ket{8}\right),\\
  \Ket{\phi_{22|0}}   &\coloneqq& \sqrt{\frac{3}{36}}\Ket{1}&\oplus\sqrt{\frac{1}{36}}\left(-\Ket{2}-\Ket{3}-\overline{\gamma_2}\Ket{8}\right),\\
  \Ket{\phi_{23|0}}   &\coloneqq&-\sqrt{\frac{3}{36}}\Ket{1}&\oplus\sqrt{\frac{1}{36}}\left(-\Ket{2}-\Ket{3}-\overline{\gamma_2}\Ket{8}\right),\\
  \Ket{\phi_{24|0}}&\coloneqq&\boldsymbol{0}&\oplus\sqrt{\frac{28}{36}}\Ket{\omega_{0}^{\left(0,4,8\right)}},\\
  \Ket{\phi_{25|0}}&\coloneqq&\boldsymbol{0}&\oplus\sqrt{\frac{28}{36}}\Ket{\omega_{1}^{\left(0,4,8\right)}},\\
  \Ket{\phi_{26|0}}&\coloneqq&\boldsymbol{0}&\oplus\sqrt{\frac{28}{36}}\Ket{\omega_{2}^{\left(0,4,8\right)}},\\
  \Ket{\phi_{27|0}}&\coloneqq&\boldsymbol{0}&\oplus\sqrt{\frac{28}{36}}\Ket{\omega_{0}^{\left(1,5,6\right)}},\\
  \Ket{\phi_{28|0}}&\coloneqq&\boldsymbol{0}&\oplus\sqrt{\frac{28}{36}}\Ket{\omega_{1}^{\left(1,5,6\right)}},\\
  \Ket{\phi_{29|0}}&\coloneqq&\boldsymbol{0}&\oplus\sqrt{\frac{28}{36}}\Ket{\omega_{2}^{\left(1,5,6\right)}},\\
  \Ket{\phi_{30|0}}&\coloneqq&\boldsymbol{0}&\oplus\sqrt{\frac{28}{36}}\Ket{\omega_{0}^{\left(2,3,7\right)}},\\
  \Ket{\phi_{31|0}}&\coloneqq&\boldsymbol{0}&\oplus\sqrt{\frac{28}{36}}\Ket{\omega_{1}^{\left(2,3,7\right)}},\\
  \Ket{\phi_{32|0}}&\coloneqq&\boldsymbol{0}&\oplus\sqrt{\frac{28}{36}}\Ket{\omega_{2}^{\left(2,3,7\right)}},
\end{align}
where $\boldsymbol{0}$ is the zero vector on $\mathbb{C}^2$.
This measurement satisfies the completeness condition
\begin{equation}
  \sum_{k=0}^{32}M_{k|0}^\dag M_{k|0}=\mathbb{1}.
\end{equation}

Similarly,
the other measurements for $A$ conditioned by $B$'s measurement outcomes $j=1$ and $j=2$, that is,
${\left\{M_{k|1}^A\right\}}_{k=0,\ldots,32}$ and ${\left\{M_{k|2}^A\right\}}_{k=0,\ldots,32}\,$, respectively,
are defined for each $k\in\left\{0,\ldots,32\right\}$ as
\begin{align}
    M_{k|1}&\coloneq M_{k|0}\left(\boldsymbol{0}\oplus{\left(X_9\right)}^3\right),\\
    M_{k|2}&\coloneq M_{k|0}\left(\boldsymbol{0}\oplus{\left(X_9\right)}^6\right).
\end{align}
These measurements satisfy the completeness condition
\begin{equation}
  \sum_{k=0}^{32}M_{k|j}^\dag M_{k|j}=\mathbb{1},
\end{equation}
for each $j\in\left\{1,2\right\}$.

In the two-way LOCC protocol for exact state merging of $\Ket{\psi}^{RAB}$ in the non-catalytic setting at zero entanglement cost,
$B$ first performs the measurement ${\left\{M_j^B\right\}}_{j=0,1,2}\,$, and the measurement outcome $j$ is sent by classical communication from $B$ to $A$.
Conditioned by $j$, the measurement ${\left\{M_{k|j}^A\right\}}_{k=0,\ldots,32}$ is performed by $A$, and the measurement outcome $k$ is sent by classical communication from $A$ to $B$.
After this LOCC measurement ${\left\{M_{k|j}^A\otimes M_j^B\right\}}_{j,k}$ by $A$ and $B$,
for any pair of measurement outcomes $j\in\left\{0,1,2\right\}$ and $k\in\left\{0,\ldots,32\right\}$,
the post-measurement state
\begin{equation}
  \frac{\left(M_{k|j}^A\otimes M_j^B\right)\Ket{\psi}^{RAB}}{\left\|\left(M_{k|j}^A\otimes M_j^B\right)\Ket{\psi}^{RAB}\right\|},
\end{equation}
is a maximally entangled state with Schmidt rank three between $R$ and $B$.
Therefore, $B$ performs local isometry conditioned by $j$ and $k$ to transform this maximally entangled state into $\Ket{\psi}^{RB^\prime B}$.
This protocol yields the conclusion.

\end{proof}

\section{\label{sec:cost}Interpretation of entanglement cost in quantum state merging}

This section discusses how entanglement cost in state merging can be interpreted.
As Theorem~\ref{thm:result} shows that entanglement cost in state merging under one-way LOCC and that two-way LOCC are different under a one-shot regime,
the entanglement cost under two-way LOCC cannot be interpreted based only on one-way communication in analogy to classical source coding with $B$'s side information.
But given the interconnection between the tasks of state merging and local state discrimination, these tasks can be interpreted as \textit{distributed decoding} of information encoded in an initially shared state.
This section first summarizes the difference in properties of $B$'s side information in asymptotic and one-shot scenarios of state merging, and then provide another interpretation of state merging based on distributed decoding.

References~\cite{H3,H4} interpret the minimal entanglement cost in the asymptotic scenario of state merging as \textit{partial quantum information conditioned by $B$'s prior quantum information}.
Consider three parties, namely, $A$, $B$, and $R$, and any tripartite pure state $\Ket{\psi}^{RAB}$ shared among $A$, $B$, and $R$.
Define a measure of partial quantum information conditioned by $B$'s prior quantum information for $\Ket{\psi}^{RAB}$ as the rate of the minimal entanglement cost in the asymptotic scenario of state merging of $\Ket{\psi}^{RAB}$, which is given by the conditional quantum entropy ${H\left(A|B\right)}_\psi$~\cite{H3,H4}.
Here, let $A$ and $B$ perform a class of operations consisting of any local preprocessing of $B$'s prior quantum information of $\psi^B$ and backward classical communication from $B$ to $A$, which is a subclass of LOCC\@.
The following proposition show that ${H\left(A|B\right)}_\psi$ is monotonically nondecreasing on average under this class of operations, and the proof is given in Appendix~\ref{sec:monotonic}.
A similar monotonic property of conditional quantum entropy induced by measurements is also discussed in Reference~\cite{C10}.
Note that while Reference~\cite{H4} discusses a case where the conditional quantum entropy is decreased by adding quantum side information to $A$, this case in Reference~\cite{H4} is different from the case discussed here, since in Reference~\cite{H4}, entanglement between $A$ and $B$ can be increased by adding the quantum side information to $A$, but in this case cannot be increased by LOCC\@.

\begin{proposition}
\label{prp:monotonic}
  \textit{Monotonic property of partial quantum information in the asymptotic scenario.}
  Given any state $\Ket{\psi}^{RAB}$,
  for any operation by $A$ and $B$ represented as
  \begin{equation}
    \label{eq:backward}
    {\left\{U_j^A\otimes M_j^B\right\}}_j\,,
  \end{equation}
  where ${\left\{M_j^B\right\}}_j$ is $B$'s measurement for preprocessing satisfying the completeness condition $\sum_j M_j^\dag M_j=\mathbb{1}$, and $U_j^A$ is $A$'s isometry conditioned by $B$'s measurement outcome $j$ sent by backward classical communication from $B$ to $A$,
  it holds that
  \begin{equation}
    {H\left(A|B\right)}_\psi\leqq\sum_j p\left(j\right){H\left(A|B\right)}_{\psi_j}\,,
  \end{equation}
  where $\Ket{\psi_j}^{RAB}$ is the post-measurement state corresponding to $j$, that is,
  \begin{align}
      \Ket{\psi_j}^{RAB}&\coloneq\sqrt{\frac{1}{p\left(j\right)}}\left(\mathbb{1}^R\otimes U_j^A\otimes M_j^B\right)\Ket{\psi}^{RAB},\\
      p\left(j\right)&\coloneq{\left\|\left(\mathbb{1}^R\otimes U_j^A\otimes M_j^B\right)\Ket{\psi}^{RAB}\right\|}^2.
  \end{align}
\end{proposition}

In contrast to this asymptotic scenario, Theorem~\ref{thm:result} indicates that entanglement cost in a one-shot scenario of state merging can be strictly decreased by the class of operations shown in Equation~\eqref{eq:backward}.
In the asymptotic scenario, Proposition~\ref{prp:monotonic} shows that the ability of performing $B$'s preprocessing and backward classical communication in two-way LOCC does not contribute to increasing $B$'s \textit{prior quantum information} from $\psi^B$.
In contrast, Theorem~\ref{thm:result} can be interpreted to say that, in a one-shot scenario, the same ability may increase $B$'s \textit{prior quantum information}.
In this sense, these interpretations provide notions of $B$'s \textit{prior quantum information} having \textit{different} properties depending on scenarios.

However, state merging can also be viewed in another way, based on the interconnection between state merging and local state discrimination.
In local state discrimination for ${\left\{\Ket{\psi_l}^{AB}\right\}}_l\,$, the index $l$ can be regarded as classical information encoded in $\Ket{\psi_l}^{AB}$, and local state discrimination aims to decode this classical information by LOCC\@.
In the same way, state merging of $\Ket{\psi}^{RAB}$ can also be regarded as distributed decoding of \textit{quantum} information by entanglement-assisted LOCC, in the sense that a protocol for state merging decodes arbitrary superposition of nonlocally shared states into the same superposition of $B$'s states, as shown in Formula~\eqref{eq:relative}.
This view of state merging as distributed decoding is generalized to more than two parties in Chapter~\ref{sec:distributed_encoding_decoding}.
These notions of information may be \textit{nonlocally} encoded in the shared quantum state~\cite{B6,B19}, in the sense that neither $A$ nor $B$ has local access to such nonlocally encoded information.

From this viewpoint, the minimal entanglement cost in state merging can be interpreted to characterize a nonlocal property of the map
\begin{equation}
  \mathcal{D}^{AB\to B}\left(\sum_l\alpha_l\Ket{\psi_l}^{AB}\right)=\sum_l\alpha_l\Ket{l}^{B}
\end{equation}
for decoding quantum information initially encoded in $\sum_l\alpha_l\Ket{\psi_l}^{AB}$, where $\sum_l\alpha_l\Ket{l}^{B}$ is locally unitarily equivalent to $\sum_l\alpha_l\Ket{\psi_l}^{B^\prime B}$.
This map $\mathcal{D}^{AB\to B}$ is an isometry map possibly defined for any given $\Ket{\psi}^{RAB}$.
Note that if catalytic use of entanglement is allowed, negative entanglement cost can also be viewed as a net gain of shared entanglement from the redundant part of $\sum_l\alpha_l\Ket{\psi_l}^{AB}$ as discussed in Chapter~\ref{sec:merge}, and the gained entanglement can be used as a resource for distributed decoding in future in the same way as the conventional interpretation~\cite{H3,H4}.

\part{\label{part:2}Operational analysis of multipartite quantum entanglement in distributed quantum information processing}

\chapter{Background and overview of Part~\ref{part:2}}

This part is motivated by the following previous studies on state constructions and transformations using entanglement-assisted LOCC\@.
In Reference~\cite{Y6} and the master thesis~\cite{Y18} of the author of this thesis, tasks of construction of a multipartite entangled state shared among spatially separated parties from a separable state have been analyzed, using a given inter-party network for quantum communication.
In the framework of local operations and classical communication (LOCC), single use of a noiseless quantum channel and that of a maximally entangled state are at equivalent cost by means of quantum teleportation simulating quantum communication~\cite{B5}.
For a bipartite state, the minimal amount of quantum communication required for preparing the state provides a well-established entanglement measure quantifying a nonlocal property of the state, called the \textit{entanglement cost} of the state~\cite{B3,H1,T1}.
The entanglement cost of a bipartite \textit{state} also generalizes to that required for spatially separated parties implementing a given nonlocal state \textit{transformation}, such as nonlocal unitaries~\cite{Z,E,C15,C16,N,Y5,C17,Y,S,S16,Y4,Y2,S17,X1,C18,V,Y3,W8,W9,W10,W12} and nonlocal measurements~\cite{J,B24,B25}, although this generalization usually accompanies challenging optimization and has been analyzed only in special cases to date.
Another direction is generalization of a \textit{bipartite} state to a \textit{multipartite} state~\cite{Y6,G9,Y7} while analysis of multipartite entanglement is also challenging~\cite{E2,W3,B26}.
To characterize  multipartite entanglement in terms of quantum communication, Reference~\cite{Y6} formulates that required for preparing the multipartite state shared among parties using a network of the noiseless quantum channels, establishing a characterization called \textit{graph-associated entanglement cost} of multipartite states.

In this part, after providing preliminaries in Chapter~\ref{sec:preliminaries_2}, the following two results are presented in Chapters~\ref{sec:distributed_encoding_decoding} and~\ref{sec:multipartite}.

\section*{Distributed encoding and decoding of quantum information over networks}

Encoding and decoding quantum information in a multipartite quantum system are fundamental building blocks in quantum information processing.
In particular, quantum error correcting codes~\cite{G,D,T10,B} require such encoding and decoding between a logical state and an entangled physical state of a multipartite system.
Quantum information is represented by this logical state, and these encoding and decoding are the inverse transformations of each other, mathematically represented by isometries.
These types of encoding and decoding have to be performed so that coherence of these states is kept; that is, an arbitrary superposition of the logical state should be preserved without revealing the classical description of the logical state.
In addition to quantum information processing, the concept of encoding and decoding nowadays has interdisciplinary roles in analyzing many-body quantum systems exhibiting nonlocal features, such as topological order in quantum phase of matter~\cite{K,K2}, holographic principle in quantum gravity~\cite{A,P}, and eigenstate thermalization hypothesis in statistical physics~\cite{F}.

\begin{figure}[t!]
    \centering
    \includegraphics[width=4in]{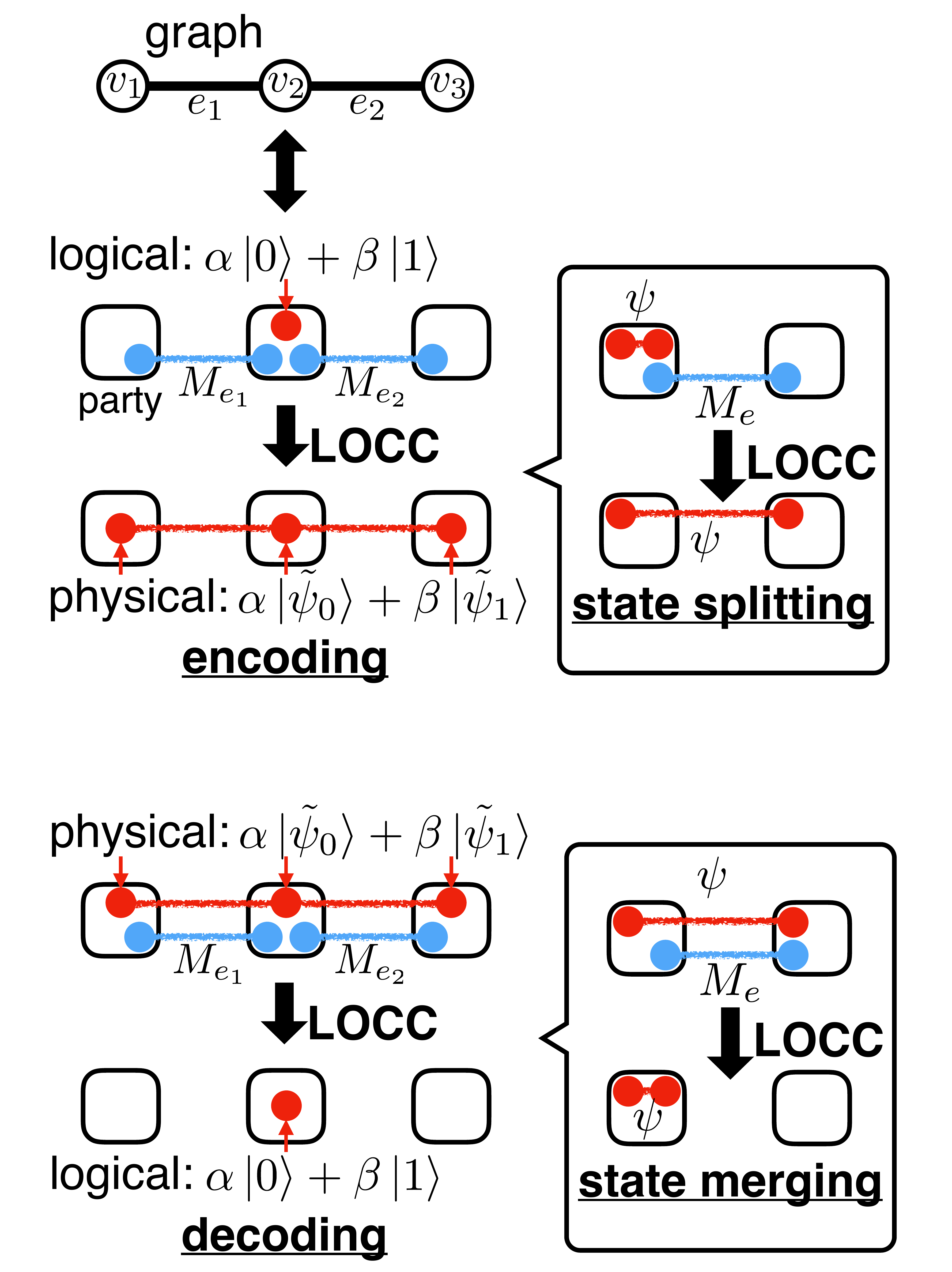}
    \caption[Encoding and decoding quantum information in a multipartite quantum system shared among spatially separated parties.]{\label{fig:intro}Encoding and decoding quantum information in a multipartite quantum system shared among spatially separated parties, where the quantum information is represented by \textit{unknown} quantum states illustrated by red circles. The parties are connected by a network of noiseless quantum channels  represented by a graph, so that the parties can sequentially apply exact state splitting to spread quantum information for encoding, and exact state merging to concentrate quantum information for decoding. Under LOCC, single use of each noiseless quantum channel represented by an edge $e$ of the graph is equivalent to that of a maximally entangled state $\Ket{\Phi_{M_e}^+}$ illustrated by a pair of blue circles connected by a line, where $M_e$ is the Schmidt rank of $\Ket{\Phi_{M_e}^+}$.}
\end{figure}

These encoding and decoding are also indispensable for distributed quantum information processing, where spatially separated parties connected by a network for quantum communication cooperate in achieving an information processing task.
Especially, encoding and decoding are crucial for some multiparty cryptographic tasks such as quantum secret sharing~\cite{B21,C,G7}.
In such distributed settings, a multipartite system for encoding a logical state is distributed among the spatially separated parties.
In this case, encoding and decoding are nonlocal transformations over all the parties, and the nonlocal properties of transformations for encoding and decoding lead to cost in implementations of the encoding and decoding.

In Chapter~\ref{sec:distributed_encoding_decoding}, entanglement costs characterizing the nonlocal properties of transformations for encoding and decoding are formalized.
Consider a setting where $N$ parties are connected by a network of the noiseless quantum channels, as illustrated in Figure~\ref{fig:intro}.
The network topology is represented by a graph in graph theory~\cite{B22} in terms of vertices and edges.
Any connected network of $N$ parties requires at least $N-1$ channels.
If an $N$-vertex connected graph has exactly $N-1$ edges, the graph is called a tree.
Using the network, the parties can spread and concentrate quantum information of \textit{unknown} states so as to encode and decode quantum information in a distributed system according to a given isometry representing the encoding and decoding.
The amount of quantum communication required for spreading and concentrating quantum information over the network characterizes nonlocal properties of the isometry.
Due to the equivalence between the noiseless quantum channel and the maximally entangled state,
a collection of maximally entangled states distributed according to the network topology comprises the resource state for spreading and concentrating quantum information by LOCC\@.
It is assumed that LOCC is free, and motivated by quantum communication on networks, Chapter~\ref{sec:distributed_encoding_decoding} considers this type of initial resource state consisting of bipartite entanglement.
The minimal total amount of quantum communication is evaluated by the entanglement entropy of the maximally entangled states for each edge, which are called the \textit{entanglement costs in spreading and concentrating quantum information}.
The entanglement cost in spreading quantum information characterizes the encoding, and that of concentrating characterizes the decoding.

Chapter~\ref{sec:distributed_encoding_decoding} evaluates the entanglement costs in spreading and concentrating quantum information over any given tree-topology network for an \textit{arbitrarily} given isometry, which differs from the works presented in References~\cite{F3,S15} for implementing \textit{particular} isometries in the context of quantum secret sharing.
To analyze the entanglement costs,
spreading and concentrating quantum information are reduced to sequential applications of exact state merging and splitting for two parties defined in Sections~\ref{sec:merge_achievability_exact} and~\ref{sec:split}, respectively, as illustrated in Figure~\ref{fig:intro}.
Regarding spreading quantum information, exact state splitting is used for providing a protocol and derive the optimal entanglement cost in spreading quantum information, which is given in terms of the rank of a state defined with respect to each edge of the given tree.
Another protocol is also shown for achieving concentrating quantum information. In particular, using exact state merging, the entanglement cost in concentrating quantum information can be reduced compared to that of spreading quantum information.
During spreading and concentrating \textit{quantum} information, coherence has to be kept, and this point is contrasted with encoding and decoding \textit{classical} information in quantum states shared among multiple parties investigated in the context of a type of quantum secret sharing based on local state discrimination~\cite{C14,R,Y10,W6,B23,L3}.
The protocols for spreading and concentrating quantum information are applicable to any isometry representing encoding and decoding and provide a protocol for one-shot distributed source compression~\cite{D8,D9,A8} applicable to arbitrarily small-dimensional systems and a general protocol for LOCC-assisted decoding of shared quantum information having studied in the context of quantum secret sharing~\cite{G8}.

\section*{When does multipartite entanglement outperform bipartite entanglement?}

In a distributed setting where spatially separable parties can freely perform LOCC, any multipartite entangled state can be prepared by LOCC from initially distributed bipartite entangled states among the parties, using quantum teleportation~\cite{B5} or less costly protocol established in Reference~\cite{Y6}.
In this regard, even if multipartite entanglement is to be used for a multiparty task in distributed information processing, initially sharing bipartite entangled states is sufficient,
and hence, it would be natural to doubt whether multipartite entanglement is necessary for performing some tasks by LOCC\@.

Aiming at showing that multipartite entanglement is still indispensable for distributed quantum information processing,
Chapter~\ref{sec:multipartite} provides \textit{nontrivial} examples demonstrating the difference in capability between entangled resource states exhibiting multipartite entanglement and those consisting only of bipartite entangled states.
This difference arises when there exists a limitation on the size of each party's local quantum system, that is, the dimension of the Hilbert space representing the local quantum system.
The comparison between bipartite and multipartite entanglement as resources for distributed quantum information processing presented in Chapter~\ref{sec:multipartite} is motivated by technological limitations on the number of qubits which can be stored in one quantum device, and is different from the comparison in the context of quantum key distribution~\cite{E3,P5}, since the cost of LOCC is considered to be negligible in Chapter~\ref{sec:multipartite}.
The difference can also be shown in a trivial example of qubits as follows.
Consider three parties $A$, $B$, and $C$ sharing two Bell states, that is, two two-qubit maximally entangled states
\begin{equation}
  \Ket{\Phi_2^+}^{AB}\otimes\Ket{\Phi_2^+}^{BC},
\end{equation}
one of which is shared between $A$ and $B$, and the other of which is shared between $B$ and $C$.
These two Bell states as a whole are regarded as a state consisting of bipartite entangled states.
In this case, once these two Bell states are given to the parties, the parties can transform the two Bell states by LOCC into any three-qubit state shared among $A$, $B$, and $C$, such as the three-qubit Greenberger-Horne-Zeilinger (GHZ) state
\begin{equation}
  \Ket{\textup{GHZ}}\coloneq\frac{1}{\sqrt{2}}\left(\Ket{000}+\Ket{111}\right),
\end{equation}
and the three-qubit $W$ state
\begin{equation}
  \Ket{W}\coloneq\frac{1}{\sqrt{3}}\left(\Ket{100}+\Ket{010}+\Ket{001}\right),
\end{equation}
both of which can be regarded as states exhibiting multipartite entanglement.
However, if each party's local system size is limited to one qubit,
the parties cannot store any state consisting of a collection of bipartite entangled states to obtain $\Ket{\textup{GHZ}}$ and $\Ket{W}$ by LOCC, while the parties can still store one of these states exhibiting multipartite entanglement as a resource for performing some tasks by LOCC\@.

\begin{figure}[t!]
    \centering
    \includegraphics[width=4in]{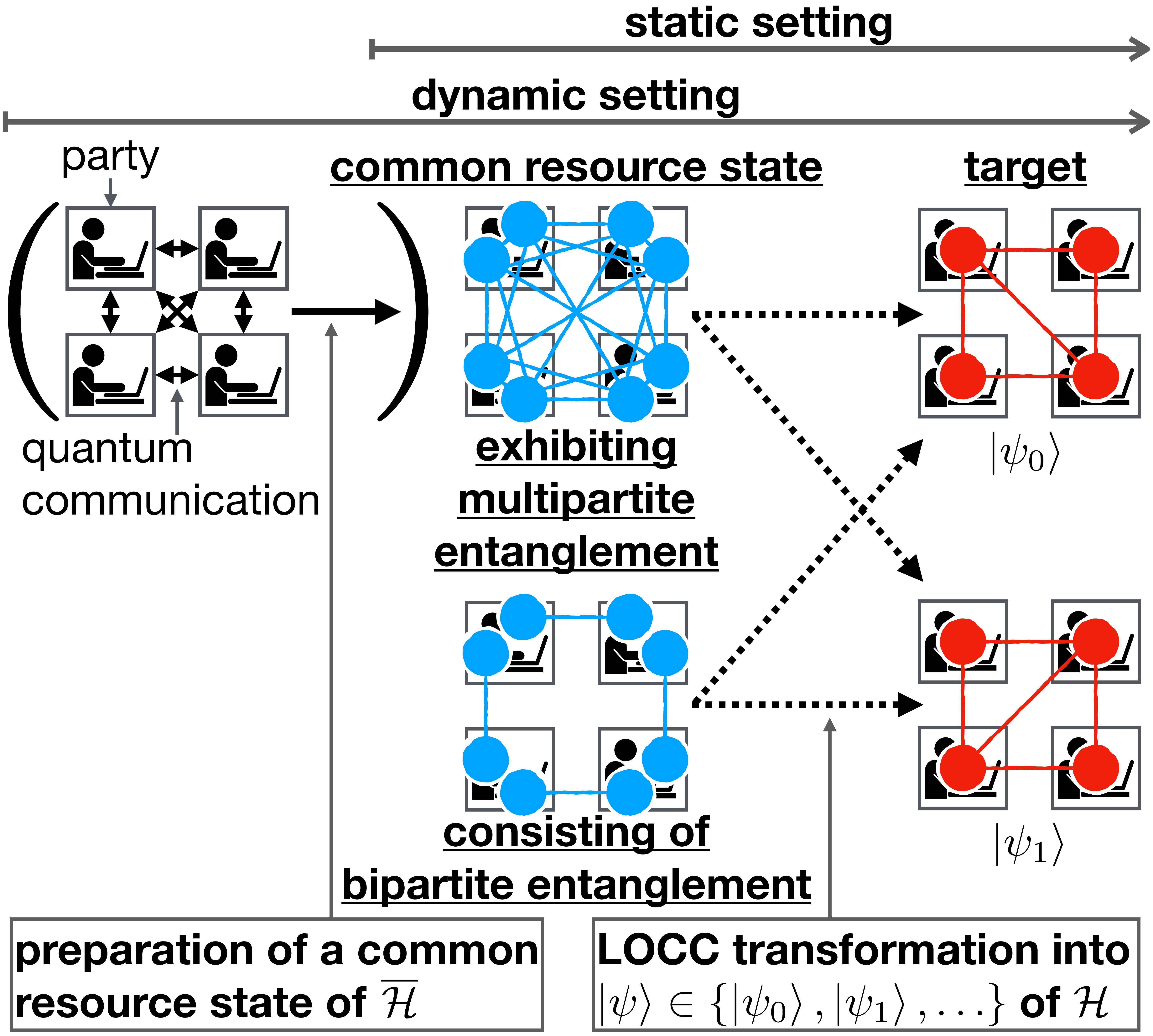}
    \caption[System-size-limited quantum state transformation.]{The tasks of system-size-limited quantum state transformation, where the parties transform a common resource state represented by blue circles by LOCC into an arbitrary state in a given target set $\left\{\Ket{\psi_0},\Ket{\psi_1},\ldots\right\}$ represented by red circles. To differentiate the capabilities of common resource states exhibiting multipartite entanglement at the top and those only consisting of bipartite entanglement at the bottom, where each connected pair of blue circles represents a bipartite entangled state, the static setting is considered where each party's local system size for storing the common resource state is limited. Also the dynamic setting is considered where the parties have to prepare a common resource state within these limitations by performing quantum communication, in addition to storing the common resource state. The difference in the capabilities arises in terms of achievability of this task.}
\label{fig:multipartite_intro}
\end{figure}

Apart from the above trivial example of qubits,
Chapter~\ref{sec:multipartite} aims to demonstrate the difference even in cases where the size of local systems of some parties is not limited to one qubit.
As illustrated in Figure~\ref{fig:multipartite_intro}, two settings of tasks aiming at preparing states in a target set from a common resource state~\cite{S18,G2} by LOCC are considered for differentiating capabilities of entangled states only consisting of a collection of bipartite entangled states and those exhibiting multipartite entanglement.
The tasks are called \textit{system-size-limited quantum state preparation},
where one of the two settings is called a \textit{static} setting, and the other is called a \textit{dynamic} setting.

In the static setting, each party's local system size is limited, and a common resource state for a given target set is stored within this limitation.
For a given target set of states of a multipartite system in general, there may not exist any common resource state in the multipartite system itself transformable by LOCC into all the states in the set.
In particular, given a multipartite system where each local dimension is $d$, almost no LOCC transformation among pure states of the system is possible~\cite{V1,S1,S2,M1,G1,S3}.
This fact implies that, in general, a common resource state for a set of multipartite states may be a state of a higher-dimensional system than that for the set itself.
If there is a limitation on each party's local system size, it may not be possible for the parties to store an entangled state of a higher-dimensional system serving as a common resource state.
Despite the efforts to understand properties of multipartite entanglement~\cite{E2,W3,B26}, general quantitative conditions of the smallest system size for common resource states have not yet been established.
Chapter~\ref{sec:multipartite} provides nontrivial examples where, within a given limitation on local system sizes, the preparation of a state in a given target set is \textit{not} achievable by any common resource state consisting of a collection of bipartite entangled states, but it is achievable by a common resource state exhibiting multipartite entanglement.
In contrast to previous studies on the LOCC convertibility between multipartite pure states of the \textit{same}-dimensional systems~\cite{V1,S1,S2,M1,G1,S3,S13,V6,T2,T4,T6}, this analysis requires LOCC transformations from a common resource state of a \textit{higher}-dimensional Hilbert space into a set of states of a \textit{lower}-dimensional Hilbert space.
The examples show the difference in the capabilities between these two types of common resource states, namely, those consisting of bipartite entanglement and those exhibiting multipartite entanglement.

As for the dynamic setting, in addition to considering a limitation on local system sizes for storing a common resource state, the parties by themselves prepare the common resource state within this limitation by performing quantum communication.
Some of the common resource states exhibiting multipartite entanglement analyzed in the static setting can be prepared within the limitation using quantum communication.
Hence, temporal uses of bipartite quantum communication resources are still sufficient for preparing such common resource states.
In contrast, Chapter~\ref{sec:multipartite} also shows other examples of states exhibiting multipartite entanglement that can be stored but cannot be prepared within a limitation on local system sizes, indicating that the common resource state used in the dynamic setting has an intermediate capability between those consisting of bipartite entanglement and those exhibiting multipartite entanglement in the static setting.

\chapter{\label{sec:preliminaries_2}Preliminaries to Part~\ref{part:2}}

This chapter provides preliminaries to Part~\ref{part:2}.
Section~\ref{sec:network} models networks in distributed quantum information processing as a collection of bipartite entangled states.
Notations for a class of networks having tree topology is summarized in Section~\ref{sec:tree} for later use.
These definitions are based on References~\cite{Y6,Y13}.
For this class of tree-topology networks, the results on constructing multipartite states investigated in Reference~\cite{Y6} are summarized in Section~\ref{sec:construction}.

\section{\label{sec:network}Quantum networks and entangled states consisting of bipartite entanglement}

A network of quantum communication channels among $N$ parties is represented by a graph~\cite{B22}.
Let
\begin{equation}
  G=(V(G),E(G))
\end{equation}
denote a simple undirected graph representing the restriction on quantum communication.
To perform arbitrary entanglement transformations over $N$ parties,  $G$ has to be a connected graph.
The arguments of $V(G)$ and $E(G)$ may be omitted to be simply written as
\begin{equation}
  G=(V,E),
\end{equation}
if obvious.
Each of the $N$ vertices
\begin{equation}
  v\in V=\{v_1\,,v_2\,,\ldots,v_N\}
\end{equation}
represents one of the $N$ parties, and each edge
\begin{equation}
  e=\{v_k\,,v_{k'}\}\in E
\end{equation}
represents a bidirectional noiseless quantum channel between $v_k$ and $v_{k'}$.
Quantum communication is only allowed between the parties directly connected by an edge.
Note that an edge $\{v_k\,,v_{k'}\}\in E$ is identified with $\{v_{k'}\,,v_{k}\}$.
Assume that the $N$ parties can freely perform LOCC\@.

When LOCC can be freely performed, quantum communication of a state of an $M_e$-dimensional system from a party $v_k$ to another party $v_{k'}$ connected by a channel $e = \left\{v_k\,,v_{k'}\right\}$ is achieved using quantum teleportation~\cite{B5} by LOCC assisted by a maximally entangled state shared between $v_k$ and $v_{k'}$
\begin{equation}
  \Ket{\Phi_{M_e}^+}^{e}\coloneq\frac{1}{\sqrt{M_e}}\sum_{l=0}^{M_e-1}\Ket{l}^{v_k}\otimes\Ket{l}^{v_{k'}}
\end{equation}
where $M_e$ is the Schmidt rank, and the superscript $e=\{v_k\,,v_{k'}\}$ represents a state shared between $v_k$ and $v_{k'}$.

In the LOCC framework, the tasks of performing a transformation of a multipartite entangled state shared among the $N$ parties under the restriction on quantum communication is equivalent to the tasks of performing the transformation by LOCC assisted by an initial resource state consisting of a set of bipartite maximally entangled states specified by a set of edges $E$.
The initial resource state for a given graph $G$ shared among the $N$ parties is represented by
\begin{align}
  \bigotimes_{e\in E}\Ket{\Phi_{M_e}^+}^e,
\end{align}
where $M_e$ is the Schmidt rank of the initial resource state specified by each edge $e$.

A more general class of this type of initial resource states can be those consisting of bipartite entanglement.
Consider a collection of bipartite entangled states distributed among the parties $v_1\,,\ldots,v_N$.
The distribution of the bipartite entangled states can also be represented by a graph $G=(V,E)$, where each vertex in the set $V=\left\{v_1\,,\ldots,v_N\right\}$ represents a party, and each edge $e=\left\{v_k\,,v_{k'}\right\}\in E$ a bipartite entangled state $\Ket{\phi_e}^e$ shared between two parties $v_k$ and $v_{k'}$.
An entangled state $\Ket{\phi}$ is called a state \textit{consisting of bipartite entanglement} if there exists a graph $G=(V,E)$ such that $\Ket\phi$ is locally unitarily equivalent to a state in the form
\begin{equation}
  \bigotimes_{e\in E}\Ket{\phi_e}^e.
\end{equation}
Note that this definition assumes \textit{pure} states consisting of bipartite entanglement, while generalization to \textit{mixed} states is straightforward.
If $\Ket{\phi}$ is fully entangled and is not a state consisting of bipartite entanglement, $\Ket{\phi}$ is called a state \textit{exhibiting multipartite entanglement}.

\section{\label{sec:tree}Tree-topology networks}

There are optimization problems on general networks that are hard to solve, such as the Hamiltonian cycle problem~\cite{K6} and the multicommodity flow problem~\cite{E6}.
Optimization of communication on networks is more involved, since the technique of network coding~\cite{A15} for reducing communication may be applied in such optimization of communication cost.
However, for a special class of networks, the optimization of communication cost can reduce to a simpler solvable case.
As for distributed quantum information processing, construction of low-noise quantum channels is challenging using current technology, and hence, it makes sense to consider networks connecting all the parties using the minimal number of quantum channels.
Such networks are represented by a class of graphs called \textit{trees} having the minimal number of edges connecting all the vertices, and are called \textit{tree-topology networks}.

\begin{figure}[t!]
    \centering
    \includegraphics[width=4in]{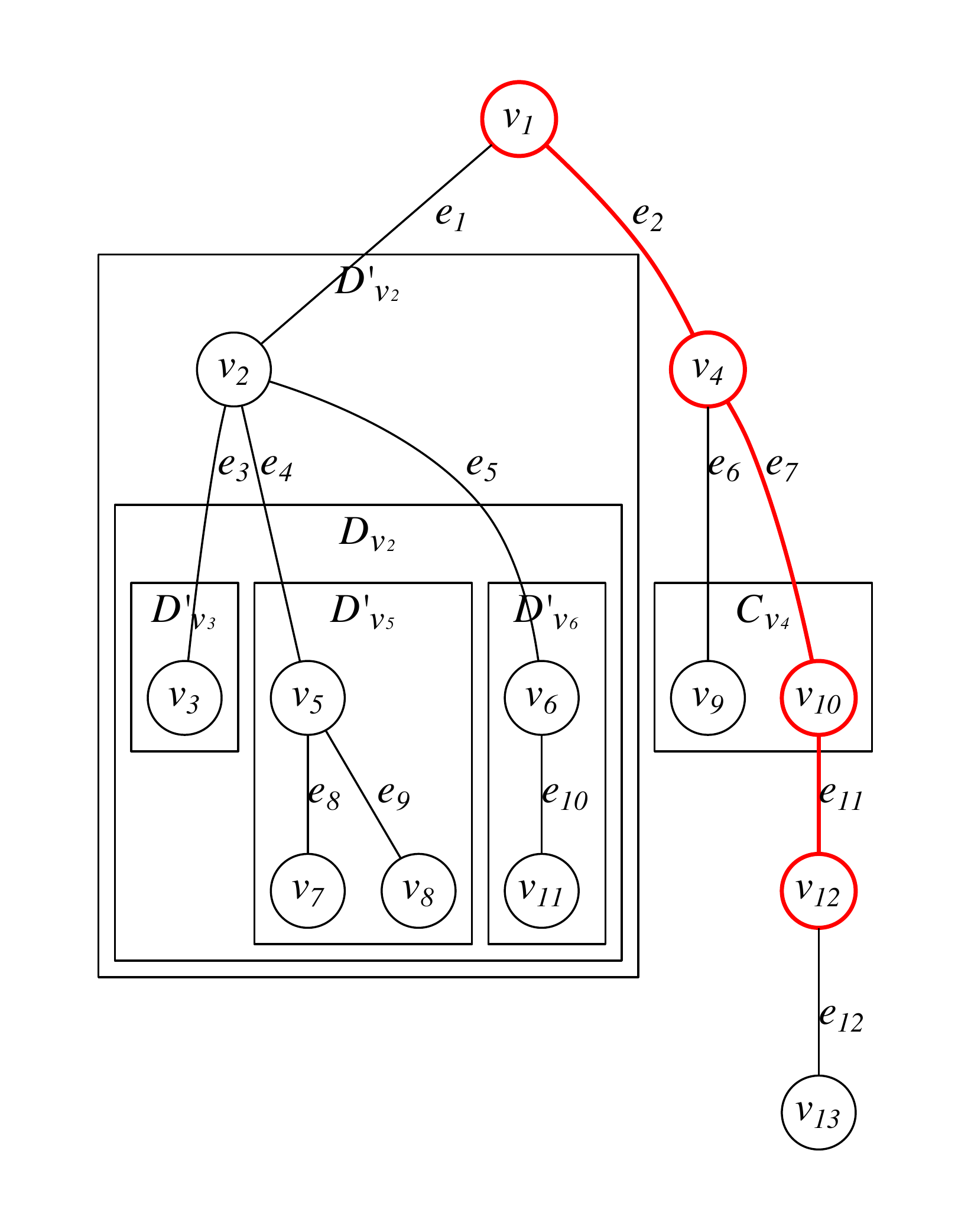}
    \caption[Notations for a tree]{Notations for a tree $T=(V,E)$. The vertex $v_1\in V$ is always designated as the root of the tree. The other vertices are labeled so that for any path connecting $v_1$ and another vertex, the closer to $v_1$ is any vertex $v_k$ on the path, the smaller is the label $k$, which is called an \textit{ascending labeling} of the vertices. For example, on the red bold path connecting the parties $v_1$ and $v_{12}$ in the figure, the vertices have to be labeled in ascending order $v_1\,,v_4\,,v_{10}\,,v_{12}$. For any $v_k\in V$, let $C_{v_k}\,$, $D_{v_k}\,$, and $D'_{v_k}$ denote the set of $v_k$'s children, the set of $v_k$'s descendants, and the set of $v_k$ itself and $v_k$'s descendants, respectively.}
\label{fig:tree_notation}
\end{figure}

Trees are connected graphs containing no cycle~\cite{B22}.
Trees with $N$ vertices have $N-1$ edges, which is the minimum to connect all the vertices.
Any connected graphs can be reduced to a tree spanning all the vertices by removing some of the edges.
Let
\begin{equation}
  T=(V,E)
\end{equation}
denote a tree.
Among the $N$ vertices of a tree $T=(V,E)$, a vertex can be designated as the root of the tree $T$, which is labeled $v_1\in V$ in the following.
In addition, the $N$ vertices $v_1\,,\ldots,v_N$ of $T$ are labeled so that, for any $v_k\neq v_1\,$, any vertex $v_{k'}$ on the path connecting $v_k$ and the root $v_1$ satisfies $k\geqq k'$; that is, the vertices are labeled in ascending order on such paths, as illustrated in Figure~\ref{fig:tree_notation}.
Call this type of labeling of the vertices an \textit{ascending labeling} of the vertices.
In the following, a tree is always regarded as a rooted tree with an ascending labeling.

A rooted tree has a recursive structure as illustrated in Figure~\ref{fig:tree_notation}.
For the root $v_1$ of any given tree $T=(V,E)$, any vertex $c$ adjacent to $v_1\,$, that is $\left\{v_1\,,c\right\}\in E$, is called a child of $v_1$.
Recursively, for any non-root vertex $v_k$ being a child of $v_{k^\prime}\,$, any vertex $c$ adjacent to $v_k$ except $v_{k^\prime}\,$, that is $\left\{v_k\,,c\right\}\in E\setminus\left\{\left\{v_{k^\prime}\,,v_k\right\}\right\}$, is called a child of $v_k$.
A vertex without any child is called a leaf.
For any non-root vertex $v_k\in V\setminus\left\{v_1\right\}$, let $p\left(v_k\right)\in V$ denote a vertex having $v_k$ as a child, which is called the parent of $v_k$.
For any vertex $v_k\in V$, $v_k$'s descendants are recursively defined as vertices being a child of $v_k$ or being a child of a descendant of $v_k$.
For any vertex $v_k\in V$, let $C_{v_k}$ denote the set of $v_k$'s children, $D_{v_k}$ the set of $v_k$'s descendants, and $D'_{v_k}$ the set of $v_k$ itself and $v_k$'s descendants.
Any edge $e\in E$ of the rooted tree can be written as $e=\{p\left(v_k\right),v_k\}$ for some $v_k\in V$.

For any $v\in V$, $D'_v$ can be decomposed by using these notations as
\begin{equation}
    D'_v = \{v\}\cup D_v =\{v\}\cup \bigcup_{c\in C_v} D'_c.
\end{equation}
The set of all vertices $V$ is represented by $V=D'_{v_1}$ for the root specified by $v_1$.  Using this decomposition, $V$ can be recursively decomposed according to the given rooted tree.

\section{\label{sec:construction}Construction of multipartite quantum states on networks}

Reference~\cite{Y6} investigates tasks of constructing multipartite quantum states shared among spatially separated parties connected by a given tree-topology network for quantum communication.
This task is called \textit{construction of a multipartite state} and defined as follows.
While Reference~\cite{Y6} analyzes two cases of exact and approximate constructions, this section summarizes the results on exact construction for later use.

\begin{definition}
    \textit{Exact construction of a multipartite state.}
    For a given graph $G=(V,E)$,
    exact construction of a given multipartite state
    \begin{equation}
      \Ket{\psi}\in\bigotimes_{v_k\in V}\mathcal{H}^{v_k},
    \end{equation}
    where $\mathcal{H}^{v_k}$ for each party represented as $v_k\in V$ is $v_k$'s system to be prepared in $\Ket{\psi}$,
    is a task of the $N$ parties $v_1\,,\ldots,v_N\in V$ preparing $\Ket{\psi}$ shared among the $N$ parties from scratch by performing an LOCC map $\mathcal{C}$ assisted by an initial resource state $\bigotimes_{e\in E}\Ket{\Phi_{M_e}^+}^e$, that is,
    \begin{equation}
        \mathcal{C}\left(\bigotimes_{e\in E}\Ket{\Phi_{M_e}^+}\Bra{\Phi_{M_e}^+}^e\right)=\Ket{\psi}\Bra{\psi}.
    \end{equation}
\end{definition}

Given any graph $G=(V,E)$ and any multipartite pure state
\begin{equation}
  \Ket{\psi}\in\bigotimes_{v_k\in V}\mathcal{H}^{v_k},
\end{equation}
consider the total amount of entanglement of initial resource states $\bigotimes_{e\in E}\Ket{\Phi_{M_e}^+}^e$ in terms of the entanglement entropy of each $\Ket{\Phi_{M_e}^+}^e$, that is,
\begin{equation}
  \sum_{e\in E}\log_2 M_e\,,
\end{equation}
and $\log_2 M_e$ for each $e\in E$ of $\bigotimes_{e\in E}\Ket{\Phi_{M_e}^+}^e$ for exact construction of $\Ket{\psi}$ for $G$ minimizing this total amount of entanglement
defines \textit{graph-associated entanglement cost} of $\Ket{\psi}$, quantitatively characterizing multipartite entanglement of $\Ket{\psi}$.
Reference~\cite{Y6} evaluates graph-associated entanglement cost of $\Ket{\psi}$ for an arbitrary given tree.
For any tree $T=(V,E)$, if an edge $e\in E$ is deleted, $T$ is divided into two disjoint trees whose vertices are represented by disjoint sets $V_e$ and $\overline{V}_e$ satisfying 
\begin{equation}
  V=V_e\cup \overline{V}_e.
\end{equation}
For any $\Ket{\psi}\in\bigotimes_{v_k\in V}\mathcal{H}^{v_k}$, let
\begin{equation}
  R_e\left(\Ket{\psi}\right)
\end{equation}
denote the Schmidt rank of $\Ket{\psi}$ with respect to the bipartition $\bigotimes_{v_k\in V_e}\mathcal{H}^{v_k}$ and $\bigotimes_{v_k\in \overline{V}_e}\mathcal{H}^{v_k}$ of $\mathcal{H}=\bigotimes_{v_k\in V}\mathcal{H}^{v_k}$.
Using this notation,
Reference~\cite{Y6} provides the necessary and sufficient condition for the initial resource state $\bigotimes_{e\in E}\Ket{\Phi_{M_e}^+}^e$ being transformable into $\Ket{\psi}$ by LOCC, as follows.
As for the proof~\cite{Y6}, the ``if'' part of this lemma is shown by construction of a protocol, and the ``only if'' part is shown from the LOCC monotonicity of the Schmidt rank.
Note that apart from this exact construction, Reference~\cite{Y6} also analyzes the task of approximately constructing multipartite states in the framework of the second-order asymptotic analysis~\cite{T9,D12}.

\begin{lemma}
\label{lem:graph_associated}
  \textit{Graph-associated entanglement cost in exact construction of multipartite states.}
  For any tree $T=(V,E)$ and any multipartite state
  \begin{equation}
    \Ket{\psi}\in\bigotimes_{v_k\in V}\mathcal{H}^{v_k},
  \end{equation}
  exact construction of $\Ket{\psi}$ for $T$ is achievable if and only if
  \begin{equation}
    \log_2 M_e\geqq R_e\left(\psi\right).
  \end{equation}
\end{lemma}

\chapter{\label{sec:distributed_encoding_decoding}Distributed encoding and decoding of quantum information over networks}

This chapter analyzes nonlocal transformations of multipartite entangled states shared among spatially separated parties, in particular, transformations for encoding and decoding quantum information in a shared multipartite quantum system.
Section~\ref{sec:spread_concentrate} defines tasks of spreading and concentrating quantum information on networks for achieving such encoding and decoding, respectively.
The former task of spreading is analyzed in Section~\ref{sec:encoding}, and the latter of concentrating is analyzed in Section~\ref{sec:decoding}.
Applications of these tasks are discussed in Section~\ref{sec:example}.

\section{\label{sec:spread_concentrate}Definition of spreading and concentrating quantum information}

State transformations for the encoding and decoding can be nonlocal transformations over multiple parties,
and local operations and classical communication (LOCC) by the parties may not be sufficient for performing such nonlocal transformations.
These encoding and decoding are achievable if the parties are allowed to communicate quantum information with each other using a network for quantum communication.
These tasks can also be regarded as tasks of spreading and concentrating quantum information according to a given isometry representing the encoding or decoding using quantum communication.

Given a network represented by any graph $G=(V,E)$ in general,
the parties aim to spread and concentrate quantum information according to a given isometry representing encoding and decoding, respectively.
A system $\mathcal{H}$ for logical states is located at one of the $N$ parties,
and the vertex labeled $v_1\in V$ is always assigned as the party where $\mathcal{H}$ is located.
Let $D$ denote the dimension of $\mathcal{H}$, that is,
\begin{equation}
  D\coloneq\dim\mathcal{H}.
\end{equation}
Write the computational basis of $\mathcal{H}$ as
\begin{equation}
  \{\Ket{l}\in\mathcal{H}:l=0,1,\ldots,D-1\}.
\end{equation}
In addition, the $N$ parties share a multipartite system $\tilde{\mathcal{H}}$ for physical states.
The system $\tilde{\mathcal{H}}$ is spanned by a set of $D$ orthonormal pure states
\begin{equation}
  \left\{\ket{\tilde{\psi}_l}^{v_1\cdots v_N}\in\tilde{\mathcal{H}}:l=0,1,\ldots,D-1\right\}.
\end{equation}
For each $v_k\in V$, let $\tilde{\mathcal{H}}^{v_k}$ denote a part of the shared multipartite system $\tilde{\mathcal{H}}$ located at the party $v_k$.
Note that $\dim\tilde{\mathcal{H}}^{v_k}$ is arbitrary as long as it holds that
\begin{equation}
  \dim\tilde{\mathcal{H}}=\dim\mathcal{H}=D,
\end{equation}
and hence, $\tilde{\mathcal{H}}$ is a \textit{subspace} of the Hilbert space consisting of these subsystems for the $N$ parties, that is,
\begin{equation}
  \tilde{\mathcal{H}}\subset\bigotimes_{v\in V}\tilde{\mathcal{H}}^{v}.
\end{equation}

Consider encoding and decoding as linear bijective maps between $\mathcal{B}\left(\mathcal{H}\right)$ and $\mathcal{B}\left(\tilde{\mathcal{H}}\right)$ mapping the basis states of $\mathcal{H}$ and $\tilde{\mathcal{H}}$ as
\begin{equation}
  \ket{l}\in\mathcal{H}\leftrightarrow\ket{\tilde{\psi}_l}\in\tilde{\mathcal{H}}
\end{equation}
for each $l\in\left\{0,\ldots,D-1\right\}$.
The encoding map is represented by an isometry $U$ from $\mathcal{H}$ to $\tilde{\mathcal{H}}$ satisfying
\begin{equation}
  \ket{\tilde{\psi}_l}=U\Ket{l}.
\end{equation}
Encoding refers to a transformation from $\rho\in\mathcal{D}\left(\mathcal{H}\right)$ into $U\rho U^\dag\in\mathcal{D}\left(\tilde{\mathcal{H}}\right)$, and decoding refers to the inverse transformation represented by $U^\dag$.

The formal definitions of the tasks of spreading and concentrating quantum information are given in terms of the LOCC framework as follows.
The tasks of spreading and concentrating quantum information are also illustrated in Figure~\ref{fig:intro}.
Note that the tasks are performed deterministically and exactly.
\begin{definition}
    \textit{Spreading and concentrating quantum information.}
    Spreading quantum information over a given graph $G=(V,E)$
    for a given isometry $U$
    is a task of the $N$ parties $v_1\,,\ldots,v_N\in V$ applying $U$ to an arbitrary \textit{unknown} input state $\rho\in\mathcal{D}\left(\mathcal{H}\right)$ of one party $v_1\in V$ to share $U\rho U^\dag\in\mathcal{D}\left(\tilde{\mathcal{H}}\right)$ among the $N$ parties by performing an LOCC map $\mathcal{S}$ assisted by an initial resource state $\bigotimes_{e\in E}\Ket{\Phi_{M_e}^+}^e$, that is,
    \begin{equation}
        \label{eq:encoding}
        \mathcal{S}\left(\rho\otimes\bigotimes_{e\in E}\Ket{\Phi_{M_e}^+}\Bra{\Phi_{M_e}^+}\right)=U\rho U^\dag.
    \end{equation}

    Concentrating quantum information over $G=(V,E)$
    and $U$
    is a task of the $N$ parties $v_1\,,\ldots,v_N\in V$ applying $U^\dag$ to a shared input state $U\rho U^\dag\in\mathcal{D}\left(\tilde{\mathcal{H}}\right)$ corresponding to an arbitrary \textit{unknown} state $\rho\in\mathcal{D}\left(\mathcal{H}\right)$ to recover $\rho$ at one party $v_1\in V$ by performing an LOCC map $\mathcal{C}$ assisted by $\bigotimes_{e\in E}\Ket{\Phi_{M_e}^+}^e$, that is,
    \begin{equation}
        \label{eq:decoding}
        \mathcal{C}\left(U\rho U^\dag\otimes\bigotimes_{e\in E}\Ket{\Phi_{M_e}^+}\Bra{\Phi_{M_e}^+}\right)=\rho.
    \end{equation}
\end{definition}

Minimum requirements for initial resource states assisting LOCC protocols achieving spreading and concentrating quantum information define entanglement cost.
In the same way as the case of analyzing the graph-associated entanglement cost in constructing multipartite states~\cite{Y6},
given any graph $G=(V,E)$,
the entanglement cost of consuming the bipartite maximally entangled state $\Ket{\Phi^+_{M_e}}^e$ for each $e\in E$ of the initial resource state $\bigotimes_{e\in E}\Ket{\Phi^+_{M_e}}$ is identified by the entanglement entropy of $\Ket{\Phi^+_{M_e}}^e$, that is,
\begin{equation}
  \log_2 M_e.
\end{equation}
If a sufficiently large amount of entanglement is available for each edge, there exist trivial protocols for achieving spreading and concentrating quantum information, simply using quantum teleportation~\cite{B5} so that the party $v_1$ can locally perform any given isometry on the unknown input state.
In contrast, the aim here is to reduce the total amount of entanglement
\begin{equation}
  \sum_{e\in E}\log_2 M_e
\end{equation}
required for spreading and concentrating quantum information, or equivalently, the total amount of quantum communication when LOCC is free.

\begin{definition}
\textit{Entanglement costs in spreading and concentrating quantum information.}
The \textit{entanglement cost in spreading quantum information} over a given graph $G=(V,E)$ for a given isometry $U$ is a family
\begin{equation}
  {\left(\log_2 M_e\right)}_{e\in E}
\end{equation}
identifying an initial resource state achieving spreading quantum information over $G$ for $U$ minimizing
\begin{equation}
  \sum_{e\in E}\log_2 M_e.
\end{equation}

The \textit{entanglement cost in concentrating quantum information} over $G$ for $U$ is a family
\begin{equation}
  {\left(\log_2 M_e\right)}_{e\in E}
\end{equation}
identifying an initial resource state achieving concentrating quantum information over $G$ for $U$ minimizing
\begin{equation}
  \sum_{e\in E}\log_2 M_e.
\end{equation}
\end{definition}

\begin{figure}[t!]
  \centering
  \includegraphics[width=4in]{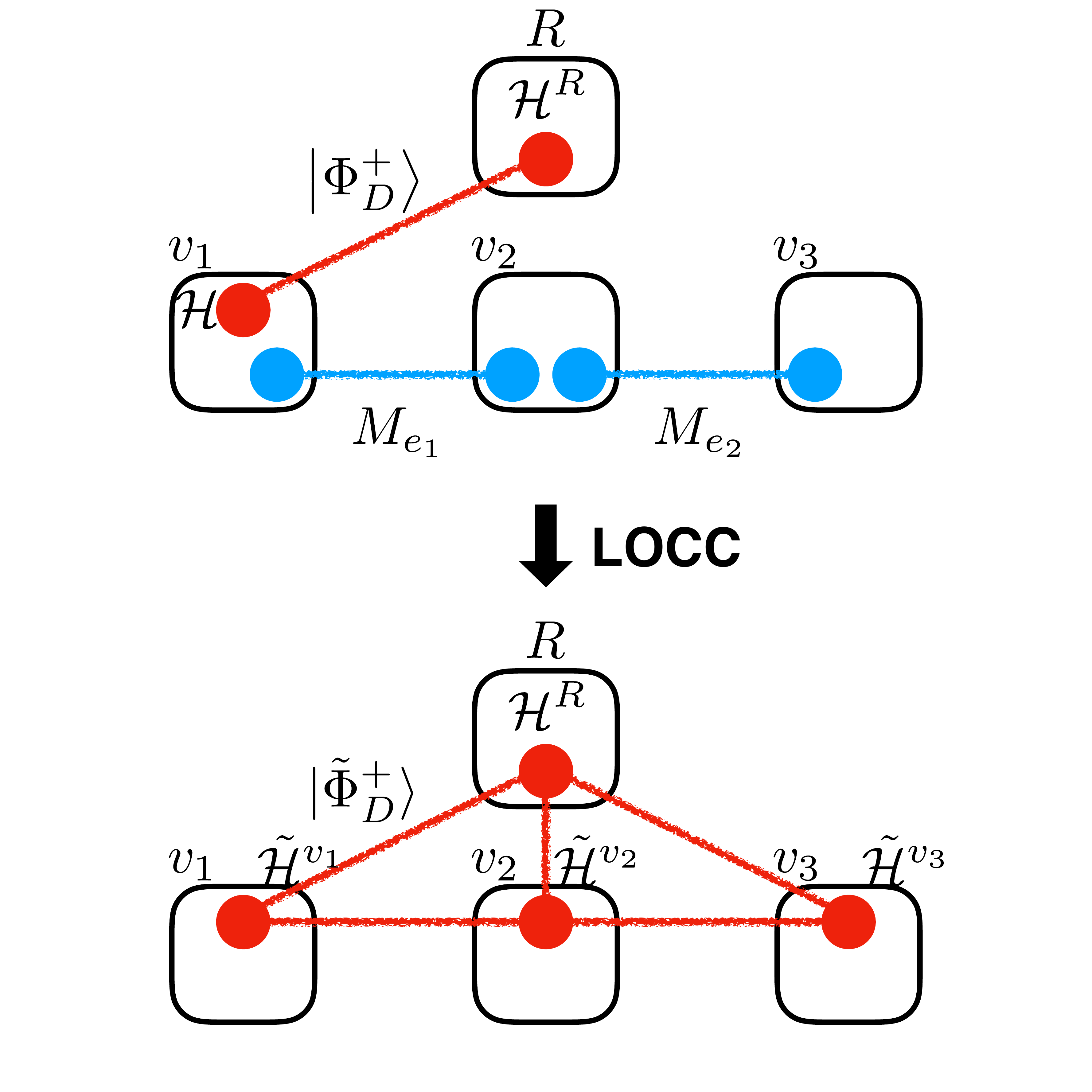}
  \caption[A state transformation task equivalent to spreading quantum information over a line-topology network.]{\label{fig:encoding}A state transformation task for three parties $v_1\,$, $v_2\,$, and $v_3$ equivalent to spreading quantum information over a line-topology network, where the system $\mathcal{H}$ for logical states is located at $v_1$. The initial state $\Ket{\Phi_D^+}$ and the final state $\ket{\tilde{\Phi_D^+}}$ are defined as Equation~\eqref{eq:data_maximally_entangled_state} and~\eqref{eq:code_maximally_entangled_state}, respectively.}
\end{figure}

\begin{figure}[t!]
  \centering
  \includegraphics[width=4in]{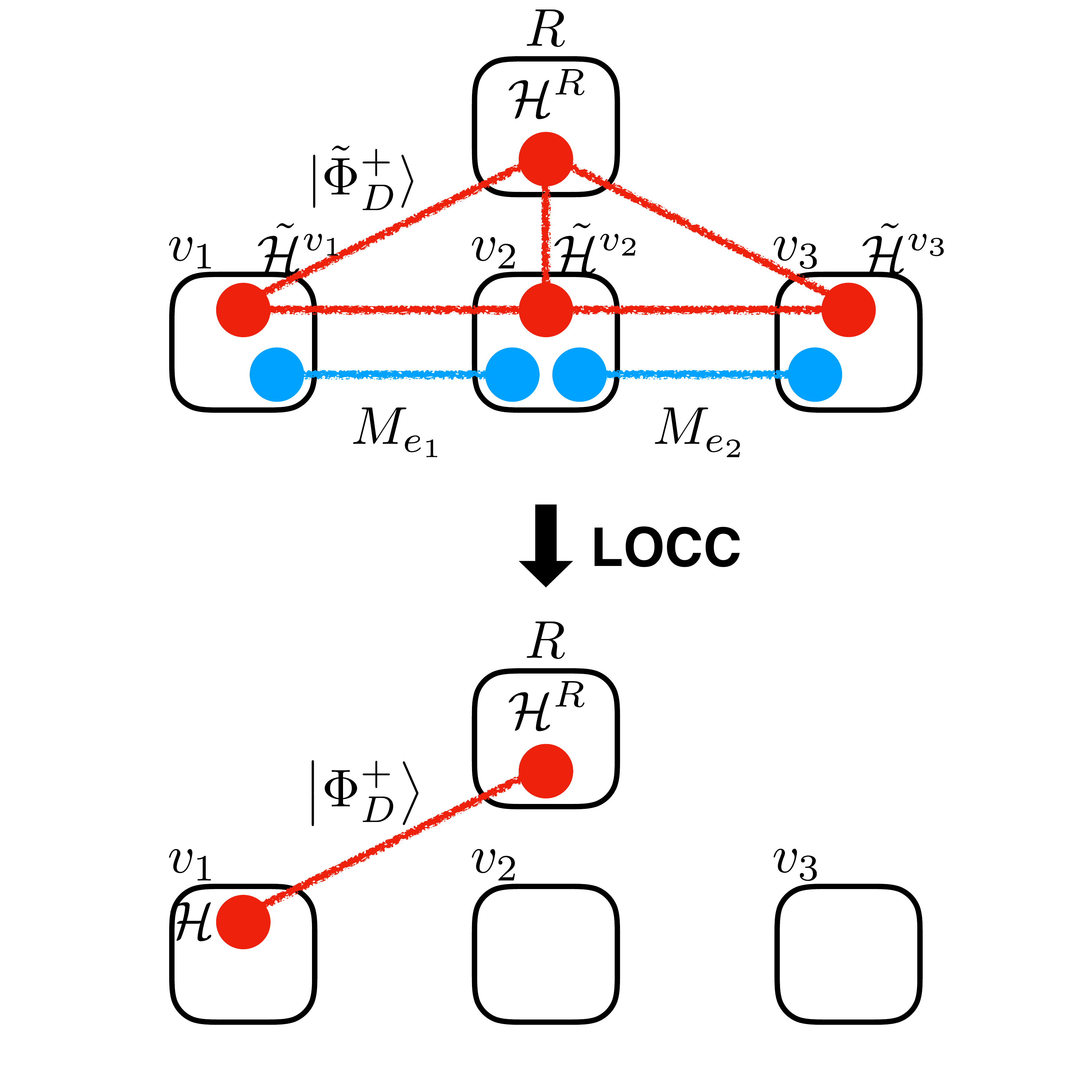}
  \caption[A state transformation task equivalent to concentrating quantum information over a line-topology network.]{\label{fig:decoding}A state transformation task equivalent to concentrating quantum information over the same network as Figure~\ref{fig:encoding}. The notations are the same as those in Figure~\ref{fig:encoding}.}
\end{figure}

To analyze the entanglement costs in spreading and concentrating quantum information, these tasks are reduced to a particular type of state transformations.
Given any graph $G=(V,E)$ and any isometry $U$,
the state transformation equivalent to spreading quantum information over $G$ for $U$ is illustrated in Figure~\ref{fig:encoding}, and the state transformation equivalent to concentrating in Figure~\ref{fig:decoding}.
To define these equivalent state transformations, consider a $D$-dimensional system $\mathcal{H}^R$ located at a party $R$ other than the $N$ parties $v_1\,,\ldots,v_N\in V$, where $\mathcal{H}^R$ is a reference system on which none of the $N$ parties can apply any operation.
Note that $D=\dim\mathcal{H}$.
Write a maximally entangled state with Schmidt rank $D$ shared between $R$ and $v_1$ as
\begin{equation}
    \label{eq:data_maximally_entangled_state}
    \Ket{\Phi^+_D}=\frac{1}{\sqrt{D}}\sum_{l=0}^{D-1}\Ket{l}\otimes\Ket{l}\in\mathcal{H}^R\otimes\mathcal{H}.
\end{equation}
Moreover, write a state obtained by performing $U$ on $\mathcal{H}$ for $\Ket{\Phi^+_D}$ as
\begin{equation}
    \label{eq:code_maximally_entangled_state}
        \ket{\tilde{\Phi}^+_D}\coloneq\left(\mathbb{1}^R\otimes U\right)\Ket{\Phi^+_D}
                              =\frac{1}{\sqrt{D}}\sum_{l=0}^{D-1}\Ket{l}\otimes\ket{\tilde{\psi}_l}\in\mathcal{H}^R\otimes\tilde{\mathcal{H}},
\end{equation}
where $\mathbb{1}^R$ is the identity operator on the system $\mathcal{H}^R$.
The equivalence between the two tasks is shown as follows, and the proof is given in Appendix~\ref{sec:equivalence_spread_concentrate},
which is a generalization of the technique of the relative state method~\cite{P3}.

\begin{proposition}
\label{lem:encoding_state_transformation}
    \textit{State transformations equivalent to spreading and concentrating quantum information.}
    Spreading quantum information over a given graph $G=(V,E)$ for a given isometry $U$ defined as Equation~\eqref{eq:encoding} is achievable if and only if there exists an LOCC map $\mathcal{S}$ by the $N$ parties assisted by the initial resource state $\bigotimes_{e\in E}\Ket{\Phi_{M_e}^+}^e$ such that
    \begin{equation}
        \label{eq:encoding_state_transformation}
            \id^R\otimes\mathcal{S}\left(\Ket{\Phi^+_D}\Bra{\Phi^+_D}\otimes\bigotimes_{e\in E}\Ket{\Phi_{M_e}^+}\Bra{\Phi_{M_e}^+}\right)
            =\ket{\tilde\Phi^+_D}\bra{\tilde\Phi^+_D},
    \end{equation}
    where $\id^R$ is the identity map on $\mathcal{H}^R$, and the states $\Ket{\Phi^+_D}$ and $\ket{\tilde\Phi^+_D}$ are defined as Equation~\eqref{eq:data_maximally_entangled_state} and~\eqref{eq:code_maximally_entangled_state}, respectively.

    Concentrating quantum information over $G=(V,E)$ for $U$ defined as Equation~\eqref{eq:decoding} is achievable if and only if there exists an LOCC map $\mathcal{C}$ by the $N$ parties assisted by $\bigotimes_{e\in E}\Ket{\Phi_{M_e}^+}^e$ such that
    \begin{equation}
        \label{eq:decoding_state_transformation}
            \id^R\otimes\mathcal{C}\left(\ket{\tilde\Phi^+_D}\bra{\tilde\Phi^+_D}\otimes\bigotimes_{e\in E}\Ket{\Phi_{M_e}^+}\Bra{\Phi_{M_e}^+}\right)
            =\Ket{\Phi^+_D}\Bra{\Phi^+_D},
    \end{equation}
    where the notations are the same as those in Equation~\eqref{eq:encoding_state_transformation}.
\end{proposition}

Calculating entanglement cost in spreading and concentrating quantum information for any network represented by an arbitrary graph is difficult due to optimization included in the definition of entanglement cost.
The following analysis is focused on a special class of graphs, \textit{trees}.
A network represented by a tree describes the situation where all parties are connected by the smallest number of channels, as discussed in Section~\ref{sec:tree}.

\section{\label{sec:encoding}Entanglement cost in spreading quantum information}

In this section, the optimal entanglement cost in spreading quantum information over any tree for any isometry is derived.
To evaluate the entanglement cost, the two-party protocol for exact state splitting shown in Theorem~\ref{thm:split} is generalized to multiple parties, so as to construct the optimal protocol for spreading quantum information over any tree-topology network connecting multiple parties.

The entanglement cost in spreading quantum information is evaluated using the following notations.
Given any tree $T=(V,E)$,
let $\tilde{\Phi}_{D,e}^{+}$ for each $e=\{p\left(v_k\right),v_k\}\in E$ denote the reduced state for $\ket{\tilde{\Phi}_D^+}$ on the system $\bigotimes_{v\in D'_{v_k}}\tilde{\mathcal{H}}^{v}$ shared among $v_k$ itself and the descendants of $v_k\,$, that is,
\begin{equation}
    \label{eq:encoding_reduced_state}
    \tilde{\Phi}_{D,e}^{+}\coloneq\tr_{R\overline{D'_{v_k}}}\ket{\tilde{\Phi}_D^+}\Bra{\tilde{\Phi}_D^+},
\end{equation}
where $\overline{D'_{v_k}}=V\setminus D'_{v_k}$ and $\tr_{R\overline{D'_{v_k}}}$ is the partial trace on $\mathcal{H}^R\otimes\bigotimes_{v\in\overline{D'_{v_k}}}\tilde{\mathcal{H}}^{v}$.
The following theorem shows an optimal protocol for the state transformation defined as Equation~\eqref{eq:decoding_state_transformation} in Proposition~\ref{lem:encoding_state_transformation} equivalent to spreading quantum information, and evaluates the optimal entanglement cost.

\begin{theorem}
\label{thm:spreading}
    \textit{Entanglement cost in spreading quantum information over trees.}
    Given any tree $T=(V,E)$ and any isometry $U$,
    spreading quantum information over $T$ for $U$ is achievable if and only if for each $e\in E$
    \begin{equation}
        \label{eq:encoding_cost}
        \log_2 M_e \geqq \log_2\rank\tilde{\Phi}_{D,e}^{+},
    \end{equation}
    where $\tilde{\Phi}_{D,e}^{+}$ is defined as Equation~\eqref{eq:encoding_reduced_state}.
\end{theorem}

\begin{proof}
    \textit{If part}:
    Given any tree $T = (V,E)$ with an ascending labeling and any isometry $U$,
    a protocol for the state transformation defined as Equation~\eqref{eq:encoding_state_transformation} in Proposition~\ref{lem:encoding_state_transformation} is constructed by applying exact state splitting in Theorem~\ref{thm:split} sequentially starting from the root party represented as $v_1\in V$,
    and the following proof also shows that this protocol achieves the equality in~\eqref{eq:encoding_cost} for each $e\in E$.
    In this protocol, the root party $v_1$ first locally applies the given isometry $U$ to
    \begin{equation}
      \Ket{\Phi_D^+}\in\mathcal{H}^R\otimes\mathcal{H}
    \end{equation}
    on $\mathcal{H}$, so as to obtain
    \begin{equation}
      \ket{\tilde{\Phi}_D^+}\in\mathcal{H}^R\otimes\tilde{\mathcal{H}},
    \end{equation}
    where $\tilde{\mathcal{H}}$ is located at $v_1$ at this moment.
    Then, the parties perform the following sub-protocol using the exact state splitting sequentially in order $v_1\,,v_2\,,\ldots,v_N$ to spread the state of $\tilde{\mathcal{H}}$.
    After all the parties performing the sub-protocol, spreading quantum information is achieved.

    The sub-protocol for each party $v_k\in V$ is shown as follows.
    At the beginning of $v_k$'s sub-protocol, it is assumed that the party $v_k$ holds the reduced state of $\ket{\tilde{\Phi}_D^+}$ on $\bigotimes_{v\in D'_{v_k}}\tilde{\mathcal{H}}^v$, that is, the system for the parties corresponding to $v_k$ itself and $v_k$'s descendants.
    Note that this assumption is satisfied because of an ascending labeling.
    If $v_k$ has no child, $v_k$'s sub-protocol terminates.
    Otherwise, for each child $c\in C_{v_k}\,$,
    $v_k$ performs the exact state splitting in Theorem~\ref{thm:split}, where $v_k$ and $c$ in the sub-protocol are regarded as $A$ and $B$ in Theorem~\ref{thm:split}, and the subsystem $\bigotimes_{v\in D'_c}\tilde{\mathcal{H}}^v$, the other subsystems of party $v_k\,$, and all the rest of the system of the parties other than $v_k$ in the sub-protocol are regarded as $\mathcal{H}^{A^\prime}$, $\mathcal{H}^A$, and $\mathcal{H}^{R}$ in Theorem~\ref{thm:split}, respectively.
    For each edge $e\in E$, Theorem~\ref{thm:split} shows that the exact state splitting in this sub-protocol achieves the equality in~\eqref{eq:encoding_cost}, where $e=\left\{v_k\,,c\right\}$ in the above case.

    \textit{Only if part}:
    The converse is derived from the LOCC monotonicity of the Schmidt rank in the state transformation defined as Equation~\eqref{eq:encoding_state_transformation} in Proposition~\ref{lem:encoding_state_transformation}.
    Consider an arbitrary edge $e=\left\{p(v_k),v_k\right\}\in E$ where $v_k\neq v_1$.
    The Schmidt rank of the initial state
    \begin{equation}
      \Ket{\Phi_D^+}^{Rv_1}\otimes\bigotimes_{e\in E}\Ket{\Phi_{M_e}^+}^{e}
    \end{equation}
    between the parties in $D'_{v_k}\,$, that is, $v_k$ itself and $v_k$'s descendants, and the other parties in $\left\{R\right\}\cup V\setminus D'_{v_k}$ is
    \begin{equation}
      M_e.
    \end{equation}
    After performing an LOCC map $\id^R\otimes\mathcal{S}$, the Schmidt rank of the final state
    \begin{equation}
      \ket{\tilde{\Phi}_D^+}^{Rv_1\cdots v_N}
    \end{equation}
    with respect to the same bipartition is
    \begin{equation}
      \rank\tilde{\Phi}_{D,e}^{+}.
    \end{equation}
    Since the Schmidt rank is monotonically nonincreasing under LOCC,
    it holds that
    \begin{equation}
      M_e\geqq\rank\tilde{\Phi}_{D,e}^{+}.
    \end{equation}
    Therefore,
    the conclusion
    \begin{equation}
      \log_2 M_e\geqq\log_2\rank\tilde{\Phi}_{D,e}^{+}
    \end{equation}
    for each $e\in E$ is obtained.
\end{proof}

\section{\label{sec:decoding}Entanglement cost in concentrating quantum information}

This section derives an upper bound of entanglement cost in concentrating quantum information over any tree for any isometry, and shows that the entanglement cost in concentrating quantum information is not larger, and can be strictly smaller, than that of spreading quantum information.
To evaluate the entanglement cost, the two-party protocol for exact state merging in the non-catalytic setting in Theorem~\ref{thm:merge_without_catalyst} is generalized to multiple parties, so as to construct a protocol for concentrating quantum information over any tree-topology network connecting multiple parties.

The entanglement cost in concentrating quantum information is evaluated using the following notations.
Given any tree $T=(V,E)$ and any isometry $U$,
the protocol achieves the state transformation defined as Equation~\eqref{eq:decoding_state_transformation} in Proposition~\ref{lem:encoding_state_transformation} equivalent to spreading quantum information.
Write the initial state shared among $R,v_1\,,\ldots,v_N$ as
\begin{equation}
  \Ket{\Phi^N}\coloneq\ket{\tilde{\Phi}_D^+}\in\mathcal{H}^R\otimes\tilde{\mathcal{H}},
\end{equation}
and the states during the protocol as a sequence
\begin{equation}
    \label{eq:sequence}
    \Ket{\Phi^N}\rightarrow\Ket{\Phi_{\boldsymbol{m}_{N-1}}^{N-1}}\rightarrow\cdots\rightarrow\Ket{\Phi_{\boldsymbol{m}_1}^1},
\end{equation}
where the subscript
$\boldsymbol{m}_k\coloneq\left({m}^{v_N},\ldots,{m}^{v_{k+1}}\right)$
denotes a tuple representing measurement outcomes obtained during the protocol and
\begin{equation}
  \label{eq:v_k}
  \Ket{\Phi_{\boldsymbol{m}_k}^k}\in\mathcal{H}^R\otimes\bigotimes_{v\in V_k}\tilde{\mathcal{H}}^{v},
  \quad V_k\coloneq\left\{v_1\,,\ldots,v_k\right\},
\end{equation}
for each $k\in\{1,\ldots,N-1\}$.
For any $\boldsymbol{m}_1\,$, the last state $\Ket{\Phi_{\boldsymbol{m}_1}^1}$ in sequence~\eqref{eq:sequence} is convertible into 
\begin{equation}
  \Ket{\Phi_D^+}\in\mathcal{H}^R\otimes\mathcal{H}
\end{equation}
by an isometry transformation by party $v_1\,$, and recurrence relation to determine sequence~\eqref{eq:sequence} is given in the proof of the following theorem (by Equation~\eqref{eq:recursive_state}).
The following theorem uses the Koashi-Imoto decomposition of $\Ket{\Phi_{\boldsymbol{m}_k}^k}$ shown in Lemma~\ref{lem:koashi_imoto_decomposition_tripartite}.
In this case, the Hilbert spaces are decomposed into
\begin{align}
  \tilde{\mathcal{H}}^{v_k}&=\bigoplus_{j=0}^{J_{\boldsymbol{m}_k}-1}\mathcal{H}_{\boldsymbol{m}_k}^{{\left(v_k\right)}_j^\textup{L}}\otimes\mathcal{H}_{\boldsymbol{m}_k}^{{\left(v_k\right)}_j^\textup{R}},\\
  \bigotimes_{v\in V_{k-1}}\tilde{\mathcal{H}}^{v}&=\bigoplus_{j=0}^{J_{\boldsymbol{m}_k}-1}\mathcal{H}_{\boldsymbol{m}_k}^{{\left(v_1\cdots v_{k-1}\right)}_j^\textup{L}}\otimes\mathcal{H}_{\boldsymbol{m}_k}^{{\left(v_1\cdots v_{k-1}\right)}_j^\textup{R}},
\end{align}
where $V_{k-1}$ is defined as Equation~\eqref{eq:v_k}.
The state is decomposed into
\begin{equation}
  \Ket{\Phi_{\boldsymbol{m}_k}^k}=\bigoplus_{j=0}^{J_{\boldsymbol{m}_k}-1}\sqrt{p_{\boldsymbol{m}_k}\left(j\right)}\Ket{\omega_{\boldsymbol{m}_k\,,j}}\otimes\Ket{\phi_{\boldsymbol{m}_k\,,j}},
\end{equation}
where $p_{\boldsymbol{m}_k}\left(j\right)$ is a probability distribution, and for each $j\in\{0,\ldots,J_{\boldsymbol{m}_k}-1\}$,
\begin{align}
  \Ket{\omega_{\boldsymbol{m}_k\,, j}}&\in\mathcal{H}_{\boldsymbol{m}_k}^{{\left(v_k\right)}_j^\textup{L}}\otimes\mathcal{H}_{\boldsymbol{m}_k}^{{\left(v_1\cdots v_{k-1}\right)}_j^\textup{L}},\\
  \Ket{\phi_{\boldsymbol{m}_k\,,j}}&\in\mathcal{H}^R\otimes\mathcal{H}_{\boldsymbol{m}_k}^{{\left(v_k\right)}_j^\textup{R}}\otimes\mathcal{H}_{\boldsymbol{m}_k}^{{\left(v_1\cdots v_{k-1}\right)}_j^\textup{R}}.
\end{align}
Also let $\lambda_{\boldsymbol{m}_k\,,0}^{{\left(v_k\right)}_j^\textup{L}}$ denote the largest eigenvalue of the reduced state of $\Ket{\omega_{\boldsymbol{m}_k\,,j}}$ on $\mathcal{H}_{\boldsymbol{m}_k}^{{\left(v_k\right)}_j^\textup{L}}$, that is,
\begin{equation}
  \tr_{{\left(v_1\cdots v_{k-1}\right)}_j^\textup{L}} \Ket{\omega_{\boldsymbol{m}_k\,,j}}\Bra{\omega_{\boldsymbol{m}_k\,,j}}.
\end{equation}

\begin{theorem}
\label{thm:concentrating}
    \textit{Entanglement cost in concentrating quantum information.}
    Given any tree $T = (V,E)$ and any isometry $U$,
    concentrating quantum information over $T$ for $U$ is achievable if there exists an ascending labeling of the vertices satisfying for each $e=\left\{p\left(v_k\right),v_k\right\}\in E$
    \begin{equation}
        \label{eq:decoding_cost_upper}
        \log_2 M_e \geqq \max_{\boldsymbol{m}_k\,,j}\left\{\log_2\left\lceil\lambda_{\boldsymbol{m}_k\,,0}^{{\left(v_k\right)}_j^\textup{L}}\dim\mathcal{H}_{\boldsymbol{m}_k}^{{\left(v_k\right)}_j^\textup{R}}\right\rceil\right\},
    \end{equation}
    where $\lceil{}\cdots{}\rceil$ is the ceiling function.
\end{theorem}

\begin{proof}
    Given any tree $T = (V,E)$ with an ascending labeling of the vertices and any isometry $U$,
    the proof is given by construction of a protocol for the state transformation defined as Equation~\eqref{eq:decoding_state_transformation} in Proposition~\ref{lem:encoding_state_transformation} achieving the equality in~\eqref{eq:decoding_cost_upper} for each $e\in E$.
    In the protocol, the parties other than the root $v_1$ sequentially perform a sub-protocol using exact state merging in the non-catalytic setting shown in Theorem~\ref{thm:merge_without_catalyst}, where each of the parties $v_N\,,\ldots,v_2$ in this order is regarded as the sender $A$ in these sequential applications of the exact state merging.
    After all of these parties performing the sub-protocol, the root party $v_1$ performs an isometry to obtain the state $\Ket{\Phi_D^+}$, which achieves concentrating quantum information.
    In the following, the sub-protocol for the non-root parties $v_N\,,\ldots,v_2$ and the isometry for the root party $v_1$ are described.

    For any $v_k\in \left\{v_N\,,\ldots,v_2\right\}$, the sub-protocol for party $v_k$ is as follows.
    The state $\Ket{\Phi^{N}}=\Ket{\tilde{\Phi_D^+}}$ may be written as $\Ket{\Phi_{\boldsymbol{m}_{N}}^{N}}$ for brevity.
    At the beginning of $v_k$'s sub-protocol, it is assumed that the party $v_k$ has the reduced state on $\tilde{\mathcal{H}}^{v_k}$ of
    \begin{equation}
       \label{eq:assumption_merge}
       \Ket{\Phi_{\boldsymbol{m}_{k}}^{k}}\in\mathcal{H}^R\otimes\tilde{\mathcal{H}}^{v_k}\otimes\bigotimes_{m=1}^{k-1}\tilde{\mathcal{H}}^{v_m}.
    \end{equation}
    Based on the classical information $m^{v_N},\ldots,m^{v_{k+1}}$ of measurement outcomes sent from other parties by classical communication, the party $v_k$ calculates the measurement basis ${\left\{\Ket{m^{v_k}}\right\}}_{m^{v_k}}$ defined as
    \begin{equation}
      \Ket{m^{v_k}}\coloneq U^A\Ket{m_1\,,m_2\,,m_3},
    \end{equation}
    where $m^{v_k}$ on the left-hand side is a label corresponding to the tuple of three labels on the right-hand side $m_1\,$, $m_2\,$, and $m_3\,$, and $U^A\Ket{m_1\,,m_2\,,m_3}$ on the right-hand side is that in Equation~\eqref{eq:merge_without_catalyst} in the proof of Theorem~\ref{thm:merge_without_catalyst} for the exact state merging of $\Ket{\Phi_{\boldsymbol{m}_{k}}^{k}}$ in the non-catalytic setting,
    in which the systems $\mathcal{H}^R$, $\tilde{\mathcal{H}}^{v_k}$, and $\bigotimes_{m=1}^{k-1}\tilde{\mathcal{H}}^{v_m}$ are regarded as $\mathcal{H}^{R}$, $\mathcal{H}^A$, and $\mathcal{H}^B$ in Equation~\eqref{eq:merge_without_catalyst}, respectively.
    The party $v_k$ performs this measurement,
    and the states in the sequence~\eqref{eq:sequence} are recursively described as
    \begin{equation}
        \label{eq:recursive_state}
        \Ket{\Phi_{\boldsymbol{m}_{k-1}}^{k-1}}=\frac{\left(\mathbb{1}\otimes\Bra{m^{v_k}}\right)\left(\Ket{\Phi_{\boldsymbol{m}_{k}}^{k}}\otimes\Ket{\Phi_{M_{e}}^+}^{e}\right)}{\left\|\left(\mathbb{1}\otimes\Bra{m^{v_k}}\right)\left(\Ket{\Phi_{\boldsymbol{m}_{k}}^{k}}\otimes\Ket{\Phi_{M_{e}}^+}^{e}\right)\right\|},
    \end{equation}
    where $\mathbb{1}$ is the identity operator on the system of the parties other than $v_k\,$, $\Ket{\Phi_{M_{e}}^+}^{e}$ with $e=\left\{p\left(v_k\right),v_k\right\}$ is the resource state shared between $v_k$ and $v_k$'s parent $p\left(v_k\right)$, and the system of party $p\left(v_k\right)$ for the resource state $\Ket{\Phi_{M_{e}}^+}^{e}$ on the right-hand side is regarded on the left-hand side as part of $\tilde{\mathcal{H}}^{p\left(v_k\right)}$ of the party $p\left(v_k\right)$.
    After this measurement, the party $v_k$ sends the measurement outcome $m^{v_k}$ to all the parties by classical communication, where the post-measurement state is represented by $\Ket{\Phi_{\boldsymbol{m}_{k-1}}^{k-1}}$.
    Note that the assumption~\eqref{eq:assumption_merge} is satisfied for the next party $v_{k-1}$ performing the sub-protocol, that is,
    \begin{equation}
       \Ket{\Phi_{\boldsymbol{m}_{k-1}}^{k-1}}\in\mathcal{H}^R\otimes\tilde{\mathcal{H}}^{v_{k-1}}\otimes\bigotimes_{m=1}^{k-2}\tilde{\mathcal{H}}^{v_m},
    \end{equation}
    because of an ascending order of the vertices.
    For each edge $e=\left\{p\left(v_k\right),v_k\right\}\in E$, Theorem~\ref{thm:merge_without_catalyst} shows that the exact state merging in this sub-protocol achieves the equality in~\eqref{eq:decoding_cost_upper}.

    As for the root party $v_1\,$, an isometry $U_{\boldsymbol{m}_1}^{v_1}$ for obtaining the state $\Ket{\Phi_D^+}$ is shown as follows.
    After the parties $v_N\,,\ldots,v_2$ performing the above sub-protocols, the shared state reduces to
    $\Ket{\Phi_{\boldsymbol{m}_{1}}^{1}}\in\mathcal{H}^R\otimes\tilde{\mathcal{H}}^{v_1}$.
    For each $v_k\in\left\{v_N\,,\ldots,v_2\right\}$,
    define an isometry
    \begin{equation}
      U_{m^{v_k}}\coloneq\left({U^{B'}}^\dag\otimes{U^B}^\dag\right) U_{m_1\,,m_2\,,m_3}U^B,
    \end{equation}
    where $m^{v_k}$ on the left-hand side is a label corresponding to the tuple of three labels on the right-hand side $m_1\,$, $m_2\,$, and $m_3\,$, and $\left({U^{B'}}^\dag\otimes{U^B}^\dag\right) U_{m_1\,,m_2\,,m_3}U^B$ on the right-hand side is that in Equation~\eqref{eq:merge_without_catalyst} in the proof of Theorem~\ref{thm:merge_without_catalyst}.
    Each $U_{m^{v_k}}$ recovers the state $\Ket{\Phi_{\boldsymbol{m}_{k}}^{k}}$ from the post-measurement state $\Ket{\Phi_{\boldsymbol{m}_{k-1}}^{k-1}}$ corresponding to $\Bra{m^{v_k}}$, that is,
    \begin{equation}
      \Ket{\Phi_{\boldsymbol{m}_{k}}^{k}}=U_{m^{v_k}}\Ket{\Phi_{\boldsymbol{m}_{k-1}}^{k-1}}.
    \end{equation}
    Repeating the above yields
    \begin{equation}
        \ket{\tilde{\Phi}_D^+}=\Ket{\Phi^N}=U_{m^{v_N}}\cdots U_{m^{v_2}}\Ket{\Phi_{\boldsymbol{m}_{1}}^{1}}.
    \end{equation}
    Consequently, the party $v_1$ obtains for any $\boldsymbol{m}_1$
    \begin{align}
        \Ket{\Phi_D^+}&=U_{\boldsymbol{m}_1}^{v_1}\Ket{\Phi_{\boldsymbol{m}_{1}}^{1}},\\
        U_{\boldsymbol{m}_1}^{v_1}&\coloneq U^\dag U_{m^{v_N}}\cdots U_{m^{v_2}}.
    \end{align}
    Note that it may not be possible for the parties $v_N\,,\ldots,v_2$ to locally perform $U_{m^{v_N}},\ldots,U_{m^{v_2}}$ during the sub-protocol, since these isometries can be nonlocal.
\end{proof}

The following theorem shows that the entanglement cost in concentrating quantum information is not larger than that of spreading quantum information.
Moreover, the former can be strictly smaller than the latter, as demonstrated in Applications~\ref{ex:distributed_source_compression} and~\ref{ex:locc_decoding} in the next section.
Note that this difference in entanglement cost arises from the difference between quantum state merging and splitting discussed in Remark~\ref{remark:merge}.

\begin{theorem}
    \textit{Comparison of entanglement cost between spreading and concentrating quantum information.}
    Given any tree $T=(V,E)$ with any ascending labeling and any isometry $U$, it holds that
    \begin{equation}
      \max_{\boldsymbol{m}_k\,,j}\left\{\log_2\left\lceil\lambda_{\boldsymbol{m}_k\,,0}^{{\left(v_k\right)}_j^\textup{L}}\dim\mathcal{H}_{\boldsymbol{m}_k}^{{\left(v_k\right)}_j^\textup{R}}\right\rceil\right\}\leqq\log_2\rank\tilde{\Phi}_{D,e}^{+}
    \end{equation}
    where the notations are the same as those in Theorems~\ref{thm:spreading} and~\ref{thm:concentrating}.
\end{theorem}

\begin{proof}
  The proof uses the LOCC monotonicity of the Schmidt rank in the state transformation defined as Equation~\eqref{eq:decoding_state_transformation} in Proposition~\ref{lem:encoding_state_transformation}, and properties of the Koashi-Imoto decomposition.
  Regard the given tree $T=(V,E)$ as the rooted tree with its root $v_1\,$, and consider an arbitrary edge $e=\left\{p(v_k),v_k\right\}\in E$ where $v_k\neq v_1$.
  The Schmidt rank of the initial state
  \begin{equation}
    \ket{\tilde{\Phi}_D^+}^{Rv_1\cdots v_N}\otimes\bigotimes_{e\in E}\Ket{\Phi_{M_e}^+}^{e}
  \end{equation}
  between the parties in $D'_{v_k}$ and the other parties in $\{R\}\cup V\setminus D'_{v_k}$ is
  \begin{equation}
    M_e\rank\tilde{\Phi}_{D,e}^{+}.
  \end{equation}
  After the parties $v_N\,,\ldots,v_{k-1}$ performing the above sub-protocols, which is an LOCC map, the state reduces to
  \begin{equation}
    \Ket{\Phi_{\boldsymbol{m}_k}^k}\otimes\bigotimes_{e\in E_k}\Ket{\Phi_{M_e}^+}^{e},
  \end{equation}
  where $\Ket{\Phi_{\boldsymbol{m}_k}^k}$ is defined as Equation~\eqref{eq:v_k} and
  $E_k\coloneq\left\{\left\{p\left(v_2\right),v_2\right\},\ldots,\left\{p\left(v_k\right),v_k\right\}\right\}$.
  The Schmidt rank of $\Ket{\Phi_{\boldsymbol{m}_k}^k}\otimes\bigotimes_{e\in E_k}\Ket{\Phi_{M_e}^+}^{e}$ with respect to the same bipartition of the parties as the above is
  \begin{equation}
    M_e \rank{\left(\Phi_{\boldsymbol{m}_k}^k\right)}^{v_k},
  \end{equation}
  where ${\left(\Phi_{\boldsymbol{m}_k}^k\right)}^{v_k}$ denotes the reduced state of the system $\tilde{\mathcal{H}}^{v_k}$ for the state $\Ket{\Phi_{\boldsymbol{m}_k}^k}$.
  Since the Schmidt rank is monotonically nonincreasing under LOCC,
  it holds that
  \begin{equation}
    M_e\rank\tilde{\Phi}_{D,e}^{+}\geqq M_e\rank{\left(\Phi_{\boldsymbol{m}_k}^k\right)}^{v_k}.
  \end{equation}
  By construction of the Koashi-Imoto decomposition, it holds that
  \begin{equation}
    \rank{\left(\Phi_{\boldsymbol{m}_k}^k\right)}^{v_k}\geqq\dim\mathcal{H}_{\boldsymbol{m}_k}^{{\left(v_k\right)}_j^R}.
  \end{equation}
  for any ${\boldsymbol{m}_k}$ and $j$.
  Since $\lambda_{\boldsymbol{m}_k\,,0}^{{\left(v_k\right)}_j^L}\leqq 1$,
  it is obtained that
  \begin{equation}
    \dim\mathcal{H}_{\boldsymbol{m}_k}^{{\left(v_k\right)}_j^R}\geqq\left\lceil\lambda_{\boldsymbol{m}_k\,,0}^{{\left(v_k\right)}_j^L}\dim\mathcal{H}_{\boldsymbol{m}_k}^{{\left(v_k\right)}_j^R}\right\rceil.
  \end{equation}
  Thus, for any ${\boldsymbol{m}_k}$ and $j$, it holds that
  \begin{equation}
    \log_2\rank\tilde{\Phi}_{D,e}^{+}\geqq\log_2\left\lceil\lambda_{\boldsymbol{m}_k\,,0}^{{\left(v_k\right)}_j^L}\dim\mathcal{H}_{\boldsymbol{m}_k}^{{\left(v_k\right)}_j^R}\right\rceil.
  \end{equation}
  Therefore, the conclusion
  \begin{equation}
    \max_{\boldsymbol{m}_k\,,j}\left\{\log_2\left\lceil\lambda_{\boldsymbol{m}_k\,,0}^{{\left(v_k\right)}_j^L}\dim\mathcal{H}_{\boldsymbol{m}_k}^{{\left(v_k\right)}_j^R}\right\rceil\right\}\leqq\log_2\rank\tilde{\Phi}_{D,e}^{+}
  \end{equation}
  for each $e=\left\{p\left(v_k\right),v_k\right\}\in E$ is obtained.
\end{proof}

\section{\label{sec:example}Applications}
Applications of the protocols for spreading and concentrating quantum information are provided in this section.
In the following, $\otimes$ may be omitted if obvious.

\begin{implication}
\label{ex:distributed_source_compression}
\textit{Application to one-shot distributed source compression for arbitrarily small-dimensional systems.}
When applied to a star-topology tree, such as
\begin{equation}
    \label{eq:star}
    \begin{split}
      T&=(V,E),\\
      V&=\left\{v_k:k=1,2,3\right\},\\
      E&=\left\{e_1=\left\{v_1\,,v_2\right\},e_2=\left\{v_1\,,v_3\right\}\right\},
    \end{split}
\end{equation}
the protocol for concentrating quantum information shown in Theorem~\ref{thm:concentrating} can be regarded as a protocol for one-shot zero-error distributed source compression~\cite{D8,D9,A8}.
Although the protocol achieves transformations between $\Ket{\Phi_D^+}$ and $\ket{\tilde{\Phi}_D^+}$, that is, maximally entangled states between $R$ and the others, it is straightforward to prove that the protocol also work for any pure state shared among the parties $R, v_1\,,\ldots,v_N\,$, which is proven for two parties in Proposition~\ref{prp:max}, and the same argument also applies to more than two parties.
Note that the protocol for concentrating quantum information is applicable to arbitrarily small-dimensional systems as well as achieving zero error, while the existing protocols for the one-shot distributed source compression~\cite{D8,D9,A8} are inefficient for small- and intermediate-scale states and cannot avoid nonzero approximation error, similarly to the case of $N=2$ discussed in Remark~\ref{remark:usefulness}.

For the network defined as Equation~\eqref{eq:star} and an isometry mapping the basis states as
\begin{equation}
  \begin{split}
    \Ket{0}&\leftrightarrow\Ket{0}^{v_1}\Ket{0}^{v_2}\Ket{0}^{v_3},\\
    \Ket{1}&\leftrightarrow\frac{1}{\sqrt{2}}\left(\Ket{0}^{v_1}+\Ket{1}^{v_1}\right)\Ket{1}^{v_2}\Ket{1}^{v_3},
  \end{split}
\end{equation}
where the three-qubit states on the right-hand sides are orthogonal to each other due to the orthogonality of $\Ket{0}$ and $\Ket{1}$,
Theorem~\ref{thm:spreading} yields the entanglement cost in spreading quantum information
\begin{equation}
  \begin{split}
    \log_2 M_{e_1}&=1,\\
    \log_2 M_{e_2}&=1,
  \end{split}
\end{equation}
and Theorem~\ref{thm:concentrating} yields a protocol for concentrating quantum information achieving
\begin{equation}
  \begin{split}
    \log_2 M_{e_1}&=1,\\
    \log_2 M_{e_2}&=0\neq 1.
  \end{split}
\end{equation}
In concentrating quantum information, the states in sequence~\eqref{eq:sequence} are calculated as
\begin{align}
   \label{eq:app1}
   &\Ket{\Phi^{3}}=\ket{\tilde{\Phi}_D^+}=\frac{1}{\sqrt{2}}\Ket{0}^R\Ket{0}^{{\left(v_3\right)}_0^R}\Ket{00}^{{\left(v_1 v_2\right)}_0^R}\oplus\left(\pm\frac{1}{\sqrt{2}}\Ket{1}^R\Ket{1}^{{\left(v_3\right)}_1^R}\Ket{+1}^{{\left(v_1 v_2\right)}_1^R}\right)\\
   \label{eq:app2}
   &\xrightarrow{\text{Measurement in }\left\{\Ket{\pm}^{v_3}\right\}}
      \Ket{\Phi_{\left(\Ket{\pm}^{v_3}\right)}^{2}}=\frac{1}{\sqrt{2}}\Ket{0}^R\Ket{0}^{{\left(v_2\right)}_0^R}\Ket{0}^{{\left(v_1\right)}_0^R}\pm\frac{1}{\sqrt{2}}\Ket{1}^R\Ket{1}^{{\left(v_2\right)}_0^R}\Ket{+}^{{\left(v_1\right)}_0^R}
\end{align}
where the right-hand sides of Equations~\eqref{eq:app1} and~\eqref{eq:app2} shows the Koashi-Imoto decomposition of the state for each step in the sequence~\eqref{eq:sequence},
and the final state shared between $R$ and $v_1$ is obtained by transferring $v_2$'s one-qubit state by quantum teleportation from $v_2$ to $v_1\,$, which requires $\log_2 M_{e_1}=1$.
The difference in the resource requirements for concentrating quantum information between the edges $e_1$ and $e_2$ arises because of the difference between the Koashi-Imoto decomposition of the state $\Ket{\Phi_{\boldsymbol{m}_{2}}^{2}}$ after the party $v_3$ performing the exact state merging and the Koashi-Imoto decomposition of the state $\Ket{\Phi^{3}}=\ket{\tilde{\Phi}_D^+}$ before.

By contrast, if the labeling of the parties $v_2$ and $v_3$ are interchanged, the tree $T$ changes to
\begin{equation}
    \begin{split}
      &T'=(V',E),\\
      &V'=\left\{v'_1=v_1\,, v'_2=v_3\,, v'_3=v_2\right\},\\
      &E=\left\{e_1=\left\{v'_1\,,v'_3\right\}=\left\{v_1\,,v_2\right\},e_2=\left\{v'_1\,,v'_2\right\}=\left\{v_1\,,v_3\right\}\right\},
    \end{split}
\end{equation}
and the protocol for concentrating quantum information applied to this tree $T'$ achieves
\begin{equation}
  \begin{split}
    \log_2 M_{e_1}&=0\neq 1,\\
    \log_2 M_{e_2}&=1.
  \end{split}
\end{equation}

This example implies that the entanglement cost in concentrating quantum information for each edge of a graph may be affected by the labeling of the vertices, that is, the order of sequential applications of exact state merging.
In this case, to obtain the entanglement cost, it is necessary to calculate the sequence~\eqref{eq:sequence} of the states during the protocol by recursively using Equation~\eqref{eq:recursive_state}.
\end{implication}

\begin{implication}
\label{ex:locc_decoding}
\textit{Application to LOCC-assisted decoding in quantum secret sharing.}
Similarly to the protocol for concentrating quantum information, Reference~\cite{G8} proposes schemes of quantum secret sharing and a protocol for decoding shared secret of quantum information, in which the parties collaboratively perform LOCC to reduce total quantum communication required for the decoding.
While the protocol in Reference~\cite{G8} works for a particular class of quantum codes,
the protocols shown in Theorems~\ref{thm:spreading} and~\ref{thm:concentrating} are applicable to any encoding and decoding in addition to this particular class.
For example, a different scheme of quantum secret sharing from those considered in Reference~\cite{G8} can be obtained from the five-qubit code~\cite{C,G7}, which maps the basis states as
\begin{equation}
  \begin{split}
    \Ket{0}\leftrightarrow\frac{1}{4}(&\Ket{00000} + \Ket{11000} + \Ket{01100} + \Ket{00110}\\
    +&\Ket{00011}+\Ket{10001}-\Ket{10100}-\Ket{01010}\\
    -&\Ket{00101}-\Ket{10010}-\Ket{01001}-\Ket{11110}\\
    -&\Ket{01111}-\Ket{10111}-\Ket{11011}-\Ket{11101}),\\
    \Ket{1}\leftrightarrow\frac{1}{4}(&\Ket{11111} + \Ket{00111} + \Ket{10011} + \Ket{11001}\\
    +&\Ket{11100} + \Ket{01110}-\Ket{01011}-\Ket{10101}\\
    -&\Ket{11010}-\Ket{01101}-\Ket{10110}-\Ket{00001}\\
    -&\Ket{10000}-\Ket{01000}-\Ket{00100}-\Ket{00010}),
  \end{split}
\end{equation}
where each qubit on the right-hand sides belongs to each of the parties $v_1\,,\ldots,v_5$.
For this isometry and a line-topology tree
\begin{equation}
    \begin{split}
      T&=(V,E),\\
      V&=\left\{v_k:k=1,\ldots,N\right\},\\
      E&=\left\{e_k=\left\{v_k\,,v_{k+1}\right\}:k=1,\ldots,N-1\right\},
    \end{split}
\end{equation}
where $N=5$,
Theorem~\ref{thm:spreading} yields the entanglement cost in spreading quantum information
\begin{equation}
    \begin{split}
      \log_2 M_{e_1}&=2,\\
      \log_2 M_{e_2}&=3,\\
      \log_2 M_{e_3}&=2,\\
      \log_2 M_{e_4}&=1,
    \end{split}
\end{equation}
and Theorem~\ref{thm:concentrating} yields a protocol for concentrating quantum information achieving
\begin{equation}
  \label{eq:ex3}
  \begin{split}
    \log_2 M_{e_1}&=0,\\
    \log_2 M_{e_2}&=0,\\
    \log_2 M_{e_3}&=0,\\
    \log_2 M_{e_4}&=0.
  \end{split}
\end{equation}
In concentrating quantum information, the states in sequence~\eqref{eq:sequence} are calculated as
\begin{align*}
    &\Ket{\Phi^{5}}=\ket{\tilde{\Phi}_D^+}\propto\Ket{+}^R\Ket{+}^{{\left(v_5\right)}_0^R}{\left(\Ket{0000}^{{\left(v_1 v_2 v_3 v_4\right)}_0^R}+\cdots\right)}\oplus\Ket{-}^R\Ket{-}^{{\left(v_5\right)}_1^R}{\left(\Ket{0000}^{{\left(v_1 v_2 v_3 v_4\right)}_1^R}+\cdots\right)}\\
    &\downarrow{\text{Measurement in }\left\{\Ket{0}^{v_5},\Ket{1}^{v_5}\right\}}\\
    &\Ket{\Phi_{\left(\Ket{0}^{v_5}\right)}^{4}}\\
    &=\frac{1}{4}\left[\Ket{0}^{R}\left(\Ket{0000}^{v_1 v_2 v_3 v_4} + \Ket{1100} + \Ket{0110} + \Ket{0011} -\Ket{1010}-\Ket{0101}-\Ket{1001}-\Ket{1111}\right)\right.\\
    &\quad\left.+\Ket{1}^R\left(\Ket{1110}^{v_1 v_2 v_3 v_4} + \Ket{0111}-\Ket{1101}-\Ket{1011} -\Ket{1000}-\Ket{0100}-\Ket{0010}-\Ket{0001}\right)\right]\\
    &\propto\Ket{+}^R\Ket{+}^{{\left(v_4\right)}_0^R}\left(\Ket{000}^{{\left(v_1 v_2 v_3\right)}_0^R}+\cdots\right)\oplus\Ket{-}^R\Ket{-}^{{\left(v_4\right)}_0^R}\left(\Ket{000}^{{\left(v_1 v_2 v_3\right)}_0^R}+\cdots\right)\\
    &\downarrow{\text{Measurement in }\left\{\Ket{0}^{v_4},\Ket{1}^{v_4}\right\}}\\
    &\Ket{\Phi_{\left(\Ket{0}^{v_5},\Ket{0}^{v_4}\right)}^{3}}\\
    &=\frac{1}{2\sqrt{2}}\left[\Ket{0}^R\left(\Ket{000}^{v_1 v_2 v_3} + \Ket{110} + \Ket{011} - \Ket{101}\right) +\Ket{1}^R\left(\Ket{111}^{v_1 v_2 v_3} - \Ket{100}-\Ket{010}-\Ket{001}\right)\right]\\
    &\propto\Ket{+}^R\Ket{+}^{{\left(v_3\right)}_0^R}\left(\Ket{00}^{{\left(v_1 v_2\right)}_0^R}+\cdots\right)\oplus\Ket{-}^R\Ket{-}^{{\left(v_3\right)}_0^R}\left(\Ket{00}^{{\left(v_1 v_2\right)}_0^R}+\cdots\right)\\
    &\downarrow{\text{Measurement in }\left\{\Ket{0}^{v_3},\Ket{1}^{v_3}\right\}}\\
    &\Ket{\Phi_{\left(\Ket{0}^{v_5},\Ket{0}^{v_4},\Ket{0}^{v_3}\right)}^{2}}\\
    &=\frac{1}{2}\left[\Ket{0}^R\left(\Ket{00}^{v_1 v_2}+\Ket{11}\right)-\Ket{1}^R\left(\Ket{01}^{v_1 v_2}+\Ket{10}\right)\right]\\
    &\propto\Ket{-}^R\Ket{+}^{{\left(v_2\right)}_0^R}\Ket{+}^{{\left(v_1\right)}_0^R}\oplus\Ket{+}^R\Ket{-}^{{\left(v_2\right)}_1^R}\Ket{-}^{{\left(v_1\right)}_1^R},\\
    &\downarrow{\text{Measurement in }\left\{\Ket{0}^{v_2},\Ket{1}^{v_2}\right\}}\\
    &\Ket{\Phi_{\left(\Ket{0}^{v_5},\Ket{0}^{v_4},\Ket{0}^{v_3},\Ket{0}^{v_2}\right)}^{1}}=\frac{1}{\sqrt{2}}\left(\Ket{0}^{R}\Ket{0}^{v_1}-\Ket{1}^R\Ket{1}^{v_1}\right)\\
    &\downarrow{\text{Local isometry by }v_1}\\
    &\frac{1}{\sqrt{2}}\left(\Ket{0}^{R}\Ket{0}^{v_1}+\Ket{1}^R\Ket{1}^{v_1}\right)
\end{align*}
where the Koashi-Imoto decomposition of the state for each step in the sequence~\eqref{eq:sequence} is shown after $\propto$ for the above states,
and while only the sequence of states for the measurement outcomes corresponding to $\Ket{0}$'s is shown in the above, those corresponding to other outcomes can be calculated in the same way.
Equation~\eqref{eq:ex3} shows that the five-qubit code can be decoded only by LOCC, \textit{i.e.,} without quantum communication.
Note that, if the protocol in Theorem~\ref{thm:concentrating} is applied to quantum secret sharing, some subsets of the parties may extract partial knowledge about the shared secret of quantum information during the protocol while this is the same situation as the existing protocol in Reference~\cite{G8}.
\end{implication}

\chapter{\label{sec:multipartite}When does multipartite entanglement outperform bipartite entanglement?}

This chapter aims at differentiating capabilities of multipartite entanglement and bipartite entanglement.
To achieve this goal, Section~\ref{sec:def} introduces the tasks of system-size-limited quantum state preparation in the static and dynamic settings.
The static setting is analyzed in Section~\ref{sec:analysis}, and the dynamic setting is analyzed in Section~\ref{sec:analysis2}.

\section{\label{sec:def}Definition of system-size-limited quantum state preparation}

This section defines the tasks of system-size-limited quantum state preparation, where difference between states exhibiting multipartite entanglement and state consisting only bipartite entanglement arises in achievability of these tasks.
These tasks are also illustrated in Figure~\ref{fig:multipartite_intro}.

Consider a scenario where a multipartite system is distributed among spatially separated parties $v_1\,,\ldots,v_N\,$, and the local system size of each party is limited.
Given a target set $S$ of multipartite states of this distributed system,
the system-size-limited quantum state preparation for $S$  is a task of the parties transforming a shared common resource state stored within the limitation of local system sizes into an arbitrary state $\Ket{\psi}\in S$ by local operations and classical communication, where use of auxiliary systems is also limited within the limitation.

The limitation on local system sizes are formulated as follows.
Assume that each party $v_k\in\left\{v_1\,,\ldots,v_N\right\}$ has a quantum system represented by a Hilbert space $\overline{\mathcal{H}}^{v_k}$, whose dimension is
\begin{equation}
    D^{\left(v_k\right)}\coloneq\dim\overline{\mathcal{H}}^{v_k}.
\end{equation}
The total system shared by the parties is denoted by
\begin{equation}
  \overline{\mathcal{H}}\coloneq \bigotimes_{k=1}^{N}\overline{\mathcal{H}}^{v_k}.
\end{equation}
The configuration of system sizes for the parties is represented as a tuple
\begin{equation}
    \boldsymbol{D}=\left(D^{\left(v_1\right)},\ldots,D^{\left(v_N\right)}\right).
\end{equation}

The parties store a common resource state within this configuration $\boldsymbol{D}$ of a given system $\overline{\mathcal{H}}$.
This common resource state is to be transformed by LOCC into a state in a given target set, so that the state in the transformed form can be used for some given task.
In general, a common resource state for a set of multipartite states may be of a higher-dimensional system than the system for the set itself.
Thus, states in the target set obtained from the common resource state by LOCC is of a subspace $\mathcal{H}$ of $\overline{\mathcal{H}}$ where each party $v_k$ has a subsystem $\mathcal{H}^{v_k}$ of $\mathcal{H}$, that is,
\begin{align}
  \mathcal{H}&\coloneq\bigotimes_{k=1}^N\mathcal{H}^{v_k},\\
  \mathcal{H}^{v_k}&\subset\overline{\mathcal{H}}^{v_k},\quad\forall v_k.
\end{align}
The target set $S$ of states to be obtained from the common resource state is given from $\mathcal{H}$.

Each party $v_k$ may perform any unitary operations and any measurement on the system $\overline{\mathcal{H}}^{v_k}$,
but $v_k$ is \textit{not} allowed to add auxiliary systems increasing the dimension of $\overline{\mathcal{H}}^{v_k}$.
Measurements can be represented by quantum instruments,
and while there exists a class of measurements called indirect measurements, which may require an auxiliary working quantum system in their implementation,
the protocols investigated in this chapter require only projective measurements, which can be considered to be implementable without such an auxiliary system.
For the completeness of the definition,  $v_k$ may be allowed to implement an indirect measurement using a projective measurement and one auxiliary working qubit in addition to the system $\overline{\mathcal{H}}^{v_k}$ itself, where the auxiliary working qubit has to be traced out after each measurement.
Note that the use of only one auxiliary working qubit is sufficient for implementing any indirect measurement~\cite{A13}.
The parties can freely perform classical information processing and classical communication, which can be performed without using a quantum system.
Given a configuration of system sizes $\boldsymbol{D}$,
assume in both the static setting and the dynamic setting that the parties can perform local operations on a limited-size quantum system in the above sense and classical communication,
and this restricted LOCC is called \textit{LOCC within the configuration $\boldsymbol{D}$}.

To compare multipartite and bipartite entanglement for the common resource states, two settings of system-size-limited quantum state preparation are defined in the following, one of which is called the \textit{static} setting, and the other the \textit{dynamic} setting.
The task of system-size-limited quantum state preparation is the static setting is defined as follows.
\begin{definition}
  \textit{System-size-limited quantum state preparation in the static setting}
  The system-size-limited quantum state preparation in the static setting for a configuration $\boldsymbol{D}$ of system sizes and a target set $S$ is a task of $N$ parties achieving the following:
  \begin{enumerate}
    \item The system $\mathcal{H}$ shared by the parties is initialized as a common resource state $\Ket\phi\in\overline{\mathcal{H}}$ for $S$;
    \item A particular target state $\Ket{\psi}\in S$ is chosen from the target set $S$, and all the parameters of $\Ket{\psi}$ for its classical description are given to all the parties. Then, the parties perform LOCC within the configuration $\boldsymbol{D}$ to transform the common resource state $\Ket\phi$ into the chosen target state $\Ket\psi$ in the target set.
  \end{enumerate}
\end{definition}
Section~\ref{sec:analysis} analyzes properties of the common resource state $\Ket\phi$ for achieving a system-size-limited quantum state preparation, that is, whether the task is achievable when the common resource state $\Ket\phi$ is a state exhibiting multipartite entanglement or consisting only of bipartite entanglement.

In the dynamic setting, in addition to LOCC within a given configuration $\boldsymbol{D}$,
allow any two parties $v_k$ and $v_{k^\prime}$ to perform quantum communication.
Each quantum communication from a party $v_k$ to another $v_{k^\prime}$ is called one \textit{round} of quantum communication.
A protocol  may include multiple rounds of quantum communication, and these multiple rounds are performed sequentially.
When $v_k$ sends a state of a $D$-dimensional system to $v_{k^\prime}$  by quantum communication, it is required that $v_k$ initially stores the state to be sent in a $D$-dimensional subsystem of $\overline{\mathcal{H}}^{v_k}$, and $v_{k^\prime}$ initializes a $D$-dimensional subsystem of $\overline{\mathcal{H}}^{v_{k^\prime}}$ as a fixed state $\Ket{0}$, so that $v_{k^\prime}$ receives the state using this subsystem.
After each quantum communication, the $D$-dimensional subsystem of $\overline{\mathcal{H}}^{v_k}$ is initialized as a fixed state $\Ket{0}$, so that $v_k$ can reuse this subsystem.
Note that quantum communication between the parties is not allowed in the static setting and is allowed only in the dynamic setting.

The task of system-size-limited quantum state preparation in the dynamic setting is defined as follows.
\begin{definition}
  \textit{System-size-limited quantum state preparation in the dynamic setting}
  The system-size-limited quantum state preparation in the dynamic setting for a configuration $\boldsymbol{D}$ of system sizes and a target set $S$ is a task of $N$ parties achieving the following:
  \begin{enumerate}
    \item The party prepare a common resource state $\Ket\phi\in\overline{\mathcal{H}}$ for $S$ by quantum communication in addition to LOCC within the configuration $\boldsymbol{D}$;
    \item A particular target state $\Ket{\psi}\in S$ is chosen from the target set $S$, and all the parameters of $\Ket{\psi}$ for its classical description are given to all the parties. Then, the parties perform LOCC within the configuration $\boldsymbol{D}$ to transform the common resource state $\Ket\phi$ into the chosen target state $\Ket\psi$ in the target set.
  \end{enumerate}
\end{definition}
In this dynamic setting, the common resource state $\Ket\phi$ is deterministically prepared by finite rounds of quantum communication, and $\Ket\phi$ may be a state exhibiting multipartite entanglement.
Note that common resource states in the dynamic setting are expected to have an intermediate capability between common resource states consisting only of bipartite entanglement and common resource states exhibiting multipartite entanglement in the static setting, since common resource states in the dynamic setting may exhibit multipartite entanglement but are prepared by only temporal uses of bipartite quantum communication resources.

\section{\label{sec:analysis}System-size-limited quantum state preparation in the static setting}

This section analyzes system-size-limited quantum state preparation in the static setting.
It is shown in this section that there exist examples of system-size-limited quantum state preparation in the static setting which is achievable by a common resource state exhibiting multipartite entanglement but not by any common resource state consisting only of bipartite entanglement.

To show such a nontrivial example, consider eight parties $v_1\,,\ldots,v_8$.
The configuration of the parties' system sizes
\begin{equation}
  \boldsymbol{D}_0=\left(D_0^{\left(v_1\right)},\ldots, D_0^{\left(v_8\right)}\right)
\end{equation}
are
\begin{equation}
    \label{eq:d}
    \begin{split}
        D_0^{\left(v_k\right)}&=\dim\overline{\mathcal{H}}^{v_k} = 4,\; \dim\mathcal{H}^{v_k} = 2, \;\forall v_k\in\{v_1\,,\ldots,v_7\};\\
        D_0^{\left(v_8\right)}&=\dim\overline{\mathcal{H}}^{v_8} =\dim\mathcal{H}^{v_8} = 2.
    \end{split}
\end{equation}
For each $v_k\in\left\{v_1\,,\ldots,v_7\right\}$,
consider the four-dimensional system $\overline{\mathcal{H}}^{v_k}$ to consist of two qubits, where one for the target set is denoted by $\mathcal{H}^{v_k}$, and the other auxiliary qubit for common resource states is denoted by $\mathcal{H}_\textup{a}^{v_k}$.
As for $v_8\,$, $\overline{\mathcal{H}}^{v_8}$ is identical to $\mathcal{H}^{v_8}$.
In the following, the systems may be written as
\begin{align}
  \overline{\mathcal{H}}^{v_k}&=\mathcal{H}^{v_k}\otimes\mathcal{H}_\textup{a}^{v_k},\quad \forall v_k\in\left\{v_1\,,\ldots,v_7\right\},\\
  \overline{\mathcal{H}}^{v_8}&=\mathcal{H}^{v_8}.
\end{align}

\begin{figure}[t!]
    \centering
    \includegraphics[width=4in]{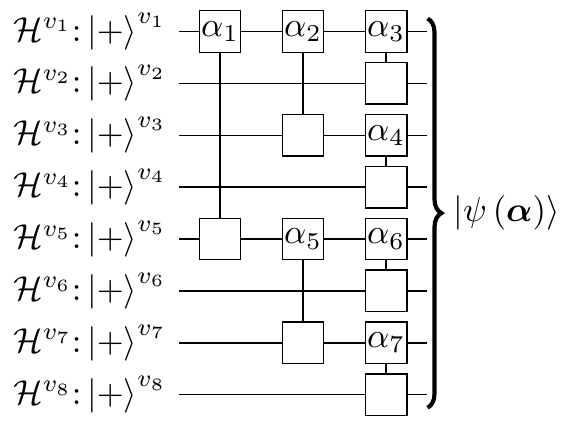}
    \caption[A quantum circuit generating all the states in a target set for a system-size-limited quantum state preparation.]{A quantum circuit generating all the states in the target set $S_0\coloneq{\left\{\Ket{\psi\left(\boldsymbol{\alpha}\right)}\right\}}_{\boldsymbol{\alpha}}$ for the system-size-limited quantum state preparation in Theorems~\ref{thm:multipartite} and~\ref{thm:bipartite}, where $\boldsymbol{\alpha}=\left(\alpha_1\,,\ldots,\alpha_7\right)$ is a tuple of the parameters representing states in $S_0$. The wires of the circuit starting from the input $\Ket{+}^{v_1},\ldots,\Ket{+}^{v_8}$ represent qubits $\mathcal{H}^{v_1},\ldots,\mathcal{H}^{v_8}$ held by the parties $v_1\,,\ldots,v_8\,$, respectively. The circuit consists of seven two-qubit gates $\exp\left(\textup{i}\alpha_k Z\otimes Z\right)$ parameterized by $\alpha_k\in\left\{\alpha_1\,,\ldots,\alpha_7\right\}$.}
\label{fig:target_set}
\end{figure}

\begin{figure}[t!]
    \centering
    \includegraphics[width=4in]{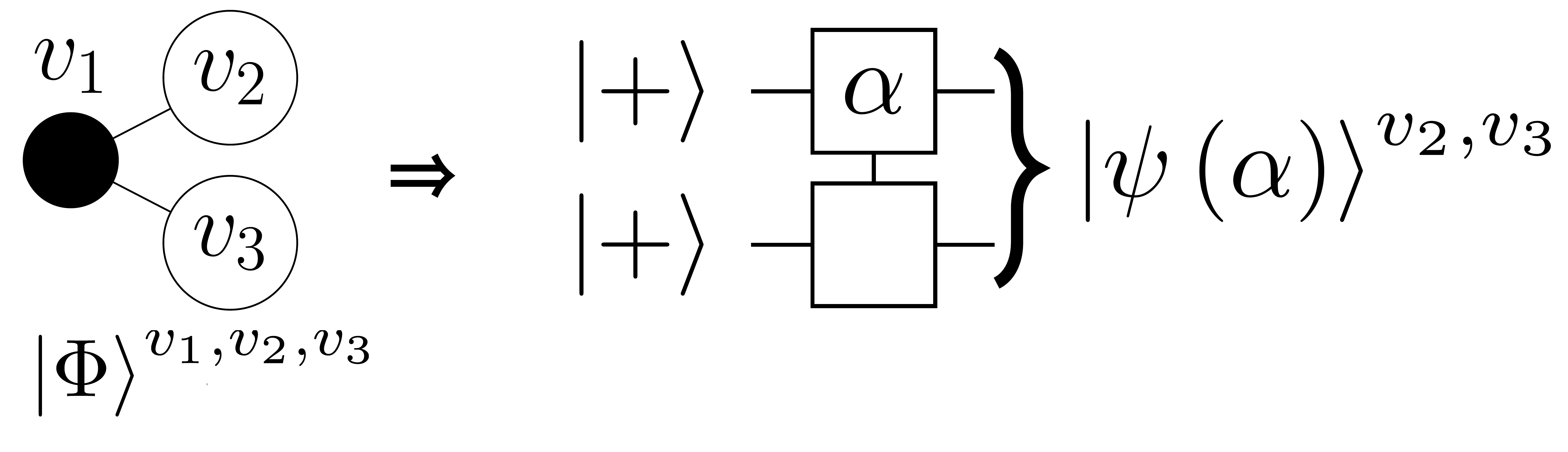}
    \caption[A simple example of a graph representing a graph state and a quantum circuit representing a class of states that can be deterministically prepared using this graph state.]{A simple example of a graph representing a graph state and a quantum circuit representing a class of states parameterized by $\alpha$ that can be deterministically prepared using this graph state. Given a graph state $\Ket{\Phi}^{v_1\,,v_2\,,v_3}$ as illustrated on the left, by performing the unitary $\exp\left(\textup{i}\alpha X^{v_1}\right)$ parameterized by $\alpha$ and a measurement in the $Z$ basis $\left\{\Ket{0},\Ket{1}\right\}$ on the qubit represented by the black vertex $v_1\,$, followed by local unitary corrections on the white vertices $v_2$ and $v_3$ conditioned by the measurement outcome, a two-qubit state $\Ket{\psi\left(\alpha\right)}$ defined as Equation~\eqref{eq:psi_alpha} represented by $v_2$ and $v_3$ can be deterministically. The state $\Ket{\psi\left(\alpha\right)}$ can also be represented as the output of the quantum circuit on the right, where a two-qubit gate $\exp\left(\textup{i}\alpha Z^{v_2}\otimes Z^{v_3}\right)$ parameterized by $\alpha$ is applied to $\Ket{+}^{v_2}\otimes\Ket{+}^{v_3}$.}
\label{fig:correspondence}
\end{figure}

\begin{figure}[t!]
    \centering
    \includegraphics[width=4in]{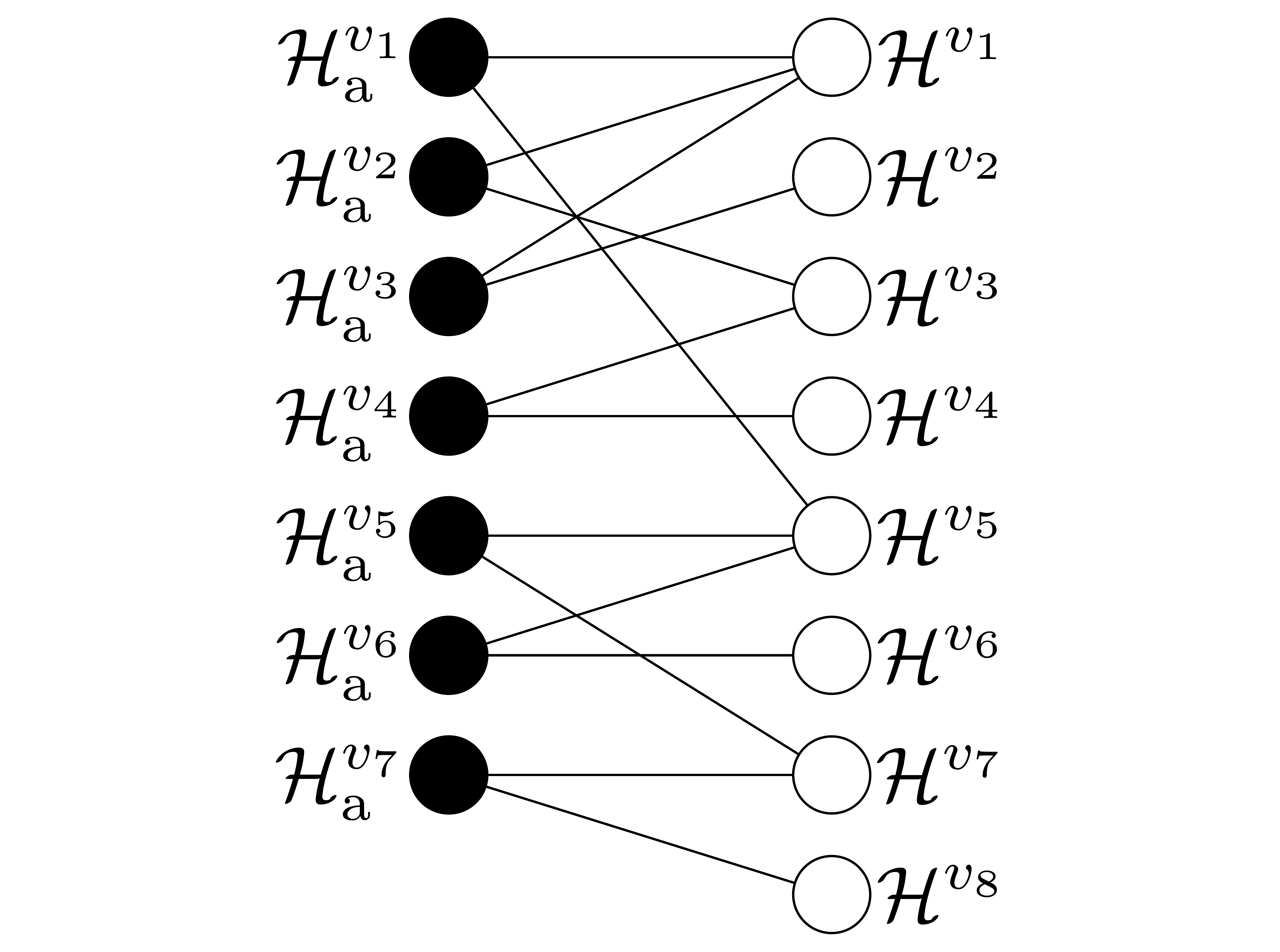}
    \caption[A graph representing a fifteen-qubit graph state used as a common resource state exhibiting multipartite entanglement.]{A graph representing a fifteen-qubit graph state $\Ket{\Phi_\textup{res}}$ used as a common resource state exhibiting multipartite entanglement in Theorem~\ref{thm:multipartite}.
    Each of the parties $v_k\in\left\{v_1\,,\ldots,v_7\right\}$ has two qubits $\mathcal{H}^{v_k}\otimes\mathcal{H}_\textup{a}^{v_k}$ while party $v_8$ has one qubit $\mathcal{H}^{v_8}$.
    Eight of the fifteen qubits $\mathcal{H}^{v_1},\ldots,\mathcal{H}^{v_8}$ represented by white vertices are qubits of which any state $\Ket{\psi\left(\boldsymbol{\alpha}\right)}$ in the target set $S_0$ is prepared.
    The other seven $\mathcal{H}_\textup{a}^{v_1},\ldots,\mathcal{H}_\textup{a}^{v_7}$ represented by black vertices are auxiliary qubits for a common resource state.
    To obtain $\Ket{\psi\left(\boldsymbol{\alpha}\right)}\in S_0$ parameterized by a tuple of parameters $\boldsymbol\alpha=\left(\alpha_1\,,\ldots\alpha_7\right)$,
    each party $v_k\in\left\{v_1\,,\ldots,v_7\right\}$ performs the following protocol in order.
    First, a unitary $\exp\left(\textup{i}\alpha_k X\right)$ parameterized by the parameter $\alpha_k$ is performed on the qubit $\mathcal{H}_\textup{a}^{v_k}$.
     Then, the qubit $\mathcal{H}_\textup{a}^{v_k}$ is measured in the $Z$ basis $\left\{\Ket{0},\Ket{1}\right\}$, and depending on the measurement outcome, local unitary corrections are applied to the qubits other than $\mathcal{H}_\textup{a}^{v_k}$.
     This protocol can deterministically transform $\Ket{\Phi_\textup{res}}$ into $\Ket{\psi\left(\boldsymbol{\alpha}\right)}\in S_0$ for any $\boldsymbol\alpha$.}
\label{fig:tree}
\end{figure}

Define a target set $S_0$ on
\begin{equation}
  \mathcal{H}=\bigotimes_{k=1}^{N}\mathcal{H}^{v_k}
\end{equation}
as the set of all the possible output states of a quantum circuit illustrated in Figure~\ref{fig:target_set}.
The circuit illustrated in Figure~\ref{fig:target_set} consists of seven two-qubit unitary gates
\begin{equation}
  \exp\left(\textup{i}\alpha_k Z\otimes Z\right)
\end{equation}
parameterized by $\alpha_i\in\left\{\alpha_1\,,\ldots,\alpha_7\right\}$, where $0\leqq\alpha_i < 2\pi$ for each $\alpha_k$.
Let
\begin{equation}
    \boldsymbol\alpha\coloneq\left(\alpha_1\,,\ldots,\alpha_7\right).
\end{equation}
denote the tuple of the seven parameters.
The input to the circuit is an eight-qubit product state
\begin{equation}
  \Ket{+}^{\otimes 8}\in\mathcal{H},
\end{equation}
where
\begin{equation}
  \Ket{+}\coloneq\frac{1}{\sqrt{2}}\left(\Ket{0}+\Ket{1}\right).
\end{equation}
The target set $S_0$ consists of the eight-qubit output states of the circuit parameterized by $\boldsymbol\alpha$ for representing the gates in the circuit, that is,
\begin{equation}
    \label{eq:s_0}
    S_0\coloneq\left\{\Ket{\psi\left(\boldsymbol\alpha\right)}\in\mathcal{H}:\boldsymbol\alpha=\left(\alpha_1\,,\ldots,\alpha_7\right)\right\},
\end{equation}
where each qubit is placed at one of the parties, as illustrated in Figure~\ref{fig:target_set}.

For example, consider the parameters
\begin{equation}
    \boldsymbol\alpha_0\coloneq\left(0,0,0,0,0,0,0\right),
\end{equation}
and the state
\begin{equation}
  \Ket{\psi\left(\boldsymbol\alpha_0\right)}=\Ket{+}^{\otimes 8}\in S_0
\end{equation}
is a product state, since each gate in the circuit reduces to the identity map.
In contrast, consider the parameters
\begin{equation}
    \boldsymbol\alpha_{\frac{\pi}{4}}\coloneq\left(\frac{\pi}{4},\frac{\pi}{4},\frac{\pi}{4},\frac{\pi}{4},\frac{\pi}{4},\frac{\pi}{4},\frac{\pi}{4}\right),
\end{equation}
and the state
\begin{equation}
  \Ket{\psi\left(\boldsymbol\alpha_\frac{\pi}{4}\right)}\in S_0
\end{equation}
is a fully entangled state, since each gate $\exp\left(\textup{i}\frac{\pi}{4} Z\otimes Z\right)$ entangles $\Ket{+}\otimes\Ket{+}$.

The configuration $\boldsymbol{D}_0$ and the target set $S_0$ defined above
yield the following two theorems on the system-size-limited quantum state preparation.
Theorem~\ref{thm:multipartite} shows achievability of the system-size-limited quantum state preparation using a common resource state exhibiting multipartite entanglement.
In contrast, Theorem~\ref{thm:bipartite} is a no-go theorem on the same system-size-limited quantum state preparation for any common resource state consisting only of bipartite entanglement.
These theorems suggest difference in achievability of system-size-limited quantum state preparation between multipartite and bipartite entanglement.

\begin{theorem}
\label{thm:multipartite}
    \textit{Multipartite entanglement for a system-size-limited quantum state preparation in the static setting.}
    The system-size-limited quantum state preparation in the static setting for the configuration $\boldsymbol{D}_0$ defined as Equation~\eqref{eq:d} and the target set $S_0$ defined as Equation~\eqref{eq:s_0}
    is achievable using a common resource state exhibiting multipartite entanglement.
\end{theorem}

\begin{theorem}
\label{thm:bipartite}
    \textit{Bipartite entanglement for a system-size-limited quantum state preparation in the static setting.}
    The system-size-limited quantum state preparation in the static setting for the configuration $\boldsymbol{D}_0$ defined as Equation~\eqref{eq:d} and the target set $S_0$ defined as Equation~\eqref{eq:s_0}
    is \textit{not} achievable using any common resource state consisting of bipartite entanglement.
\end{theorem}

Note that while shallower quantum circuits having a similar structure to the circuit in Figure~\ref{fig:target_set} are not sufficient for proving the difference between multipartite and bipartite entanglement, the example in Theorems~\ref{thm:multipartite} and~\ref{thm:bipartite} is not necessarily the simplest, and other target sets of states having similar properties also exist.
In particular, the following theorem shows another example, where the target set $S_1$ is of $2m$-qudit states, the size of each qudit is $D\geqq 2$, and each state in $S_1$ has the maximal Schmidt rank with respect to any bipartition between $m$ qudits and the other $m$ qudits.
Random weighted graph states or random pure states fulfill this condition, for which the reduced states have almost maximum entropy for any bipartition~\cite{C20,P7,H15}.
For any resource state consisting of bipartite entanglement to obtain states in $S_1$ by LOCC, or even by stochastic LOCC, there has to be at least one party for which the local quantum system size for storing this resource state needs to be almost quadratically larger than $D$, that is, greater than or equal to $D^{2-\frac{1}{m}}$.
Also note that for some special configurations of local system sizes, these differences between multipartite and bipartite entanglement do not arise, especially in cases of~\cite{R8}
\begin{equation}
  \dim\overline{\mathcal{H}}^{v_1}\geqq\prod_{k=2}^{N}\dim\overline{\mathcal{H}}^{v_k}.
\end{equation}

\begin{theorem}
\label{prp:max_rank}
    \textit{Requirement for resource states consisting of bipartite entanglement for preparing a multipartite entangled state having maximal Schmidt ranks.}
    Consider a $2m$-qudit state $\Ket{\psi}\in\mathcal{H}\coloneq{\left(\mathbb{C}^d\right)}^{\otimes 2m}$ of local system size $D$ which has the maximal Schmidt rank with respect to bipartite cuts between any $m$ qudits and the other $m$ qudits; that is, for any such bipartite cut, the Schmidt rank is $D^m$.
    If $2m$ parties $v_1\,,\ldots,v_{2m}$ prepare $\Ket{\psi}$ by LOCC from any resource state only consisting of bipartite entanglement,
    then there has to exist at least one party $v\in V\coloneq\left\{v_1\,,\ldots,v_{2m}\right\}$ for which the local system size $\dim\overline{\mathcal{H}}^{v}$ for storing this resource state is almost quadratically larger, that is,
    \begin{equation}
        \label{eq:full_schmidt_rank}
        \max_{v_k\in V} \left\{\dim\overline{\mathcal{H}}^{v_k}\right\} \geqq D^{2- \frac{1}{m}}.
    \end{equation}
\end{theorem}

Note that the lower bound of local system sizes in Inequality~\eqref{eq:full_schmidt_rank} is almost sufficient for fulfilling the necessary condition~\eqref{eq:schmidt_rank_condition} on the Schmidt ranks in the proof of Theorem~\ref{prp:max_rank} by storing a symmetric distribution of maximally entangled states shared between all pairs of the parties.
In this case, since each party shares maximally entangled states with the other $2m-1$ parties, the maximally entangled state corresponding to each $e\in E$ satisfies
\begin{equation}
  M_e=\left\lceil D^\frac{1}{m}\right\rceil,
\end{equation}
and the local system size for each $v\in V$ is
\begin{equation}
  {\left\lceil D^\frac{1}{m}\right\rceil}^{2m-1},
\end{equation}
where $\lceil{}\cdots{}\rceil$ is the ceiling function.

The proofs of Theorems~\ref{thm:multipartite},~\ref{thm:bipartite}, and~\ref{prp:max_rank} are as follows.

\begin{proof}[Proof of Theorem~\ref{thm:multipartite}]
  The proof is by construction of a common resource state exhibiting multipartite entanglement for the target set $S_0$.
  As a common resource state exhibiting multipartite entanglement, a class of graph states proposed in Reference~\cite{S18} can be used.
  A graph state~\cite{H12,H13} is a multi-qubit entangled state characterized by a graph $G=(V,E)$.
  Note that while graphs in this thesis also represent distribution of bipartite entanglement, a graph state is a different concept, which is a state exhibiting multipartite entanglement obtained for a graph $G=(V,E)$ as follows:
  first, for each vertex $v_k\in V$, a qubit labeled as $v_k$ is initialized as
  \begin{equation}
    \Ket{+}^{v_k}\coloneq\frac{1}{\sqrt{2}}\left(\Ket{0}^{v_k}+\Ket{1}^{v_k}\right),
  \end{equation}
  and then, for each edge $e=\left\{v_k\,,v_{k^\prime}\right\}\in E$, the controlled-$Z$ gate
  \begin{equation}
    \label{eq:cz}
      CZ^{v_k\,,v_{k^\prime}}
      \coloneq{\left(\Ket{00}\Bra{00}+\Ket{01}\Bra{01}+\Ket{10}\Bra{10}-\Ket{11}\Bra{11}\right)}^{v_k\,,v_{k^\prime}}
  \end{equation}
  is applied to two qubits labeled as $v_k$ and $v_{k^\prime}$.
  Reference~\cite{S18} proposes an LOCC protocol for preparing any pure state of an arbitrary number of qubits by performing sequential projective measurements and local unitary corrections on a particular type of graph states.
  To see how this protocol works, consider the three-vertex graph shown in Figure~\ref{fig:correspondence} as a simple example.
  The graph state $\Ket{\Phi}^{v_1\,,v_2\,,v_3}$ represented by this graph is invariant under a local unitary transformation $X^{v_1}\otimes Z^{v_2}\otimes Z^{v_3}$, that is,
  \begin{equation}
    X^{v_1}\otimes Z^{v_2}\otimes Z^{v_3}\Ket{\Phi}^{v_1\,,v_2\,,v_3}=\Ket{\Phi}^{v_1\,,v_2\,,v_3},
  \end{equation}
  where $X$ and $Z$ are the Pauli operators on a qubit.
  Thus, if the unitary operator $\exp\left(\textup{i}\alpha X^{v_1}\right)$ parameterized by $\alpha$ is performed on qubit $v_1\,$, the action is equivalent to
  \begin{equation}
    \begin{split}
      \exp\left(\textup{i}\alpha X^{v_1}\right)\otimes\mathbb{1}^{v_2}\otimes\mathbb{1}^{v_3}\Ket{\Phi}^{v_1\,,v_2\,,v_3}
      &=\mathbb{1}^{v_1}\otimes\exp\left(\textup{i}\alpha Z^{v_2}\otimes Z^{v_3}\right) \Ket{\Phi}^{v_1\,,v_2\,,v_3},
    \end{split}
  \end{equation}
  which can be shown using the Taylor series of the exponential function.
  Then, it is straightforward to verify that if $\exp\left(\textup{i}\alpha X^{v_1}\right)$ and a measurement in $Z$ basis $\left\{\Ket{0},\Ket{1}\right\}$ are performed on the qubit $v_1\,$, the post-measurement state of two qubits $v_2$ and $v_3$ can be deterministically transformed by local unitary corrections $\mathbb{1}^{v_2}\otimes\mathbb{1}^{v_3}$ or $Z^{v_2}\otimes Z^{v_3}$ conditioned by the measurement outcome $\Ket{0}$ or $\Ket{1}$, respectively, into
  \begin{equation}
    \label{eq:psi_alpha}
    \Ket{\psi\left(\alpha\right)}^{v_2\,,v_3}\coloneq\exp\left(\textup{i}\alpha Z^{v_2}\otimes Z^{v_3}\right)\left(\Ket{+}^{v_2}\otimes\Ket{+}^{v_3}\right).
  \end{equation}
  In the same way, it is shown in Reference~\cite{S18} that any quantum circuit consisting of one-qubit Clifford gates and multi-qubit gates
  \begin{equation}
    \exp\left(\textup{i}\alpha Z\otimes Z\otimes\cdots\otimes Z\right)
  \end{equation}
  parameterized by $\alpha$ can be implemented by performing sequential projective measurements and local unitary corrections on a particular graph state corresponding to the quantum circuit.
  In addition, it is shown that any pure state of an arbitrary number of qubits is locally unitarily equivalent to a pure state generated by a quantum circuit consisting of these types of gates.

  As for the common resource state for the target set $S_0\,$, the fifteen-qubit graph state $\Ket{\Phi_\textup{res}}$ illustrated in Figure~\ref{fig:tree} held by the parties $v_1\,,\ldots,v_8$ can be used.
  In the same way as explained above, there exists a protocol for transforming the graph state $\Ket{\Phi_\textup{res}}$ in Figure~\ref{fig:tree} into any state $\Ket{\psi\left(\boldsymbol{\alpha}\right)}\in S_0$.
  In this protocol, each of parties $v_k\in\left\{v_1\,,\ldots,v_7\right\}$ performs a unitary $\exp\left(\textup{i}\alpha_k X^{v_k}\right)$ parameterized by $\alpha_k$ and a measurement in the $Z$ basis $\left\{\Ket{0},\Ket{1}\right\}$ on the auxiliary qubit represented by $\mathcal{H}_\textup{a}^{v_k}$, followed by local unitary corrections on qubits other than $\mathcal{H}_\textup{a}^{v_k}$ conditioned by the measurement outcome.
  After the parties performing this protocol, the parties obtain $\Ket{\psi\left(\boldsymbol\alpha\right)}\in S_0$ deterministically for any parameters $\boldsymbol\alpha=\left(\alpha_1\,,\ldots,\alpha_7\right)$.
\end{proof}

\begin{proof}[Proof of Theorem~\ref{thm:bipartite}]
  In this proof, a necessary condition of a resource state consisting of bipartite entanglement for preparing a state $\Ket{\psi\left(\boldsymbol\alpha_{\frac{\pi}{4}}\right)}\in S_0$ by LOCC within the configuration $\boldsymbol D_0$ is first derived, and then, it is shown that any resource state consisting of bipartite entanglement for preparing $\Ket{\psi\left(\boldsymbol\alpha_{\frac{\pi}{4}}\right)}$ cannot satisfy this necessary condition.

    A necessary condition for preparing the state $\Ket{\psi\left(\boldsymbol\alpha_{\frac{\pi}{4}}\right)}\in S_0$ from a resource state consisting of bipartite entanglement by LOCC within the configuration $\boldsymbol D_0$ is derived as follows.
    Observe that the state $\Ket{\psi\left(\boldsymbol\alpha_{\frac{\pi}{4}}\right)}$ is fully entangled, that is, entangled with respect to any bipartition of the eight qubits.
    To prepare a fully entangled state, the resource state at party $v_8$ has to be entangled with some other parties.
    Since
    \begin{equation}
      \dim \overline{\mathcal{H}}^{v_8}=2,
    \end{equation}
    the party $v_8$ can store only one qubit of a bipartite resource state entangled with another party, which is labeled
    \begin{equation}
      u_7\in\{v_1\,,\ldots,v_7\}.
    \end{equation}
    The quantum system $\overline{\mathcal{H}}^{u_7}$ at $u_7$ is decomposed into
    \begin{equation}
      \overline{\mathcal{H}}^{u_7}=\mathcal{H}^{u_7}_{\{u_7\,,v_8\}}\otimes\mathcal{H}^{u_7}_\textup{r},
    \end{equation}
    where $\mathcal{H}^{u_7}_{\{u_7\,,v_8\}}$ is a system of more than one dimension for the bipartite entangled resource state shared with $v_8\,$, and $\mathcal{H}^{u_7}_\textup{r}$ the remaining quantum system.
    It is necessary that
    \begin{equation}
        \label{eq:dim}
        \begin{split}
            &\dim\mathcal{H}^{u_7}_{\{u_7\,,v_8\}}=2,\\
            &\dim\mathcal{H}^{u_7}_\textup{r}=2,
        \end{split}
    \end{equation}
    which can be shown by contradiction as follows.
    Assume that
    \begin{equation}
      \dim\mathcal{H}^{u_7}_{\{u_7\,,v_8\}}>2.
    \end{equation}
    Then, it is necessary that
    \begin{equation}
      \dim\mathcal{H}^{u_7}_\textup{r}<2,
    \end{equation}
    and the resource state shared between the parties $u_7$ and $v_8$ cannot be entangled with any of the other parties.
    This contradicts the assumption that a fully entangled state can be prepared, and Equation~\eqref{eq:dim} is shown.
    Since it holds that
    \begin{equation}
      \dim\mathcal{H}^{u_7}_\textup{r}=2,
    \end{equation}
    the party $u_7$ can store another single qubit of a bipartite resource state entangled with a party other than $v_8\,$, which is labeled
    \begin{equation}
      u_6\in\{v_1\,,\ldots,v_7\}\setminus\{u_7\}.
    \end{equation}
    By iterating this argument, any resource state consisting of bipartite entanglement for preparing a fully entangled state by LOCC within the configuration $\boldsymbol{D}_0$ is required to be seven two-qubit entangled states shared between $u_1$--$u_2\,$, $\ldots$, $u_6$--$u_7\,$, and $u_7$--$v_8\,$, respectively, where
    \begin{equation}
        \label{eq:perm}
        (u_1\,,\ldots,u_7)~\text{is a permutation of}~(v_1\,,\ldots,v_7).
    \end{equation}
    Note that although $u_1$ uses only one qubit in this case, the remaining system of $u_1\,$, which is two dimensional, cannot be used for sharing an entangled state with the other parties, since there is no dimension left in the quantum systems of the other parties.
    Therefore, the distribution of the two-qubit entangled states is represented by a line-topology graph, as illustrated in Figure~\ref{fig:permutation}.
    Note that this line-topology graph is a tree.
    Since the target set $S_0$ includes a fully entangled state $\Ket{\psi\left(\boldsymbol\alpha_\frac{\pi}{4}\right)}$,
    it is necessary that any common resource state consisting of bipartite entanglement for $S_0$ within the configuration $\boldsymbol D_0$ is a state consisting of seven two-qubit entangled states represented by the line-topology tree as shown in Figure~\ref{fig:permutation}.

\begin{figure}[t!]
    \centering
    \includegraphics[width=4in]{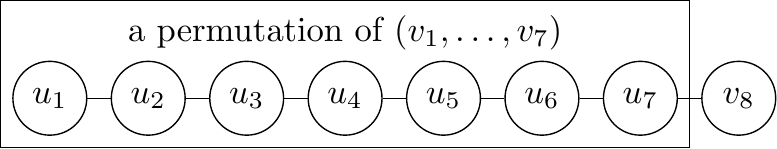}
    \caption[A line-topology tree representing a resource state consisting of bipartite entanglement to prepare a fully entangled state within a configuration of local system sizes.]{A line-topology tree representing a resource state consisting of bipartite entanglement to prepare a fully entangled state within the configuration $\boldsymbol{D}_0$ of local system sizes. Since the target set $S_0$ includes a fully entangled state $\Ket{\psi\left(\boldsymbol\alpha_\frac{\pi}{4}\right)}$, the common resource states consisting of bipartite entanglement for $S_0$ have to be represented by the line-topology tree in the figure, which leads to a contradiction with the condition given in Lemma~\ref{lem:graph_associated} as shown in the main text.}
\label{fig:permutation}
\end{figure}

    It is shown that the state $\Ket{\psi\left(\boldsymbol\alpha_\frac{\pi}{4}\right)}$ cannot be prepared from any such resource state as follows.
    Since any two-qubit entangled state can be obtained by LOCC from a Bell state
    \begin{equation}
      \frac{1}{\sqrt{2}}\left(\Ket{0}\otimes\Ket{0}+\Ket{1}\otimes\Ket{1}\right),
    \end{equation}
    it suffices to consider resource states consisting of seven Bell states represented by the line-topology tree.
    Thus, the condition on the Schmidt ranks given in Lemma~\ref{lem:graph_associated} implies that the state $\Ket{\psi\left(\boldsymbol\alpha_\frac{\pi}{4}\right)}$ can be prepared from resource states consisting of seven Bell states represented by a line-topology tree if and only if for any edge $e$ of the line-topology tree
    \begin{equation}
        R_e\left(\Ket{\psi\left(\boldsymbol\alpha_\frac{\pi}{4}\right)}\right)\leqq 2,
    \end{equation}
    where the notations are the same as those in Lemma~\ref{lem:graph_associated}.
    In other words, the Schmidt rank of $\Ket{\psi\left(\boldsymbol\alpha_\frac{\pi}{4}\right)}$ with respect to each edge of the line-topology tree needs to be smaller than or equal to two.
    However, the explicit calculation of $R_e\left(\Ket{\psi\left(\boldsymbol\alpha_\frac{\pi}{4}\right)}\right)$ for all the edges $e$ of all the $7!=5040$ different trees obtained from the permutations of $v_1\,,\ldots,v_7$ in Equation~\eqref{eq:perm} shows that, for any of the permutations, there exists an edge $e$ such that
    \begin{equation}
        \label{eq:r_e}
        R_e\left(\Ket{\psi\left(\boldsymbol\alpha_\frac{\pi}{4}\right)}\right)>2.
    \end{equation}
    The Schmidt rank $R_e\left(\Ket{\psi\left(\boldsymbol\alpha_\frac{\pi}{4}\right)}\right)$ in Inequality~\eqref{eq:r_e} can be exactly calculated with the help of a computer program.
    Although computers cannot calculate irrational numbers exactly, the Schmidt rank $R_e\left(\Ket{\psi\left(\boldsymbol\alpha_\frac{\pi}{4}\right)}\right)$ of a vector $\Ket{\psi\left(\boldsymbol\alpha_\frac{\pi}{4}\right)}$ with irrational elements can be reduced to that of a vector only with integer elements.
    To remove irrational coefficients for normalization of the state $\Ket{+}$ and the gates $\exp\left(\textup{i}\frac{\pi}{4}Z\otimes Z\right)$,
    substitute $\Ket{+}$ and $\exp\left(\textup{i}\frac{\pi}{4}Z\otimes Z\right)$ in the circuit in Figure~\ref{fig:target_set} with $\sqrt{2}\Ket{+}$ and $\sqrt{2}\exp\left(\textup{i}\frac{\pi}{4}Z\otimes Z\right)$, respectively.
    The resulting vector
    \begin{equation}
      \Ket{\tilde\psi\left(\boldsymbol\alpha_\frac{\pi}{4}\right)}\coloneq 2^{\frac{15}{2}}\Ket{\psi\left(\boldsymbol\alpha_\frac{\pi}{4}\right)}
    \end{equation}
    has the same Schmidt ranks as $\Ket{\psi\left(\boldsymbol\alpha_\frac{\pi}{4}\right)}$ for any bipartition, and all the elements of $\Ket{\tilde\psi\left(\boldsymbol\alpha_\frac{\pi}{4}\right)}$ are complex numbers whose real and imaginary parts are both integers by construction.
    Therefore, Schmidt ranks of $\Ket{\psi\left(\boldsymbol\alpha_\frac{\pi}{4}\right)}$ can be exactly obtained by calculating those of $\Ket{\tilde\psi\left(\boldsymbol\alpha_\frac{\pi}{4}\right)}$ by computer.

    The calculation of $R_e\left(\Ket{\psi\left(\boldsymbol\alpha_\frac{\pi}{4}\right)}\right)$ implies that the state $\Ket{\psi\left(\boldsymbol\alpha_\frac{\pi}{4}\right)}$ cannot be prepared from any resource state consisting of the seven Bell states.
    Due to this calculation, it is concluded that there exists no common resource state consisting of bipartite entanglement for the target set $S_0$ within the configuration $\boldsymbol D_0$.
\end{proof}

\begin{proof}[Proof of Theorem~\ref{prp:max_rank}]
    Since any bipartite state can be obtained from a maximally entangled state, it suffices to evaluate $\dim\overline{\mathcal{H}}^{v_k}$ for storing a resource state consisting of bipartite maximally entangled states distributed according to the complete graph $K=\left(V,E\right)$, that is, the fully connected graph for the $2m$ parties.
    Let $M_e\in\left\{1,2,\ldots\right\}$ denote the Schmidt rank of the maximally entangled state for each edge $e\in E$.

    First, a lower bound of the total system size for storing $\bigotimes_{e\in E}\Ket{\Phi_{M_e}^+}^e$, that is, $\prod_{e\in E}{\left(M_e\right)}^2$, is derived.
    Consider an edge cut $C$~\cite{B22} of $K$ between any $m$ vertices and the other $m$ vertices.
    Since the Schmidt rank is monotonically nonincreasing under LOCC,
    it is necessary that, for any $C$,
    \begin{equation}
        \label{eq:schmidt_rank_condition}
        \prod_{e\in C}M_e\geqq D^m.
    \end{equation}
    Considering Inequality~\eqref{eq:schmidt_rank_condition} for all the
    \begin{equation}
      {{2m\choose m}}/2
    \end{equation}
    possible choices of $C$ between any $m$ vertices and the other $m$ vertices and taking the products of the right- and left-hand sides of these inequalities yield
    \begin{equation}
      \prod_{C} \prod_{e\in C} M_e\geqq D^{m{\frac{{2m\choose m}}{2}}}.
    \end{equation}
    Since $M_e$ for each $e \in  E$ appears ${2m-2\choose m-1}$ times in the product on the left-hand side, the last inequality can be written as
    \begin{equation}
        \prod_{C} \prod_{e\in C} M_e =\prod_{e\in E}{\left(M_e\right)}^{{2m-2}\choose{m-1}}\geqq D^{m\frac{{2m\choose m}}{2}}.
    \end{equation}
    Therefore, a lower bound of the total system size is
    \begin{equation}
        \prod_{e\in E}{\left(M_e\right)}^2\geqq D^{2\left(2m-1\right)}.
    \end{equation}

    Since the total system size for storing $\bigotimes_{e\in E}\Ket{\Phi_{M_e}^+}^e$ is written as
    \begin{equation}
        \dim\overline{\mathcal{H}}=\prod_{v_k\in V}\dim\overline{\mathcal{H}}^{v_k},
    \end{equation}
    it holds that
    \begin{equation}
        \prod_{v_k\in V}\dim\overline{\mathcal{H}}^{v_k}\geqq\prod_{e\in E}{\left(M_e\right)}^2\geqq D^{2\left(2m-1\right)}.
    \end{equation}
    Therefore, it is obtained that
    \begin{equation}
        \max_{v_k\in V} \left\{\dim\overline{\mathcal{H}}^{v_k}\right\}\geqq{\left(\prod_{v_k\in V}\dim\overline{\mathcal{H}}^{v_k}\right)}^{\frac{1}{2m}}\geqq D^{2-\frac{1}{m}},
    \end{equation}
    which yields the conclusion.
\end{proof}

\section{\label{sec:analysis2}System-size-limited quantum state preparation in the dynamic setting}
This section analyzes the difference in system-size-limited quantum state preparation appearing in the dynamic setting.
Before analyzing multipartite cases, a simpler bipartite case is discussed to clarify the difference between the static setting and the dynamic setting.
Consider two parties $v_1$ and $v_2\,$, where each party has two qubits; that is, the configuration $\left(D^{\left(v_1\right)},D^{\left(v_2\right)}\right)$ is given by
\begin{align}
    D^{\left(v_1\right)}&=\dim\overline{\mathcal{H}}^{v_1} = 4,\\
    D^{\left(v_2\right)}&=\dim\overline{\mathcal{H}}^{v_1} = 4.
\end{align}
In this case, these two parties can store an entangled resource state with Schmidt rank four in the static setting.
However, in the dynamic setting, the parties can prepare an entangled resource state with Schmidt rank at most two, which is shown as follows.
Consider any shared state $\Ket{\phi}^{v_1\,,v_2}$ after the last round of quantum communication for preparing $\Ket{\phi}^{v_1\,,v_2}$, where it is assumed that the direction of the quantum communication in the last round is from $v_1$ to $v_2$ without loss of generality.
Since the quantum communication sends out at least one qubit from $v_1\,$, the rank of $v_1$'s reduced state for $\Ket{\phi}^{v_1\,,v_2}$ is at most two; that is, the Schmidt rank of $\Ket{\phi}^{v_1\,,v_2}$ is at most two.
Since the Schmidt rank is monotonically nonincreasing by LOCC, $v_1$ and $v_2$ after the last round of quantum communication cannot prepare an entangled resource state with Schmidt rank more than two, which yields the conclusion.

Although this two-party example is trivial, nontrivial cases of more than two parties are shown as follows.
Theorem~\ref{prp:bipartite_dynamic} shows that the common resource states available in the dynamic setting can still have more capability than any common resource state consisting of bipartite entanglement in the static setting, as well as the common resource states exhibiting multipartite entanglement in the static setting.
In contrast, Theorem~\ref{prp:multipartite_dynamic} shows the existence of common resource states which cannot be prepared in the dynamic setting by the parties within a limitation of local system sizes while the common resource states can still be stored within the limitation in the static setting.
This implies that the common resource states in the dynamic setting have in this case less capability than a common resource state exhibiting multipartite entanglement in the static setting.

\begin{figure}[t!]
    \centering
    \includegraphics[width=5.0in]{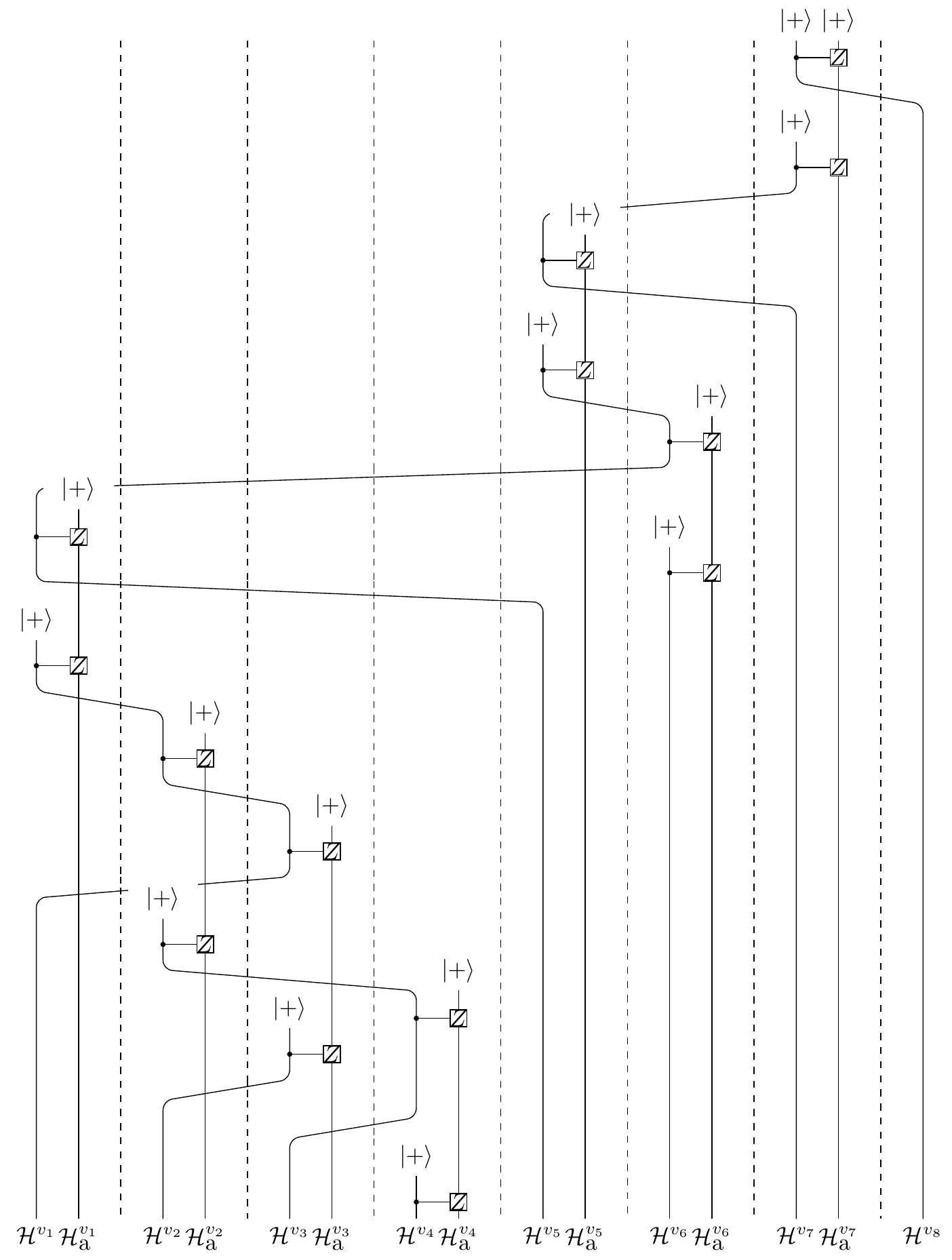}
    \caption[A quantum circuit representing a protocol for preparing a common resource state for a target set by quantum communication in addition to LOCC within a configuration of local system sizes.]{A quantum circuit representing a protocol for preparing the common resource state $\Ket{\Phi_\textup{res}}$ for the target set $S_0$ by quantum communication in addition to LOCC within the configuration $\boldsymbol{D}_0$ of local system sizes defined as Equation~\eqref{eq:d}. Each of the parties $v_k\in\left\{v_1\,,\ldots v_7\right\}$ can perform local operations on at most two qubits $\mathcal{H}^{v_k}\otimes\mathcal{H}_\textup{a}^{v_k}$ while the party $v_8$ can perform local operations on one qubit $\mathcal{H}^{v_8}$. The dashed lines represent the separation of the parties. Each wire of the circuit corresponds to a qubit corresponding to the Hilbert space on the right, and the circuit consists of the controlled-$Z$ gates $CZ$ defined as Equation~\eqref{eq:cz} and quantum communication represented by crossings of the wires.}
\label{fig:multipartite}
\end{figure}

\begin{theorem}
\label{prp:bipartite_dynamic}
    \textit{A common resource state in the dynamic setting having more capability than any common resource state consisting of bipartite entanglement.}
    The state $\Ket{\Phi_\textup{res}}$ in the proof of Theorem~\ref{thm:multipartite} and in Figure~\ref{fig:tree} can be used as a common resource state for achieving the system-size-limited quantum state preparation in the dynamic setting for the configuration $\boldsymbol{D}_0$ defined as Equation~\eqref{eq:d} and the target set $S_0$ defined as Equation~\eqref{eq:s_0}, while the system-size-limited quantum state preparation in the static setting for $\boldsymbol{D}_0$ and $S_0$ cannot be achieved by any common resource state consisting of bipartite entanglement due to Theorem~\ref{thm:bipartite}.
\end{theorem}

\begin{theorem}
\label{prp:multipartite_dynamic}
    \textit{Common resource states exhibiting multipartite entanglement which cannot be prepared in the dynamic setting.}
    Consider four parties $v_1\,$, $v_2\,$, $v_3\,$, and $v_4$.
    Given a configuration $\boldsymbol{D}_1=\left(D_1^{\left(v_1\right)},D_1^{\left(v_2\right)},D_1^{\left(v_3\right)},D_1^{\left(v_4\right)}\right)$, where
    \begin{align}
        D_1^{\left(v_1\right)}&=\dim\overline{\mathcal{H}}^{v_1} = 4,\\
        D_1^{\left(v_k\right)}&=\dim\overline{\mathcal{H}}^{v_k} = 2, \; \forall v_k\in\left\{v_2\,,v_3\,,v_4\right\},
    \end{align}
    any fully entangled common resource state $\Ket{\phi}\in\overline{\mathcal{H}}$ whose Schmidt rank with respect to the bipartition between $v_1$ and $v_2 v_3 v_4$ is more than two cannot be prepared in the dynamic setting, although there exists such a common resource state which can be stored in the static setting.
\end{theorem}

Note that under the limitation in Theorem~\ref{prp:multipartite_dynamic}, the parties can prepare any state whose Schmidt rank with respect to the bipartition between $v_1$ and $v_2 v_3 v_4$ is not more than two.
This is because $v_1$'s reduced state can be represented by one qubit in this case, and hence, the parties can perform quantum communication to bring arbitrary two qubits to $v_1$ to perform any two-qubit gates.
As for another remark, while it is assumed in the definition of the dynamic setting that quantum communication is performed sequentially, one can also consider simultaneous quantum communication between two parties, which is considered as a swap operation between the two.
However, this simultaneous quantum communication yields a trivial result since the parties under the limitation in Theorem~\ref{prp:multipartite_dynamic} can prepare any state $\Ket{\Phi}\in\overline{\mathcal{H}}$ using swap operations for letting $v_1$ perform arbitrary two-qubit gates.

\begin{proof}[Proof of Theorem~\ref{prp:bipartite_dynamic}]
    It is shown that the common resource state $\Ket{\Phi_\textup{res}}$ in the proof of Theorem~\ref{thm:multipartite} and in Figure~\ref{fig:tree} can be prepared by the parties using quantum communication in addition to LOCC within the configuration $\boldsymbol{D}_0$.
    The protocol for preparing $\Ket{\Phi_\textup{res}}$ is represented by a quantum circuit illustrated in Figure~\ref{fig:multipartite}.
    In this circuit, the parties repeatedly perform $CZ$ gates defined as Equation~\eqref{eq:cz} to entangle qubits initialized as $\Ket{+}$, distribute one qubit of the entangled state by quantum communication, and perform a $CZ$ gate again to entangle the remaining part of the entangled state with another qubit initialized as $\Ket{+}$.
    After this protocol, the state $\Ket{\Phi_\textup{res}}$ is shared among the parties $v_1\,,\ldots,v_8$.
\end{proof}

\begin{proof}[Proof of Theorem~\ref{prp:multipartite_dynamic}]
The proof is given in a similar way to the example given at the beginning of this section.
Consider any fully entangled state $\Ket{\phi}^{v_1\,,v_2\,,v_3\,,v_4}$ shared among $v_1\,$, $v_2\,$, $v_3\,$, and $v_4$ after the last round of quantum communication for preparing $\Ket{\phi}^{v_1\,,v_2\,,v_3\,,v_4}$.
The direction of the quantum communication in the last round is either of the following three possibilities:
\begin{enumerate}
    \item from $v_1$ to $v_k$ where $k\in\left\{2,3,4\right\}$;
    \item from $v_k$ to $v_{k^\prime}$ where $k,k^\prime\in\left\{2,3,4\right\}$ and $k\neq k^\prime$;
    \item from $v_k$ to $v_1$ where $k\in\left\{2,3,4\right\}$.
\end{enumerate}
Since $\Ket{\phi}^{v_1\,,v_2\,,v_3\,,v_4}$ is fully entangled, the latter two possibilities 2 and 3, which lead to a product state between $v_k$ and the others, are excluded.
Regarding possibility 1, after sending at least one qubit from $v_1$ to $v_k\,$, the rank of $v_1$'s reduced state for $\Ket{\phi}^{v_1\,,v_2\,,v_3\,,v_4}$ is at most two; that is, the Schmidt rank of $\Ket{\phi}^{v_1\,,v_2\,,v_3\,,v_4}$ with respect to the bipartition between $v_1$ and $v_2 v_3 v_4$ is at most two.
Since the Schmidt rank is monotonically nonincreasing by LOCC, the parties after the last round of quantum communication cannot prepare any common resource state whose Schmidt rank with respect to the bipartition between $v_1$ and $v_2 v_3 v_4$ is more than two, which yields the conclusion.
\end{proof}

\part{\label{part:conclusion}Conclusion and outlook}

\chapter{\label{sec:summary_1}Conclusion of Part~\ref{part:1}}

Part~\ref{part:1} analyzed entanglement cost, or equivalently, quantum communication cost under LOCC, required for one-shot quantum state merging, aimed at investigating properties of transferring quantum information between two parties on small and intermediate scales.
The following two results are obtained in this part.
Being complementary to existing protocols achieving nearly optimal one-shot state merging on a large scale,
these results open the way to another direction for future research on small and intermediate scales.

\section*{Quantum state merging for arbitrarily small-dimensional systems}

Chapter~\ref{sec:merge} constructed protocols for one-shot state merging under one-way LOCC, which work for any state of an arbitrarily small-dimensional system and satisfy arbitrarily high fidelity requirements.
The protocols retain the essential feature of state merging; that is, entanglement cost can be reduced by exploiting a structure of a given state.
This feature arises because the Koashi-Imoto decomposition of the given state shows the classical part, the quantum part, and the redundant part of the state, and entanglement can be gained from the redundant part by entanglement distillation, while the classical part can be merged at zero entanglement cost by a measurement followed by classical communication of the measurement outcome.
In these protocols, it is crucial to coherently combine different subprocesses, namely, entanglement distillation from the redundant part and quantum teleportation of the quantum part, using controlled measurements and controlled isometries.

In addition to achievability bounds for an arbitrarily small-dimensional system derived from the protocols for exact state merging, improved converse bounds of entanglement cost in exact state merging are shown and this bound is proven to be optimal when a purification of the state to be merged is a three-qubit state.
These results on exact state merging can also be extended to its approximate versions by means of smoothing~\cite{R2,T5,T11}, while exact state merging suffices in cases relevant to distributed quantum information processing, such as the cases of code states for quantum error correcting codes~\cite{G,D,T10,B}.

These results yield protocols for one-shot quantum state merging applicable even to small- and intermediate-scale states, and further research will be needed to establish general strategies for state merging achieving both small-scale applicability and asymptotic optimality.

\section*{One-shot quantum state merging under one-way and two-way communication}

Chapter~\ref{sec:two_way} proved that the minimal entanglement cost in state merging under \textit{one-way} LOCC and that under \textit{two-way} LOCC can be different in a one-shot scenario, while they have shown to coincide in the asymptotic scenario.
The analysis in Chapter~\ref{sec:two_way} employs interconnection between state merging and local state discrimination, to demonstrate a provable separation between one-way LOCC and two-way LOCC in state merging, whose \textit{asymptotically non-surviving} property is different from the known separations in Table~\ref{table:compare}.
Based on this interconnection, state merging and local state discrimination can also be interpreted as distributed decoding of nonlocally encoded information.

These results suggest that in state merging from $A$ to $B$ under a one-shot regime,
preprocessing of quantum side information at $B$ and backward classical communication from $B$ to $A$ may increase usability of the quantum side information for reducing entanglement cost of protocols, while further research will be needed to construct general protocols for one-shot state merging using two-way communication.
Even if construction of an optimal two-way protocol for one-shot quantum state merging should be challenging,
the difference between one-way and two-way LOCC in entanglement cost in state merging may also appear in the framework of second-order asymptotic analysis~\cite{T9}.
In particular,
entanglement cost of non-catalytic approximate state merging of $\Ket{\psi}^{RAB}$ within $\epsilon$ may be in the form of
\begin{equation}
  \min\left\{\log_2 K\right\}=n{H\left(A|B\right)}_\psi+\sqrt{n}C\left(\epsilon\right)+\mathcal{O}\left(\log n\right),
\end{equation}
where the minimum is taken over all the protocols achieving approximate state merging of ${\left(\Ket{\psi}^{RAB}\right)}^{\otimes n}$ within $\epsilon$, and $C\left(\epsilon\right)$ is a function of $\epsilon$ for the coefficient of the second term.
Whereas whether one-way or two-way LOCC is allowed in quantum state merging does not affect the first coefficient ${H\left(A|B\right)}_\psi\,$, it is left as an open question whether $C\left(\epsilon\right)$ is affected or not.

\chapter{\label{sec:summary_2}Conclusion of Part~\ref{part:2}}

Part~\ref{part:2} analyzed properties of multipartite entanglement in distributed quantum information processing.
The following two results are obtained in this part.
These results facilitate operational understanding and efficient use of multipartite entanglement in the context of distributed quantum information processing.

\section*{Distributed encoding and decoding of quantum information over networks}

Chapter~\ref{sec:distributed_encoding_decoding} quantitatively characterized nonlocal properties of multipartite quantum transformations for encoding and decoding quantum information in a multipartite system in terms of the entanglement cost.
For any tree-topology network connecting spatially separated parties $v_1\,,\ldots,v_N\,$,
the entanglement costs required for performing an isometry $U:\mathcal{H}\to\bigotimes_{k=1}^{N}\tilde{\mathcal{H}}^{v_k}$ representing encoding and the inverse $U^\dag:\bigotimes_{k=1}^{N}\tilde{\mathcal{H}}^{v_k}\to\mathcal{H}$ representing decoding are evaluated, where the system $\mathcal{H}$ for logical states is located at one of the parties and each subsystem $\tilde{\mathcal{H}}^{v_k}$ for physical states is located at each party $v_k$.
Regarding the encoding, a protocol for spreading quantum information is constructed, and this protocol is proven to achieve the optimal entanglement cost.
As for the decoding, the protocol for concentrating quantum information is also constructed and this protocol can reduce the entanglement cost compared to that of spreading quantum information.
Hence, while $U$ and $U^\dag$ are inverse of each other, a bound is derived for quantitatively differentiating nonlocal properties of $U$ for encoding and $U^\dag$ for decoding in terms of entanglement cost.
Applications of these protocols to multiparty tasks are also demonstrated, such as one-shot distributed source compression~\cite{D8,D9,A8} and LOCC-assisted decoding in quantum secret sharing~\cite{G8}.

The concept of encoding and the decoding represented by isometries has pivotal roles not only in quantum information science,
and further investigation of applications within and beyond quantum information science is left for future works.

\section*{When does multipartite entanglement outperform bipartite entanglement?}

Chapter~\ref{sec:multipartite} introduced and analyzed the task of system-size-limited quantum state preparation for comparing multipartite and bipartite entanglement from the viewpoint of local quantum system sizes of the parties in the distributed settings.
Introducing limitations on the size of the local system of each party,
Chapter~\ref{sec:multipartite} analyzes the capabilities of common resource states exhibiting multipartite entanglement for a given target set of quantum states and those consisting of bipartite entanglement.

By showing nontrivial examples, the capabilities of these common resource states are differentiated in terms of achievability of the system-size-limited quantum state preparations for the same target set in the static setting where a common resource state has to be stored within a given limitation of local system sizes.
In addition to this static setting, the dynamic setting is considered where the parties may use a common resource state exhibiting multipartite entanglement, but this common resource state has to be prepared by temporal uses of bipartite quantum communication resources within the limitation of local system sizes.
As for the dynamics setting, examples shown in Chapter~\ref{sec:multipartite} imply that common resource states in the dynamic setting have an intermediate capability between the common resource states exhibiting multipartite entanglement and those consisting of bipartite entanglement.

These results provide examples indicating that multipartite entanglement outperforms bipartite entanglement when limitations on the local system sizes matter in both the static setting and the dynamic setting.
Further research will be needed to establish more general connections between the system sizes for common resource states and properties differentiating multipartite and bipartite entanglement.

\chapter{Concluding remarks and outlook}

This thesis has established a paradigm for investigating multipartite entanglement based on distributed quantum information processing over networks, progressing beyond applications of resource-theoretic analyses based on the state convertibility introducing the partial order of entanglement in the LOCC framework.
In particular, properties of multipartite entanglement are characterized in terms of entanglement costs and local quantum system sizes required for distributed quantum information processing.
For such characterization, this thesis has constructed and analyzed the protocols for one-shot quantum state merging and splitting applicable to arbitrarily small-dimensional quantum systems, which can be used as fundamental building blocks of further theoretical analyses and experimental implementations of distributed quantum information processing.

Besides applying the framework of distributed quantum information processing over networks established in this thesis to further investigations of multipartite entanglement,
the obtained results in this thesis are also related to the following open questions in broader research fields.

\section*{One-shot quantum information theory on small and intermediate scales}
Chapter~\ref{sec:merge} discusses the cases where asymptotically optimal protocols in one-shot quantum information theory are not necessarily efficient on small and intermediate scales relevant to distributed quantum information processing. As for one-shot quantum state merging, answers to the following questions are beneficial to constructing more efficient protocols on the small and intermediate scales than those obtained in Chapter~\ref{sec:merge}: Under what condition does the obtained protocols for one-shot quantum state merging based on the Koashi-Imoto decomposition become optimal, and more generally, how is the minimal cost of one-shot state merging on the small and intermediate scales characterized?

\section*{Entanglement catalysis}
As discussed in Chapter~\ref{sec:two_way}, there are tasks of which catalytic use of entanglement affects achievability, and further investigations of this entanglement catalysis lead to advantageous distributed quantum information processing and better understanding of quantum entanglement.
How is it possible in the catalytic setting to prove the asymptotically non-surviving separation between one-way LOCC and two-way LOCC in one-shot quantum state merging shown in Chapter~\ref{sec:two_way}, or conversely, does this separation disappear by allowing catalytic use of entanglement?

\section*{Characterization of properties of multipartite entanglement beyond quantification}
While entanglement costs in spreading and concentrating quantum information are defined for a general network, it is crucial to consider tree-topology networks in the evaluation of the entanglement costs in Chapter~\ref{sec:distributed_encoding_decoding}. If entanglement costs can be evaluated for a more general class of networks than tree-topology networks, difference in the entanglement costs arising from topologies of the given networks may provide a characterization of multipartite entanglement based on not only quantities but also the network topologies. Is there another class of networks than tree-topology networks over which such evaluation of entanglement costs is possible?

\section*{Causality and entanglement in communication tasks}
Chapter~\ref{sec:multipartite} shows cases where tasks using multipartite entanglement cannot be achieved using bipartite entanglement. Combination of bipartite entanglement with classical communication between two parties achieves quantum communication by means of quantum teleportation, where entanglement can be regarded as a spatial resource shared between spatially separated parties, and classical communication can be regarded as a temporal resource introducing a temporal causal order into distributed quantum information processing. As for multipartite cases, no analogous correspondence between spatial resources of states exhibiting multipartite entanglement and temporal resources for achieving communication tasks is known in general.  Is there a situation where a resource state exhibiting multipartite entanglement can also be interpreted as a resource for a multipartite version of some quantum communication task, and if such a situation exists, what serves as a multipartite temporal resource introducing causality among multiple parties, corresponding to the multipartite spatial resource, \textit{i.e.}, multipartite entanglement?

\appendix

\part*{Appendix}

\chapter{\label{sec:koashi_imoto}How to obtain Koashi-Imoto decomposition}

This appendix demonstrates how to obtain the Koashi-Imoto decomposition of a given tripartite pure state $\Ket{\psi}^{RAB}$ shown in Lemma~\ref{lem:koashi_imoto_decomposition_tripartite}. The Koashi-Imoto decomposition of $\Ket{\psi}^{RAB}$ follows from that of the corresponding set $S_\psi^{A|R}$ defined as Equation~\eqref{eq:psi_lambda}, as discussed in Section~\ref{sec:decomposition}. This appendix summarizes an algorithm shown in Reference~\cite{K3} for obtaining the Koashi-Imoto decomposition of any given set of states and provide an example of how to obtain the Koashi-Imoto decomposition of a given tripartite pure state using this algorithm.

The algorithm shown in Reference~\cite{K3} works by iteratively refining decompositions of the Hilbert space $\mathcal{H}^A$ in the form of
\begin{equation}
  \label{eq:decomposition_form}
  \mathcal{H}^A=\bigoplus_{j=0}^{J-1}\mathcal{H}^{a_j^\textup{L}}\otimes\mathcal{H}^{a_j^\textup{R}}.
\end{equation}
For a decomposition in this form, let $\Pi^{a_j^\textup{L}}$ and $\Pi^{a_j^\textup{R}}$ denote the projectors onto $\mathcal{H}^{a_j^\textup{L}}$ and $\mathcal{H}^{a_j^\textup{R}}$, respectively.
The degree of refinement is evaluated by an index $r$ defined for the decomposition in the form of Equation~\eqref{eq:decomposition_form} as
\begin{equation}
  r\coloneq\frac{1}{2}\left(\sum_{J=0}^{J-1}\dim\mathcal{H}^{a_j^\textup{R}}\right)\left(\sum_{J=0}^{J-1}\dim\mathcal{H}^{a_j^\textup{R}}+1\right)-J+1.
\end{equation}
The algorithm begins with initially regarding $\mathcal{H}^A$ as
\begin{equation}
  \mathcal{H}^A=\mathcal{H}^{a_0^\textup{L}},
\end{equation}
where $J=1$, the index is initially
\begin{equation}
  r=1,
\end{equation}
and $\mathcal{H}^{a_0^\textup{R}}$ does not explicitly appear since
\begin{equation}
  \dim\mathcal{H}^{a_0^\textup{L}}=\dim\mathcal{H}^A,\quad\dim\mathcal{H}^{a_0^\textup{R}}=1.
\end{equation}
Then, the refinement can be performed by two types of procedures, which are referred to as the \textit{$\textup{L}$-decomposing procedure} and the \textit{$\textup{R}$-combining procedure}.
According to the given set of states, the $\textup{L}$-decomposing procedure decomposes a Hilbert space $\mathcal{H}^{a_{j_0}^\textup{L}}$ in an intermediate decomposition in the form of Equation~\eqref{eq:decomposition_form} into two subspaces, and the $\textup{R}$-combining procedure combines two different Hilbert spaces $\mathcal{H}^{a_{j_0}^\textup{R}}$ and $\mathcal{H}^{a_{j_1}^\textup{R}}$ in an intermediate decomposition in the form of Equation~\eqref{eq:decomposition_form} into one, as discussed later.
Each procedure increases the index $r$ representing the degree of refinement of the decomposition,
and the algorithm repeatedly applies either of the two procedures, until a decomposition maximizing $r$ is obtained.
Since $r$ is an integer bounded by
\begin{equation}
  1\leqq r\leqq\frac{1}{2}\left(\dim\mathcal{H}^A\right)\left(\dim\mathcal{H}^A+1\right),
\end{equation}
the algorithm terminates after applying these procedures
\begin{equation}
  O\left({\left(\dim\mathcal{H}^A\right)}^2\right)
\end{equation}
times in total.
The decomposition maximizing $r$ is uniquely determined and is said to be maximal in Reference~\cite{K3}, satisfying the conditions shown in Lemma~\ref{lem:koashi_imoto_decomposition_set}.
For obtaining the Koashi-Imoto decomposition of a given bipartite state $\psi^{RA}$,
whether the decomposition in the form of Equation~\eqref{eq:decomposition_form} is maximal can also be checked by calculating operators on $\mathcal{H}^\textup{R}\otimes\mathcal{H}^{a_j^\textup{L}}\otimes\mathcal{H}^{a_j^\textup{R}}$ for all $j$
\begin{equation}
  \label{eq:product_operator}
    \psi^{R a_j^\textup{L} a_j^\textup{R}}
    \coloneq\left(\mathbb{1}^{R}\otimes\Pi^{a_j^\textup{L}}\otimes\Pi^{a_j^\textup{R}}\right)\psi^{RA}\left(\mathbb{1}^{R}\otimes\Pi^{a_j^\textup{L}}\otimes\Pi^{a_j^\textup{R}}\right),
\end{equation}
and if the decomposition is maximal, each of these operators is a tensor product of operators of $\mathcal{H}^\textup{R}\otimes\mathcal{H}^{a_j^\textup{R}}$ and $\mathcal{H}^{a_j^\textup{L}}$.

In the following, how to perform the $\textup{L}$-decomposing procedure and the $\textup{R}$-combining procedure is discussed in the case of the Koashi-Imoto decomposition of $S_\psi^{A|R}\coloneq\left\{\psi^A\left(\Lambda^{R}\right):\Lambda^{R}\geqq 0\right\}$ defined as Equation~\eqref{eq:psi_lambda}.

\textit{The $\textup{L}$-decomposing procedure}: (See also Lemma~3 in Reference~\cite{K3}.)
Given an intermediate decomposition in the form of Equation~\eqref{eq:decomposition_form},
the $\textup{L}$-decomposing procedure aims to decompose a Hilbert space $\mathcal{H}^{a_{j_0}^\textup{L}}$ in this given decomposition into two subspaces,
so that the decomposition is refined as
\begin{equation}
  \mathcal{H}^{a_{j_0}^\textup{L}}\otimes\mathcal{H}^{a_{j_0}^\textup{R}}=\left(\mathcal{H}_{+}^{a_{j_0}^\textup{L}}\otimes\mathcal{H}^{a_{j_0}^\textup{R}}\right)\oplus\left(\mathcal{H}_{-}^{a_{j_0}^\textup{L}}\otimes\mathcal{H}^{\tilde{a}_{j_0}^\textup{R}}\right),
\end{equation}
where the right-hand side represents subspaces in a refined decomposition satisfying
\begin{equation}
    \mathcal{H}^{a_{j_0}^\textup{L}}=\mathcal{H}_{+}^{a_{j_0}^\textup{L}}\oplus\mathcal{H}_{-}^{a_{j_0}^\textup{L}}.
\end{equation}
For the Koashi-Imoto decomposition of $S_\psi^{A|R}$,
this refinement is achieved in the following way.
\begin{enumerate}[{Step $\textup{L}$}-1:]
  \item Find $j_0\in\left\{0,\ldots,J-1\right\}$, $\Ket{a}\in\mathcal{H}^{a_{j_0}^\textup{R}}$, $\Ket{b}\in\mathcal{H}^{a_{j_0}^\textup{R}}$, and $\Lambda^{R}\geqq 0$ such that for any $c\geqq 0$
    \begin{equation}
      \rho\neq c\rho^\prime,
    \end{equation}
    where $\rho$ and $\rho^\prime$ are operators on $\mathcal{H}^{a_{j_0}^\textup{L}}$ defined as
    \begin{align}
        \rho&\coloneq\left(\Pi^{a_{j_0}^\textup{L}}\otimes\Bra{a}^{a_{j_0}^\textup{R}}\right)\psi^A\left(\Lambda^{R}\right)\left(\Pi^{a_{j_0}^\textup{L}}\otimes\Ket{a}^{a_{j_0}^\textup{R}}\right),\\
        \rho^\prime&\coloneq\left(\Pi^{a_{j_0}^\textup{L}}\otimes\Bra{b}^{a_{j_0}^\textup{R}}\right)\psi^A\left(\mathbb{1}^{R}\right)\left(\Pi^{a_{j_0}^\textup{L}}\otimes\Ket{b}^{a_{j_0}^\textup{R}}\right).
    \end{align}
  \item Calculate the spectral decomposition of an operator on $\mathcal{H}^{a_{j_0}^\textup{L}}$
    \begin{equation}
      \frac{\rho}{\tr\rho}-\frac{\rho^\prime}{\tr\rho^\prime}=\sum_l \lambda_l\Ket{l}\Bra{l}.
    \end{equation}
    Using the subspaces spanned by eigenvectors of this operator corresponding to the positive eigenvalues and the non-positive eigenvalues, decompose $\mathcal{H}^{a_{j_0}^\textup{L}}$ into
    \begin{equation}
      \mathcal{H}^{a_{j_0}^\textup{L}}=\mathcal{H}_{+}^{a_{j_0}^\textup{L}}\oplus\mathcal{H}_{-}^{a_{j_0}^\textup{L}},
    \end{equation}
    where the subspaces on the right-hand side are defined as
    \begin{align}
      \mathcal{H}_{+}^{a_{j_0}^\textup{L}}&\coloneq \spn \left\{\Ket{l}\in\mathcal{H}^{a_{j_0}^\textup{L}}:\lambda_l>0\right\},\\
      \mathcal{H}_{-}^{a_{j_0}^\textup{L}}&\coloneq\spn\left\{\Ket{l}\in\mathcal{H}^{a_{j_0}^\textup{L}}:\lambda_l\leqq 0 \right\}.
    \end{align}
    Note that $\mathcal{H}_{+}^{a_{j_0}^\textup{L}}$ and $\mathcal{H}_{-}^{a_{j_0}^\textup{L}}$ are nonzero subspaces.
  \item Define a refined decomposition as
    \begin{align}
        \mathcal{H}^A&=\bigoplus_{j=0}^{\tilde{J}-1}\mathcal{H}^{\tilde{a}_j^\textup{L}}\otimes\mathcal{H}^{\tilde{a}_j^\textup{R}},\\
        \tilde{J}&\coloneq J+1,\\
        \mathcal{H}^{\tilde{a}_j^\textup{L}}&\coloneq\begin{cases}
          \mathcal{H}^{a_j^\textup{L}}&\textup{if } 0\leqq j\leqq j_0-1,\\
          \mathcal{H}^{a_{j-1}^\textup{L}}&\textup{if } j_0\leqq j\leqq J-2,\\
          \mathcal{H}_{+}^{a_{j_0}^\textup{L}}&\textup{if } j=J-1,\\
          \mathcal{H}_{-}^{a_{j_0}^\textup{L}}&\textup{if } j=J,\\
        \end{cases}\\
        \mathcal{H}^{\tilde{a}_j^\textup{R}}&\coloneq\begin{cases}
          \mathcal{H}^{a_j^\textup{R}}&\textup{if } 0\leqq j\leqq j_0-1,\\
          \mathcal{H}^{a_{j-1}^\textup{R}}&\textup{if } j_0\leqq j\leqq J-2,\\
          \mathcal{H}^{a_{j_0}^\textup{R}}&\textup{if } j=J-1,\, J.
        \end{cases}
    \end{align}
\end{enumerate}

\textit{The $\textup{R}$-combining procedure}: (See also Lemma~4 in Reference~\cite{K3}.)
Given an intermediate decomposition in the form of Equation~\eqref{eq:decomposition_form},
the $\textup{R}$-combining procedure aims to combine two different Hilbert spaces $\mathcal{H}^{a_{j_0}^\textup{R}}$ and $\mathcal{H}^{a_{j_1}^\textup{R}}$ in this given decomposition into one,
so that the decomposition is refined as
\begin{equation}
  \begin{split}
    &\left(\mathcal{H}^{a_{j_0}^\textup{L}}\otimes\mathcal{H}^{a_{j_0}^\textup{R}}\right)\oplus\left(\mathcal{H}^{a_{j_1}^\textup{L}}\otimes\mathcal{H}^{a_{j_1}^\textup{R}}\right)\\
    &=\left(\mathcal{H}^{a_{j_0\cap j_1}^\textup{L}}\otimes\left(\mathcal{H}^{a_{j_0}^\textup{R}}\oplus\mathcal{H}^{a_{j_1}^\textup{R}}\right)\right)
    \oplus\left(\mathcal{H}_\perp^{a_{j_0}^\textup{L}}\otimes\mathcal{H}^{a_{j_0}^\textup{R}}\right)\oplus\left(\mathcal{H}_\perp^{a_{j_1}^\textup{L}}\otimes\mathcal{H}^{a_{j_1}^\textup{R}}\right),
  \end{split}
\end{equation}
where the right-hand side represents subspaces in a refined decomposition satisfying
\begin{align}
  \mathcal{H}^{a_{j_0}^\textup{L}}&=\mathcal{H}^{a_{j_0\cap j_1}^\textup{L}}\oplus\mathcal{H}_\perp^{a_{j_0}^\textup{L}},\\
  \mathcal{H}^{a_{j_1}^\textup{L}}&=\mathcal{H}^{a_{j_0\cap j_1}^\textup{L}}\oplus\mathcal{H}_\perp^{a_{j_1}^\textup{L}}.\\
\end{align}
For the Koashi-Imoto decomposition of $S_\psi^{A|R}$,
this refinement is achieved in the following way.
\begin{enumerate}[{Step $\textup{R}$}-1:]
  \item Find $j_0\in\left\{0,\ldots,J-1\right\}$, $j_1\in\left\{0,\ldots,J-1\right\}$, $\Ket{a}\in\mathcal{H}^{a_{j_0}^\textup{R}}$, $\Ket{b}\in\mathcal{H}^{a_{j_1}^\textup{R}}$, and $\Lambda^{R}\geqq 0$ such that $j_0 < j_1$ and
    \begin{equation}
      \begin{split}
        &\supp\left(\left(\Pi^{a_{j_0}^\textup{L}}\otimes\Bra{a}^{a_{j_0}^\textup{R}}\right)\psi^A\left(\Lambda^{R}\right)\left(\Pi^{a_{j_0}^\textup{L}}\otimes\Ket{a}^{a_{j_0}^\textup{R}}\right)\right)
        =\mathcal{H}^{a_{j_0}^\textup{L}},\\
        &\supp\left(\left(\Pi^{a_{j_1}^\textup{L}}\otimes\Bra{b}^{a_{j_1}^\textup{R}}\right)\psi^A\left(\Lambda^{R}\right)\left(\Pi^{a_{j_1}^\textup{L}}\otimes\Ket{b}^{a_{j_1}^\textup{R}}\right)\right)
        =\mathcal{H}^{a_{j_1}^\textup{L}},\\
        &\sigma\neq\boldsymbol{0},
      \end{split}
    \end{equation}
    where $\supp(\cdots)$ represents the support, $\boldsymbol{0}$ is the zero operator, and $\sigma$ is an operator from $\mathcal{H}^{a_{j_0}^\textup{L}}$ to $\mathcal{H}^{a_{j_1}^\textup{L}}$ defined as
    \begin{equation}
      \sigma\coloneq\left(\Pi^{a_{j_1}^\textup{L}}\otimes\Bra{b}^{a_{j_1}^\textup{R}}\right)\psi^A\left(\Lambda^{R}\right)\left(\Pi^{a_{j_0}^\textup{L}}\otimes\Ket{a}^{a_{j_0}^\textup{R}}\right).
    \end{equation}
  \item Calculate the singular value decomposition of $\sigma$
    \begin{equation}
      \sigma=\sum_{l=0}^{R-1} \sigma_l\Ket{l}^{a_{j_1}^\textup{L}}\Bra{l}^{a_{j_0}^\textup{L}},
    \end{equation}
    where $R\coloneq\rank\sigma$, and $\sigma_0\,,\ldots,\sigma_{R-1}$ are the positive singular values.
    Using the subspace spanned by the singular vectors $\left\{\Ket{0},\ldots,\Ket{R-1}\right\}$ of $\sigma$ corresponding to the positive singular values, decompose $\mathcal{H}^{a_{j_0}^\textup{L}}$ and $\mathcal{H}^{a_{j_1}^\textup{L}}$ into
    \begin{align}
      \mathcal{H}^{a_{j_0}^\textup{L}}&=\mathcal{H}^{a_{j_0\cap j_1}^\textup{L}}\oplus{\mathcal{H}_\perp^{a_{j_0}^\textup{L}}},\\
      \mathcal{H}^{a_{j_1}^\textup{L}}&=\mathcal{H}^{a_{j_0\cap j_1}^\textup{L}}\oplus{\mathcal{H}_\perp^{a_{j_1}^\textup{L}}},
    \end{align}
    where the subspaces on the right-hand side are defined as
    \begin{align}
      \mathcal{H}^{a_{j_0\cap j_1}^\textup{L}}&\coloneq\spn\left\{\Ket{0},\ldots,\Ket{R-1}\right\},\\
      \mathcal{H}_\perp^{a_{j_0}^\textup{L}}&\coloneq\supp\left(\Pi^{a_{j_0}^\textup{L}}-\sum_{l=0}^{R-1}\Ket{l}\Bra{l}^{a_{j_0}^\textup{L}}\right),\\
      \mathcal{H}_\perp^{a_{j_1}^\textup{L}}&\coloneq\supp\left(\Pi^{a_{j_1}^\textup{L}}-\sum_{l=0}^{R-1}\Ket{l}\Bra{l}^{a_{j_1}^\textup{L}}\right).
    \end{align}
    Note that $\mathcal{H}_\perp^{a_{j_0}^\textup{L}}$ and $\mathcal{H}_\perp^{a_{j_0}^\textup{L}}$ may be zero, and define flags indicating whether $\mathcal{H}_\perp^{a_{j_0}^\textup{L}}$ and $\mathcal{H}_\perp^{a_{j_0}^\textup{L}}$ are zero as
    \begin{align}
      s_{j_0}&\coloneq\begin{cases}
        0&\textup{if } \mathcal{H}_\perp^{a_{j_0}^\textup{L}}=\left\{\boldsymbol{0}\right\},\\
        1&\textup{otherwise},
      \end{cases}\\
      s_{j_1}&\coloneq\begin{cases}
        0&\textup{if } \mathcal{H}_\perp^{a_{j_1}^\textup{L}}=\left\{\boldsymbol{0}\right\},\\
        1&\textup{otherwise}.
      \end{cases}
    \end{align}
  \item Define a refined decomposition as
    \begin{align}
        \mathcal{H}^A&=\bigoplus_{j=0}^{\tilde{J}-1}\mathcal{H}^{\tilde{a}_j^\textup{L}}\otimes\mathcal{H}^{\tilde{a}_j^\textup{R}},\\
        \tilde{J}&\coloneq J-1+s_{j_0}+s_{j_1}\,,\\
        \mathcal{H}^{\tilde{a}_j^\textup{L}}&\coloneq\begin{cases}
          \mathcal{H}^{a_j^\textup{L}}&\textup{if } 0\leqq j\leqq j_0-1,\\
          \mathcal{H}^{a_{j+1}^\textup{L}}&\textup{if } j_0\leqq j\leqq j_1-2,\\
          \mathcal{H}^{a_{j+2}^\textup{L}}&\textup{if } j_1-1\leqq j\leqq J-3,\\
          \mathcal{H}^{a_{j_0\cap j_1}^\textup{L}}&\textup{if } j=J-2,\\
          \mathcal{H}_\perp^{a_{j_0}^\textup{L}}&\textup{if } j = J-2+s_{j_0}\\
                                       &\textup{and }s_{j_0}=1,\\
          \mathcal{H}_\perp^{a_{j_1}^\textup{L}}&\textup{if } j = J-2+s_{j_0}+s_{j_1}\\
                                       &\textup{and }s_{j_1}=1,\\
        \end{cases}\\
        \mathcal{H}^{\tilde{a}_j^\textup{R}}&\coloneq\begin{cases}
          \mathcal{H}^{a_j^\textup{R}}&\textup{if } 0\leqq j\leqq j_0-1,\\
          \mathcal{H}^{a_{j+1}^\textup{R}}&\textup{if } j_0\leqq j\leqq j_1-2,\\
          \mathcal{H}^{a_{j+2}^\textup{R}}&\textup{if } j_1-1\leqq j\leqq J-3,\\
          \mathcal{H}^{a_{j_0}^\textup{R}}\oplus\mathcal{H}^{a_{j_1}^\textup{R}}&\textup{if } j=J-2,\\
          \mathcal{H}^{a_{j_0}^\textup{R}}&\textup{if } j = J-2+s_{j_0}\\
                                 &\textup{and } s_{j_0}=1,\\
          \mathcal{H}^{a_{j_1}^\textup{R}}&\textup{if } j = J-2+s_{j_0}+s_{j_1}\\
                                 &\textup{and } s_{j_1}=1.
        \end{cases}
    \end{align}
\end{enumerate}

The following example demonstrates how to obtain the Koashi-Imoto decomposition of a tripartite pure state using the above algorithm.

\begin{example}
  \textit{Koashi-Imoto decomposition of a tripartite pure state.}
Consider a tripartite pure state
\begin{equation}
  \begin{split}
    &\Ket{\psi}^{RAB}\\
    &\coloneq\frac{1}{2\sqrt{2}}{\left(\Ket{0}^{R}\otimes\Ket{0}^{A_1}+\Ket{1}^{R}\otimes\Ket{1}^{A_1}\right)}
    \otimes{\left(\Ket{0}^{A_2}\otimes\Ket{0}^{B}+\Ket{1}^{A_2}\otimes\Ket{1}^{B}\right)}\\
    &\quad +\frac{1}{\sqrt{2}}\Ket{2}^{R}\otimes\Ket{2}^{A_1}\otimes\Ket{0}^{A_2}\otimes\Ket{2}^{B},
  \end{split}
\end{equation}
where $\mathcal{H}^\textup{R}$ is of $3$ dimension, $\mathcal{H}^A=\mathcal{H}^{A_1}\otimes\mathcal{H}^{A_2}$ of $3\times 2 = 6$ dimension, and $\mathcal{H}^{B}$ of $3$ dimension.
The Koashi-Imoto decomposition can be algorithmically obtained as follows,
where the order of subspaces in intermediate decompositions is sorted for readability.
\begin{enumerate}[{Step} 1:]
  \item Initially, regard $\mathcal{H}^A$ as
    \begin{equation}
      \label{eq:1}
      \mathcal{H}^A=\mathcal{H}^{a_0^\textup{L}}.
    \end{equation}
  \item Apply the $\textup{L}$-decomposing procedure to the intermediate decomposition given by Equation~\eqref{eq:1}, where $j_0=0$, $\Ket{a}=1$, $\Ket{b}=1$, and $\Lambda^{R}=\Ket{0}\Bra{0}$, and $\mathcal{H}^A$ is decomposed into
    \begin{equation}
      \label{eq:2}
      \mathcal{H}^A=\mathcal{H}^{a_0^\textup{L}}\oplus\mathcal{H}^{a_1^\textup{L}},
    \end{equation}
    where $\dim\mathcal{H}^{a_0^\textup{R}}=\dim\mathcal{H}^{a_1^\textup{R}}=1$ and
    \begin{align}
      \mathcal{H}^{a_0^\textup{L}}=\spn\Big\{&\Ket{0}^{A_1}\otimes\Ket{0}^{A_2},\Ket{0}^{A_1}\otimes\Ket{1}^{A_2}\Big\},\\
      \mathcal{H}^{a_1^\textup{L}}=\spn\Big\{&\Ket{1}^{A_1}\otimes\Ket{0}^{A_2},\Ket{1}^{A_1}\otimes\Ket{1}^{A_2},
      \Ket{2}^{A_1}\otimes\Ket{0}^{A_2},\Ket{2}^{A_1}\otimes\Ket{1}^{A_2}\Big\},\\
    \end{align}
  \item Apply the $\textup{L}$-decomposing procedure to the intermediate decomposition given by Equation~\eqref{eq:2}, where $j_0=1$, $\Ket{a}=1$, $\Ket{b}=1$, and $\Lambda^{R}=\Ket{1}\Bra{1}$, and $\mathcal{H}^A$ is decomposed into
    \begin{equation}
      \label{eq:3}
      \mathcal{H}^A=\mathcal{H}^{a_0^\textup{L}}\oplus\mathcal{H}^{a_1^\textup{L}}\oplus\mathcal{H}^{a_2^\textup{L}},
    \end{equation}
    where $\dim\mathcal{H}^{a_0^\textup{R}}=\dim\mathcal{H}^{a_1^\textup{R}}=\dim\mathcal{H}^{a_2^\textup{R}}=1$ and
    \begin{align}
      \mathcal{H}^{a_0^\textup{L}}=\spn\Big\{&\Ket{0}^{A_1}\otimes\Ket{0}^{A_2},\Ket{0}^{A_1}\otimes\Ket{1}^{A_2}\Big\},\\
      \mathcal{H}^{a_1^\textup{L}}=\spn\Big\{&\Ket{1}^{A_1}\otimes\Ket{0}^{A_2},\Ket{1}^{A_1}\otimes\Ket{1}^{A_2}\Big\},\\
      \mathcal{H}^{a_2^\textup{L}}=\spn\Big\{&\Ket{2}^{A_1}\otimes\Ket{0}^{A_2},\Ket{2}^{A_1}\otimes\Ket{1}^{A_2}\Big\}.
    \end{align}
  \item Apply the $\textup{R}$-combining procedure to the intermediate decomposition given by Equation~\eqref{eq:3}, where $j_0=0$, $j_1=1$, $\Ket{a}=1$, $\Ket{b}=1$, and $\Lambda^{R}=\Ket{0}\Bra{0}+\Ket{0}\Bra{1}+\Ket{1}\Bra{0}+\Ket{1}\Bra{1}$, and $\mathcal{H}^A$ is decomposed into
    \begin{equation}
      \label{eq:4}
      \mathcal{H}^A=\left(\mathcal{H}^{a_0^\textup{L}}\otimes\mathcal{H}^{a_0^\textup{R}}\right)\oplus\mathcal{H}^{a_1^\textup{L}},
    \end{equation}
    where $\dim\mathcal{H}^{a_1^\textup{R}}=1$ and
    \begin{align}
      \mathcal{H}^{a_0^\textup{L}}=\spn\Big\{&\Ket{0}^{A_2},\Ket{1}^{A_2}\Big\},\\
      \mathcal{H}^{a_0^\textup{R}}=\spn\Big\{&\Ket{0}^{A_1},\Ket{1}^{A_1}\Big\},\\
      \mathcal{H}^{a_1^\textup{L}}=\spn\Big\{&\Ket{2}^{A_1}\otimes\Ket{0}^{A_2},\Ket{2}^{A_1}\otimes\Ket{1}^{A_2}\Big\}.
    \end{align}
  \item Terminate the algorithm, since for each $j$, the operator $\psi^{R a_j^\textup{L} a_j^\textup{R}}$ defined as Equation~\eqref{eq:product_operator} is a tensor product of operators of $\mathcal{H}^\textup{R}\otimes\mathcal{H}^{a_j^\textup{R}}$ and $\mathcal{H}^{a_j^\textup{L}}$, and hence, the decomposition in Equation~\eqref{eq:4} is maximal. In this case, $\Ket{\psi}^{RAB}$ is decomposed into
    \begin{equation}
      \begin{split}
        &\Ket{\psi}^{RAB}\\
        &=\frac{1}{2\sqrt{2}}{\left(\Ket{0}^{R}\otimes\Ket{0}^{a_0^\textup{R}}+\Ket{1}^{R}\otimes\Ket{1}^{a_0^\textup{R}}\right)}
        \otimes{\left(\Ket{0}^{a_0^\textup{L}}\otimes\Ket{0}^{b_0^\textup{L}}+\Ket{1}^{a_0^\textup{L}}\otimes\Ket{1}^{b_0^\textup{L}}\right)}\\
        &\quad \oplus\frac{1}{\sqrt{2}}\left(\Ket{2}^{R}\otimes{\left(\Ket{2}\otimes\Ket{0}\right)}^{a_1^\textup{L}}\otimes\Ket{2}^{b_1^\textup{L}}\right),
      \end{split}
    \end{equation}
\end{enumerate}
\end{example}

\chapter{\label{sec:equivalence}Tasks equivalent to exact state merging}

This appendix provides the proof of Proposition~\ref{prp:max} on the tasks equivalent to exact state merging,
in the sense that the tasks shown in Proposition~\ref{prp:max} are achievable at the same entanglement cost using the same protocol.

\begin{proof}[Proof of Proposition~\ref{prp:max}]
  The equivalence in the catalytic setting is shown in the following, while the statement in the non-catalytic setting follows from the same argument setting $\log_2 L=0$.
  It is shown that each of Statements~1--3 holds if and only if
  \begin{equation}
    \label{eq:m}
    \mathcal{M}\left({\Ket{\psi_l}\Bra{\psi_{l'}}}^{AB}\otimes{\Phi^+_K}^{\overline{A}\overline{B}}\right)={\Ket{\psi_l}\Bra{\psi_{l'}}}^{B'B}\otimes{\Phi^+_L}^{\overline{A}\overline{B}}
  \end{equation}
  holds for any $l$ and $l'$.

  \textit{Statement~1 $\Leftrightarrow$ \textup{Equation}~\eqref{eq:m}}:
  Assume Statement~1; that is, an LOCC map $\mathcal{M}$ by $A$ and $B$ achieves the following exact state merging of $\Ket{\psi}^{RAB}$
  \begin{equation}
    \id^R\otimes\mathcal{M}\left({\psi}^{RAB}\otimes{\Phi^+_K}^{\overline{A}\overline{B}}\right)={\psi}^{RB'B}\otimes{\Phi^+_L}^{\overline{A}\overline{B}}.
  \end{equation}
  The left-hand side and the right-hand side are written as
  \begin{equation}
    \begin{split}
      \id^R\otimes\mathcal{M}\left({\psi}^{RAB}\otimes{\Phi^+_K}^{\overline{A}\overline{B}}\right)
      &=\sum_{l,l'}\frac{1}{\sqrt{\lambda_l\lambda_{l'}}}{\Ket{l}\Bra{l'}}^R\otimes\mathcal{M}\left({\Ket{\psi_l}\Bra{\psi_{l'}}}^{AB}\otimes{\Phi^+_K}^{\overline{A}\overline{B}}\right),\\
      {\psi}^{RB'B}\otimes{\Phi^+_L}^{\overline{A}\overline{B}}
      &=\sum_{l,l'}\frac{1}{\sqrt{\lambda_l\lambda_{l'}}}{\Ket{l}\Bra{l'}}^R\otimes{\Ket{\psi_l}\Bra{\psi_{l'}}}^{B'B}\otimes{\Phi^+_L}^{\overline{A}\overline{B}}.
    \end{split}
  \end{equation}
  Due to the linear independence, Equation~\eqref{eq:m} holds for any $l$ and $l'$.
  The converse follows from the linearity of $\mathcal{M}$.

  \textit{Statement~2 $\Leftrightarrow$ \textup{Equation}~\eqref{eq:m}}: This equivalence can be shown in the same way as the equivalence between Statement~1 and Equation~\eqref{eq:m}, by substituting $\psi$ with ${\Phi}_D^+\left(\psi\right)$.

  \textit{Statement~3 $\Leftrightarrow$ \textup{Equation}~\eqref{eq:m}}: Assume Statement~3. For each $l$,
  \begin{equation}
    \mathcal{M}\left({\Ket{\psi_l}\Bra{\psi_{l}}}^{AB}\otimes{\Phi^+_K}^{\overline{A}\overline{B}}\right)={\Ket{\psi_l}\Bra{\psi_{l}}}^{B'B}\otimes{\Phi^+_L}^{\overline{A}\overline{B}}
  \end{equation}
  holds as a special case of Statement~3.
  For any different $l$ and $l^\prime$,
  consider two cases of choosing $\psi_{\boldsymbol{\alpha}}^{AB}\in S_\psi^{AB}$ as
  \begin{equation}
      \frac{1}{2}\Ket{\psi_l}\Bra{\psi_l}+\frac{1}{2}\Ket{\psi_l}\Bra{\psi_{l^\prime}}+\frac{1}{2}\Ket{\psi_{l^\prime}}\Bra{\psi_l}+\frac{1}{2}\Ket{\psi_{l^\prime}}\Bra{\psi_{l^\prime}}
  \end{equation}
  and
  \begin{equation}
      \frac{1}{2}\Ket{\psi_l}\Bra{\psi_l}+\frac{\textup{i}}{2}\Ket{\psi_l}\Bra{\psi_{l^\prime}}-\frac{\textup{i}}{2}\Ket{\psi_{l^\prime}}\Bra{\psi_l}+\frac{1}{2}\Ket{\psi_{l^\prime}}\Bra{\psi_{l^\prime}}.
  \end{equation}
  Applying Statement~3 to these two states and using the linearity of $\mathcal{M}$ yield
  \begin{align}
    \mathcal{M}\left({\Ket{\psi_l}\Bra{\psi_{l^\prime}}}^{AB}\otimes{\Phi^+_K}^{\overline{A}\overline{B}}\right)&={\Ket{\psi_l}\Bra{\psi_{l^\prime}}}^{B'B}\otimes{\Phi^+_L}^{\overline{A}\overline{B}},\\
    \mathcal{M}\left({\Ket{\psi_{l^\prime}}\Bra{\psi_{l}}}^{AB}\otimes{\Phi^+_K}^{\overline{A}\overline{B}}\right)&={\Ket{\psi_{l^\prime}}\Bra{\psi_{l}}}^{B'B}\otimes{\Phi^+_L}^{\overline{A}\overline{B}}.
  \end{align}
  Therefore, Equation~\eqref{eq:m} holds for any $l$ and $l^\prime$.
  The converse follows from the linearity of $\mathcal{M}$.
\end{proof}

\chapter{\label{sec:monotonic}Monotonic property of conditional quantum entropy}

This appendix provides the proof of Proposition~\ref{prp:monotonic} on the monotonically nondecreasing property of conditional quantum entropy ${H\left(A|B\right)}_\psi$ under $B$'s preprocessing and backward classical communication from $B$ to $A$.

\begin{proof}[Proof of Proposition~\ref{prp:monotonic}]
  This proof shows that $B$'s preprocessing ${\left\{M_j^{B}\right\}}_j$ does not decrease the conditional quantum entropy on average, and backward classical communication and $A$'s isometry $U_j^A$ do not change the conditional quantum entropy.
  Performing ${\left\{U_j^A\otimes M_j^B\right\}}_j$ is equivalent to sequentially performing the following steps.
  First, the measurement ${\left\{M_j^B\right\}}_j$ can be regarded as $B$'s local channel
  transforming $\psi^{RAB}$ into
  \begin{equation}
    {\psi^\prime}^{XRAB}\coloneq\sum_j p\left(j\right)\Ket{j}\Bra{j}^{X}\otimes\frac{M_j^B \psi^{RAB} {M_j^B}^\dag}{p\left(j\right)},
  \end{equation}
  where $\mathcal{H}^{X}$ is $B$'s system for storing the measurement outcome.
  Next, the backward classical communication transforms ${\psi^\prime}^{XRAB}$ into
  \begin{equation}
    \begin{split}
      &{\psi^{\prime\prime}}^{X^\prime XRAB}\\
      &\coloneq\sum_j p\left(j\right)\Ket{j}\Bra{j}^{X^\prime}\otimes\Ket{j}\Bra{j}^{X}\otimes\frac{M_j^B \psi^{RAB} {M_j^B}^\dag}{p\left(j\right)},
    \end{split}
  \end{equation}
  where $\mathcal{H}^{X^\prime}$ is $A$'s system for storing the measurement outcome.
  Finally, the isometry $U_j^A$ transforms ${\psi^{\prime\prime}}^{X^\prime XRAB}$ into
  \begin{equation}
    {\psi^{\prime\prime\prime}}^{X^\prime XRAB}\coloneq\sum_j p\left(j\right)\Ket{j}\Bra{j}^{X^\prime}\otimes\Ket{j}\Bra{j}^{X}\otimes\psi_j^{RAB}.
  \end{equation}

  The conditional quantum entropy for each of these steps is evaluated as follows.
  Regarding the measurement ${\left\{M_j^B\right\}}_j\,$, the data processing inequality yields
  \begin{equation}
    {H\left(A|B\right)}_\psi\leqq{H\left(A|XB\right)}_{\psi^\prime}.
  \end{equation}
  As for the backward classical communication, it holds that
  \begin{equation}
    {H\left(A|XB\right)}_{\psi^\prime}={H\left(A|XB\right)}_{{\psi^{\prime\prime}}}={H\left(X^\prime A|XB\right)}_{{\psi^{\prime\prime}}}.
  \end{equation}
  Since the isometry $U_j^A$ for each $j$ can be performed using a controlled isometry independent of $j$
  \begin{equation}
    \sum_j\Ket{j}\Bra{j}^{X^\prime}\otimes U_j^A,
  \end{equation}
  it holds that
  \begin{equation}
    {H\left(X^\prime A|XB\right)}_{{\psi^{\prime\prime}}}={H\left(X^\prime A|XB\right)}_{\psi^{\prime\prime\prime}}.
  \end{equation}
  Therefore, it is obtained that
  \begin{equation}
    \begin{split}
      {H\left(A|B\right)}_\psi&\leqq{H\left(X^\prime A|XB\right)}_{\psi^{\prime\prime\prime}}\\
                              &={H\left(A|XB\right)}_{\psi^{\prime\prime\prime}}\\
                              &=\sum_j p\left(j\right){H\left(A|B\right)}_{\psi_j}\,,
    \end{split}
  \end{equation}
  which yields the conclusion.
\end{proof}

\chapter{\label{sec:equivalence_spread_concentrate}Tasks equivalent to spreading and concentrating quantum information}

This appendix provides the proof of Proposition~\ref{lem:encoding_state_transformation} on state transformations equivalent to spreading and concentrating quantum information over networks,
in the sense that the tasks shown in Proposition~\ref{lem:encoding_state_transformation} are achievable at the same entanglement cost using the same protocol.

\begin{proof}[Proof of Proposition~\ref{lem:encoding_state_transformation}]
    The statement on spreading quantum information is proven in the following, while the statement on concentrating quantum information also follows from the same argument by substituting $\rho$, $U\rho U^\dag$, $\Ket{l}\Bra{l'}$, $\ket{\tilde{\psi}_l}\bra{\tilde{\psi}_{l'}}$, and $\mathcal{S}$ in the following with $U\rho U^\dag$, $\rho$, $\ket{\tilde{\psi}_l}\bra{\tilde{\psi}_{l'}}$, $\Ket{l}\Bra{l'}$, and $\mathcal{C}$, respectively.

    \textit{If part}:
    If there exists an LOCC map $\mathcal{S}$ defined as Equation~\eqref{eq:encoding} for any input state $\rho$,
    Equation~\eqref{eq:encoding_state_transformation} holds as a special case of Equation~\eqref{eq:encoding} in which the input state $\rho$ is a completely mixed state.

    \textit{Only if part}:
    Assume that there exists an LOCC map $\mathcal{S}$ defined as Equation~\eqref{eq:encoding_state_transformation}.
    Due to the linearity of the map $\mathcal{S}$, Equation~\eqref{eq:encoding_state_transformation} yields
    \begin{equation}
      \frac{1}{D}\sum_{l,l'=0}^{D-1}\Ket{l}\Bra{l'}\otimes\mathcal{S}\left(\Ket{l}\Bra{l'}\otimes\bigotimes_{e\in E}\Ket{\Phi_{M_e}^+}\Bra{\Phi_{M_e}^+}\right)
      =\frac{1}{D}\sum_{l,l'=0}^{D-1}\Ket{l}\Bra{l'}\otimes\ket{\tilde{\psi}_l}\bra{\tilde{\psi}_{l'}}.
    \end{equation}
    Since the set ${\left\{\Ket{l}\Bra{l'}\right\}}_{l,l'}$ of operators on the system $\mathcal{H}^R$ is linearly independent, it holds that
    \begin{equation}
        \mathcal{S}\left(\Ket{l}\Bra{l'}\otimes\bigotimes_{e\in E}\Ket{\Phi_{M_e}^+}\Bra{\Phi_{M_e}^+}\right) = \ket{\tilde{\psi}_l}\bra{\tilde{\psi}_{l'}},
    \end{equation}
    for each $l,l'\in\left\{0,\ldots,D-1\right\}$.
    Therefore, writing any operators $\rho\in\mathcal{D}\left(\mathcal{H}\right)$ and $U\rho U^\dag\in\mathcal{D}\left(\tilde{\mathcal{H}}\right)$ as
    \begin{equation}
        \rho=\sum_{l,l'=0}^{D-1}c_{l,l'}\Ket{l}\Bra{l'},\quad U\rho U^\dag=\sum_{l,l'=0}^{D-1}c_{l,l'}\ket{\tilde{\psi}_l}\bra{\tilde{\psi}_{l'}}
    \end{equation}
    yields Equation~\eqref{eq:encoding}
    \begin{equation}
      \begin{split}
        &\mathcal{S}\left(\rho\otimes\bigotimes_{e\in E}\Ket{\Phi_{M_e}^+}\Bra{\Phi_{M_e}^+}\right)\\
        &=\sum_{l,l'=0}^{D-1}c_{l,l'}\mathcal{S}\left(\Ket{l}\Bra{l'}\otimes\bigotimes_{e\in E}\Ket{\Phi_{M_e}^+}\Bra{\Phi_{M_e}^+}\right)\\
        &=\sum_{l,l'=0}^{D-1}c_{l,l'}\ket{\tilde{\psi}_l}\bra{\tilde{\psi}_{l'}}\\
        &=U\rho U^\dag.
      \end{split}
    \end{equation}
\end{proof}

\chapter{\label{sec:one_shot_entropies}Min- and max-entropies}

This appendix summarizes entropic functions used in analyses of one-shot quantum state merging, such as min- and max-entropies~\cite{R2,T5,T11}.

Given any quantum state $\psi^{AB}\in\mathcal{D}\left(\mathcal{H}^A\otimes\mathcal{H}^B\right)$,
the conditional min-entropy $H_{\min}$ and the conditional max-entropy $H_{\max}$ of $A$ conditioned by $B$ are defined as
\begin{align}
  &{H_{\min}\left(A|B\right)}_\psi\coloneq\max_{\sigma^B\in\mathcal{D}\left(\mathcal{H}^B\right)}\sup\left\{\lambda\in\mathbb{R}:\psi^{AB}\leqq\frac{\mathbb{1}^A\otimes\sigma^B}{2^\lambda}\right\},\\
  &{H_{\max}\left(A|B\right)}_\psi\coloneq\max_{\sigma^B\in\mathcal{D}\left(\mathcal{H}^B\right)}\log_2 {\left\|\sqrt{\psi^{AB}}\sqrt{\mathbb{1}^A\otimes\sigma^B}\right\|}_1^2.
\end{align}
These entropies are defined so that the duality is satisfied; that is, for any pure state $\Ket{\psi}^{RAB}$, it holds that
\begin{equation}
  -{H_{\min}\left(A|R\right)}_\psi={H_{\max}\left(A|B\right)}_\psi.
\end{equation}
The definition of min- and max-entropy of $A$ is obtained by considering $\dim\mathcal{H}^B=1$ in the above definition of the conditional min- and max-entropies, that is,
\begin{align}
  &{H_{\min}\left(A\right)}_\psi\coloneq\sup\left\{\lambda\in\mathbb{R}:\psi^{A}\leqq\frac{\mathbb{1}^A}{2^\lambda}\right\}=\log_2\frac{1}{\lambda_0},\\
  &{H_{\max}\left(A\right)}_\psi\coloneq\log_2{\left\|\sqrt{\psi^{A}}\right\|}_1^2=2\log_2\tr\sqrt{\psi^A},
\end{align}
where $\lambda_0$ is the largest eigenvalue of $\psi^A$.

The smoothed versions of these entropies are defined using optimization over states that are sufficiently close to the given state, and this technique is called smoothing.
In the following, the set of sub-normalized operators on a Hilbert space $\mathcal{H}^A$ is denoted by
\begin{equation}
  \mathcal{D}_{\leqq}\left(\mathcal{H}^A\right)\coloneq\left\{\psi^A\in\mathcal{B}\left(\mathcal{H}^A\right):\psi^A\geqq 0, \tr\psi^A\leqq 1\right\}.
\end{equation}
Given any state $\psi^{AB}\in\mathcal{D}\left(\mathcal{H}^A\otimes\mathcal{H}^B\right)$ and any error threshold $\epsilon\in\left[0,\tr\sqrt{\psi^{AB}}\right]$ for smoothing,
define the $\epsilon$-ball of states around $\psi^{AB}$ as
\begin{equation}
  \mathcal{B}^\epsilon\left(\psi^{AB}\right)\coloneq\left\{\sigma^{AB}\in\mathcal{D}_{\leqq}\left(\mathcal{H}^A\otimes\mathcal{H}^B\right):P\left(\psi^{AB},\sigma^{AB}\right)\leqq\epsilon\right\},
\end{equation}
where $P\left(\psi^{AB},\sigma^{AB}\right)$ is the purified distance between sub-normalized states defined as Equation~\eqref{eq:purified_distance_subnormalized}.
The $\epsilon$-smooth conditional min-entropy $H_{\min}^\epsilon$ and the $\epsilon$-smooth conditional max-entropy $H_{\max}^\epsilon$ of $A$ conditioned by $B$ are defined as
\begin{align}
  &{H_{\min}^\epsilon\left(A|B\right)}_\psi\coloneq\max_{\tilde{\psi}^{AB}\in\mathcal{B}^\epsilon\left(\psi^{AB}\right)}{H_{\min}\left(A|B\right)}_{\tilde{\psi}}\,,\\
  &{H_{\max}^\epsilon\left(A|B\right)}_\psi\coloneq\min_{\tilde{\psi}^{AB}\in\mathcal{B}^\epsilon\left(\psi^{AB}\right)}{H_{\max}\left(A|B\right)}_{\tilde{\psi}}.
\end{align}
Note that the optimal states in the smoothing of these definitions are not necessarily normalized.
The definition of the $\epsilon$-smooth min- and max-entropy of $A$ is also obtained by considering $\dim\mathcal{H}^B=1$ in the above definition of the $\epsilon$-smooth conditional min- and max-entropies.
These smoothed entropies converge to the non-smoothed ones as $\epsilon\to 0$.

\bibliography{citation_bibtex}

\end{document}